\documentclass[11pt]{article}
\usepackage[margin=1in]{geometry}
\usepackage[latin2]{inputenc}
\usepackage[english]{babel}
\usepackage[cmex10]{amsmath}
\usepackage{amsthm}
\usepackage{amssymb}
\usepackage{cite}
\usepackage{tabularx}
\usepackage{todonotes}
\usepackage{comment}
\usepackage{caption}
\usepackage{subcaption}
\usepackage{letltxmacro}

\newtheorem{theorem}{Theorem}[section]
\newtheorem{lemma}[theorem]{Lemma}
\newtheorem{claim}[theorem]{Claim}
\newtheorem{corollary}[theorem]{Corollary}

\newtheorem{observation}[theorem]{Observation}
\theoremstyle{definition}
\newtheorem{defin}[theorem]{Definition}

\def\cqedsymbol{\ifmmode$\lrcorner$\else{\unskip\nobreak\hfil
\penalty50\hskip1em\null\nobreak\hfil$\lrcorner$
\parfillskip=0pt\finalhyphendemerits=0\endgraf}\fi} 

\newcommand{\cqed}{\renewcommand{\qed}{\cqedsymbol}}

\graphicspath{{}}

\newcommand{\executeiffilenewer}[3]{%
\ifnum\pdfstrcmp{\pdffilemoddate{#1}}%
{\pdffilemoddate{#2}}>0%
{\immediate\write18{#3}}\fi%
} 

\newcommand{\svg}[2]{\def\svgwidth{#1}%
\executeiffilenewer{#2.svg}{#2.pdf}%
{inkscape -z -D --file=#2.svg %
--export-pdf=#2.pdf --export-latex}%
{\input{#2.pdf_tex}}}

\DeclareMathAlphabet{\mathcal}{OMS}{cmsy}{m}{n}

\newcommand{\probshort}{\textsc{$k$-DPP}}

\newcommand{\group}{\Lambda}
\newcommand{\vnoose}{\vec{\nu}}
\newcommand{\snoose}{\vec{\mu}}
\newcommand{\alt}{\mathfrak{a}}
\newcommand{\Z}{\mathbb{Z}}

\newcommand{\Ff}{\mathcal{F}}
\newcommand{\Pp}{\mathcal{P}}
\newcommand{\Qq}{\mathcal{Q}}
\newcommand{\Rr}{\mathcal{R}}
\newcommand{\Aa}{\mathcal{A}}
\newcommand{\Cc}{\mathcal{C}}
\newcommand{\Ge}{\overline{G}}
\newcommand{\De}{\overline{D}}
\newcommand{\We}{\overline{W}}
\newcommand{\Pe}{\overline{P}}
\newcommand{\Ppe}{\overline{\mathcal{P}}}
\newcommand{\Gf}{\widehat{G}}

\newcommand{\Nf}{\widehat{N}}
\newcommand{\Ppf}{\widehat{\mathcal{P}}}
\newcommand{\Qqf}{\widehat{\mathcal{Q}}}
\newcommand{\Rrf}{\widehat{\mathcal{R}}}
\newcommand{\bwholes}{\mathfrak{p}}
\newcommand{\compgraph}{\ensuremath{G_{\textrm{comp}}}}
\newcommand{\bundlegraph}{\ensuremath{G_{\textrm{bundle}}}}
\newcommand{\barcs}{\ensuremath{\hat{E}}}
\newcommand{\bundleset}{\ensuremath{\mathbb{B}}}
\newcommand{\bundleword}{\mathtt{bw}}
\newcommand{\bundleprof}{\mathtt{bp}}
\newcommand{\refcurve}{\gamma_\mathrm{ref}}
\newcommand{\ringcl}{\mathrm{cl}}

\newcommand{\anycomp}{\ensuremath{H}}
\newcommand{\disccomp}{\ensuremath{\anycomp^{\textrm{disc}}}}
\newcommand{\ringcomp}{\ensuremath{\anycomp^{\textrm{ring}}}}
\newcommand{\decomp}{\ensuremath{\mathcal{D}}}
\newcommand{\incarcs}[1]{\ensuremath{\hat{E}(#1)}}

\newwrite\tempfile
\immediate\openout\tempfile=references.txt
\newcommand{\writeref}[1]{\immediate\write\tempfile{\unexpanded{#1}}}
\newcommand{\writerefe}[1]{\immediate\write\tempfile{\expandafter{#1}}}
\LetLtxMacro{\oldref}{\ref}
\LetLtxMacro{\oldsection}{\section}
\renewcommand{\ref}[1]{\oldref{#1}\writeref{\oldref{#1} (#1)}\writeref{}}

\title{The planar directed $k$-Vertex-Disjoint Paths problem is fixed-parameter tractable}

\author{
  Marek Cygan
  \thanks{
    Institute of Informatics, University of Warsaw, Poland,
      \texttt{cygan@mimuw.edu.pl}.
  }
  \and
  D\'aniel Marx
  \thanks{
    Computer and Automation Research Institute, Hungarian Academy of Sciences (MTA SZTAKI), Hungary,
      \texttt{dmarx@cs.bme.hu}.
  }
  \and
  Marcin Pilipczuk
  \thanks{
    Institute of Informatics, University of Warsaw, Poland,
      \texttt{malcin@mimuw.edu.pl}.
  }
  \and
  Micha\l{} Pilipczuk
  \thanks{
    Department of Informatics, University of Bergen, Norway, \texttt{michal.pilipczuk@ii.uib.no}.
  }
  }

\date{}

\begin{document}

\maketitle

\begin{abstract}
  Given a graph $G$ and $k$ pairs of vertices $(s_1,t_1)$, $\dots$,
  $(s_k,t_k)$, the $k$-Vertex-Disjoint Paths problem asks for pairwise
  vertex-disjoint paths $P_1$, $\dots$, $P_k$ such that $P_i$ goes
  from $s_i$ to $t_i$. Schrijver \cite{schrijver:xp} proved that the
  $k$-Vertex-Disjoint Paths problem on planar directed graphs can be
  solved in time $n^{O(k)}$. We give an algorithm with running time
  $2^{2^{O(k^2)}}\cdot n^{O(1)}$ for the problem, that is, we show the
  fixed-parameter tractability of the problem.
\end{abstract}
\clearpage
\tableofcontents
\clearpage
\renewcommand{\section}[1]{\oldsection{#1}%
\writeref{\bigskip}\writeref{}\writeref{\textbf{Section}}%
\writeref{\textbf}\writerefe{\thesection:}%
\writeref{\textbf{#1}\bigskip}%
\writeref{}}

\section{Introduction}

A classical problem of combinatorial optimization is finding disjoint paths with specified endpoints:

\begin{center}
\fbox{\parbox{0.7\linewidth}{
$k$-Vertex-Disjoint Paths Problem (\probshort)

\begin{tabularx}{\linewidth}{rX}
\textbf{Input:} &A graph $G$ and $k$ pairs of vertices $(s_1,t_1)$, $\dots$, $(s_k,t_k)$.\\
\textbf{Question:} & Do there exist $k$ pairwise vertex-disjoint paths $P_1$, $\dots$, $P_k$ such that $P_i$ goes from $s_i$ to $t_i$?
\end{tabularx}}}
\end{center}

We consider only the vertex-disjoint version of the problem in this
paper; disjoint means vertex disjoint if we do not specify otherwise.
If the number $k$ of paths is part of the input, then the problem is
NP-hard even on undirected planar graphs \cite{MR94m:05114}. However,
for every fixed $k$, Robertson and Seymour showed that there is a
cubic-time algorithm for the problem in general undirected graphs
\cite{MR97b:05088}. Their proof uses the structure theory of graphs
excluding a fixed minor and therefore extremely complicated. More
recently, a significantly simpler, but still very complex algorithm
was announced by Kawarabayashi and Wollan
\cite{DBLP:conf/stoc/KawarabayashiW10}.  Obtaining polynomial running
time for fixed $k$ is significantly simpler in the special case of
planar graphs \cite{DBLP:journals/jct/RobertsonS88}; see also the
self-contained presentations of Reed et
al.~\cite{DBLP:conf/gst/ReedRSS91} or Adler et al.~\cite{isolde}.

The problem  becomes dramatically harder for directed graphs: it is
NP-hard even for $k=2$ in general directed graphs
\cite{MR81e:68079}. Therefore, we cannot expect an analogue of the
undirected result of Robertson and Seymour \cite{MR97b:05088} saying
that the problem is polynomial-time solvable for fixed $k$. For
directed planar graphs, however, Schrijver gave an algorithm with
polynomial running time for fixed $k$:

\begin{theorem}[Schrijver \cite{schrijver:xp}]
The $k$-Vertex-Disjoint Paths Problem on directed planar graphs can be solved in time $n^{O(k)}$. 
\end{theorem}

The algorithm of Schrijver is based on enumerating all possible
homology types of the solution and checking in polynomial time whether
there is a solution for a fixed type. Therefore, the running time is
mainly dominated by the number $n^{O(k)}$ of homology types.  Our main
result is improving the running time by removing $k$ from the exponent
of $n$:

\begin{theorem}\label{th:main}
The $k$-Vertex-Disjoint Paths Problem on directed planar graphs can be solved in time $2^{2^{O(k^2)}}\cdot n^{O(1)}$.
\end{theorem}
In other words, we show that the $k$-Disjoint Paths Problem is
fixed-parameter tractable on directed planar graphs. The
fixed-parameter tractability of this problem was asked as an open
question by Bodlaender, Fellows, and Hallett \cite{DBLP:conf/stoc/BodlaenderFH94} already in 1994, in one
of the earliest papers on parameterized complexity. The question was
reiterated in the open problem list of the classical monograph of
Downey and Fellows \cite{MR2001b:68042} in 1999. Note that, for
undirected planar graphs, the algorithm with best dependence on $k$ is
due to Adler et al.~\cite{isolde} and has running time
$2^{2^{O(k)}}\cdot n^{O(1)}$. Therefore, for the more general directed
version of the problem, we cannot expect at this point a running time
with better than double-exponential dependence on $k$.

For general undirected graphs, the algorithm of Robertson and
Seymour~\cite{MR97b:05088} relies heavily on the structure theory of
graphs excluding a fixed minor; in fact, this algorithm is one of the
core achievements of the Graph Minors series.  More recent results on
finding subdivisions \cite{grohe-stoc2011-topminor} or
parity-constrained disjoint paths
\cite{DBLP:conf/focs/KawarabayashiRW11} also build on this framework.
Even in the much simpler planar case, the algorithm presented by Adler
et al.~\cite{isolde} uses the concepts and tools developed in the
study of excluded minors. In a nutshell, their algorithm has three
main components. First, if treewidth (a measure that plays a crucial
role in graph structure theory) is bounded, then standard algorithmic
techniques can be used to solve the $k$-Vertex-Disjoint Paths
Problem. Second, if treewidth is large, then (planar version of) the
Excluded Grid Theorem \cite{DBLP:journals/jct/DiestelJGT99,DBLP:journals/jct/RobertsonST94,DBLP:journals/jct/RobertsonS86,DBLP:journals/algorithmica/GuT12} implies that the graph contains a subdivision of
a large wall, which further implies that there is a vertex enclosed by
a large number of disjoint concentric cycles, none of them enclosing
any terminals. Finally, Adler~\cite{isolde} et al.~show that such a
vertex is irrelevant, in the sense that it can be removed without
changing the answer to the problem. Thus by iteratively removing such
irrelevant vertices, one eventually arrives to a graph of bounded
treewidth.

Can we apply a similar deep and powerful theory in the directed version
of the problem? There is a notion of directed
treewidth~\cite{dirgrids} and an excluded grid theorem holds at least
for planar graphs \cite{dirgrids-planar} (and more generally, for
directed graphs whose underlying undirected graph excludes a fixed
minor \cite{dirgrids-excluded-minors}). However, the other two
algorithmic components are missing: it is not known how to solve the
$k$-Vertex Disjoint Paths problem in $f(k)\cdot n^{O(1)}$ time on directed
graphs having bounded directed treewidth and the directed grids
excluded by these theorems do not seem to be suitable for excluding
irrelevant vertices. There are other notions that try to generalize
treewidth to directed graphs, but the algorithmic applications are
typically quite limited \cite{DBLP:journals/jct/BerwangerDHKO12,DBLP:conf/iwpec/GanianHKMORS10,DBLP:conf/iwpec/GanianHKLOR09,DBLP:conf/wg/MeisterTV07,DBLP:conf/mfcs/Safari05,DBLP:journals/tcs/KreutzerO11}. In
particular, the $k$-Vertex-Disjoint Paths Problem is known to be
W[1]-hard on directed acyclic graphs
\cite{DBLP:journals/siamdm/Slivkins10}, which is strong evidence that
any directed graph measure that is small on acyclic graphs is not
likely to be of help.

Our algorithm does not use any tool from the structure theory of
undirected graphs, or any notion of treewidth for directed graphs.
The only previous results that we use are the results of Ding,
Schrijver, and Seymour
\cite{dss:green-line,DBLP:journals/siamdm/DingSS92} on various special
cases of the directed disjoint paths problem, the cohomology
feasibility algorithm of Schrijver~\cite{schrijver:xp}, and a
self-contained combinatorial argument from Adler et al.~\cite{isolde}.
Therefore, we have to develop our own tools and in particular a new
type of decomposition suitable for the problem. A concept that appears
over and over again in this paper is the notion of alternation: we are
dealing with sequences of paths and cycles having alternating
orientation (i.e., each one has an orientation that is the opposite of
the next one), we measure the ``width'' of a sequence of arcs by the
number of alternations in the sequence, and we measure ``distance''
between faces by the minimum alternation on any sequence of arcs
between them. Section~\ref{sec:overview} gives a high-level overview
of the algorithm. Let us highlight here the most important steps and
the main ideas:
\begin{itemize}
\item  \textbf{Irrelevant vertices.} Analogously to Adler et
  al.~\cite{isolde}, we prove that a vertex enclosed by a large set of
  concentric cycles having {\em alternating orientation} and not
  enclosing any terminals is irrelevant. As expected, the proof is
  more complicated and technical than in the undirected case
  (Section~\ref{sec:irr}).
\item \textbf{Duality of alternation.} We show that alternation has
  properties that are similar to the classical properties of undirected
  planar graphs (Section~\ref{sec:spiral}).  We prove an approximate duality between
  alternating paths and the minimum alternation size of a cut
  (reminiscent of max-flow min-cut duality), and between concentric cycles and
  alternation distance (reminiscent of the fact that two faces far
  away in a planar graph are separated by many disjoint cycles).
\item \textbf{Decomposition.} We present a novel kind of decomposition
  into ``disc'' and ``ring'' components (Section~\ref{sec:decomp}). The
  crucial property of the decomposition is that the set of arcs
  leaving a component has bounded alternation. That is, the components
  are connected by a bounded number of {\em bundles,} each containing
  a set of ``parallel'' arcs with the same orientation.
\item \textbf{Handling ring components.} Ring components pose a
  particular challenge: we have to understand how many turns a path
  of the solution does when connecting the inside and the outside. We
  prove a rerouting argument showing that only a bounded number of
  possibilities has to be taken into account for the winding numbers of
  these paths.
\item \textbf{Guessing bundle words.} Given a decomposition, a path of
  the solution can be described by a word consisting of a sequence of
  symbols representing the bundles visited by path, in the order they
  appear in the path. Note that a bundle can be used several times by
  a path of the solution, thus the word can be very long. Our goal is
  to enumerate a bounded number of possible bundle words for each path
  of the solution. These words, together with our understanding of
  what is going on inside the rings, allow us to guess the homology
  type of the solution, and then invoke Schrijver's cohomology
  feasibility algorithm to check if there is a solution with this
  homology type.
\end{itemize}
In the next section we present an informal overview of the algorithm;
all formal arguments follow in the rest of the paper.

The techniques introduced in this paper were developed specifically
with the $k$-Vertex-Disjoint Paths Problem in mind. It is likely that
some of the duality arguments or decomposition techniques can have
applications for other problems involving planar directed graphs.

In general directed graphs, vertex-disjoint and edge-disjoint versions
of the disjoint paths problems are equivalent: one can reduce the
problems to each other by simple local transformations (e.g.,
splitting a vertex into an in-vertex and an out-vertex). However, such
local transformations do not preserve planarity. Therefore, our result
has no implications for the edge-disjoint version of the problem on
planar directed graphs. Let us note that in planar graphs the
edge-disjoint version seems very different from the vertex-disjoint
version: as the paths can cross at vertices, the solution does not
have a topological structure of the type that is exploited by both
Schrijver's algorithm \cite{schrijver:xp} and our algorithm. The
complexity of the planar edge-disjoint version for fixed $k$ remains
an open problem; it is possible that, similarly to general graphs
\cite{MR81e:68079}, it is NP-hard even for $k=2$.

One can define a variant of the planar edge-disjoint problem where
crossings are not allowed. That is, in the noncrossing edge-disjoint
version paths are allowed to share vertices, but if edge $e_1$
entering $v$ is followed by $e_2$, and edge $f_1$ entering $v$ is
followed by $f_2$, then the cyclic order of these edges cannot be
$(e_1,f_1,e_2,f_2)$ or $(e_1,f_2,e_2,f_1)$ around $v$. It is easy to
see that this version can be reduced (in a planarity-preserving way)
to the vertex disjoint version by replacing each vertex by a large
bidirected grid. Therefore, our algorithm can solve the noncrossing
edge-disjoint version of the $k$-Disjoint Paths Problem as well.


\section{Overview of the algorithm}\label{sec:overview}

The goal of this section is to give an informal overview of our main result --- the fixed-parameter
algorithm for finding $k$ disjoint paths in directed planar graphs.

\subsection{Irrelevant vertex rule}\label{ss:view:irr}

Let us first recall how to solve the $k$-disjoint paths problem in the undirected (even non-planar) case.
The algorithm of Robertson and Seymour \cite{MR97b:05088} considers two cases. If the treewidth of
the input graph $G$ is bounded by a function of the parameter ($k$, the number of terminal pairs),
then the problem can be solved by a standard dynamic programming techniques on a tree decomposition of small width
of $G$. Otherwise, by the Excluded Grid Theorem \cite{DBLP:journals/jct/DiestelJGT99,DBLP:journals/jct/RobertsonST94,DBLP:journals/jct/RobertsonS86,DBLP:journals/algorithmica/GuT12}, $G$ contains a large grid as a minor. 

The idea now is to distinguish a vertex $v$ of $G$, whose deletion does not change the answer of the problem;
that is, there exist the required $k$ disjoint paths in $G$ if and only if they exist in $G \setminus v$.
Note that the disjoint paths problem can become only harder if we delete a vertex; thus, to pronounce $v$ irrelevant,
one needs to prove that any solution using the vertex $v$ can be redirected to a similar one, omitting $v$.

In the case of planar graphs one may apply the following quite intuitive reasoning. Assume that $G$ contains a large grid as a minor;
as there are at most $2k$ terminals, a large part of this grid does not enclose any terminal.
In such a part, a vertex $v$ hidden deep inside the grid seems irrelevant: any solution using $v$ needs to traverse a large part of the
grid to actually contain $v$, and it should be possible to ``shift'' the paths a little bit to omit $v$.
This reasoning can be made formal, and Adler et al. \cite{isolde} proved that, in undirected planar graphs,
the middle vertex of a grid of exponential (in $k$) size is irrelevant. In fact, they show a bit stronger statement:
if we have sufficiently many (around $2^k$) concentric cycles on the plane, such that the outermost cycle does not enclose any terminal,
then any vertex on the innermost cycle is irrelevant.

One of the main argument in the proof of Adler et al. \cite{isolde} is as follows. Assume that there are many pairwise disjoint 
segments of the solution that cross sufficiently many orthogonal paths (henceforth called {\em{chords}}) in the graph; see Figure \ref{fig:view:isolde}.
Assume moreover that the aforementioned segments are the only parts of the solution that appear in the area enclosed by the outermost segments and chords (i.e., in the part of the plane depicted on Figure \ref{fig:view:isolde}). Then, if the number of segments is more than $2^k$, one can
redirect some of them, using the chords, and shortening the solution. Thus, in a minimal (in some carefully chosen sense) solution, a set
of more than $2^k$ paths cannot go together for a longer period of time.  

\begin{figure}
\begin{center}
\includegraphics{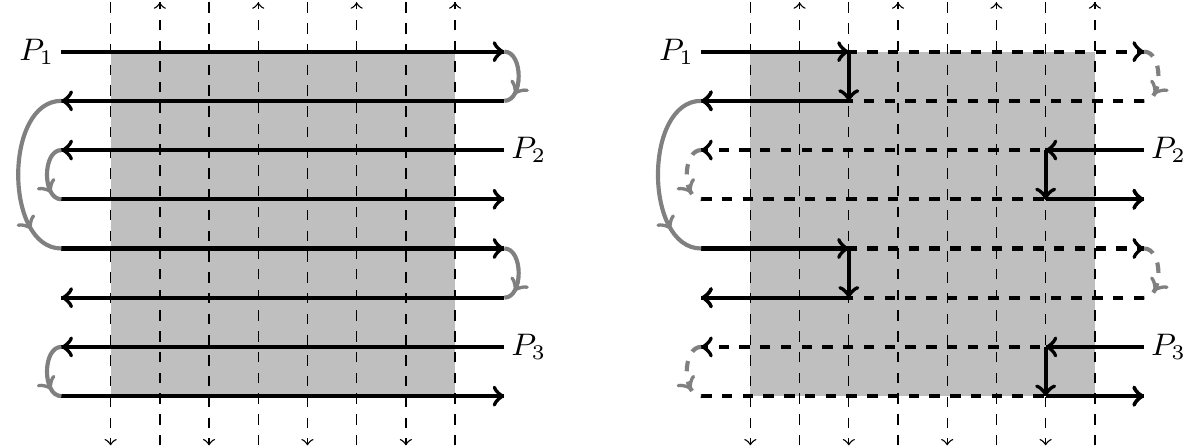}
\caption{A situation where a shortcut can be made and how it can be made. There are more than $2^k$ horizontal segments of the paths,
  crossed by sufficiently many vertical chords, that, in the directed setting, are required to be of alternating orientation.
  Moreover, it is assumed that no other path of any path intersects the gray area, so that the paths remain pairwise disjoint after rerouting.}
  \label{fig:view:isolde}
  \end{center}
  \end{figure}

  The argument of Adler et al.~\cite{isolde} described in the previous
  paragraph redirects the paths of the solution using the chords in an
  undirected way, and hence the direction in which a chord is used is
  unpredictable, depending on the order in the which the segments appear on the
  paths of the solution.  Hence, if we want to transfer this argument
  to the directed setting, then we need to make some assumption on the
  direction of the chords. It turns out that what we need is that the
  chords are directed paths with alternating orientation. This ensures
  that we always have a chord going in the right direction at any
  place we would possibly need it.

  If a set of paths intersect the innermost cycle, then they need to
  traverse all cycles. Adler et~al.~\cite{isolde} show how to find a
  subset of these paths and how to cut out chords from the cycles in a
  way that satisfies the conditions of the rerouting argument.  In the
  directed setting, in order to obtain chords of alternating
  orientation, we need to assume that the cycles have alternating
  orientation too.
\begin{defin}
We say that cycles $C_1$, $\dots$, $C_d$ form a {\em sequence of concentric cycles with alternating orientation} in a plane graph $G$ if
\begin{enumerate}
\item they are pairwise vertex disjoint,
\item for every $1\le i <d$, cycle $C_i$ encloses $C_{i+1}$, and
\item for every $1\le i <d$, exactly one of the cycles $C_i$ and $C_{i+1}$ is oriented clockwise.
\end{enumerate}
\end{defin}
Luckily, it turns out that such a sequence of cycles is sufficient for
the irrelevant vertex rule. In Section \ref{sec:irr} we prove the
following.
\begin{theorem}[Irrelevant vertex rule]\label{thm:view:irr}
For any integer $k$, there exists $d = d(k) = 2^{O(k^2)}$ such that the following holds.
Let $G$ be an instance of \probshort{} and let $C_1, C_2, C_3,\ldots,C_d$ bet a sequence
of concentric cycles in $G$ with alternating orientation, where $C_1$ is the outermost cycle.
Assume moreover that $C_1$ does not enclose any terminal. Then any vertex of $C_d$ is irrelevant.
\end{theorem}
At the heart of the proof of Theorem \ref{thm:view:irr} lies the
rerouting argument described above, which states that a solution can
be rerouted and shortened if a set of more than $2^k$ paths travel
together through sufficiently many (exponential in $k$) chords cut out
from the alternating cycles $C_i$. However, it is much harder to prove
the existence of these paths and chords needed for the rerouting
argument than in the undirected case, and we now sketch how it could
be done.

\begin{figure}
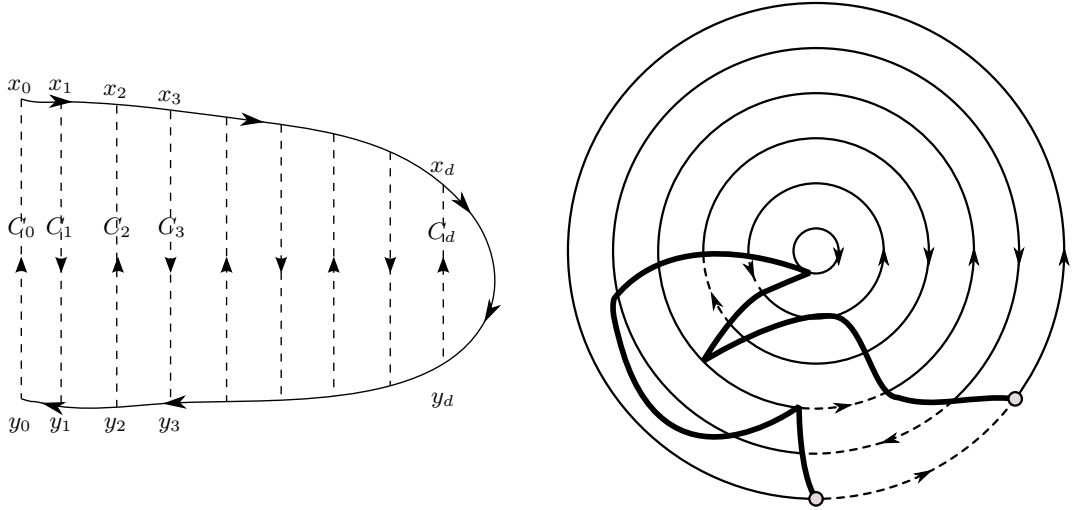

\centering

\begin{subfigure}{.45\textwidth}
{\footnotesize 
  \centering
  \svg{.9\linewidth}{bend}
}
\end{subfigure}
\begin{subfigure}{.45\textwidth}
{\footnotesize 
  \centering
  \svg{.9\linewidth}{concentricbend2}
}
\end{subfigure}
\caption{A $d$-bend $B$ with chords $C_0, C_1, \ldots, C_d$, and how it can be cut out from concentric cycles, using parts of the cycles as chords.}\label{fig:view:bend}
\end{figure}

Consider the situation assumed in Theorem \ref{thm:view:irr} and
assume we have a solution where one path, say $P$, intersects the
innermost cycle. On one side of $P$ we obtain a structure we call a
{\em{bend}}, depicted on Figure \ref{fig:view:bend}.  The parts of the
cycles are called {\em{chords}}, a bend with $d$ chords is a
$d$-{\em{bend}}.  Moreover, {\em{the type of the bend}} is the number
of different paths from the solution that intersect the interior or
the boundary of the bend; our initial bend is of type at most $k$. Our
main technical claim in the proof of Theorem \ref{thm:view:irr} is
that in a (somehow defined) minimal solution there do not exist
$d$-bends of type $t$, for $d > f(k,t)$ and some function $f(k,t) =
2^{O(kt)}$, that do not enclose any terminals.

\begin{figure}
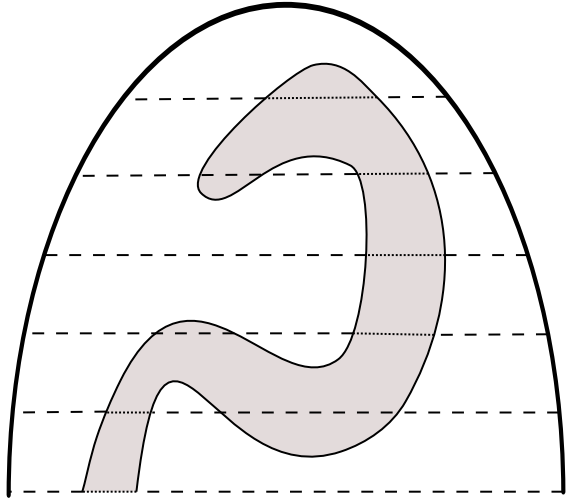

{\footnotesize
\begin{center}
\svg{0.45\linewidth}{bendinbend3}
\end{center}}
\caption{Part of a path creates a bend inside another bend.}\label{fig:view:bend-in-bend}
\end{figure}

Assume we have a $d$-bend $B$ of type $t$, for some large $d$,
enclosed by a part of a path $P_{i}$ of the solution. We analyze the
{\em segments} of the solution: the maximal subpaths of the paths
$P_1$, $\dots$, $P_k$ in the interior of the bend. If the interval
vertices of the last two chords are not intersected by any segment,
then one of these two chords has the right orientation to serve as a
shortcut for the path $P_i$, contradicting the minimality of the
solution. Therefore, we can assume that all but the last two chords
are intersected by segments.  If any segment of $P_{i'}$ intersects
the $j$-th chord of $B$, then it itself induces a $j'$-bend $B'$
inside $B$, for some $j' = j-O(1)$ (see Figure
\ref{fig:view:bend-in-bend}).  Hence, if the path $P_i$ itself does
not intersect the interior of the $d$-bend $B$, any bend inside $B$ is
of strictly smaller type, and the claim is proven by induction on $t$.

Otherwise, we can argue that several segments of $P_i$ enter the
interior of the $d$-bend $B$. Our goal is to prove that there is a
large set of segments of $P_i$ entering $B$ that form a nested
sequence and they travel together through a large number of chords
deep inside the bend, with no other segment of $P_i$ between
them. Then we can argue that any other segment of some $P_{i'}$ with
$i'\neq i$ intersecting these chords is also nested with these
segments, otherwise they would create a large bend of strictly smaller
type, and induction could be applied. Therefore, we get a large set of
paths travelling together through a large number of chords, and the
rerouting argument described above can be invoked.

To find this nested sequence of segments of $P_i$, we analyze how
$P_i$ intersects the chords of $B$.  We construct the following
auxiliary graph $H$: start with a subgraph of $G$ consisting of $P$
and the chords of $B$ and suppress all vertices of degree $2$. Let
$H^*$ be the dual of $H$ and $T^*$ be subgraph of $H^*$ consisting
only of chord arcs.  It is not hard to see that $T^*$ is a tree, with
at least $d-O(1)$ vertices; see Figure \ref{fig:view:dual} for an
illustration.

  \begin{figure}
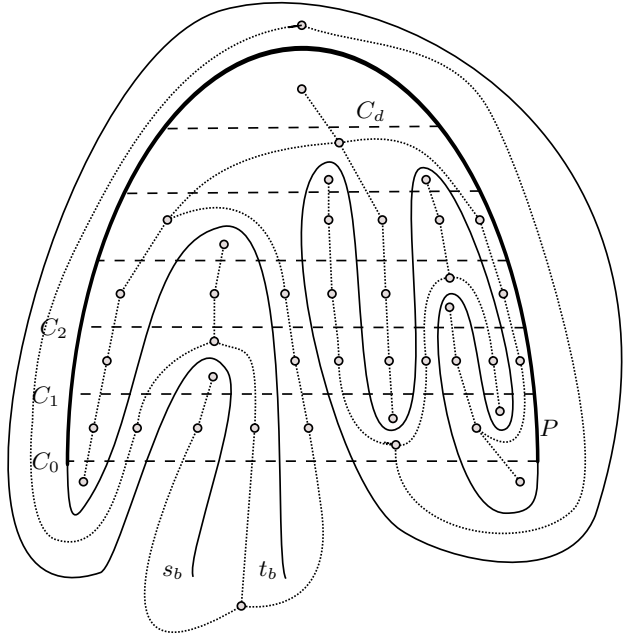

{\footnotesize 
\begin{center}
\svg{0.5\linewidth}{dual}
\end{center}
}
\caption{A $d$-bend with chords $C_0, C_1, \ldots, C_d$
appearing on a path $P_b$. The dotted lines show the edges of the tree $T^*$.}\label{fig:view:dual}
\end{figure}

Roughly speaking, if the segments of $P_i$ reaching deep inside the
bend are not nested, then we can find a face $f$, hidden deeply inside the $d$-bend $B$, such that at least three of the
connected components of $T^* \setminus f$ cross many chords. If
one of these components has the property that none of the faces
appearing in this component contains a terminal, then the part of the
path $P_i$ that encloses this component is, by the definition of
$T^*$, encloses a bend of type at most $t-1$. Thus, by the induction
hypothesis, it cannot cross more than $f(k,t-1)$ chords of the
$d$-bend $B$. However, if all such connected components of $T^* \setminus
f$ contain faces with terminals inside, we cannot argue anything about
$f$: the path $P$ may need to do travel in such strange manner in
order to go around some terminals. The crucial observation is that
there are $O(k)$ faces for which this situation can arise: as there are $2k$
terminals, there are only $O(k)$ vertices of the tree $T^*$ such that
at least three components of $T^*\setminus f$ contain faces with
terminals.  Therefore, if we avoid these $O(k)$ special faces, then we
can find the required set of nested segments and we can find a place
to apply the rerouting argument.  This finishes the sketch of the
proof of the irrelevant vertex rule (Theorem \ref{thm:view:irr}).

We would like to note that we can test in polynomial time if the
irrelevant vertex rule applies: if we guess one faces enclosed
by $C_r$ and the orientation of $C_r$, we can construct the cycles in
a greedy manner, packing the next cycle as close as possible to the
previously constructed one. However, we do not use this property in
our algorithm: the decomposition algorithm, described in the next
subsection, returns an irrelevant vertex situation if it fails to
produce a suitable decomposition.

We would also like to compare the assumptions of Theorem
\ref{thm:view:irr} with the conjectured canonical obstruction for
small directed treewidth, depicted on Figure
\ref{fig:view:dirgrid}. It has been shown that a planar graph
\cite{dirgrids-planar}, or, more generally, a graph excluding a fixed
undirected minor \cite{dirgrids-excluded-minors}, has small directed
treewidth unless it contains a large directed grid (as in Figure
\ref{fig:view:dirgrid}), in some minor-like fashion, and this
statement is conjectured to be true for general graphs
\cite{dirgrids}.  Although the assumption of bounded directed
treewidth may be easier to use than the bounded-alternation
decomposition presented in the next subsection, we do not know how to
argue about irrelevancy of some vertex or arc in the directed
grid. Thus, we need to stick with our irrelevant vertex rule with
relatively strong assumptions (a large number of alternating cycles),
and see in the rest of the proof what can be deduced if such a
situation does not occur.

\begin{figure}
\begin{center}
\includegraphics{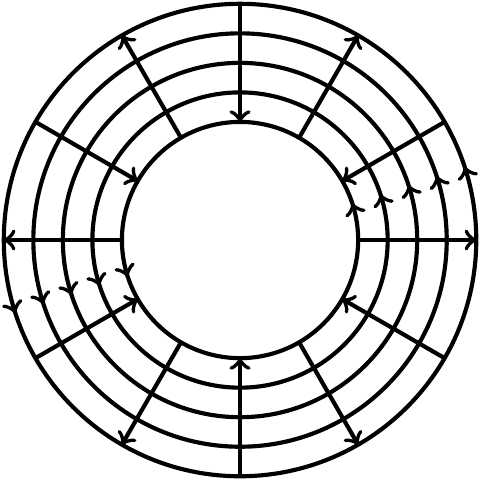}
\caption{A directed grid --- a conjectured canonical obstacle for small directed treewidth.}
  \label{fig:view:dirgrid}
  \end{center}
  \end{figure}

\subsection{Decomposition and duality theorems}\label{ss:view:decomp}

Once we have proven the irrelevant vertex rule (Theorem \ref{thm:view:irr}), we may see what can be deduced about the structure of the graph
if the irrelevant vertex rule does not apply. Recall that in the undirected case the absence of an irrelevant vertex implied a bound on the 
treewidth of the graph, and hence the problem can be solved by a standard dynamic programming algorithm.

In our case the situation is significantly different. As we shall see, the assumptions in Theorem \ref{thm:view:irr} are rather strong,
and, if the irrelevant vertex rule is not applicable, the problem does not become as easy as in the bounded-treewidth case.
Recall that Theorem \ref{thm:view:irr} assumed a large number of cycles of alternating orientation, and these alternations were
crucial for the rerouting argument.
It turns out that, if such cycles cannot be found, we can decompose the graph into relatively simple pieces using cuts of bounded alternation.

Consider a directed curve $\gamma$ on the plane that intersects the plane graph $G$ only in a finite number of points
(i.e., $\gamma$ does not ``slide'' along any arc of $G$). 
For any point $p \in \gamma \cap G$ we define $S(\gamma,p) \subseteq \{-1,+1\}$ as follows: $-1 \in S(\gamma,p)$ if it is possible for a path in $G$
to cross $\gamma$ in $p$ from left to right, and $+1 \in S(\gamma,p)$
if it possible to cross $\gamma$ from right to left (see Figure \ref{fig:view:Sp}). The {\em{alternation}} of $\gamma$ is the length of the longest
sequence of alternating $+1$ and $-1$s that is embeddable (in a natural way) into the sequence
$S(\gamma,p)_{p \in \gamma \cap G}$.

\begin{figure}
\begin{center}
\includegraphics{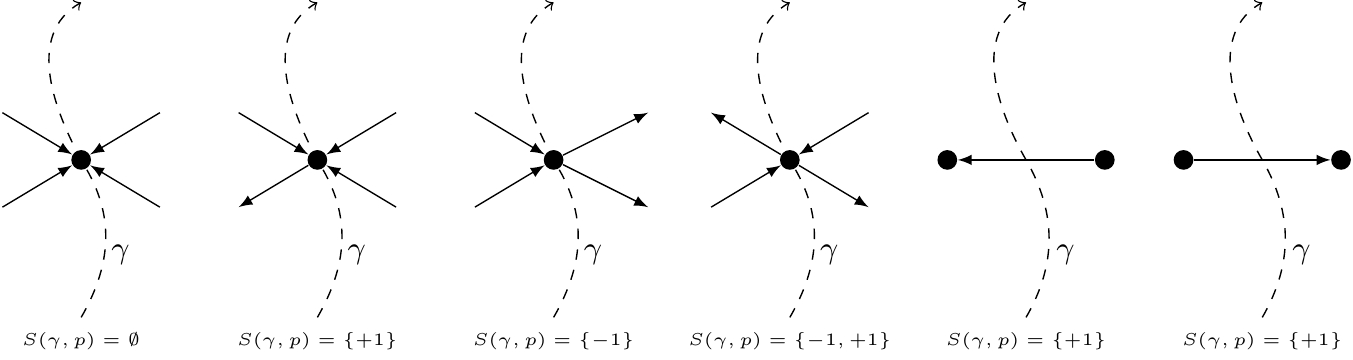}
\caption{An illustration of the definition of $S(\gamma,p)$ for $p \in \gamma \cap G$.}
  \label{fig:view:Sp}
  \end{center}
  \end{figure}

  Note that the existence of a curve $\gamma$ with alternation $\alt$
  connecting faces $f_1$ and $f_2$ proves that $f_1$ and $f_2$ cannot
  be separated by a sequence of more than $\alt$ concentric cycles of
  alternating orientation. Thus, a curve of bounded alternation is in
  some sense dual to the notion of concentric cycles of bounded
  alternation. It turns out that this duality is tight: such a curve
  of bounded alternation is the only obstacle that prevents the
  existence of these concentric cycles. One can also formulate a
  duality statement similar to the classical max-flow min-cut duality,
  with a set of paths of alternating orientation playing the role of
  the flow and a curve of bounded alternation playing the role of a
  cut. The following lemma states both types of duality in an informal
  way (see Figures \ref{fig:view:dualities1} and \ref{fig:view:dualities2} for illustration).

\begin{lemma}[Alternation dualities, informal statement.]\label{lem:view:duality}
Let $G$ be a graph embedded in a subset of a plane homeomorphic to a ring, and let $f_{in}$ and $f_{out}$ be the two faces of $G$
that contain the inside and the outside of the ring, respectively. Let $r$ be an even integer. Then, in polynomial time, one can in $G$:
\begin{enumerate}
\item either find a sequence of $r$ cycles of alternating orientation, separating $f_{in}$ from $f_{out}$, or
find a curve connecting $f_{in}$ with $f_{out}$ with alternation at most $r$ (Figure~\ref{fig:view:dualities1}); and
\item either find a sequence of $r$ paths, connecting $f_{in}$ and $f_{out}$, with alternating orientation, or
find a closed curve separating $f_{in}$ from $f_{out}$ with alternation at most $r+4$ (Figure~\ref{fig:view:dualities2}).
\end{enumerate}
\end{lemma}

\begin{figure}[t]
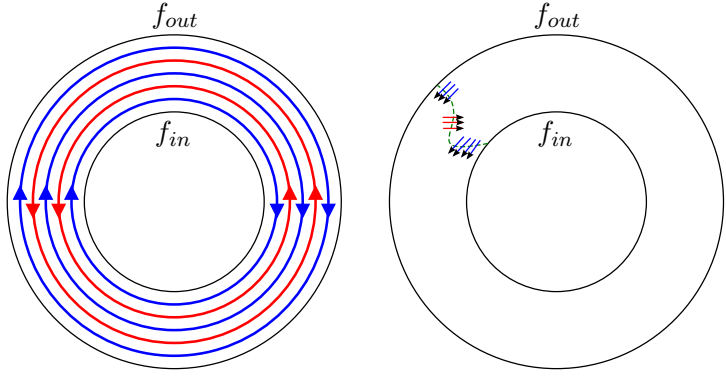

\begin{center}
\begin{subfigure}{0.3\linewidth}
\svg{0.9\linewidth}{dual1a} 
\end{subfigure}
\begin{subfigure}{0.3\linewidth}
\svg{0.9\linewidth}{dual1b} 
\end{subfigure}

\caption{Two cases in Lemma \ref{lem:view:duality}(1): cycles of alternating orientation between $f_{in}$ and $f_{out}$ or a curve of bounded alternation connecting $f_{in}$ and $f_{out}$.}
  \label{fig:view:dualities1}
  \end{center}
  \end{figure}

\begin{figure}[t]
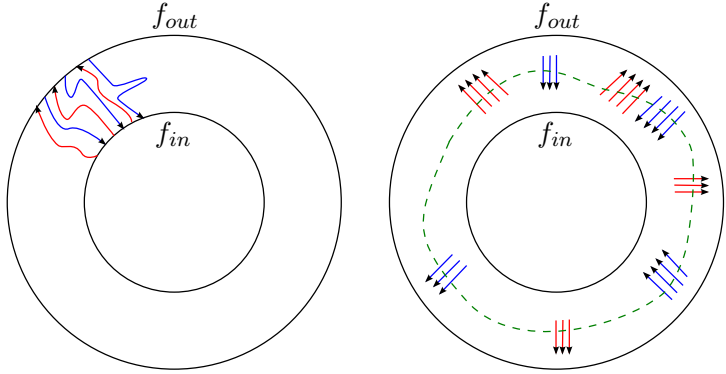

\begin{center}
\begin{subfigure}{0.3\linewidth}
\svg{0.9\linewidth}{dual2a} 
\end{subfigure}
\begin{subfigure}{0.3\linewidth}
\svg{0.9\linewidth}{dual2b} 
\end{subfigure}

\caption{Two cases in Lemma \ref{lem:view:duality}(2): paths of alternating orientation connecting $f_{in}$ and $f_{out}$ or a curve of bounded alternation separating $f_{in}$ and $f_{out}$.}
  \label{fig:view:dualities2}
  \end{center}
  \end{figure}

  Let us now give intuition on how to prove statements like Lemma
  \ref{lem:view:duality}.  If we identify $f_{in}$ and $f_{out}$, or
  more intuitively, extend the surface with a handle connecting $f_{in}$
  and $f_{out}$, we can perceive $G$ as a graph on a torus.  After
  some gadgeteering, we may use the following result of Ding, Schrijver, and Seymour~\cite{dss:green-line}: if one wants (in a graph $G$ on a torus)
  to route a set of vertex-disjoint cycles with prescribed homotopy
  class and directions, a canonical obstacle is a face-vertex curve
  $\gamma$ (of some other homotopy class), where the sequence
  $S(\gamma,p)_{p \in \gamma \cap G}$ does not contain the expected
  subpattern of $+1$ and $-1$s.  Note that such a curve is not far
  from the curves promised by Lemma~\ref{lem:view:duality}.

  Equipped with this understanding of alternation, in Section
  \ref{sec:decomp} we prove a decomposition theorem that is crucial
  for our algorithm.  We state this theorem here informally (see
  Figure \ref{fig:view:decomp} for an illustration); the precise
  statement appears in Theorem~\ref{th:finddecomp}.

\begin{figure}
\begin{center}
\includegraphics{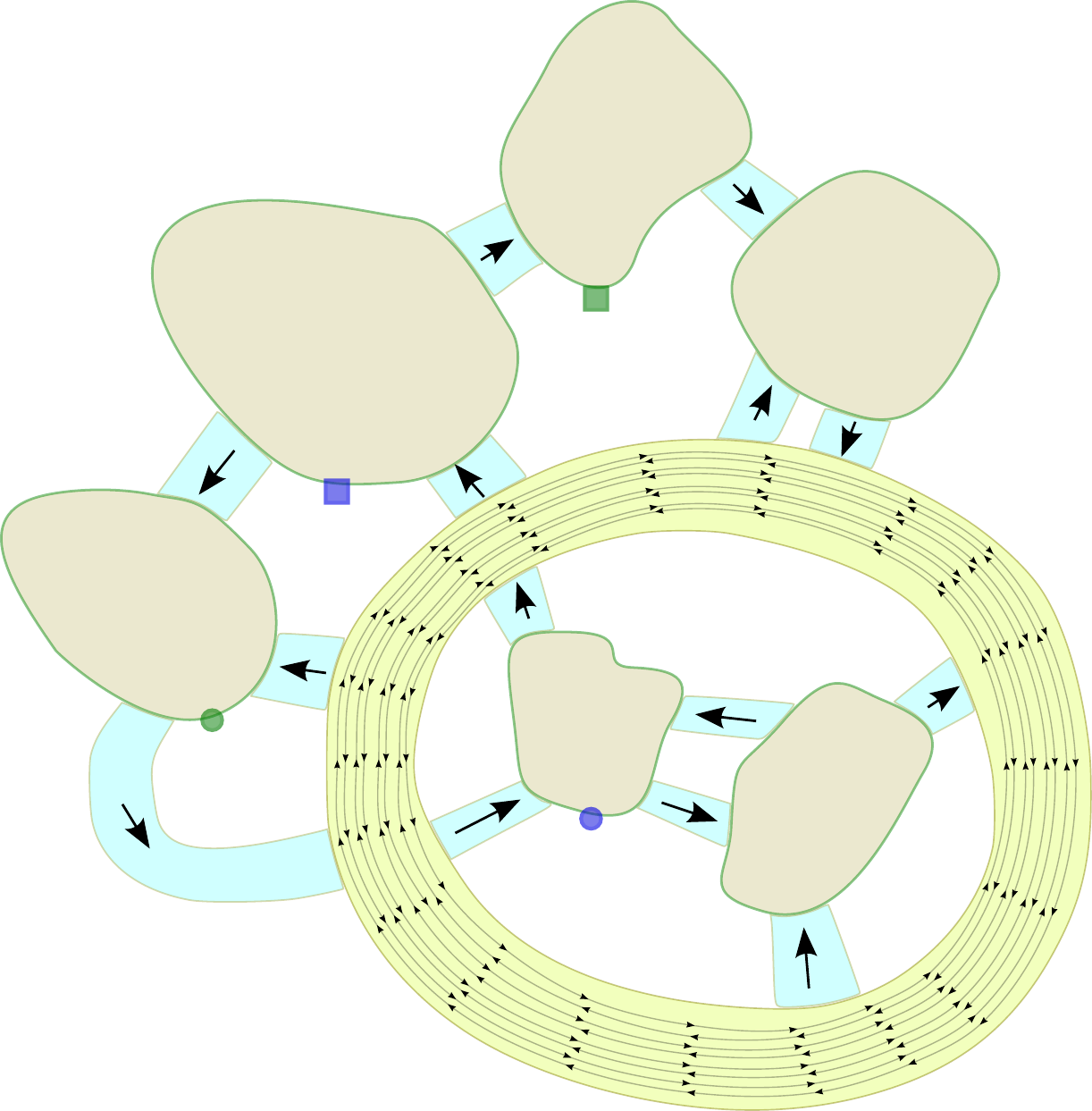}
\caption{An example of a decomposition. The disc components are red, the ring component is yellow.}
  \label{fig:view:decomp}
  \end{center}
  \end{figure}

\begin{theorem}[Decomposition theorem, informal statement]\label{thm:view:decomp}
Assume that $G$ is a plane graph with $k$ terminal pairs to which the irrelevant vertex rule is not applicable.
Then one can  partition the graph $G$
into a bounded (in $k$) number of {\em{disc}} and {\em{ring components}},
using cuts of bounded total alternation;
a disc (resp., ring) component occupies a subset of the plane that is isomorphic
to a disc (resp., ring). Moreover, each terminal lives on the border of a
disc component, and each ring component contains many concentric
cycles of alternating orientation, separating the inside from the outside.
\end{theorem}

The decomposition of Theorem~\ref{thm:view:decomp} is obtained by
iteratively refining a decomposition, moving a terminal to the
boundary of a component in each step. If a disc component contains a
terminal such that there is a curve of bounded alternation from the
terminal to the boundary of the component, then the terminal can be
moved to the boundary by removing the arcs intersected by the
curve. This operation increases the alternation of the cut separating
the component from the rest of graph only by a bounded number, thus we
can afford to perform one such step for each terminal. Otherwise, if
there is no such curve, then Lemma~\ref{lem:view:duality}(1) implies
that there is a large set of concentric cycles of alternating
orientations separating all the terminals in the component from the
boundary. We again consider two cases. If there is large set of paths
with alternating orientations crossing these cycles, then the paths
and cycles together form some kind of grid, and we can easily identify
a vertex that is separated from all the terminals by a large set of
concentric cycles with alternating orientation. Such a vertex is
irrelevant by Theorem~\ref{thm:view:irr}, and hence can be removed. On
the other hand, if there is no such set of paths, then
Lemma~\ref{lem:view:duality}(2) implies that there is a curve of
bounded alternation separating the terminals of the component from the
boundary of the component. We can use this curve to cut away a ring
component and we can do this in such a way that after removing the
ring component, one of the terminals is close to the boundary of the
remaining part of the disc component (in the sense that there is a
curve of bounded alternation connecting it to the
boundary). Therefore, we can apply the argument described above to
move this terminal to the boundary. Iteratively applying these steps
until all the terminals are on the boundary of its component produces
the required decomposition.

How can we use the decomposition of Theorem~\ref{thm:view:decomp} to
solve \probshort?  The disc components are promising to work with, as
the \probshort{} problem is polynomial if all terminals lie on the
outer face of the graph \cite{DBLP:journals/siamdm/DingSS92}.  In a
topological sense, if the terminals are on the outer face, then the
solutions are equivalent, whereas if the terminals are on the inside
and outside boundaries of a ring, then the solutions can differ in how
many ``turns'' they do along the ring. This possible difference in the
number of turns create particular challenges when we are trying to
apply the techniques of Schrijver \cite{schrijver:xp} to find a
solution based on a fixed homotopy class.

Theorem~\ref{thm:view:decomp} would be nicer and more powerful if we
could always obtain a decomposition using only disc components, but as
Figure \ref{fig:view:ringcut} indicates, this does not seem to be
possible in general.  Assume that one part of the graph (inside) is
separated from the outside by a large number of concentric cycles of
alternating orientation, but with cuts of bounded alternation between
each consecutive cycles. Suppose that there are terminals inside the
innermost cycle and outside the outermost cycle. Then the irrelevant
vertex rule of Theorem~\ref{thm:view:irr} is not applicable, as we
cannot find suitable set of cycles without enclosing some terminals
inside the innermost cycle. Moreover, if we aim for bounded
alternation cuts, we cannot cut through too many cycles.  Thus, this set of concentric cycles
need to be embedded in a separate ring component.

\begin{figure}
\begin{center}
\includegraphics{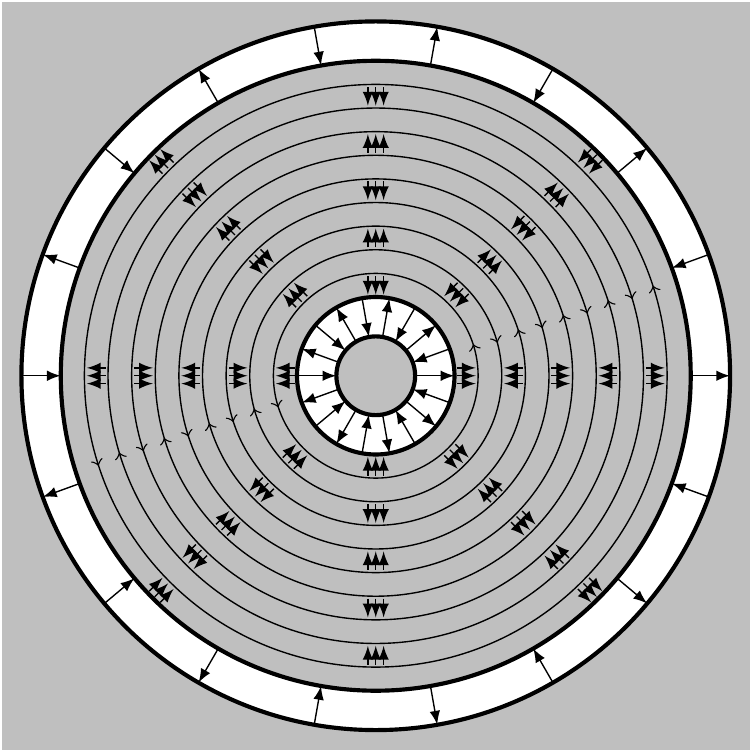}
\caption{A situation, where a ring component is necessary: there are many concentric cycles of alternating orientation,
  but the arcs between the cycles have bounded alternation.}
  \label{fig:view:ringcut}
  \end{center}
  \end{figure}

The formal statements of Lemma \ref{lem:view:duality} are proven in Section \ref{sec:spiral},
whereas the decomposition theorem is proven in Section \ref{sec:decomp}.

\subsection{Bundles and bundle words}\label{ss:view:bundles}

From the previous subsection we know that, if the irrelevant vertex rule is not applicable,
one may decompose the graph into a bounded number of disc and ring components, using cuts of bounded alternation.
Let us reformulate this statement: we can decompose the graph into a bounded number of disc and ring components,
connected by a bounded number of {\em{bundles}}; a {\em{bundle}} is a set of arcs of $G$ that form a directed path
in the dual of $G$, such that no other arc nor vertex of $G$ is drawn between the consecutive arcs of the bundle.
Thus, we obtain something we call {\em{bundled instance}} $(G,\decomp,\bundleset)$: a graph $G$ with terminals,
a family of components $\decomp$ and a family of bundles $\bundleset$. On Figure \ref{fig:view:decomp}
one can see a partition of arcs between components into bundles. With any path $P$ in a bundled instance $(G,\decomp,\bundleset)$
we can associate its {\em{bundle word}}, denoted $\bundleword(P)$: we follow the path $P$ from start to end
and append a symbol $B \in \bundleset$ whenever we traverse an arc belonging to a bundle $B$. That is,
$\bundleword(P)$ is a word over alphabet $\bundleset$; see Figure \ref{fig:view:decomp2} for an example.

\begin{figure}
\begin{center}
\includegraphics{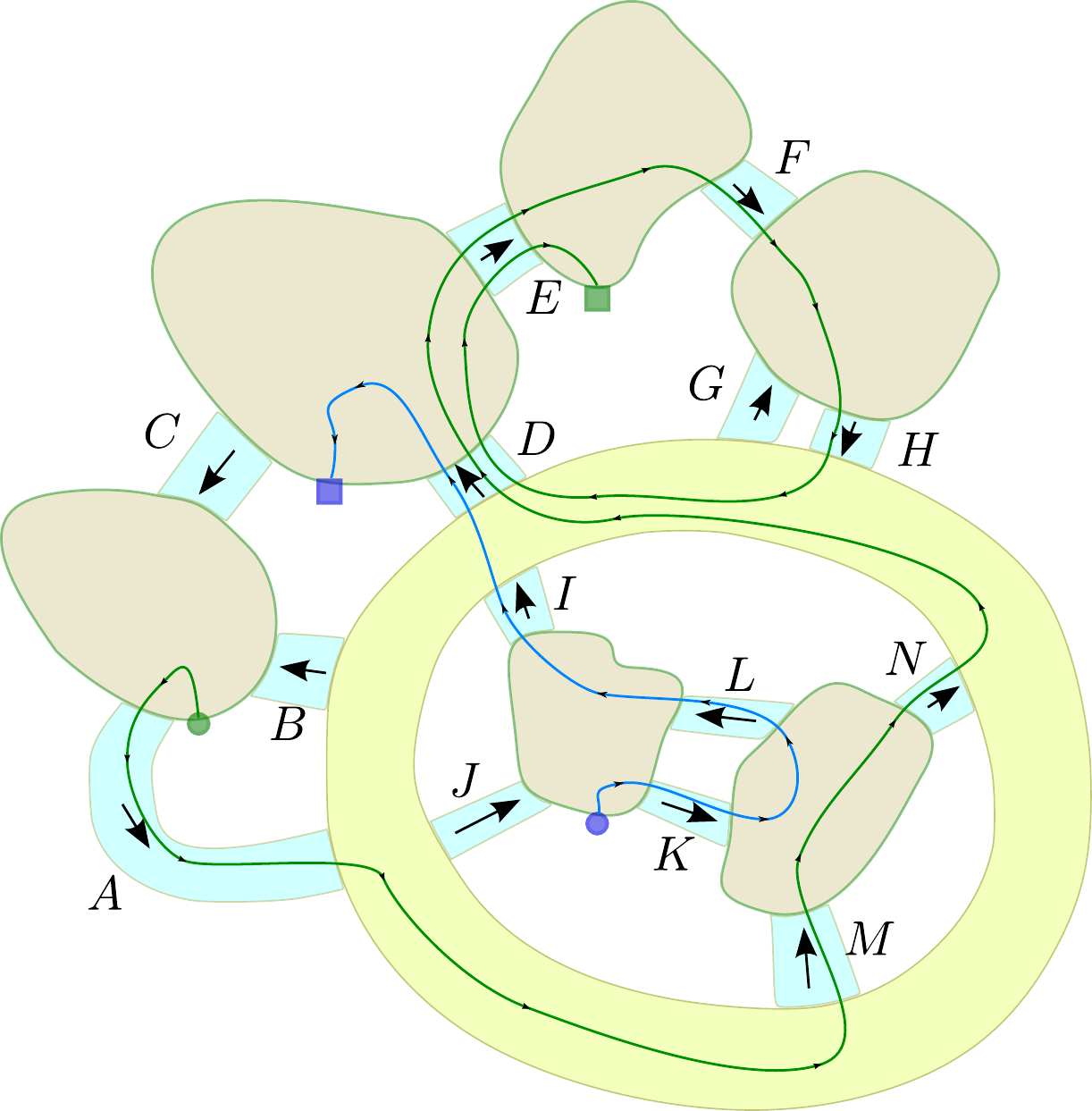}
\caption{An example of a solution. The green path has bundle word $ANMDEFHDE$ and the blue path has bundle word $KLID$.}
  \label{fig:view:decomp2}
  \end{center}
  \end{figure}

Assume for a moment that there are no ring components in the decomposition; ring components present their own challenges requiring an additional
layer of technical work, but they do not alter the main line of reasoning.
Assume moreover that we have computed somehow bundle words $\bundleword(P_i)$, $i=1,2,\ldots,k$ for some solution $(P_i)_{i=1}^k$
for \probshort{} on $G$. The important observation is that, given the bundle words, the cohomology feasibility algorithm
of Schrijver \cite{schrijver:xp} is able to extract (an approximation of) the paths $P_i$ in polynomial time.

To show this, let us recall the algorithm of Schrijver \cite{schrijver:xp} that solves \probshort{} in polynomial time
for every fixed $k$. The heart of the result of Schrijver lies in the proof that \probshort{} is polynomial-time solvable if we are given a 
homotopy class of the solution. In simpler words, given a ``pre-solution'', where many paths can traverse the same arc, even in wrong direction
(but they cannot cross), one can in polynomial time check if the paths can be ``shifted'' (i.e., modified by a homotopy) so that they
become a feasible solution. In such a ``shift'' (homotopy) one can move a path over a face, but not over a face that contains a terminal.\footnote{Note that we can assume that each terminal is of degree one, and the notion of ``face containing a terminal''  is well-defined.}
See Figure \ref{fig:view:homotopies} for an illustration of different homotopy classes of a solution.

\begin{figure}
\begin{center}
\includegraphics{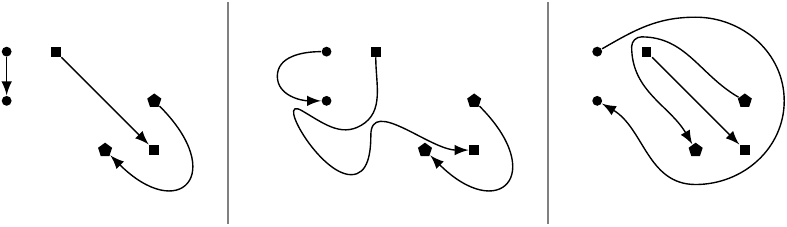}
\caption{Different homotopy classes of a solution: in the first two figures, the solutions are of the same class,
 whereas on the third figure the homotopy class is different.}
  \label{fig:view:homotopies}
  \end{center}
  \end{figure}

In our setting, we note that, in the absence of ring components, two solutions with the same set of bundle words of each paths are homotopical;
thus, given bundle words of a solution, one can use the algorithm of Schrijver to verify if there is a solution consistent with the given
set of bundle words. However, one should note that the homotopies are allowed to do much more than to only move paths within a bundle; formally,
using the Schrijver's algorithm we can either (i) correctly conclude that there is no solution with given set of bundle words $(p_i)_{i=1}^k$, or (ii)
compute a solution $(P_i)_{i=1}^k$ such that the bundles of $\bundleword(P_i)$ is a subset (as a multiset) of the bundles of $p_i$.

Unfortunately, if the decomposition contains ring components as well,
then the bundle words of a solution does not describe the homotopy
class of the solution.  What do we need to learn, apart from bundle
words of the solution, to identify the homotopy class of the solution
if ring components are present?  The answer is not very hard to see:
for any subpath of a path in a solution that crosses some ring
component (i.e., goes from the inside to the outside of vice versa) we
need to know how many times it ``turns'' inside the ring component; we
call it a {\em{winding number}} inside a component.\footnote{Formally,
  for a part of a path that starts and ends on the same side of the
  ring component we need to know also on which side it leaves the
  other side of the ring component; however, it turns out that this
  knowledge is quite easy to guess or deduce and we ignore this issue
  in the overview.}

\begin{figure}
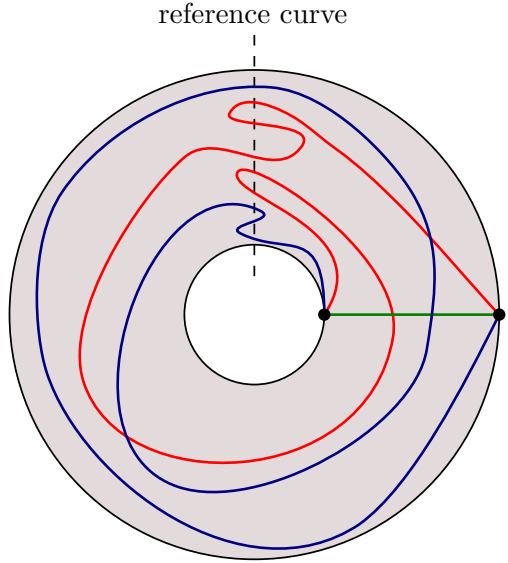

\begin{center}
\svg{0.4\linewidth}{fig-view-winding2}
\caption{Winding numbers inside a ring component; formally defined as the signed number of crosses of some fixed chosen reference curve.
  The green line has winding number $0$, the blue $+2$ and the red $-1$.}
  \label{fig:view:winding}
  \end{center}
  \end{figure}

Thus, our goal is to compute a small family of possible bundle words and winding numbers such that, if there exists a solution to \probshort{}
on $G$, there exists a solution consistent with one of the elements of the family.
In fact, our main goal in the rest of the proof is to show that one can compute
such family of size bounded in the parameter $k$.

\subsection{Guessing bundle words}\label{ss:view:spirals}

Assume again that there are no ring components; we are to guess the bundle words of one of the solutions.
Recall that the number of bundles, $|\bundleset|$, is bounded in $k$. Thus, if a bundle word of some path $P_i$ from a solution $(P_i)_{i=1}^k$
is long, it needs to contain many repetitions of the same bundles.

Let us look at one such repetition: let $uB$ be a subword of $\bundleword(P_i)$, where $B$ is the first symbol of $u$ and
$u$ contains each symbol of $\bundleset$ at most once. We call such place a {\em{spiral}}. Note that this spiral separates the graph
into two parts, the inside and the outside, where any other path can cross the spiral only in a narrow place inside bundle $B$ (see Figure~\ref{fig:view:spiraling-cut}).
As the arcs of $B$ go in one direction, any path $P_j$, $j \neq i$ can cross the spiral only once, in the same direction as $P_i$, and the
path $P_i$ cannot cross the spiral $uB$ again. Note that we know exactly which paths cross the spiral $uB$: the paths that have terminals
on different sides of the spiral $uB$.

\begin{figure}
\begin{center}
\includegraphics[width=0.7\textwidth]{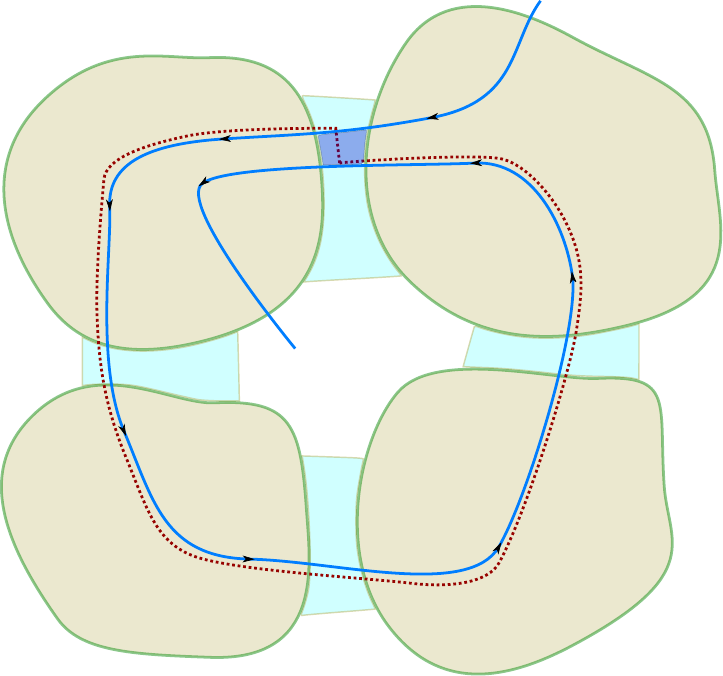}
\caption{A spiral. Any other path may cross the dashed curve only in a narrow place in the top bundle (highlighted with darker blue colour).}
  \label{fig:view:spiraling-cut}
  \end{center}
  \end{figure}

  There is also one important corollary of this observation on
  spirals. If $\bundleword(P_i) = x u B y$ for some words $x,y$ and
  spiral $uB$, then, for any bundle $B'$ that does not appear in $u$,
  only one of the words $x$ or $y$ may contain $B'$: the bundle $B'$
  is either contained inside the spiral $uB$ or outside it. By some
  quite simple word combinatorics, we infer that $\bundleword(P_i)$
  can be decomposed as $u_1^{r_1} u_2^{r_2} u_3^{r_3} \ldots
  u_s^{r_s}$, where each word $u_i$ contains each symbol of
  $\bundleset$ at most once, each $r_i$ is an integer and $s \leq
  2|\bundleset|$.  Note that if we aim to guess bundle words of some
  solution $(P_i)_{i=1}^k$, there is only a bounded number of choices
  for the length $s$ and the words $u_i$; the difficult part is to
  guess the exponents $r_i$, if they turn out to be big (unbounded in
  $k$).  That is, we can easily guess the global structure of the
  spirals (how they are nested, etc.), but we cannot easily guess the
  ``width'' of the spirals (how many turns of the same type they
  do). We need further analysis and insight in order to be able to
  guess these numbers as well.

Let us focus on a place in a graph where a path $P_i$ in the solution $(P_i)_{i=1}^k$ contains a subword $u^rB$ of $\bundleword(P_i)$
for some large integer $r$, where $B$ is the first symbol of $u$.
The situation, depicted on Figure \ref{fig:view:spiraling-ring}, looks like a large spiral;
the {\em{spiraling ring}} is the area between the first and last spiral $uB$ in the subword $u^rB$.
Note that any path $P_j$ that enters this area, actually needs to traverse all $r$ spirals $uB$ and $\bundleword(P_j)$ contains a subword $u^{r-1}B$;
let $I \subseteq \{1,2,\ldots,k\}$ be the set of indices $j$ such that $P_j$ traverses $uB$.
Moreover, note that, since $B$ contains arcs going in only one direction, these parts of paths $(P_j)_{j \in I}$
are the only intersections of the solution $(P_i)_{i=1}^k$
with the spiraling ring.

\begin{figure}
\begin{center}
\includegraphics[width=0.7\textwidth]{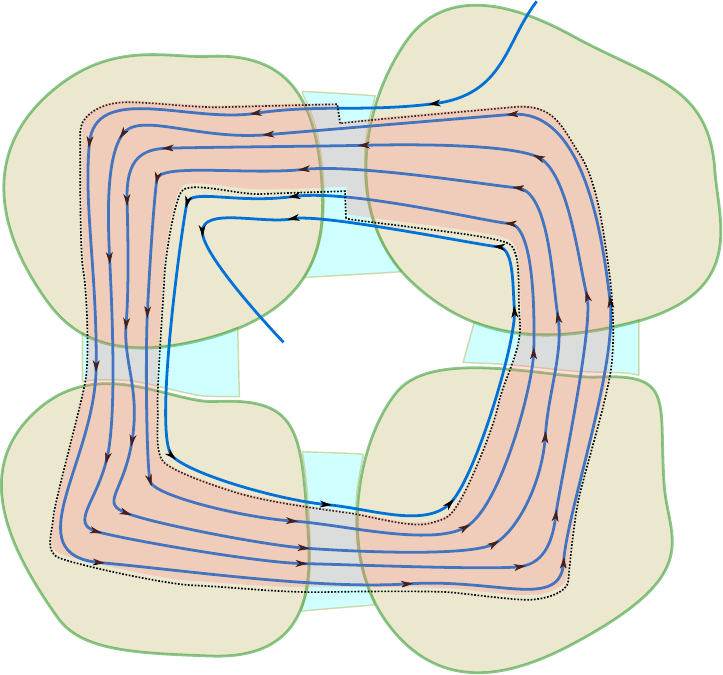}
\caption{A spiraling ring; the dashed lines are its borders.}
  \label{fig:view:spiraling-ring}
  \end{center}
  \end{figure}

Our main claim is that we can choose $r$ to be as small as possible, just to be able to route the desired paths through the spiraling ring in $G$.
More formally,
we prove that if we can route $|I|$ directed paths through the spiraling ring, such that each path traverses $B$ roughly $r^\ast$ times
(but they may start and end in different places than the parts of the solution $(P_i)_{i=1}^k$ traversing the spiraling ring),
then we can modify the solution $(P_i)_{i=1}^k$ inside the spiraling ring such that $r \leq r^\ast + O(1)$.
If we choose $r^\ast$ to be minimal possible, we have $|r-r^\ast| = O(1)$ and we can guess $r$.

To prove that the paths can be rerouted, we show the following in Section \ref{sec:spiral}.
\begin{theorem}[rerouting in a ring, informal statement]\label{thm:view:ring-routing}
Let $G$ be a plane graph embedded in a ring, with outer face $f_{out}$
and inner face $f_{in}$. Assume that there exist two sequences of
vertex-disjoint paths $(P_i)_{i=1}^s$ and $(Q_i)_{i=1}^s$
connecting $f_{in}$ with $f_{out}$, such that $P_i$ goes in the same direction
(from $f_{in}$ to $f_{out}$ or vice versa) as $Q_i$, and the endpoints of $(P_i)_{i=1}^s$
lie in the same order on $f_{in}$ as the endpoints of $(Q_i)_{i=1}^s$.
Then one can reroute $(P_i)_{i=1}^s$ inside $G$, keeping the endpoints,
such that the winding number of $P_1$ differs from
the winding number of $Q_1$ by no more than $6$.
\end{theorem}

How to prove such a rerouting statement?
We again make use of the 
results of Ding et al. \cite{dss:green-line} on a canonical obstacle for routing a set of directed cycles on a torus
(as in the proof of Lemma \ref{lem:view:duality}).
We connect $f_{in}$ and $f_{out}$ with a handle, and perceive the paths
$P_i$ and $Q_i$ as a cycles on a torus. The winding number $w_P$ of $P_1$ determines
the homotopy class of the cycles $(P_i)_{i=1}^s$, and the winding number $w_Q$ of $Q_1$
determines the homotopy
class of the cycles $(Q_i)_{i=1}^s$.
Now we observe that an obstacle for routing
the same set of cycles with ``homotopy'' $w$ for $w_Q + O(1) < w < w_P - O(1)$ (or, symmetrically, $w_P + O(1) < w < w_Q - O(1)$)
would be an obstacle for ``homotopy'' either $w_Q$ or $w_P$,
a contradiction. Hence, almost all ``homotopies'' between $w_Q$ and $w_P$ are realizable.
By some gadgeteering, we may force
the cycles to use the same endpoints as the paths $(P_i)_{i=1}^s$, at the cost of $O(1)$ loss in the ``homotopy'' class.

We would like to note that Theorem \ref{thm:view:ring-routing} is
a cornerstone of our result. The exponential time complexity of
the algorithm of Schrijver \cite{schrijver:xp} comes from the fact that the number of homotopy classes
of a solution solution cannot be bounded by a function of $k$, because the number of possibilities for the number of turns of the solution 
in some ring-like parts of the graph cannot be bounded by a function of $k$. Theorem \ref{thm:view:ring-routing} overcomes this obstacle by showing that that for each such ring, one can choose a canonical number of turns
(that depends only on the ring, not how it is connected to the outside)
and the solution can be assumed to spiral a similar number of turns than the canonical
pass. In other words, if one would try to construct a $W[1]$-hardness proof
of \probshort{} by a reduction from, say, $k$-\textsc{Clique},
one cannot expect to obtain a gadget that encodes a choice of a vertex or edge of the clique
by a number of turns a solution path makes in some ring-like part of the graph
--- and such an encoding seems natural, taking into account the source
of the exponential-time complexity of the algorithm of Schrijver \cite{schrijver:xp}.

However, it still requires significant work to make use of Theorem \ref{thm:view:ring-routing}.
In the case of spiraling rings, the question that remains is how to get minimal exponent $r^\ast$ such that $|I|$ paths can be routed through a spiraling ring
with $r^\ast$ turns. The idea is to isolate a part of the graph and parts of the bundle words of the solution
where only one exponent is unknown, and then apply Schrijver's algorithm for 
different choices of exponent; the desired value $r^\ast$ is the smallest exponent for which Schrijver's algorithm returns a solution.

However, in this approach we need to cope with two difficulties. First, the Schrijver's algorithm may find a solution that 
follows our guidelines (bundle words) in a very relaxed manner. However, as at each step we choose the exponent $r$ to be very close
to the minimum possible number of turns in a spiraling ring, the bundle words of the solution found by the algorithm cannot differ much
from the given ones, as they need to spiral at least the number of times we have asked them to (up to an additive constant).

Second, it is not always easy to isolate a part of the graph where one spiraling occurs. A natural thing to do is to cut the graph
along bundles not used in the spiral and attach auxiliary terminals; however, in a situation on Figure \ref{fig:view:nontrivial-spiral}
we cannot separate the middle spiraling ring from the two shorter ones outside and inside it. Luckily, it turns out that here
the additional spirals use always {\em{strictly smaller}} number of bundles, and we can guess exponents $r$ in terms $u^r$
in the increasing order of $|u|$.

\begin{figure}[ht]
\begin{center}
\includegraphics[width=0.6\textwidth]{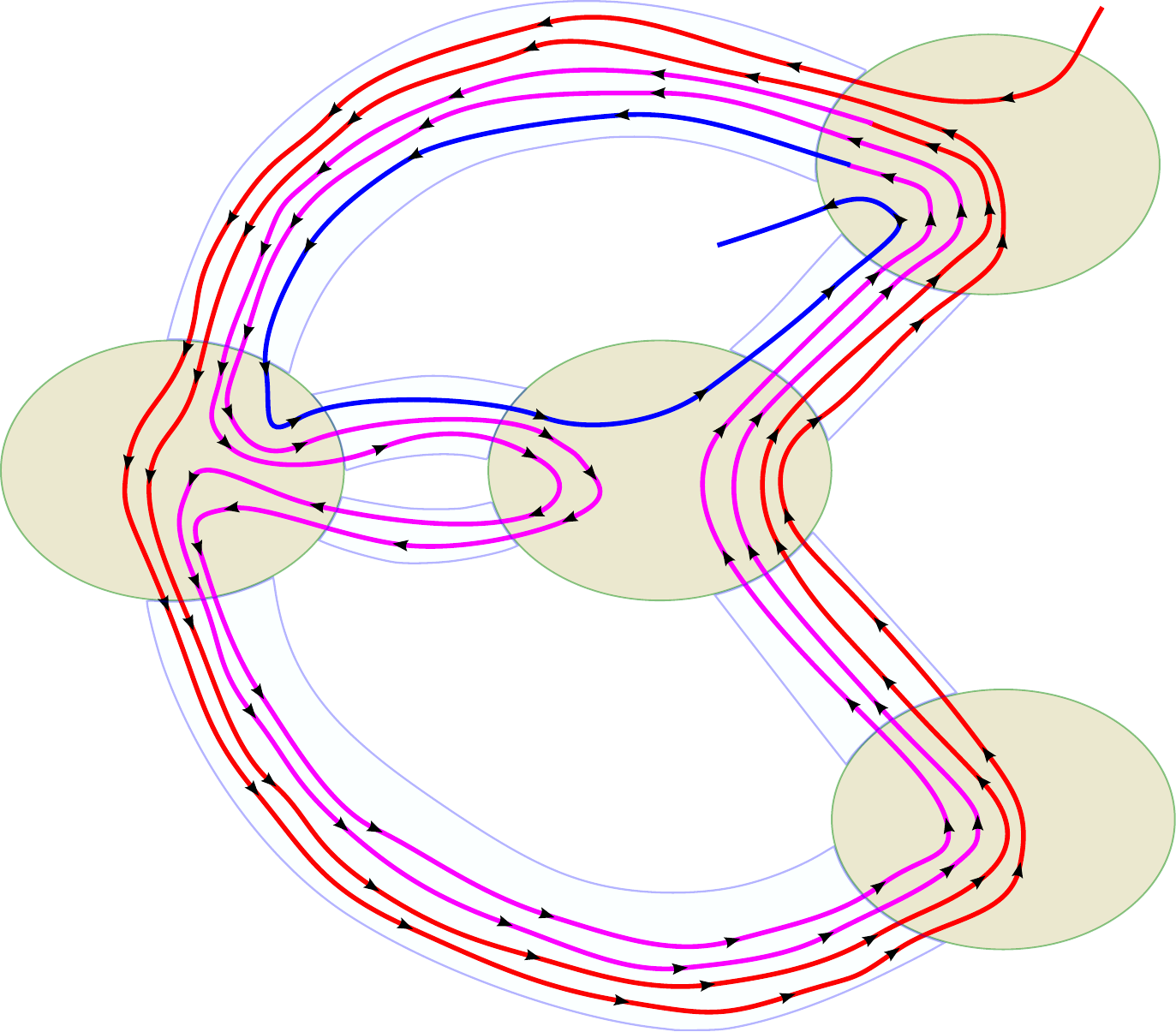}
\caption{A nontrivial situation around a spiral on $P_i$. Path $P_i$ first spirals two times using four bundles (red spiral), then makes the spiral we want to investigate by spiraling two times using six bundles (pink spiral), and finally makes one turn of a spiral using three bundles (blue spiral). Note that the set of bundles used in the red spiral and in the blue spiral is a subset of the set of bundles used in the pink spiral, which we aim to measure. The crucial observation is that these sets of bundles must be in fact {\bf{proper}} subsets of the set used by the pink spiral.}
\label{fig:view:nontrivial-spiral}
\end{center}
\end{figure}

\subsection{Handling a ring component}\label{ss:view:ring}

In the previous subsection we have sketched how to guess bundle words of the solution in absence of ring components.
Recall that, for a ring component, and for any part of a path that traverses a ring component (henceforth called {\em{ring passage}};
note that ring passages are visible in bundle words of paths) we need to know its {\em{winding number}}: the number of times it turns
inside the ring component.

As we have learnt already how to route paths in rings, it is tempting to use the aforementioned techniques to guess winding numbers:
guess how many ring passages there are, and find one winding number $w^\ast$ for which routing is possible (using Schrijver's algorithm)\footnote{Note that all winding numbers of passages in one ring component do not differ by more than one, and in this overview we assume they are equal.};
the actual solution should be reroutable to a winding number $w$ close to $w^\ast$.

However, there are two major problems with this approach. 
First, not only the ring passages of the solution use the arcs and vertices of a ring component, but 
parts of paths from the solution that start and end on the same side of the ring component (henceforth called {\em{ring visitors}}) may also be present.
Luckily, we may assume that the ring components contain many concentric cycles of alternating orientation, as otherwise the decomposition
algorithm would cut it though to obtain a disc component. If a ring visitor goes too deeply into the ring component, it creates a $d$-bend
for too large $d$ and we can reroute it.
Thus, the ring visitors use only a thin layer of the ring component, and we can argue that we can still conduct the rerouting argument
in the ring component without bigger loss on the bound on $|w-w^\ast|$; see Figure \ref{fig:view:passage} for an illustration.

\begin{figure}[ht]
\begin{center}
\includegraphics[width=0.6\textwidth]{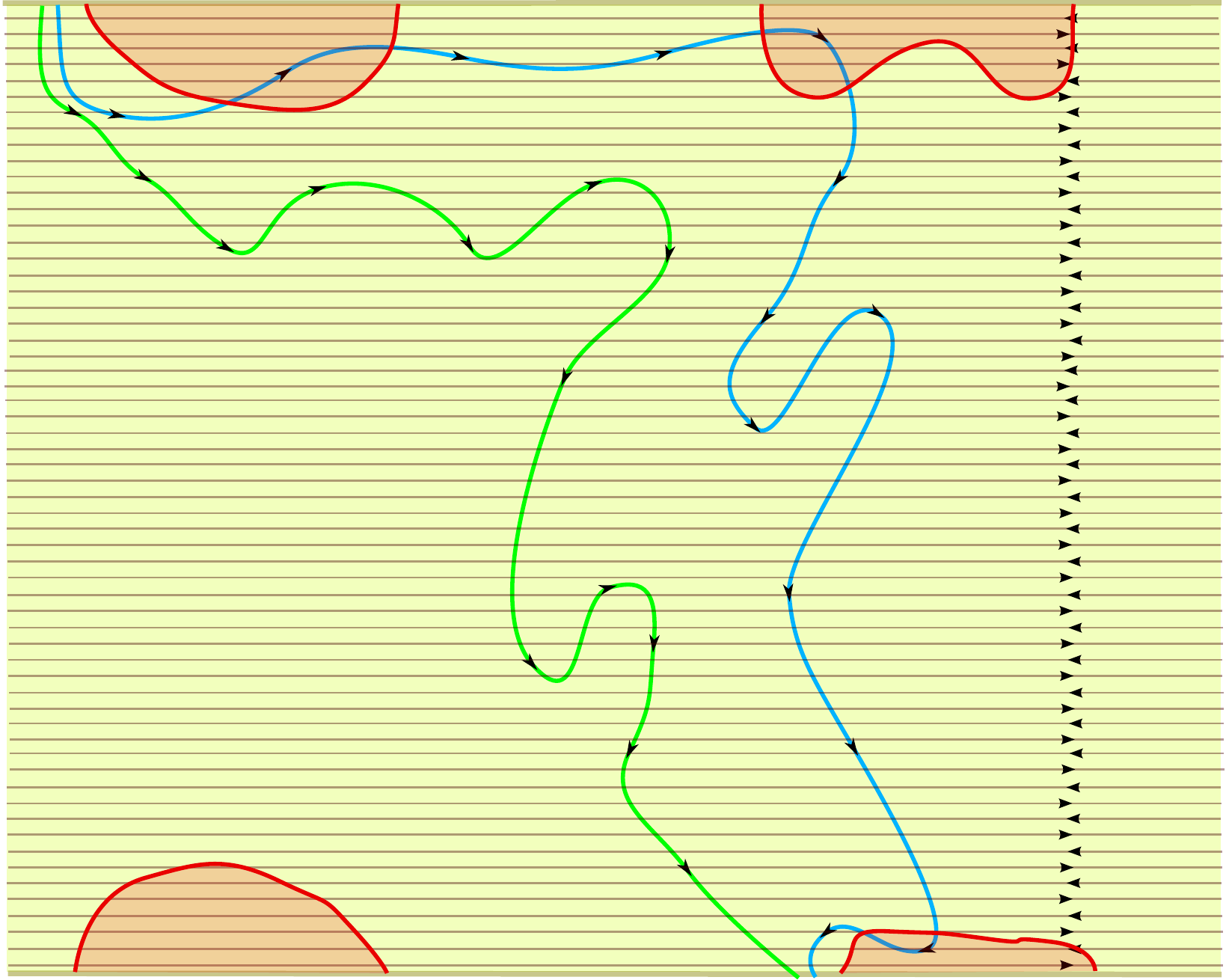}
\caption{Guessing winding number in a (flattened) ring component.
  The green pass is the actual solution we seek for, and the blue pass is the path we found.
  The red parts are used by ring visitors in the solution, thus we are not able to use the rerouting argument
  to the whole ring component. However, we are able to use it for the ring component with the red parts removed;
  we then argue, basing on the shallowness of ring visitors,
  that the winding numbers of the entire blue pass and the subpath of the blue pass that connects the red parts on the top
  with the red parts on the bottom do not differ much.
}
\label{fig:view:passage}
\end{center}
\end{figure}

Second, we do not have yet any means to control the number of ring passages,
and the previous techniques of guessing bundle words have significant technical problems
if we try to handle spiraling rings involving ring components.
Hence, it is not trivial to guess the set of ring passages traversing a ring component.
Here again we may use the concentric cycles hidden inside a ring component, as well as bounded alternation cuts that can be found repeatedly
inside the ring component if the irrelevant vertex rule is not applicable.
We argue that, if we have many ring passages, a large number of them need to travel together via a large number of concentric cycles
of alternating orientation and we can use a rerouting argument in the spirit of the one used by Adler et al.~\cite{isolde}.\footnote{It is worth noticing that the bound on directions make it possible to use a simple flow argument instead of the techniques of Adler et al.~\cite{isolde}.}
Overall, we obtain that there exists a solution with a bounded number of ring passages, and we are able to guess a good candidate for a winding number inside a ring component.
To merge the techniques of the previous and this subsection, we need to handle ring visitors when guessing bundle words:
these visitors may take part in some large spiraling ring. Luckily, as ring visitors are shallow in ring components, we can ``peel'' the ring components:
separate a thin layer that may contain ring visitors, cut it through and treat is as disc component for the sake of bundle word guessing
algorithm.

The analysis of bundle words, ring visitors and ring passages is done in Section \ref{sec:bundles-and-bundle-words}.
The Schrijver's algorithm is recalled in Section \ref{s:words-to-paths}.
The final guessing arguments --- both for bundle words and winding numbers --- are described in Section \ref{s:word-guessing}.

\subsection{Summary}

We conclude with a summary of the structure of the algorithm.

First, we invoke decomposition algorithm of Theorem \ref{thm:view:decomp}. If it fails, it exhibits a place where the irrelevant vertex rule is applicable; apply the rule and restart the algorithm. Otherwise, compute bundled instance $(G,\decomp,\bundleset)$, with $|\decomp|$ and $|\bundleset|$
bounded in $k$.

Given the bundled instance $(G,\decomp,\bundleset)$, we aim to branch into subcases whose number is bounded by a function of $k$, in each subcase
guessing a set of bundle words for the solution, as well as winding numbers of each ring passage. Our branching will be exhaustive
in the following sense: if $G$ is a YES-instance to \probshort{}, there will be a guess consistent with some solution
(but not {\em{all}} solutions will have their consistent branches).

We branch in two phases. First, we guess the bundle words;
the hard part is to guess exponents in spiraling rings $u^r$, where we argue that we can choose an exponent close to minimal possible
number of turns in a spiraling ring, and detect this number using an application of Schrijver's algorithm for a carefully chosen
subgraph of $G$. Second, we guess the winding numbers; here we argue that the winding numbers of the solution can be assumed to be close
to a winding number of an arbitrarily chosen way to route ring passages through the ring component, ignoring the ring visitors.

Finally, given bundle words and winding numbers, we deduce the homotopy type of the solution and invoke Schrijver's algorithm
on the entire graph to verify whether our guess is a correct one.



\newcommand{\cC}{{E_C}}
\newcommand{\cP}{{E_P}}

\section{Irrelevant vertex rule}
\label{sec:irr}

We prove the validity of the irrelevant vertex rule in this section.

\begin{defin}
Consider an instance of the $k$-path problem with graph $G$ and terminals $T$. Let $v$ be a nonterminal vertex. We say that $v$ is {\em irrelevant} if the instance on $G$ has a solution if and only if the instance on $G\setminus v$ has.
\end{defin}

The main result of the section is the following:
\begin{theorem}[Irrelevant Vertex Rule]\label{th:irrelevant}
  Consider an instance of the $k$-path problem with a graph $G$
  embedded in the plane and a set $T$ of terminals. There is a $d:=d(k)=2^{O(k^2)}$
  such that the following holds. Let $C_0$, $\dots$, $C_d$ be an alternating sequence of concentric cycles ($C_0$ is outside) such that
  there is no terminal enclosed by $C_0$. If $v$ is a vertex of $C_d$,
  then $v$ is irrelevant.
\end{theorem}

We prove Theorem~\ref{th:irrelevant} by formulating a statement about unique solutions.
\begin{defin}\label{def:expmin}
  Let $G$ be a graph with a set $T$ of terminals. We say that a
  solution $P_1$, $\dots$, $P_k$ is {\em unique} if for every solution
  $P'_1$, $\dots$, $P'_k$ we have $P_i=P'_i$ for every $i$.  Given a
  solution $P_1$, $\dots$, $P_k$, we say that an arc is {\em free} if
  it is not used by any $P_i$.
\end{defin}
The main technical statement that we prove is that a unique solution cannot enter a large set of concentric free cycles:
\begin{lemma}\label{lem:unique}
  There is a function $d(k)=2^{O(k^2)}$ such that if an instance has a
  unique solution such that there is an alternating sequence of $d(k)$
  free concentric cycles not enclosing any terminals, then the
  solution does not intersect the innermost cycle.
\end{lemma}
The setting of Lemma~\ref{lem:unique} somewhat simplifies the
proof, as we have to argue about a solution that is disjoint from the
concentric cycles. Theorem~\ref{th:irrelevant} follows from
Lemma~\ref{lem:unique} by the combination of a minimal choice
argument and contracting arcs that are shared by the solution and the
cycles.
\begin{proof}[Proof (or Theorem~\ref{th:irrelevant})] Let $\cC$ be the
  union of the arc sets of every cycle $C_i$.  If the instance has no
  solution, then $v$ is irrelevant by definition. Otherwise, let $P_1$,
  $\dots$, $P_k$ be a solution using the minimum number of arcs {\em
    not} in $\cC$; let $\cP$ be union of the arc sets of every path
  $P_i$.  We may assume that this solution uses vertex $v$, otherwise
  we are done. We create a new instance and a corresponding solution
  the following way.  First, let us remove every arc not in $\cP\cup
  \cC$. Next, we contract every arc of $\cP\cap \cC$. Note that it is
  not possible that the solution uses every arc of a cycle $C_j$,
  thus each cycle $C_j$ is contracted into a cycle $C'_j$. (The length
  of the cycle $C'_j$ might be 1 or 2; in order to avoid dealing with
  loops and parallel arcs, we may subdivide such arcs without
  changing the problem.)  Let $G'$ be the graph obtained this way.
  Clearly, each path $P_i$ becomes a (possibly shorter) path $P'_i$
  with the same endpoint and the paths remain disjoint, hence
  $P'_1$, $\dots$, $P'_k$ is a solution of the new instance that is
  disjoint from the cycles $C'_1$, $\dots$, $C'_d$. The arcs of $\cC\setminus \cP$ are free, and hence the cycles $C'_1$, $\dots$, $C'_d$ are free.

  As one of the paths $P_i$ used vertex $v$ of $C_d$, there is a path
  $P'_i$ intersecting $C'_d$.  Therefore, by
  Lemma~\ref{lem:unique}, $P'_1$, $\dots$, $P'_k$ is not the unique
  solution. Let $Q'_1$, $\dots$, $Q'_k$ be a solution different from
  $P'_1$, $\dots$, $P'_k$. As the arcs of $\cP\setminus \cC$ form
  disjoint paths connecting the terminals, there has to be an arc
  $e^*\in \cP\setminus \cC$ not used by any $Q'_i$. We construct a
  solution $Q_1$, $\dots$, $Q_k$ of the original instance from $Q'_1$,
  $\dots$, $Q'_k$ by uncontracting the arcs contracted during the
  construction of the new instance. We have to verify that each $Q_i$
  is a directed path. In general, if we contract an arc
  $\overrightarrow{xy}$ into a single vertex $v_{xy}$ in a directed
  graph, then a path $Q$ going through $v_{xy}$ in the contracted
  graph might not correspond to any directed path in the original
  graph: it is possible that $Q$ enters $v_{xy}$ on an arc that
  corresponds to an arc entering $y$ and it leaves $v_{xy}$ on an
  arc that leaves $x$, hence replacing $v_{xy}$ with the arc
  $\overrightarrow{xy}$ in $Q$ does not result in a directed
  path. However, in our case, when we contracted an arc
  $\overrightarrow{xy}$, then $\overrightarrow{xy}$ is the only arc
  leaving $x$ and it is the only arc entering $y$: this follows from
  the fact that we remove every arc not on any path $P_i$ or on any
  cycle $C_j$. Therefore, the paths $Q_1$, $\dots$, $Q_k$ obtained by
  reversing the contractions form a solution of the original
  instance. Let us observe that, as no $Q'_i$ uses the arc  
  $e^*\in \cP\setminus \cC$ in $G'$, no $Q_i$ uses the corresponding arc in $G$, which is
  an arc not in $\cC$. Moreover, if some $Q_i$ uses an arc that is
  not in $\cC$, then this arc is in $\cP$ and hence used by some $P_{i'}$ (otherwise we
  would have removed it in the construction of $G'$). Thus $Q_1$,
  $\dots$, $Q_k$ is a solution that uses strictly fewer arcs not in
  $\cC$ than $P_1$, $\dots$, $P_k$, which contradicts the minimality
  of the choice of $P_1$, $\dots$, $P_k$.
\end{proof}

\subsection{Bends}
\begin{defin}
Let $G$ be an embedded planar graph, $T$ a set of terminals, and $P_1$, $\dots$, $P_k$ a solution.
A {\em $d$-bend} $B=(P;C_0,\dots,C_d)$ consists of
\begin{itemize}
\item A subpath $P$ of some $P_i$ with endpoints $x_0$ and $y_0$, and
  distinct vertices $x_0$, $\dots$, $x_d$, $y_d$, $\dots$, $y_0$
  appearing on it in this order or in the reverse of this order.
\item For every every $0\le i \le d$, a free path $C_i$ with endpoints $\{x_i,y_i\}$ (the {\em chords})
such that
\begin{itemize}
\item The paths $C_i$ are pairwise vertex disjoint.
\item The orientations of the paths are alternating, i.e., for every $0\le i < d$, vertex $x_i$ is the start vertex of $C_i$ if and only if $x_{i+1}$ is the end vertex of $C_{i+1}$.
\item The internal vertices of $C_i$ are not on $P$.
\item The undirected cycle formed by $C_i$ and $P[x_i,y_i]$ (see footnote\footnote{
If $P$ is a directed path, then $P[x,y]$ is the subpath with endpoints $x$ and $y$. We do not specify which of $x$ and $y$ is the start vertex, i.e., $P[x,y]=P[y,x]$.})
encloses $C_j$ for every $j>i$.
\end{itemize}
\end{itemize}
We say that $d$-bend $B$ {\em (strictly) encloses} a vertex/face/path if the undirected cycle formed by $P$ and $C_0$ (strictly) encloses it. We say that $B$ is {\em terminal free} if it does not enclose any terminals.
We say that $d$-bend $B$ {\em appears} on a path $Q$ if $P$ is a subpath of $Q$.
\end{defin}
Note that we do not specify the orientation of the paths $P$ and $C_0$,
(only that the orientations of $C_0$, $\dots$, $C_d$  are alternating).
\begin{figure}
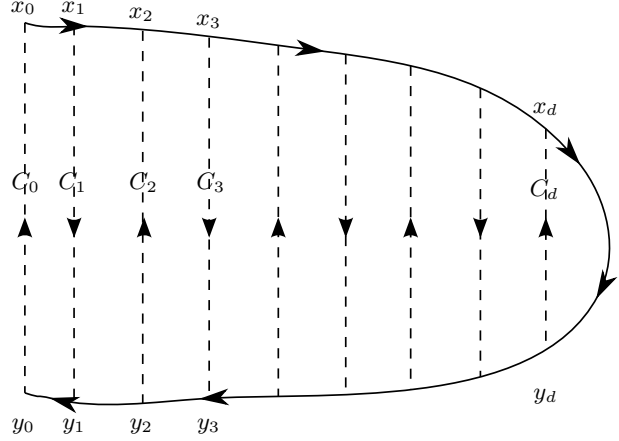

{\footnotesize 
\begin{center}
\svg{0.5\linewidth}{bend}
\end{center}
}
\caption{A $d$-bend $B=(P;C_0,\dots,C_d)$.}\label{fig:bend}
\end{figure}

Observe that if $B=(P;C_0,\dots,C_d)$ is a $d$-bend of $G$, then it is
a $d$-bend of every (embedded) supergraph of $G$: adding new arcs in
an embedding-preserving way does not ruin any of the properties. It
also follows that if $B$ is a $d$-bend of a subgraph $G$, then it is a
$d$-bend in $G$ as well.

\begin{defin}
  Given a solution $P_1$, $\dots$, $P_k$ and $d$-bend
  $B=(P;C_0,\dots,C_d)$, the {\em type} of a $d$-bend is the number of
  paths $P_i$ that have a vertex enclosed by $B$.
\end{defin}
Note that if $B$ appears on some $P_i$, then a subpath of $P_i$ is
enclosed by $B$ and hence counted in the type of $B$.
We prove the following statement by induction on type $t$. 
\begin{lemma}\label{lem:bendbound}
There is a function $f(k,t)$ such that if $P_1$, $\dots$, $P_k$ is a unique solution, then no $f(k,t)$-bend of type $t$.
\end{lemma}

The following simple parity argument will be used several times:

\begin{lemma}\label{lem:pathparity0}
  Let $C$ be an undirected cycle, $Q$ a subpath of $C$, and $L$ a path
  with endpoints not enclosed by $C$ such that $L$ does not share any
  arcs with $C$, does not go through the endpoints of $Q$ and
  crosses $Q$ an odd number of times. Then there is a vertex $w_1$ of
  $Q\cap L$ and a vertex $w_2$ of $(C\setminus Q)\cap L$ such that the
  subpath $L[w_1,w_2]$ is enclosed by $C$ and the internal vertices of
  $L[w_1,w_2]$ are disjoint from $C$.
\end{lemma}
\begin{proof}
  Clearly, $L$ crosses $C$ an even number of times; let $p_1$,
  $\dots$, $p_{2f}$ be these crossing points in the order they appear
  on $L$ (see Figure~\ref{fig:parity}). As the endpoints of $L$ are not enclosed by $C$, the subpath
  $L[p_{2j-1},p_{2j}]$ is enclosed by $C$ for every $1\le j \le
  f$. Path $L$ crosses $Q$ an odd number of times, which means that
  there is a $1\le j \le f$ such that exactly one of the crossing
  points $p_{2j-1}$ and $p_{2j}$ is on $Q$. If $L[p_{2j-1},p_{2j}]$
  has no internal vertex on $Q$, then we are done. Otherwise, as there
  are no crossing points between $p_{2j-1}$ and $p_{2j}$, path $L$
  touches $C$ at every such internal vertex. Suppose that there are
  $s$ such touching points (possibly $s=0$); let they be $q_1$,
  $\dots$, $q_s$ in the order they appear on $L$ from $p_{2j-1}$ and
  $p_{2j}$. Let $q_0=p_{2j-1}$ and $q_{s+1}=p_{2j}$. As exactly one of
  $q_0$ and $q_{s+1}$ is on $Q$, there has to be an $0\le i \le s$
  such that exactly one of $q_i$ and $q_{i+1}$ is on $Q$.  Let $w_1\in
  \{q_s,q_{s+1}\}$ be the vertex on $Q$ and let $w_2$ be the other
  one; the subpath $L[w_1,w_2]$ now has no internal vertex on $C$ and
  is enclosed by $C$ (as $L[q_0,q_{s+1}]$ is enclosed by $C$), thus
  $L[w_1,w_2]$ satisfies the requirements.
\end{proof}
\begin{figure}
{\footnotesize 
\begin{center}
\svg{0.7\linewidth}{parity}
\end{center}
}
\caption{Lemma~\ref{lem:pathparity0}.}\label{fig:parity}
\end{figure}

The following lemma shows that if a solution uses the innermost cycle
of an alternating sequence of concentric cycles, then there is a
$d$-bend. This allows us to use the bound on $d$-bends to show that a
unique solution cannot use the innermost cycle.  That is, the following
lemmas show that Lemma~\ref{lem:bendbound} implies
Lemma~\ref{lem:unique}.

\begin{lemma}\label{lem:concentricbend2}
Let $C_0$, $\dots$, $C_{d+1}$ be an alternating sequence of free
concentric cycles. Assume that a path $P_b$ of the solution
contains a subpath $P := P_b[x_0,y_0]$ such that $x_0,y_0 \in C_0$,
$P$ does not contain any other vertex of $C_0$ except for $x_0$ and $y_0$,
and $P$ intersects $C_{d+1}$.
Moreover, assume that for one of the two subpaths $C_0'$ of $C_0$ between $x_0$
and $y_0$, the cycle $P \cup C_0'$ does not enclose any terminals. Then there
exists a terminal free $d$-bend $(P,C_0',C_1',\ldots,C_d')$, where $C_i'$ is a subpath of $C_i$
for $1 \leq i \leq d$.
\end{lemma}
\begin{proof}
  Let $z$ be an arbitrary vertex of $C_{d+1}\cap P$.
  We may also assume that the
  graph has no other arc than the arcs of $P$ and the arcs of the
  $C_i$'s: if we can find a $d$-bend after removing the additional
  arcs, then there is a $d$-bend in the original graph as well.

  Consider the cycle $C:=C'_0\cup P$. As $C$ intersects $C_{d+1}$, it
  cannot enclose any $C_i$ with $i\le d$. Therefore, every $C_i$ with
  $0\le i \le d$ has an arc $e_i$ not enclosed by $C$. For every
  $0\le i \le d$, let us subdivide $e_i$ with two new vertices and let
  us remove the arc between these two new vertices.  Now these two
  new vertices are the endpoints of a path $L_i$; it is clear that
  the endpoints of $L_i$ are not enclosed by $C$.

  Let $Q$ be the subpath $P[x_0,z]$. Clearly, $Q$ crosses each $C_i$
  with $1 \le i \le d$ an odd number of times. As $Q$ does not use the
  arc $e_i$, this implies that $L_i$ crosses $Q$ an odd number
  times. Applying Lemma~\ref{lem:pathparity0} on the path $L_i$ and
  the subpath $Q$ of cycle $C$, we get that for every $1\le i \le d$,
  there is at least one subpath $L_i[w_1,w_2]$ enclosed by $C$ such
  that $w_1\in Q$, $w_2\not\in Q$, and the path has no internal vertex
  on $P$. We will call such a subpath of $L'_i$ an {\em $i$-chord} (see Figure~\ref{fig:chords}). In
  particular, $C'_0$ is the unique 0-chord. Note that these chords are
  internally vertex disjoint (but two $i$-chords can share an
  endpoint), thus they have a natural ordering along $Q$, starting
  with $C'_0$. Let us observe that if an $i$-chord and a $j$-chord are
  consecutive in this ordering, then $|i-j|\le 1$: this is because
  then there is a face whose boundary contains both chords and no such
  face is possible if $|i-j|>1$ (recall that we assumed that there is
  no other arc in the graph than the arcs of $P$ and the arcs of
  the $C_i$'s).
\begin{figure}
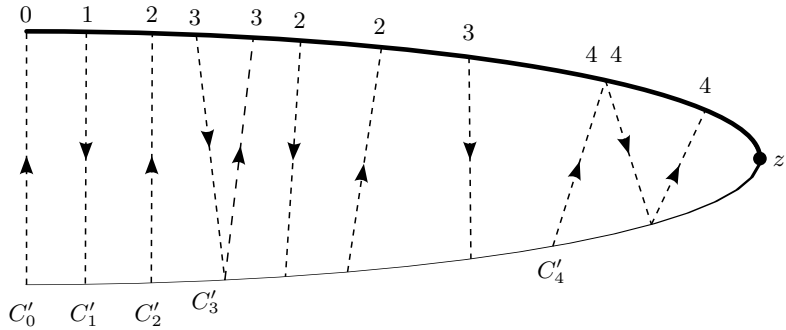

{\footnotesize 
\begin{center}
\svg{0.65\linewidth}{chords}
\end{center}
}
\caption{A possible layout of the $i$-chords and the selection of the chords $C'_i$. Note that each chord is enclosed by the cycle $C$.}\label{fig:chords}
\end{figure}

  Let $C'_i=L_i[x_i,y_i]$ be the first $i$-chord in this ordering. We
  claim that for every $0\le i \le d$, the orientation of $C'_i$ in
  the cycle $C'_i\cup P[x_i,y_i]$ is the same as the orientation of
  the cycle $C_i$.  For $i=0$, this follows from the fact that $P$ is
  enclosed by $C_0$. For some $1 \le i \le d$, let $K$ be the last
  chord before $C'_i$. We claim that $K$ is an $(i-1)$-chord. By our
  observation in the previous paragraph, it is either an $(i-1)$-chord,
  $i$-chord, or $(i+1)$-chord. As $C'_i$ is the first $i$-chord, $K$
  cannot be an $i$-chord. If $K$ is an $(i+1)$-chord, then there has
  to be another $i$-chord between $C'_0$ and $K$ (as the difference is
  at most one between consecutive chords), again contradicting the
  choice of $C'_i$. Thus $K$ is an $(i-1)$-chord. Consider the face
  $F$ whose boundary contains both $K$ and $C'_i$: this face is
  between $C_{i-1}$ and $C_i$.  If $C_i$ is oriented clockwise (the
  other case follows by symmetry), then $C'_i$ goes counterclockwise
  on the the boundary of $F$.  The undirected cycle $C'_0\cup
  P[x_0,x_i]\cup C'_i \cup P[y_0,y_i]$ encloses $C'_{i-1}$, hence it
  encloses $F$ as well, which means that the undirected cycle
  $C'_i\cup P[x_i,y_i]$ does not enclose $F$. It follows that if the
  orientation of $C'_i$ is counterclockwise on the boundary of $F$,
  then it is clockwise on the undirected cycle $C'_i\cup P[x_i,y_i]$,
  i.e., same as the orientation of $C_i$.  Therefore, we have shown
  that the orientations of the chords $C'_0$, $\dots$, $C'_d$ are
  alternating, hence they form a $d$-bend.
  Moreover, as $P \cup C_0'$ does not enclose any terminals, the bend
  is terminal free.
\end{proof}

\begin{lemma}\label{lem:concentricbend}
  Let $C_0$, $\dots$, $C_{d+1}$ be an alternating sequence of free
  concentric cycles such that the outermost
  cycle $C_0$ does not enclose any terminals. If a path $P_b$ of the solution intersects $C_{d+1}$,
  then there is a terminal-free $d$-bend appearing on $P_b$.
\end{lemma}
\begin{figure}
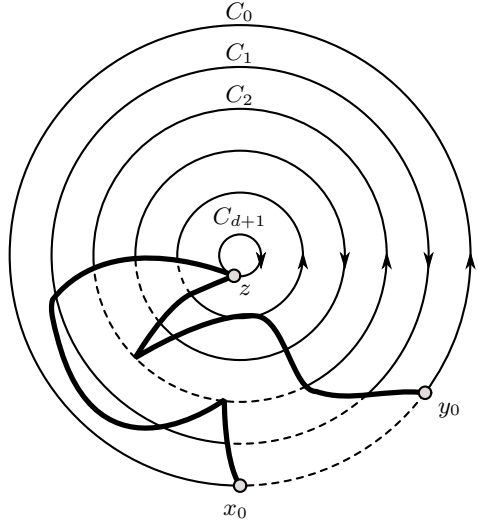

{\footnotesize 
\begin{center}
\svg{0.4\linewidth}{concentricbend}
\end{center}
}
\caption{Proof of Lemma~\ref{lem:concentricbend} via Lemma~\ref{lem:concentricbend2}.
  The dashed segments
  are the $i$-chords (note that that there are three 2-chords, sharing
  some endpoints).}\label{fig:concentricbend}
\end{figure}
\begin{proof}
  Let $z$ be an arbitrary vertex of $C_{d+1}\cap P_i$. As the
  endpoints of $P_b$ are not enclosed by $C_0$, both subpaths of $P$
  from $z$ to its endpoints have to intersect $C_0$. Let $x_0$ and
  $y_0$ be the intersections on these two subpaths closest to $z$. Let
  $P=P_b[x_0,y_0]$. Note that $P$ is enclosed by $C_0$ and $x_0,y_0$
  are the only vertices of $P$ on $C_0$.  Let $C'_0$ be one of the two
  subpaths of $C_0$ between $x_0$ and $y_0$ chosen arbitrarily; by our assumption, no
  internal vertex of $C'_0$ is on $P$. Moreover, as $C_0$ does not enclose any terminals,
  neither does $C_0' \cup P$. The claim follows from Lemma \ref{lem:concentricbend2}
  on $P$, $C_0'$ and cycles $C_0,C_1,\ldots,C_{d+1}$.
\end{proof}

\subsection{Paths and segments in bends}
In this section, we discuss how new bends can be formed from paths
appearing in a bend and how this can be used to obtain bounds on
bends. From now on, we fix $t$ and assume that
Lemma~\ref{lem:bendbound} holds for $t-1$, i.e., $f(k,t-1)$ is
defined.

The following lemma will be our main tool in constructing a new bend
whenever there is a path that starts on some chord, visits many
chords, and returns to the same chord. The proof is very similar to
the proof of Lemma~\ref{lem:concentricbend}, but note that here we
cannot assume that the path is enclosed by the bend $B$ (see
Figure~\ref{fig:bendinbend}).

\begin{lemma}\label{lem:bendinbendmain}
  Consider a solution $P_1$, $\dots$, $P_k$ and let
  $B=(P,C_0, \dots,C_{d})$ be a terminal-free $d$-bend appearing on
  some $P_{b}$. Let $P'$ be a subpath of $P_{b'}$ (possibly $b=b'$)
  that is vertex-disjoint from $P$ and going from $v_1\in C_x$ to
  $v_2\in C_x$ and intersecting $C_y$. Assume furthermore that the internal vertices of
  $C_x[v_1,v_2]$ are not on $P'\cup P_b$ and the cycle
  $C_x[v_1,v_2]\cup P'$ does not enclose any terminals. If $|x-y|\ge
  d'+2$, then there is a terminal-free $d'$-bend
  $B'=(P';C'_0,\dots,C'_{d'})$ with $C'_0=C_x[v_1,v_2]$.
\end{lemma}
\begin{figure}
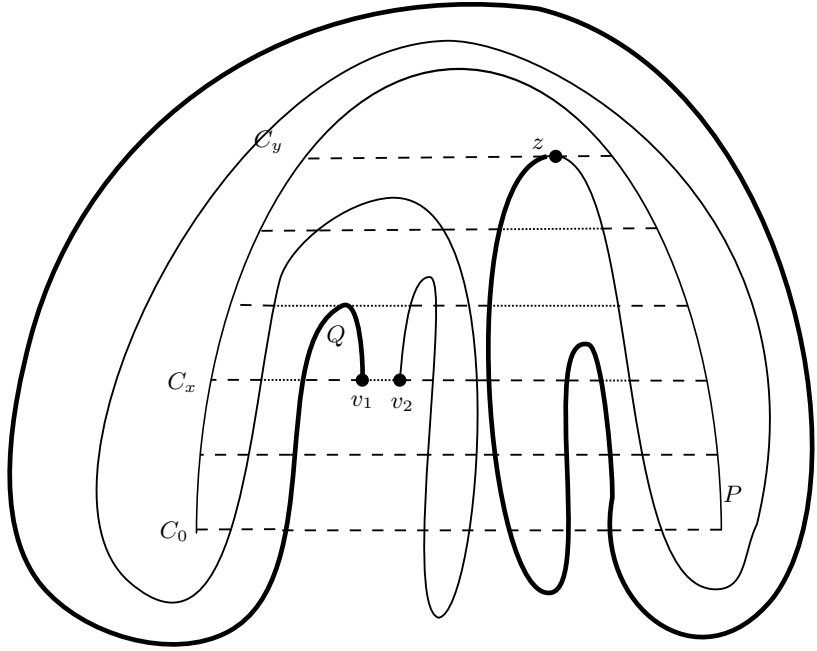

{\footnotesize 
\begin{center}
\svg{0.65\linewidth}{bendinbend}
\end{center}
}
\caption{Proof of Lemma~\ref{lem:bendinbendmain}. The dotted segments are the $i$-chords, and the strong line is $Q$.}\label{fig:bendinbend}
\end{figure}
\begin{proof}
  We claim that the cycle $C:=C_x[v_1,v_2]\cup P'$ does not enclose
  any vertex $v$ of $P$. Since the endpoints of $P_b$ (which are
  terminals) are not enclosed by $C$, the two subpaths of $P_b$ from
  $v$ to the endpoints of $P_b$ intersect the cycle at two distinct
  points. This is only possible if $P_b=P_{b'}$ and $v_1$, $v_2$ are
  the two intersection points. However, then path $P'$ and the subpath
  of $P_b$ going from $v_1$ to $v_2$ via $v$ forms a cycle, a
  contradiction. Therefore, no vertex of $P$ is enclosed by $C$; in
  particular, the endpoints of $C_i$ are not enclosed by $C$ for any
  $0 \le i \le d$.

  In the rest of the proof, we will consider only the subgraph
  containing only the $d$-bend $B$ and the path $P'$. It is clear that
  if we find the required $d'$-bend $B'$ in this subgraph, then it is
  a $d'$-bend of the original graph as well.

   Let $C^*_i=C_{x+i}$ if $y>x$ and let $C^*_i=C_{x-i}$ if $y<x$.  Let
   $z$ be an arbitrary vertex of $C_{y}\cap P'$ and let $Q$ be the
   subpath $P'[v_1,z]$. Clearly, $Q$ crosses each $C^*_i$ with $1 \le i
   \le d'+1$ an odd number of times. We have seen that $C$ does not
   enclose the endpoints of $C'_i$. Applying
   Lemma~\ref{lem:pathparity0} on the path $C^*_i$ and the subpath $Q$
   of cycle $C$, we get that for every $1\le i \le d'+1$, there is at
   least one subpath $C^*_i[w_1,w_2]$ with $w_1\in Q$, $w_2\not\in Q$,
   and no internal vertex on $P$. We will call such a subpath of
   $C^*_i$ an {\em $i$-chord.} In particular, $C^*_0[v_1,v_2]=C_x[v_1,v_2]$ is a 0-chord, but
   there could be other 0-chords as well. Note that these chords are
   internally vertex disjoint (but two $i$-chords can share an
   endpoint), thus they have a natural ordering along $Q$. Let us
   observe that if an $i$-chord and a $j$-chord are consecutive in
   this ordering, then $|i-j|\le 1$: this is because then there is a
   face whose boundary contains both chords and no such face is
   possible if $|i-j|>1$.

   Let $C'_i=C^*_i[x'_i,y'_i]$ be the first $i$-chord in this
   ordering. We claim that for every $1\le i \le d'$, the orientation
   of $C'_i$ (i.e., whether it goes from $x'_i$ to $y'_i$ or the other
   way around) and the orientation of $C'_{i+1}$ are opposite.  Let
   $K$ be the last chord before $C'_i$. We claim that $K$ is an
   $(i-1)$-chord. By our observation in the previous paragraph, it is
   either an $(i-1)$-chord, $i$-chord, or $(i+1)$-chord. As $C'_i$ is
   the first $i$-chord, $K$ cannot be an $i$-chord. If $K$ is an
   $(i+1)$-chord, then there has to be another $i$-chord between
   $C'_0$ and $K$ (as the difference is at most one between
   consecutive chords), again contradicting the choice of $C'_i$. Thus
   $K$ is an $(i-1)$-chord. Recall that we are considering only the
   subgraph consisting of the $d$-bend $B$ and the path
   $P'$. Therefore, there is a face $F$ whose boundary contains both
   $K$ and $C'_i$.  Observe that $C'_i$ and $C'_{i+1}$ have the same
   orientation along the the boundary of $F$ and this orientation
   depends on the parity of $i$. Note furthermore that  face $F$ is
   enclosed by the undirected cycle $C'_0\cup P'[x'_0,x'_{i}]\cup C'_{i}
   \cup P'[y'_i,y'_{0}]$ and it is {\em not} enclosed by the cycle
   $C'_i\cup P[x'_i,y'_i]$. Therefore, the orientation of $C'_i$ along
   the cycle $C'_i\cup P[x'_i,y'_i]$ is the opposite of its orientation
   along the boundary of the face $F$, that is, this orientation also
   depends on the parity of $i$.  This proves that the orientation of
   $C'_i$ and $C'_{i+1}$ are opposite for every $1\le i \le d'$.  If
   the orientation of $C'_0$ and $C'_1$ are also opposite (see Figure~\ref{fig:bendinbend2} for an example), then we get
   a $d'$-bend $B'=(P;C'_0,C'_1,\dots,C'_{d'})$ and we are done. If
   the orientation of $C'_0$ is the same as the orientation of $C'_1$
   (implying that it is the opposite of $C'_2$), then we get a
   $d'$-bend $B'=(P;C'_0,C'_2,C'_3,\dots,C'_{d'+1})$.
\end{proof}

\begin{figure}
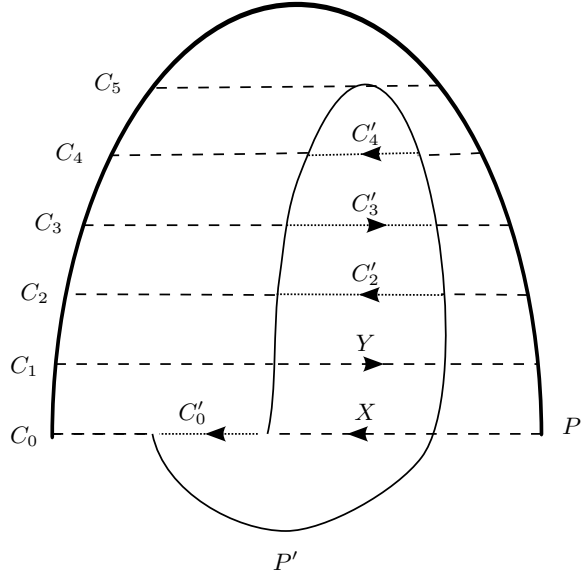

{\footnotesize 
\begin{center}
\svg{0.5\linewidth}{bendinbend2}
\end{center}
}
\caption{Lemma~\ref{lem:bendinbendmain}: An example where path $P'$ starts on $C_0$, intersects $C_5$, and there is a 3-bend $(P;C'_0,C'_2,C'_3,C'_4)$. Note that $X$ cannot be a chord as it intersects $C'_0$ and $Y$ cannot be the chord after $C'_0$ as it has the wrong orientation.}\label{fig:bendinbend2}
\end{figure}

Note that the example on Figure~\ref{fig:bendinbend2} that the
requirement $|x-y|\ge d'+2$ in Lemma~\ref{lem:bendinbendmain} is
tight: path $P'$ intersects $C_0$, $\dots$, $C_5$, but there is only a
3-bend enclosed by $P'$.

\begin{defin}
  Let $B=(P;C_0, \dots,C_d)$ be a $d$-bend in a solution $P_1$, $\dots$, $P_k$. A {\em segment} (with
  respect to $B$) is a subpath $S$ of some $P_b$ with endpoints on $C_0$, enclosed by
  $B$, and no internal vertex on $C_0$. A segment of path $Q$ is a
  segment $S$ with respect to $B$ that is a subpath of $Q$.  We say
  that a segment $S$ with endpoints $q_1,q_2\in C_0$ (strictly)
  encloses $X$ if the undirected cycle $S\cup C_x[q_1,q_2]$ (strictly) encloses
  $X$.  Two segments are {\em nested} if one encloses the other; a set
  of segments is nested if they are pairwise nested.  If segment
  $S_1$ encloses segment $S_2$, then we say that $X$ is {\em between}
  $S_1$ and $S_2$ if $X$ is enclosed by $S_1$ and no vertex of $X$ is
  strictly enclosed by $S_2$.
\end{defin}

In particular, $P$ is a segment of bend $B=(P;C_0, \dots,C_d)$. As the
chords are free, every segment is arc disjoint from every chord. Note
that two distinct segments of a path $P_b$ are not necessarily
disjoint: they can share endpoints. However, they cannot share both
endpoints (otherwise they would form a cycle in the path $P_b$) and
cannot share any internal vertices (since every internal vertex of the
path $P_b$ has exactly two arcs of $P_b$ incident to it).

The parity argument of Lemma~\ref{lem:pathparity0} can be stated with
a chord $C_i$ playing the role of $L$ and a segment playing the role of
the cycle $C$:
\begin{lemma}\label{lem:pathparity}
  Let $B=(P;C_0, \dots, C_d)$ be a $d$-bend, let $S$ be a segment, and
  let $Q$ be a subpath of $S$ with endpoints on $C_x$ and $C_y$. For
  every $x< i <y$, there are vertices $w_1,w_2\in S\cap C_i$ such that $w_1\in
  Q$, $w_2\not\in Q$, $C_i[w_1,w_2]$ is enclosed by $S$ and has no internal vertex on
  $S$, and $S[w_1,w_2]$ includes an endpoint of $Q$.
\end{lemma}
\begin{proof}
  Clearly, $Q$ crosses $C_i$ an odd number of times.  Let $v_1$ and
  $v_2$ be the endpoints of $S$ on $C_0$.  We apply
  Lemma~\ref{lem:pathparity0} on the line $L=C_i$ and the cycle
  $C=S\cup C_0[v_1,v_2]$.  As $w_1\in Q$ and $w_2\not\in Q$, the path
  $S[w_1,w_2]$ includes an endpoint of $Q$.
\end{proof}

In the following lemmas, we prove certain bounds on $d$-bends under
the assumption that the solution is unique.  Note that the following
proof and the proof of Lemma~\ref{lem:nestedrerouting} are the only
places where we directly use the fact that a solution is unique; all
the other arguments build on these two proofs.
\begin{lemma}\label{lem:bendinbendx}
  Consider a unique solution $P_1$, $\dots$, $P_k$ and let
  $(P;C_0, \dots,C_d)$ be a terminal-free $d$-bend of type $t$
  appearing on $P_b$. If no internal vertex of $C_0$ is on $P_b$, then
  $d\le f(k,t-1)+2$.
\end{lemma}
\begin{proof}
  Observe that $P$ is the only segment of $P_b$ with respect to $B$:
  as $C_0$ has no internal vertex on $P_b$, the endpoints of a segment
  $S$ of $P_b$ have to be $x_0$ and $y_0$, and therefore if $S$ and
  $P$ are different, they would form a cycle in $P_b$, a
  contradiction. This means that $C_d$ and $C_{d-1}$ have no internal
  vertex on $P_b$. Now for some $i\in \{d,d-1\}$, the orientation of
  $C_i$ is such that paths $P[x_0,x_i]$, $C_i$, $P[y_0,y_i]$ can be
  concatenated to obtain a directed path $P'$. If no internal vertex
  of $C_i$ is used by the solution, then we can replace $P$ by $P'$ in
  $P_b$ to obtain a new solution, contradicting the assumption that
  $P_1$, $\dots$, $P_k$ is unique. Thus $C_i$ is intersected by a
  segment $S$ of some $P_{b'}$ with $b\neq b'$. Now
  Lemma~\ref{lem:bendinbendmain} implies that there is a $d'$-bend
  $B'=(S;C'_0,\dots,C'_{d'})$ with $d'=i-2\ge d-3$ and $C'_0$ being
  the subpath of $C_0$ connecting the endpoints of $S$. Observe that
  $S\cup C'_0$ is enclosed by $B$, thus the type of $B'$ is at most
  the type of $B$.  In fact, the type of $B'$ is strictly smaller: as
  $S$ and $C'_0$ are disjoint from $P_b$, no vertex of $P_b$ is
  enclosed by $B$. By the induction hypothesis of
  Lemma~\ref{lem:bendbound}, this implies $d'<f(k,t-1)$ and therefore
  $d\le f(k,t-1)+2$ follows from $d'\ge d-3$.
\end{proof}

\begin{lemma}\label{lem:bendinbendxx}
  Consider a unique solution $P_1$, $\dots$, $P_k$ and let $B=(P,C_0,
  \dots,C_d)$ be a terminal-free $d$-bend of type $t$ appearing on some
  $P_{b}$. Let $Q$ be an subpath of  $P_{b}$ with endpoints $v_1,v_2\in
  C_x$ that intersects $C_y$. If $C_x[v_1,v_2]$
  has no internal vertex on $P_b$ and the cycle $C_x[v_1,v_2]\cup Q$
  does not enclose any terminals, then $|x-y|\le f(k,t-1)+4$ holds.
\end{lemma}
\begin{proof}
  If $C_x[v_1,v_2]$ has no internal vertex on $P_b$, then
  Lemma~\ref{lem:bendinbendmain} implies the existence of a
  terminal-free $d'$-bend $B'=(Q;C'_0,\dots,C'_{d'})$ with $d'\ge
  |x-y|-2$ and $C'_0=C_x[v_1,v_2]$. However, since $C'_0$ has no
  internal vertex on $P_b$, Lemma~\ref{lem:bendinbendx} implies $d'\le
  f(k,t-1)+2$ and the claim follows.
\end{proof}

\begin{lemma}\label{lem:bendinbend5}
  Consider a unique solution $P_1$, $\dots$, $P_k$ and let $(P;C_0,
  \dots,C_d)$ be a terminal-free $d$-bend of type $t$ appearing on
  some $P_b$. Let $S$ be a segment of $P_b$ that intersects $C_y$ for
  some $y\ge f(k,t-1)+6$.  Then $S$ encloses another segment of $P_b$
  that intersects $C_{y'}$ for $y'=y-f(k,t-1)-5 \ge 1$.
\end{lemma}
\begin{proof}
  Let us first verify the claim for the case when $P=S$. If no segment
  of $P_b$ intersects the internal vertices of $C_{y'}$, then
  $B'=(P[x_{y'},y_{y'}];C_{y'},\dots,C_d)$ is a $d'$-bend with $d'=d-y'\ge f(k,t-1)+5$
  such that no internal vertex of $C_{y'}$ is on $P_b$, which would
  contradict Lemma~\ref{lem:bendinbendx}. Thus in the following, we
  assume that $P\neq S$.

  Let $z$ be a vertex of $S\cap C_y$ and let $s_1,s_2$ be the
  endpoints of $S$ on $C_0$. Let $Q$ be the subpath $S[s_1,z]$. By
  Lemma~\ref{lem:pathparity}, there is a subpath $C_{y'}[w_1,w_2]$
  enclosed by $S$ and $w_1\in Q$, $w_2\in S\setminus Q$ and having no
  internal vertex on $S$ (note that $y'\ge 1$ and the endpoints of
  $C_{y'}$ are not enclosed by $S$ as $S\neq P$). This implies that
  $S[w_1,w_2]$ contains $z$ and hence intersects $C_{y}$.  If the
  internal vertices of $C_{y'}[w_1,w_2]$ are intersected by $P_b$,
  then we are done, as every such vertex is on a segment $S'$
  (different from $S$) enclosed by $S$. Otherwise, we apply
  Lemma~\ref{lem:bendinbendxx} on the segment $S[w_1,w_2]$ (note that
  $P\neq S$ implies that $S[w_1,w_2]$ is disjoint from $P$), and
  $y-y'\ge f(k,t-1)+5$ gives a contradiction.
\end{proof}

\subsection{Rerouting nested segments}
Our goal is to show that if we have a large nested set of segments
crossing many chords and having no other segment between them, then we
can simplify the solution by rerouting. However, it seems hard to
ensure the requirement that there are no other segments in between;
therefore, we weaken the requirement by asking that there is a large
area (intersecting many chords) such that there are no other
segments between our nested segments in this area (see
Figure~\ref{fig:rerouting}).
\begin{figure}
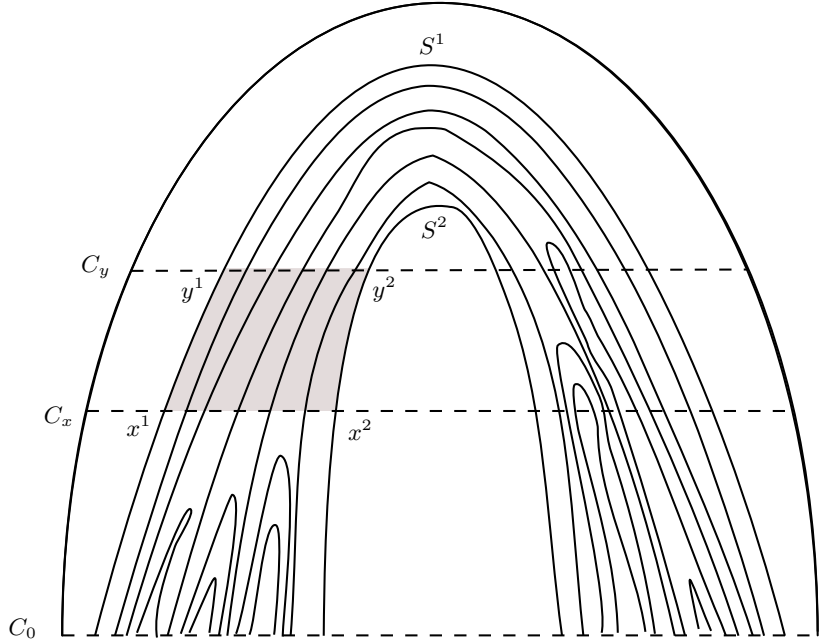

{\footnotesize 
\begin{center}
\svg{0.7\linewidth}{rerouting}
\end{center}
}
\caption{Condition of Lemma~\ref{lem:nestedrerouting}: the segments
intersecting $C_x[x_1,x_2]$ or $C_y[y_1,y_2]$ (i.e., segments entering the gray area) are nested.}\label{fig:rerouting}
\end{figure}

\begin{lemma}\label{lem:nestedrerouting}
  Consider a unique solution $P_1$, $\dots$, $P_k$. Let
  $B=(P;C_0, \dots,C_d)$ be a $d$-bend and let $S^1$ and $S^2$ be two
  nested segments with respect to $B$. Let $x^1,y^1\in S^1$ and
  $x^2,y^2\in S^2$ be vertices with $x^1,x^2\in C_{x}$ and $y^1,y^2\in
  C_y$ for some $y\ge x+2(2^{k}+1)$ such that subpaths $C^*_x:=C_x[x^1,x^2]$ and
  $C^*_y:=C_{y}[y^1,y^2]$ are between $S^1$ and $S^2$. If the set of segments
  intersecting $C^*_{x}$ or $C^*_{y}$ (including $S^1$
  and $S^2$) is nested between $S^1$ and $S^2$, then this set has size
  at most $2^k$.
\end{lemma}
The proof of Lemma~\ref{lem:nestedrerouting} requires two tools.  As
an important step in the proof of the undirected irrelevant vertex
rule of Adler et al.~\cite[Lemma 6]{isolde}, it is proved that if
there is a disc in the plane such that a solution of the $k$-disjoint
paths problem that contains nothing else than a set of $\ell > 2^k$
parallel paths, then one can ``redraw'' the solution by replacing the
part of the solution in the disc by introducing a set of strictly less
than $\ell$ new noncrossing edges inside the disc. Then if these new
``virtual'' edges can be actually realized by the edges of the graph
inside the disc, then it follows that we can modify the solution,
and hence it is not unique.

The following lemma is a directed analogue of the statement of Adler
et al.~\cite[Lemma 6]{isolde}. The same proof\footnote{The proof
  appears in the full version, which can be accessed at
  http://www.ii.uib.no/\textasciitilde
  daniello/papers/planarUniqueLinkage.}  works for the directed case;
in fact, the proof in \cite{isolde} first assigns an arbitrary
orientation to each path, thus it is clear that there is a correct
directed solution in the modified graph.
\begin{lemma}\label{lem:discrerouting}
  Let $G$ be a graph embedded in the plane with $k$ pairs of terminals
  with a solution $P_1$, $\dots$, $P_k$. Let $D$ be a closed disc not
  containing any terminals and suppose that $x_1$, $\dots$, $x_\ell$,
  $y_\ell$, $\dots$, $y_1$ appear on the boundary of $D$ (in this
  order) for some $\ell>2^k$. For every $1\le i \le \ell$, let $Q_i$
  be a directed path with endpoints $\{x_i,y_i\}$ (going in arbitrary
  direction) such that these paths are pairwise disjoint and contained
  in $D$. Suppose that every $Q_i$ appears as a subpath of the
  solution and the solution uses no other vertex in $D$. Let $V_S$ and
  $V_E$ be the starting points and end points of all the $Q_i$'s
  (clearly, $|V_S|=|V_E|=\ell$). Then there is a noncrossing matching
  of size strictly less than $\ell$ between $V_S$ and $V_E$ such that
  if $G'$ is the planar graph obtained by removing every arc in $D$
  and adding every $e_j$ oriented from $V_S$ to $V_E$, then there is a
  solution in $G'$.
\end{lemma}

To use Lemma~\ref{lem:discrerouting} to modify a solution, we have to
show that there are pairwise disjoint directed paths inside the disc
$D$ that correspond to the noncrossing matching. The existence of such
directed paths connecting vertices on the boundary of the disc can be
conveniently shown by the sufficient and necessary condition given by
Ding, Schrijver, and Seymour \cite{DBLP:journals/siamdm/DingSS92}, which will be our second tool
in the proof of Lemma~\ref{lem:nestedrerouting}. We review this
condition next.

Let $G$ be a directed graph embedded in the plane; for simplicity we
will treat only the case when the boundary of the infinite face is a
simple cycle, as it is already sufficient for our purposes. Let
$(r_1,s_1)$, $\dots$, $(r_k,s_k)$ be pairs of terminals on the
boundary such that all $2k$ terminals are different. We say that
$(r_i,s_i)$ and $(r_j,s_j)$ {\em cross,} i.e., they appear in the order
$r_i$, $r_j$, $s_i$, $s_j$ or in the $r_i$, $s_j$, $s_i$, $r_j$ on the
boundary of the infinite face. We say that the {\em noncrossing
  condition} holds if no two pairs cross. We say that the {\em cut
  condition} holds if the following is true for every curve $C$ of the
disc $D$ going from a point on the boundary of $D$ to another point on
the boundary of $D$, intersecting $G$ only in the vertices, and not
intersecting any terminal:
\begin{quote}
If $C$ separates $(r_{i_1},s_{i_1})$, $\dots$, $(r_{i_n},s_{i_n})$ in this order, then $C$ contains vertices $p_1$, $\dots$, $p_n$, in this order so that for each $j=1,\dots,n$:
\begin{itemize}
\item if $r_{i_j}$ is at the left-hand side of $C$, then at least one
  arc of $G$ is entering $C$ at $p_j$ from the left and at least one
  arc of $G$ is leaving $C$ at $p_j$ from the right.
\item if $r_{i_j}$ is at the right-hand side of $C$, then at least one arc of $G$ is entering $C$ at $p_j$ from the right and at least one arc of $G$ is leaving $C$ at $p_j$ from the left.
\end{itemize}
\end{quote}
If is easy to see that both the noncrossing and the cut conditions are necessary for the existence of pairwise vertex-disjoint paths connecting the terminal pairs. The result of Ding, Schrijver, and Seymour states that these two conditions are sufficient:
\begin{theorem}[\cite{DBLP:journals/siamdm/DingSS92}]\label{th:discrerouting}
  Let $(r_1,s_1)$, $\dots$, $(r_k,s_k)$ be pairs of terminals on the
  infinite face of an embedded planar graph. There exist $k$ pairwise
  vertex-disjoint directed paths $P_1$, $\dots$, $P_k$ such that $P_i$
  goes from $r_i$ to $s_i$ if and only if the noncrossing and the cut
  conditions hold.
\end{theorem}

We are now ready to prove Lemma~\ref{lem:nestedrerouting}. Let us
remark that this proof is the second and last place (after
Lemma~\ref{lem:bendinbendx}) where the uniqueness of the solution is
directly used.
\begin{proof}[Proof (of Lemma~\ref{lem:nestedrerouting})]
  Let $S_1$, $\dots$, $S_\ell$ be the segments nested between $S^1$
  and $S^2$ intersecting $C^*_x$ or $C^*_y$ in the order they are
  nested ($S_1=S^1$, $S_{\ell}=S^2$). If $\ell\le 2^k$, then we are
  done. Otherwise, we argue that it can be assumed that
  $\ell=2^k+1$. Let $S':=S_{2^k+1}$. Let $x'$ be the vertex of $S'$ on
  $C^*_x$ closest to $x^1$ one $C^*_x$ (such a vertex exists, as $S'$
  is nested between $S^1$ and $S^2$). Similarly, let $y'$ be the
  vertex of $S'$ on $C^*_y$ closest to $y^1$. Now we can replace
  $S^2$, $x^2$, $y^2$, with $S'$, $x'$, $y'$, respectively. Note that
  $C_x[x^1,x']$ is between $S^1$ and $S'$, hence only the segments
  $S_1$, $\dots$, $S_{2^k+1}=S'$ can intersect it, satisfying the
  conditions of the lemma being proved. Thus if we can arrive to a
  contradiction for $\ell=2^k+1$, the lemma follows.

\begin{figure}
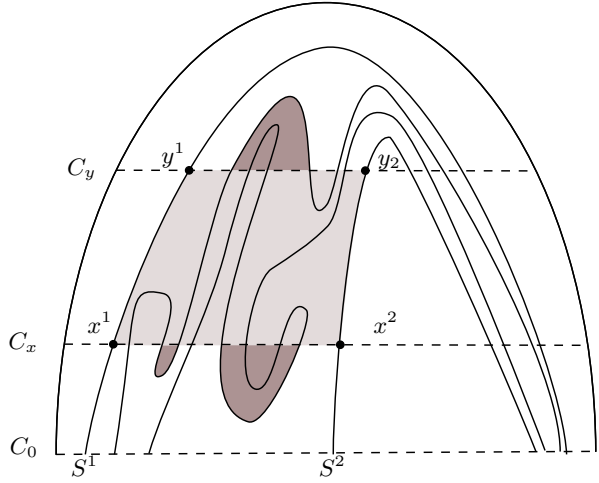

{\footnotesize 
\begin{center}
\svg{0.5\linewidth}{rerouting2}
\end{center}
}
\caption{Proof of Lemma~\ref{lem:nestedrerouting}. The figure shows
  two nested segments between $S^1$ and $S^2$. The light gray shows
  the disc $D_0$, which is extended to $D$ by the dark gray
  area.}\label{fig:rerouting2}
\end{figure}

  Let $D_0$ be the disc enclosed by the cycle $S^1[x^1,y^1]\cup C^*_x
  \cup S^2[x^2,y^2]\cup C^*_y$. It is tempting to try to apply
  Lemma~\ref{lem:discrerouting} on the disc $D_0$ and the subpaths of
  the $S_i$'s inside $D_0$, but as Figure~\ref{fig:rerouting2} shows, these subpaths are
  not necessarily parallel (and therefore Figure~\ref{fig:rerouting}
  gives only a simplified picture of what is happening inside the gray
  area). To avoid this difficulty, we extend $D_0$ to a disc $D$ the
  following way: for every subpath $Q$ of every $S_i$, if $Q$ has both
  endpoints on $C^*_x$ or both endpoints on $C^*_y$ and no other
  vertex in $D_0$, then we add to $D$ the disc enclosed by $Q$ and the
  subpath of $C^*_x$ or $C^*_y$ between the endpoints of $Q$. This
  process creates a disc $D$ whose boundary consists of the subpath
  $S^1[x^1,y^1]$, a path $B_x$ from $x^1$ to $x^2$, the subpath
  $S^2[x^2,y^2]$, and a path $B_y$ from $y^1$ to $y^2$.

\begin{claim}
Only the segments $S_1$, $\dots$, $S_\ell$ intersect
  $D$.
\end{claim}
\begin{proof}
Suppose that a segment $S$ contains a vertex $x$ in $D$.  If
  $x$ is in $D_0$, then it is clear that $S$ is part of the nested
  sequence of segments. Otherwise, $x$ is in $D$ because it is
  enclosed by a cycle formed by a subpath $Q$ of some $S_i$ and a
  subpath of $C^*_x$ or $C^*_y$.  Since every segment
  intersecting $C^*_x$ or $C^*_y$ is in the nested
  sequence by assumption, we can conclude that $S$ is also one of
  these nested segments.
\cqed\end{proof}

Among all intersections of $S_i$ with $C^*_x$ or $C^*_y$ consider
those two that are closest to the endpoints of $S_i$; one of these
intersections (denote it by $x_i$) has to be on $C^*_x$, the other
(denote it by $y_i$) has to be on $C^*_y$. 
Let $Q_i:=S_i[x_i,y_i]$.
\begin{claim}
For every $1\le i \le \ell$,
\begin{enumerate}
\item vertex $x_i$ is on $B_x$,
\item vertex $y_i$ is on $B_y$, and
\item $Q_i$ is enclosed by $D$.
\end{enumerate}
\end{claim}
\begin{proof}
If $x_i$ is not on the
boundary of $D$, then it is enclosed by a subpath $Q$ of some $S_j$
and a subpath of $C^*_x$, which contradicts the assumption that there
is a subpath of $S_i$ from $x_i$ to an endpoint of $S_i$ and not
intersecting $C^*_x\cup C^*_y$. Similarly, every $y_i$ is on the
boundary segment $B_y$ of $D$.

Suppose that a vertex $z$ of $Q_i$ is not enclosed by $D$ (and hence
by $D_0$). Let $Q'$ be the subpath of $Q_i$ containing $z$ with both
endpoints on the boundary of $D_0$ and no internal vertex enclosed by
$D_0$.  Then either both endpoints are on $C^*_x$ or both endpoints
are on $C^*_y$. In either case, the definition of $D$ would add the
disc enclosed by $Q'$ to $D$, contradiction that $z$ is not enclosed
by $D$.  \cqed
\end{proof}
It is clear that the solution uses no other
vertex of $D$ than the vertices of the $Q_i$'s: we have seen that
every vertex in $D$ belongs to some $S_i$ and this vertex of $S_i$ has
to be part of $Q_i$. Therefore, the conditions of
Lemma~\ref{lem:discrerouting} hold for $Q_1$, $\dots$, $Q_{\ell}$ and
we may assume the existence of the matching $M=\{e_1,\dots,
e_{\ell'}\}$ of size $\ell'<2^k+1$. We consider the arcs
$e_1$, $\dots$, $e_{\ell'}$ to be directed, with orientation as given
by Lemma~\ref{lem:discrerouting}.

\begin{claim}
There exist pairwise vertex-disjoint paths $R_1$,
  $\dots$, $R_{\ell'}$ enclosed by $D$ such that $R_j$ has the same
  start vertex $r_j$ and end vertex $s_j$ as the start and end vertex of $e_j$,
  respectively.
\end{claim}
\begin{proof}
We use Theorem~\ref{th:discrerouting} to prove the
  existence of these paths. The fact that the matching is noncrossing
  implies that the noncrossing condition holds. Suppose now that a
  curve $C$ in $D$ with endpoints on the boundary of $D$ violates the
  cut condition. If $C$ intersects both $B_x$ and $B_y$, then it has a subpath with an endpoint on $C^*_x$ and an
  endpoint on $C^*_y$. Therefore, $C$ crosses all the chords between
  $C_x$ and $C_y$, thus we can find the required vertices $p_1$,
  $\dots$, $p_{2^k+1}$ (note that $y\ge x+2(2^{k}+1)$).  Therefore, we
  may assume that $C$ is disjoint from $B_y$ (the case when $C$ has
  no endpoint on $B_x$ is similar).

  Without loss of generality, we may assume that $C$ separates the
  pairs $(r_1,s_1)$, $\dots$, $(r_q,s_q)$ only and in this order.  As
  $C$ is disjoint from $B_y$, there is a subpath $B'$ of the boundary
  of $D$ that connects the endpoints of $C$ and is disjoint from
  $B_y$; let $b_1,b_2$ be the endpoints of $B'$. Now $C$ separates the
  pair $(r_j,s_j)$ if and only if exactly one of $r_j$ and $s_j$ is on
  $B'$; denote this vertex by $\beta_j\in \{r_j,s_j\}$. Moreover, the
  order in which $C$ separates the pairs correspond to the order in
  which the vertices $\beta_j$ corresponding to the separated pairs
  appear on $B'$, that is, vertices $b_1$, $\beta_1$, $\dots$,
  $\beta_z$, $b_2$ appear in this order on $B'$. Each vertex $\beta_j$
  is a vertex $x_{i_j}$ for some $1 \le i_j \le 2^k+1$ and there is a
  corresponding path $Q_{i_j}$ that enters $\beta_j$ if $\beta_j=r_j$
  and leaves $\beta_j$ if $\beta_j=s_j$. Let $p_j$ to be the first
  intersection of $C$ with $Q_{i_j}$ (note that this cannot be
  $x_{i_j}$ or $y_{i_j}$, as $C$ does not intersect the
  terminals). The path $Q_{i_j}$ shows that $p_j$ has the required
  arcs entering and leaving, and therefore this sequence witnesses
  that $C$ does not violate the cut condition.  \cqed\end{proof}

As the matching $M$ was given by Lemma~\ref{lem:discrerouting}, if we
remove every arc enclosed by $D$ and add the arcs $e_1$, $\dots$,
$e_{\ell'}$, then there is a solution. Therefore, if we remove every
arc enclosed by $D$ except those that are on some $R_i$, then there
is still a solution. This solution is different from $P_1$, $\dots$,
$P_k$: the paths in the solution have at most $\ell'<2^k+1$ maximal
subpaths enclosed by $D$. This contradicts the assumption that $P_1$,
$\dots$, $P_k$ is a unique solution.
\end{proof}

\subsection{Rerouting in a bend}
Before beginning the main part of the inductive proof of
Lemma~\ref{lem:unique}, we need to introduce one more technical tool. We have to
be careful to avoid certain subpaths of the solution, as they enclose terminals and hence the inductive argument cannot be applied on them. The following definitions will be helpful in analyzing this situation:
\begin{defin}\label{def:bridgetree}
  Let $P_1$, $\dots$, $P_b$ be a solution, let $B=(P',C_0,\dots, C_d)$ be a $d$-bend appearing on a path $P_b$. Let
  $H$ be the undirected {\em $P_b$-graph} formed by the paths $P_b$ and $C_0$,
  $\dots$, $C_d$, with every degree-2 vertex suppressed. We call an edge of
  $H$ a {\em $P_b$-arc} if it corresponds to a subpath of $P_b$ or a {\em
    chord arc} if it corresponds to a subpath of some $C_i$.  A
  subpath of $C_i$ that corresponds to a chord arc (i.e., the
  endpoints of the subpath is on $P_b$ and the internal vertices are
  disjoint form $P_b$) is called a {\em $P_b$-bridge.} The {\em dual
    $P_b$-graph} is the dual $H^*$ of $H$. The subgraph $T^*$ of $H^*$
  containing the chord arcs is the {\em $P_b$-tree} of $B$.
\end{defin}
\begin{figure}
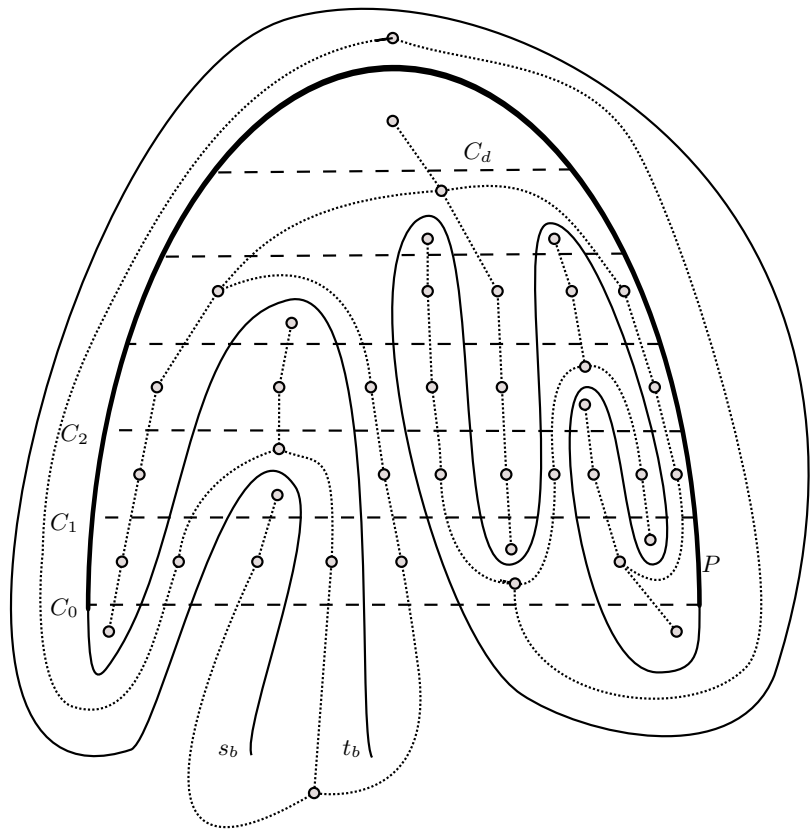

{\footnotesize 
\begin{center}
\svg{0.65\linewidth}{dual}
\end{center}
}
\caption{A $d$-bend $(P;C_0,\dots,C_d)$ appearing on a path $P_b$. The dotted lines show the arcs of the $P_b$-tree.}\label{fig:dual}
\end{figure}
Note that the graphs defined in Definition~\ref{def:bridgetree} are all undirected. The following lemma justifies the name $P_b$-tree:
\begin{lemma}\label{lem:bendtree}
  Let $B=(P',C_0,\dots, C_d)$ be a $d$-bend appearing on an a path $P_b$ of the solution.
   The $P_b$-tree $T^*$ of $B$ is a spanning tree of the dual $P_b$-graph.
\end{lemma}
\begin{proof}
  Suppose that there is a cycle $C^*$ in $T^*$. This cycle $C^*$
  corresponds to a cut $C$ in the primal graph $H$, thus removing the
  edges of $C$ disconnects the graph $H$. However, $P_b$ is a connected
  spanning subgraph of $H$ not containing any edge of the cut $C$, a
  contradiction. To see that $T^*$ is connected and spanning, suppose
  that the dual graph $H^*$ has a a minimal cut consisting only of
  $P_b$-edges. Then the primal graph $H$ contains a cycle consisting
  only of $P_b$ edges, meaning that there is cycle in $P_b$, a
  contradiction.
\end{proof}

A variant of the segment is the $j$-segment, which has its endpoints on $C_j$:

\begin{defin}
  Let $P_1$, $\dots$, $P_k$ be a solution and let $B=(P;C_0, \dots,C_d)$ be a $d$-bend.  A {\em 
    $j$-segment} is a subpath $Q$ of some $P_b$ with endpoints on $C_j$, no internal vertex on $C_j$, and enclosed by the cycle $C_j\cup P[x_j,y_j]$.
\end{defin}
Note that every $j$-segment $S_j$ is on a unique segment $S$, but
segment $S$ can contain multiple $j$-segments. Observe that if
$j$-segment $S_j$ is on segment $S$, then the subpath of $C_j$
connecting the endpoints of $S_j$ are not necessarily enclosed by $S$
(see Figure~\ref{fig:jsegment}). Therefore, to avoid confusion, we do not define
the notion of ``enclosing'' for $j$-segments.
\begin{figure}
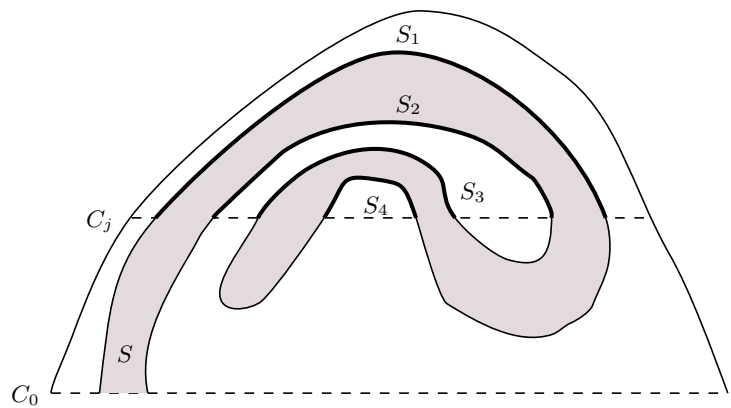

{\footnotesize
\begin{center}
\svg{0.6\linewidth}{jsegment}
\end{center}}
\caption{A segment $S$ in a bend and four $j$-segments $S_1$, $\dots$, $S_4$ on $S$.}\label{fig:jsegment}
\end{figure}

We need the following technical lemma in the proof of
Claim~\ref{cl:seeingsegments}. Intuitively, if $S$ is a $j$-segment
and $e$ is a $P_b$-bridge having an endpoint on $S$, then the
$P_b$-bridges having an endpoint on $S$ give two paths from $e$ to
$C_j$ in the $P_b$-tree $T^*$. For example, in
Figure~\ref{fig:twodisjoint}, one can see the two paths from
$C_i[\alpha_1,\beta_1]$ to $C_j$.  However, these two paths are not
necessarily edge disjoint: for example, in
Figure~\ref{fig:twodisjoint}, there are no two disjoint paths from
$C_i[\alpha_2,\beta_2]$ to $C_j$.  The following lemma shows that the
two disjoint paths always exist if the $P_b$-bridge $e$ has {\em
  exactly} one endpoint on $S$ (as it is the case with
$C_i[\alpha_1,\beta_1]$, but not with $C_i[\alpha_2,\beta_2]$).
\begin{figure}
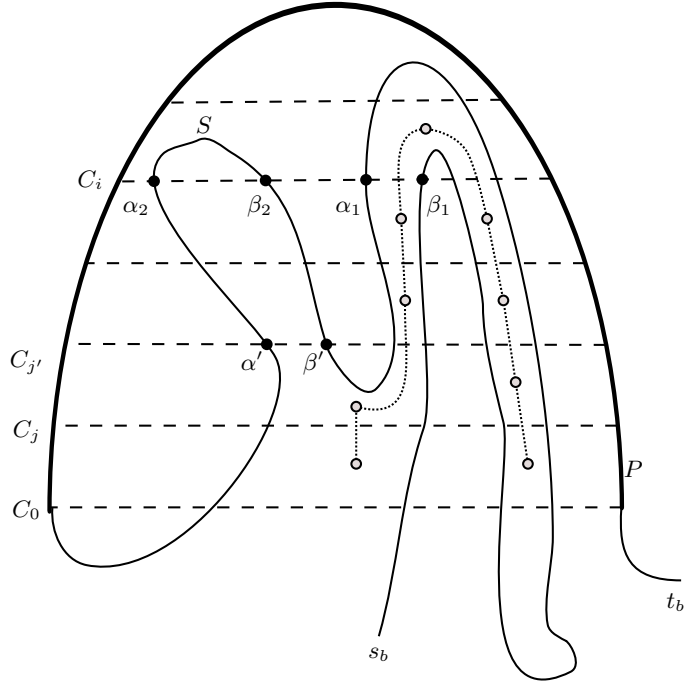

{\footnotesize 
\begin{center}
\svg{0.6\linewidth}{twodisjoint}
\end{center}
}
\caption{In $T^*_b$, there are two disjoint paths from
  $C_{i}[\alpha_1,\beta_1]$ to the faces between $C_{j}$ and $C_{j-1}$,
  but there are no two disjoint paths from $C_{i}[\alpha_2,\alpha_2]$, as
  $C_{i'}[\alpha',\beta']$
 is a separator.}
\label{fig:twodisjoint}
\end{figure}

\begin{lemma}\label{lem:segmenttwopaths}
  Let $B=(P;C_0, \dots, C_d)$ be a $d$-bend appearing on, let $S$ be a $j$-segment
  of some $P_b$, and let edge $e$ of the $P_b$-tree $T^*_b$ correspond
  to $P_b$-bridge $C_i[\alpha,\beta]$ having exactly one endpoint on
  $S$.  Let $T^*_S$ be the subgraph of $T^*_b$ containing edges
  corresponding to $P_b$-bridges having an endpoint on $S$. Then there
  are two edge-disjoint paths $Q_1,Q_2$ in $T^*_S$ such that the first
  vertex of each $Q_i$ is an endpoint of $e$ and its last vertex
  corresponds to a face of the $P_b$-graph between $C_j$ and
  $C_{j-1}$.
\end{lemma}
\begin{proof}
  Suppose without loss of generality that $\alpha\in S$.  Note that
  the endpoints of $e$ correspond to to faces of the $P_b$-graph, one
  between $C_i$ and $C_{i+1}$, the other between $C_i$ and $C_{i-1}$.
  By Menger's Theorem, if there are no such paths, then there is an
  edge $e'$ of $T^*_S$ separating the endpoints of $e$ from the faces
  between $C_j$ and $C_{j-1}$. Suppose that $e'$ corresponds to
  $P_b$-bridge $C_{j'}[\alpha',\beta']$ for some $\alpha',\beta'\in
  P_b$. Observe that both $\alpha'$ and $\beta'$ have to be on $S$,
  otherwise $S$ has a subpath from $\alpha$ to one of its endpoints on
  $C_j$ that is disjoint from $\{\alpha',\beta'\}$, and the
  $P_b$-bridges along this subpath of $S$ give a path in $T^*_S$
  avoiding $e'$, contradicting the assumption that $e'$ is a
  separator. By the same reason, $\alpha$ has to appear between
  $\alpha'$ and $\beta'$ on $S$, that is, $S[\alpha',\beta']$ contains
  $\alpha$.

  Let $C_S$ be the cycle formed by $S$ and the subpath of $C_j$
  connecting the endpoints of $S$.  Suppose first that $C_S$ encloses
  $\beta$.  Consider the cycle $C=C_{j'}[\alpha',\beta']\cup
  S[\alpha',\beta']$. This cycle encloses $\alpha$ (as $\alpha\in
  S[\alpha',\beta']$) but cannot enclose $\beta$, which is on a
  $j$-segment of $P_b$ different from $S$: $C_{j'}[\alpha',\beta']$ has
  no internal vertex on $P_b$. This is only possible if
  $C_{j'}[\alpha',\beta']$ is not enclosed by $C_S$. Thus the edges of
  $T^*_S$ corresponding to $P_b$-bridges enclosed by $C_S$ are
  disjoint from $e'$ and connect $e$ to the faces between $C_j$ and
  $C_{j-1}$.

  The argument is similar if $C_S$ does not enclose $\beta$. Again,
  $C=C_{j'}[\alpha',\beta']\cup S[\alpha',\beta']$ encloses $\alpha$
  but does not enclose $\beta$. This is only possible if
  $C_{j'}[\alpha',\beta']$ is enclosed by $C_S$. Thus the edges of
  $T^*_S$ corresponding to $P_b$-bridges not enclosed by $C_S$ are
  disjoint from $e'$ and connect $e$ to the faces between $C_j$ and
  $C_{j-1}$.
\end{proof}

We are now ready to start the main part of the proof.
\begin{lemma}
If $f(k,t-1)$ is defined, then $f(k,t)$ is defined.
\end{lemma}
\begin{proof}
We define the following constants:
\begin{align*}
s:=2^k+1 &&& m:=f(k,t-1)+5 &&& M:=40sm, &&& f(k,t):=M(2k+4).
\end{align*}
Note that this recursive definition of $f(k,t)$ implies that $f(k,t)=2^{O(kt)}$.
  Suppose that a terminal-free $d$-bend $B=(P;C_0,\dots,C_d)$ of type $t$ appears on path $P_b$ in a unique solution for some $d\ge f(k,t)$.
  Let $T^*_b$ be the $P_b$-tree as in Definition~\ref{def:bridgetree}. We define a set $F$ of special faces of the $P_b$-graph $H$ containing
\begin{itemize}
\item the infinite face,
\item faces strictly enclosing a terminal $s_{b'}$ or $t_{b'}$ for some $b'\neq b$,
\item the at most two faces whose boundaries contain   $s_b$ and $t_b$ (which are degree-1 vertices), and
\item the two faces whose boundary contains the arc of $P$ incident to $x_0$.
\end{itemize}
Note that $|F|\le 2k+3$. Let
  $T^*_0$ be the minimal subtree of $T^*_b$ containing every vertex that
  corresponds to a face in $F$. Let $X$ be the set of vertices of
  $T^*_b$ that have degree at least 3 in $T^*_0$. As $T^*_0$ has at most
  $2k+3$ leaves, we have $|X|\le 2k+1$.

Consider the faces of $H_b$ enclosed by the $d$-bend $B$. At most $2k+1$
of these faces correspond to elements of $X$, thus $d\ge f(k,t)$
implies that there is an $h>M$ such that no face corresponding to $X$
appears between $C_h$ and $C_{h-M}$ in the $d$-bend $B$. Let us fix
such an $h$ and let $h^*=h-30sm$.

  Let $\alpha$ and $\beta$ be two vertices of $P_b$ enclosed by $B$. We say
  that $\alpha$ {\em sees $\beta$ from the inside} if $\alpha$ and $\beta$ are both on the same $C_i$, vertex $\alpha$ is strictly
  enclosed by the segment of $P_b$ containing $\beta$, and the 
  subpath $C_i[\alpha,\beta]$ does not have any internal vertex on
  $P_b$. Note that this subpath $C_i[\alpha,\beta]$ may intersect some $P_{b'}$
  with $b'\neq b$.

\begin{figure}
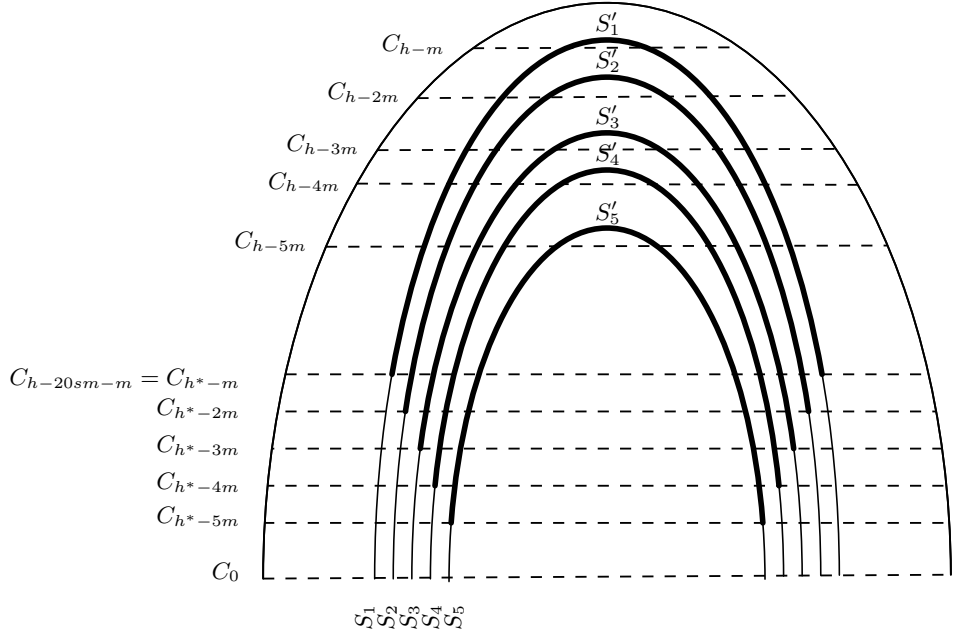

{\footnotesize 
\begin{center}
\svg{0.8\linewidth}{si-segments}
\end{center}
}
\caption{The segments in Claim~\ref{cl:seeingsegments}.}\label{fig:si-segments}
\end{figure}

First we find a sequence of nested segments $S_1$, $\dots$, $S_s$ of
$P_b$ (see Figure~\ref{fig:si-segments}) such that, in a weak but
precise technical sense, there are no further segments of $P_b$ between
them. What we show is that each segment has a subpath that is seen
from the inside only by vertices of the next segment in the sequence.
\begin{claim}\label{cl:seeingsegments}
  There are distinct nested segments $S_1$, $\dots$, $S_s$ of $P_b$ with respect to $B$ and an
  $(h^*-im)$-segment $S'_i$ of each $S_i$ such that
\begin{enumerate}
\item $S'_i$ does not intersect $C_h$.
\item $S'_i$ intersects $C_{h-im}$,
\item every vertex of $P_b$ seeing a vertex of $S'_i$ from the inside is on $S'_{i+1}$.
\end{enumerate}
\end{claim}
\begin{proof}
  There is at least one segment of $P_b$ intersecting $C_{h-m}$: $P$
  itself is such a segment.  Let $S_1$ be a segment of $P_b$
  intersecting $C_{h-m}$, but not enclosing any other segment of $P_b$
  intersecting $C_{h-m}$. It follows that $S_1$ does not intersect
  $C_h$: otherwise Lemma~\ref{lem:bendinbend5} implies that it
  encloses another segment intersecting
  $C_{h-f(k,t-1)-5}=C_{h-m}$. The segments $S_2$, $\dots$, $S_s$ we
  construct in the rest of the proof are all enclosed by $S_1$, thus
  they do not intersect $C_h$ either. Let $S'_1$ be an arbitrary
  $(h^*-m)$-segment of $S_1$ intersecting $C_{h-m}$.

  Suppose that we have constructed such a sequence up to $S_{i-1}$ and
  $S'_{i-1}$. We find $S_i$ and $S'_i$ the following way.  Again by
  Lemma~\ref{lem:bendinbend5},  there is a segment of $P_b$ enclosed by
  $S_{i-1}$ that intersects $C_{h-im}$, thus there is at least one
  $S_i$ enclosed by $S_{i-1}$. Therefore, there is at least one $(h^*-im)$-segment
  $S'_i$ intersecting $C_{h-im}$. We show that there is at
  most one such $(h^*-im)$-segment that contains vertices seeing
  $S'_{i-1}$ from inside, thus we can define $S'_{i}$ satisfying property (3).

  Suppose $(h^*-im)$-segments $S'$ and $S''$ of $P_b$ contain vertices
  $\alpha'$ and $\alpha''$ that see $\beta',\beta''\in S'_{i-1}$ from
  the inside, respectively. Let $S^*$ be the $(h^*-im)$-segment of
  $S'_{i-1}$ (note that $S'_{i-1}$ is a $(h^*-(i-1)m)$-segment, thus
  its $(h^*-im)$-segment is unique).  Let $L$ be the vertices of
  $T^*_b$ corresponding to faces between $C_{h^*-im}$ and
  $C_{h^*-im-1}$.

  Let $Z$ be the subtree of $T^*_b$ that corresponds to $P_b$-bridges
  enclosed by $S_{i-1}$ and having an endpoint on $S'_{i-1}$; clearly,
  $Z$ is connected.  The edges corresponding to $P_b$-bridges with an
  endpoint on $S^*$ contain a path $Q$ from $Z$ to $L$ (see Figure~\ref{fig:3way}).  By
  Lemma~\ref{lem:segmenttwopaths} applied on the $P_b$-bridge
  $\alpha'\beta'$ and the segment $S'$ (resp., $\alpha''\beta''$ and
  $S''$), we get two edge-disjoint paths $Q'_1,Q'_2$ (resp.,
  $Q''_1,Q''_2$) that go between $Z$ and $L$ and each edge of the
  paths corresponds to a $P_b$-bridge having an endpoint on $S'$
  (resp., $S''$).
\begin{figure}
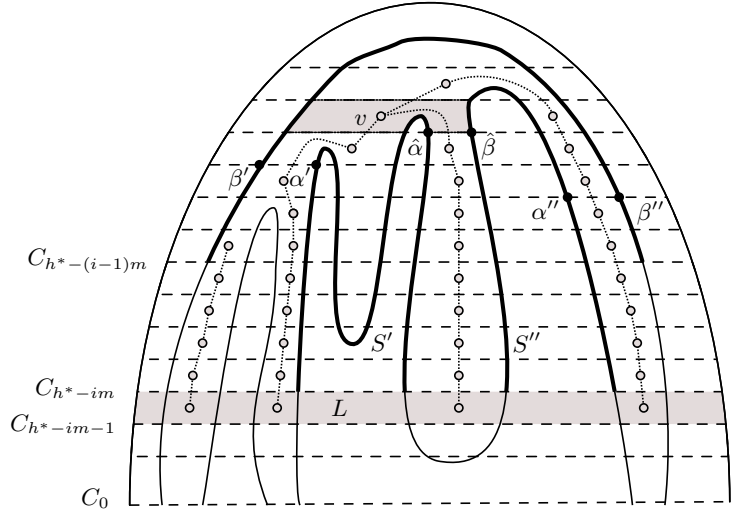

{\footnotesize 
\begin{center}
\svg{0.6\linewidth}{3way}
\end{center}
}
\caption{Proof of Claim~\ref{cl:seeingsegments}. In $T^*_b$, there are
  two disjoint paths from $\alpha'\beta'$ to $L$, and
two disjoint paths from  $\alpha''\beta''$ to $L$ using the edges shown on the figure. There
  is also a path from $Z$ to $L$ using edges that correspond to
  $P_b$-bridges having an endpoint on $S_{i-1}$. There are three
  edge-disjoint paths from vertex $v$ to $L$ in $T^*_b$, and we arrive to a
  contradiction using the path $P_b[\hat\alpha,\hat\beta]$ intersecting many chords but not
  enclosing any terminals.}
\label{fig:3way}
\end{figure}
  We show that in $T^*_b$ there are three edge-disjoint paths between $Z$
  and $L$. If there are no such paths, then Menger's Theorem implies
  that there are two edges $e_1$ and $e_2$ covering all such
  paths. Observe that each edge can be contained in at most two out of
  the five paths $Q$, $Q'_1$, $Q'_2$, $Q''_1$, $Q''_2$:
  otherwise the $P_b$-bridge corresponding to the edge would have an
  endpoint on all three of the $(h^*-im)$-segments $S^*$, $S'$, and
  $S''$, which is impossible. Therefore, one of these paths is
  disjoint from $e_1$ and $e_2$, a contradiction.  

  Let $Q_1$, $Q_2$, $Q_3$ be three edge-disjoint paths from $Z$ to
  $L$; we can assume that they start at (possibly not distinct)
  vertices $v_1$, $v_2$, $v_3$ of $Z$ and they contain no other
  vertices of $Z$. Let $Z'$ be the minimal subtree of $Z$ containing
  $v_1$, $v_2$, and $v_3$. This subtree $Z'$ has a vertex $v$
  (possibly $v\in \{v_1,v_2,v_3\}$) and three edge-disjoint paths
  $Z_1$, $Z_2$, $Z_3$ (possibly of length 0) where $Z_i$ goes from $v$
  to $v_i$. Now the concatenation of $Z_i$ and $Q_i$ for $i=1,2,3$
  gives three edge-disjoint paths from $v$ to $L$. As $v\in Z$, the length
  of the paths is at least $m= f(k,t-1)+5$. Let $\hat T_1,\hat
  T_2,\hat T_3$ be the components of $T^*_b\setminus v$ that contain these
  paths (minus $v$). We have chosen $h$ such that no face of $X$
  appears between $C_h$ and $C_{h-M}$; in particular, as $h^*-im>h-M$,
  vertex $v$ is not in $X$.  Therefore, it cannot happen that all
  three of $\hat T_1,\hat T_2,\hat T_3$ contain vertices from $T^*_0$:
  this would imply that $v\in T^*_0$ and has degree at least 3 in
  $T^*_0$, i.e., $v\in X$ by the definition of $X$. Suppose that $\hat
  T\in \{\hat T_1,\hat T_2,\hat T_3\}$ is disjoint from $T^*_0$ and
  let $\hat e$ be the edge connecting $\hat T$ and $v$.  Suppose that
  $\hat e$ corresponds to $P_b$-bridge $C_{\hat j}[\hat \alpha,\hat \beta]$. Note
  that $\hat \alpha, \hat \beta\not\in P$, as they are enclosed by
  $S_{i-1}$ (which is different from $P$).

  We would like to invoke Lemma~\ref{lem:bendinbendxx} with
  $Q=P_b[\hat \alpha,\hat \beta]$ to arrive to a contradiction, but we
  need to verify the conditions that no terminal is enclosed and this
  path is disjoint from $P$. This is the part of the proof where the
  definition of $T^*_0$ and the fact that $\hat T$ is disjoint from
  $T^*_0$ comes into play.  Consider the cycle $C$ formed by $\hat
  \alpha \hat \beta$ and $P_b[\hat \alpha, \hat \beta]$. Removing the
  edge $\hat e$ from $T^*$ splits $T^*$ into two parts, one of which
  is $\hat T$. The cycle $C$ encloses the faces corresponding to one
  of these two parts; more precisely, it encloses the part that does
  not contain the infinite face.  As $\hat T$ is disjoint from
  $T^*_0$, it cannot contain the infinite face, thus $C$ encloses
  exactly the faces of $\hat T$. This means that $C$ does not enclose
  any face of $T^*_0$ and hence does not enclose any
  terminals. Moreover, we claim that $C$ is disjoint from $P$. As
  $\hat \alpha,\hat \beta\not\in P$, if $P_b[\hat \alpha,\hat \beta]$
  contains a vertex of $P$, then it fully contains $P$, including the
  arc of $P$ incident to $x_0$. Both faces incident to this arc is in
  $T^*_0$, thus if $C$ contains this arc, then $C$ encloses a face of
  $T^*_0$, a contradiction. Thus $P_b[\hat \alpha,\hat \beta]$ is
  disjoint from $P$. The tree $\hat T$ contains a vertex of $L$ (since
  it contains one of the three paths $Q_1$, $Q_2$, $Q_3$, minus $v$),
  which means that $C$ encloses an arc of $C_{h^*-im}$. Therefore,
  $P_b[\hat \alpha,\hat \beta]$ intersects $C_{h^*-im}$ and $\hat
  \alpha,\hat \beta$ has no internal vertex on $P_b$. Thus we arrive
  to a contradiction by Lemma~\ref{lem:bendinbendxx}.
  \cqed\end{proof} Note that the statement of
Claim~\ref{cl:seeingsegments} is somewhat delicate. It
does not claim that every vertex of $S_i$ on $C_j$ for some $j\ge
h^*-im$ is only seen from inside only by $S'_{i+1}$; it claims this
only for a one specific $(h^*-im)$-segment $S'_i$ of $S_i$. Also, the
$(h^*-(i+1)m)$-segment $S'_{i+1}$ could contain more than one
$(h^*-im)$-segments (see Figure~\ref{fig:si-segments2}).
\begin{figure}
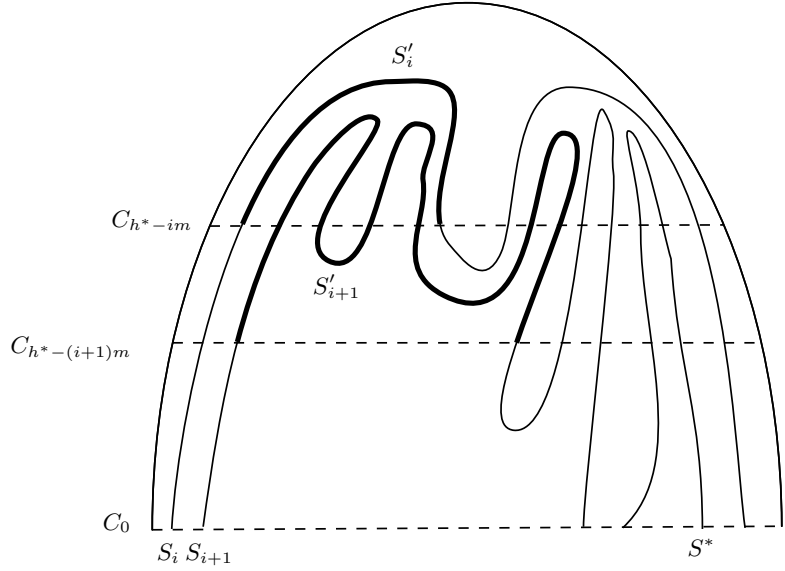

{\footnotesize 
\begin{center}
\svg{0.65\linewidth}{si-segments2}
\end{center}
}
\caption{A more complicated example of the segments in Claim~\ref{cl:seeingsegments}. Segment $S_i$ has an $(h^*-im)$-segment different from $S'_i$, which is seen from inside by vertices on a segment $S^*$ different from $S_{i+1}$, and also by vertices on a $(h^*-(i+1)m)$-segment of $S_{i+1}$ different from $S'_{i+1}$.}\label{fig:si-segments2}
\end{figure}

Not all vertices of $S'_i$ are seen from inside by the vertices of
$S'_{i+1}$: for example, if for some $v_1,v_2\in S'_i\cap C_j$, the
path $C_j[v_1,v_2]$ is enclosed by $S_i$ and does not have any
internal vertex on $P_b$, then $v_1$ and $v_2$ are not seen from
inside by any vertex of $S_{i+1}$. Nevertheless, the following claim shows
that if a subpath of $S'_i$ intersects many chords, then it contains a
vertex seen from inside by a vertex of $S'_{i+1}$.
\begin{claim}\label{cl:longseen}
  Let $Q$ be a subpath of $S'_i$ having an endpoint on $C_x$ and an
  endpoint on $C_y$. If $x+f(k,t-1)+5 \le j \le y-f(k,t-1)-5$ holds
  for some $j$, then there exists at least one vertex of $Q$ on $C_j$ that is seen from inside by a
  vertex of $S'_{i+1}$.
\end{claim}
\begin{proof}
  By Lemma~\ref{lem:pathparity}, there are vertices $w_1,w_2\in
  C_j\cap S_i$ such that subpath $C_j[w_1,w_2]$ is enclosed by $S_i$
  and $S_i[w_1,w_2]$ contains an endpoint of $Q$, which implies that
  $S_i[w_1,w_2]$ intersects $C_x$ or $C_y$. Moreover, the internal
  vertices of $C_j[w_1,w_2]$ are not on $S_i$. Note that
  $S_i[w_1,w_2]$ is disjoint from $P$. As $|x-i|,|y-i|\ge f(k,t-1)+5$,
  Lemma~\ref{lem:bendinbendxx} implies that a segment of $P_b$
  different from $S_i$ intersects $C_j[w_1,w_2]$. Let $z$ be the
  vertex of $C_j[w_1,w_2]$ closest to $w_1$ (but different from $w_1$)
  that is in $P_b$. As the internal vertices of $C_j[w_1,w_2]$ are not
  on $S_i$, vertex $z$ is not on $S_i$. Now $z$ sees $w_1$ from inside,
  hence $z$ is on $S'_{i+1}$.  \cqed\end{proof}

We would like to say that a subpath of some $C_j$ between $S'_1$ and
$S'_s$ intersects no other segment of $P_b$ than $S_1$, $\dots$,
$S_s$. This is not completely trivial, as we can use the third
property of Claim~\ref{cl:seeingsegments} for a vertex of $S_i$ only
if we show that it is on $S'_i$ as
well. Claims~\ref{cl:crosspath}--\ref{cl:pathexists} provide such
paths.

\begin{claim}\label{cl:crosspath}
  Let $x$ be a vertex of $S'_1\cap C_j$ for some $j\ge h^*+2sm$. Then
  for every $1\le i \le s$, there is an undirected path $W_i$ from $x$
  to a vertex $y\in S'_i$ such that
\begin{itemize}
\item $W_i$ is between $S_1$ and $S_i$,
\item $W_i$ is between $C_{j-2im}$ and $C_{j+2im}$,
\item $W_i$ does not intersect any segment of $P_b$ different from $S_1$, $\dots$, $S_s$.
\end{itemize}
\end{claim}
\begin{proof}
  The proof is by induction on $i$.  A path $W_1$ consisting of only
  $x=y_1$ shows that the statement is true for $i=1$.  Suppose that
  $y_i$ is on $C_{j_i}$. Let $q$ be a vertex of $S'_i$ either on
  $C_{j_i-2m}$ or $C_{j_i+2m}$ such that $S'_i[y_i,q]$ has no internal
  vertex on either $C_{j_i-2m}$ or $C_{j_i+2m}$ (as the endpoints of
  $S'_i$ are on $C_{h^*-im}$ and $j_i-2m \ge j-2im-2m \ge
  h^*+2sm-2im-2m \ge h^*-im$ holds, such a vertex $q$ has to
  exists). Suppose therefore that $q$ is on $C_{j_{i+1}}$, where
  $j_{i+1}$ is either $j_{i}-2m$ or $j_{i}+2m$.  Let
  $Q_i=S'_i[y_i,q]$.  By Claim~\ref{cl:longseen}, there is a vertex
  $z$ of $Q_i$ on $C_{j_{i+1}}$ seen from inside by a vertex
  $y_{i+1}\in S'_{i+1}$. Appending $Q_i[y_i,z]$ and the subpath of
  $C_{j_{i+1}}$ between $z$ and $y_{i+1}$ to the path $W_i$ gives the
  required undirected path $W_{i+1}$ ending at $y_{i+1}$. If path
  $W_i$ is between $C_{j-2im}$ and $C_{j+2im}$, then $W_{i+1}$ is
  between $C_{j-2(i+1)m}$ and $C_{j+2(i+1)m}$.  \cqed\end{proof}

\begin{claim}\label{cl:pathallnested}
  For some $j\ge h^*+4sm+1$, let $C^*=C_j[v_1,v_2]$ be a subpath
  between $S_1$ and $S_s$ with $v_1\in S'_1$, $v_2\in S_s$, and having
  no internal vertex on $S_1$ or $S_s$. Then $S_1$, $\dots$, $S_s$ are
  the only segments of $P_b$ intersecting $C^*$.
\end{claim}
\begin{proof}
  Let $S$ be a segment intersecting $C^*$, which means that
  $S$ is enclosed by $S_1$. 
  The path $C^*$ splits the area between $S_1$ and $S_s$ into two
  regions. At least one of these two regions contains an endpoint of
  segment $S$; let $R^*$ be such a region and let $S'$ be a subpath of
  $S$ in this region between $C_0$ and $C^*$.
Let $R^*$ be enclosed by subpaths $S^*_1$ of $S_1$, subpath $S^*_s$ of
$S_s$, subpath $C^*_0$ of $C_0$, and $C^*$.

Let $u$ be a vertex of $S^*_1$ on $C_{j-2sm-1}$ closest to $v_1$; as
$v_1\in S'_1$ and $j-2sm-1\ge h^*+2sm$, we have $u\in S'_1$. Let $W_s$
be the undirected path given by Claim~\ref{cl:crosspath}. As path
$W_s$ connects $u$ to $S_s$, contained between $S_1$ and $S_s$, and
contained also between $C_{j-4sm-1}$ and $C_{j-1}$, it cannot
intersect $C^*$ and hence it is in the region $R^*$. Now path $W_s$
separates $C^*$ and $C^*_0$ in $R^*$, thus $W_s$ intersects
$S'$. Since $S_1$, $\dots$, $S_s$ are the only segments of $P_s$ that
$W_s$ intersects, it follows that $S$ is one of these segments.
\cqed\end{proof}

\begin{claim}\label{cl:pathexists}
  Let $Q$ be a subpath of $S'_1$ from a vertex of $C_x$ to a vertex of
  $C_y$.  Suppose that $x+2sm < j < y-2sm$ and $j\ge h^*+4sm+1$ holds for some $j$. Then there
  is a vertex $\alpha\in Q\cap C_j$ and vertex $\beta\in S_s\cap C_j$
  such that the subpath $C_j[\alpha,\beta]$ is between $S_1$ and
  $S_s$, has no internal vertex on $S_1$ and $S_s$, and intersects no
  segment of $P_b$ other than $S_1$, $\dots$, $S_s$.
\end{claim}
\begin{proof}
  Consider the subpath $C_j[w_1,w_2]$ given by
  Lemma~\ref{lem:pathparity} with $w_1,w_2\in S'_1$. If it intersects
  $S_s$, then let $\beta$ be the vertex of $S_s$ closest to
  $\alpha:=w_1$ on this subpath and we are done by
  Claim~\ref{cl:pathallnested}. Otherwise, we arrive to a
  contradiction as follows. The path $S_1[w_1,w_2]$ contains an
  endpoint of $Q$, hence it contains a vertex $z$ that is either on
  $C_x$ or on $C_y$. Applying Claim~\ref{cl:crosspath} on this vertex
  $z$ gives a path $W_s$ from $z$ to $S_s$ and enclosed by $S_1$.
  Path $W_s$ cannot intersect $C_j[w_1,w_2]$ as $|x-j|,|y-j|>
  2sm$. Thus $W_s$ is enclosed by the cycle $S_1[w_1,w_2]\cup
  C_j[w_1,w_2]$. As $W$ has an endpoint on $S_s$, this contradicts the
  assumption that $C_j[w_1,w_2]$ does not intersect $S_s$.
  \cqed\end{proof}

We are now ready to find the area required by Lemma~\ref{lem:nestedrerouting}, where all the segments are nested.
Let $j_i=(h^*+4sm+1)+2smi$ for $i=0,\dots, 8$ (note that $j_i\le h^*+20sm+1<h-m$ for every
such $j_i$). As $S'_1$ is an $(h^*-m)$-segment and intersects
$C_{h-m}$, we can choose vertices $v_0$, $\dots$, $v_8$ appearing on
$S'_1$ in this order such that $v_j$ is on $C_{j_i}$. For $q=1,3,5,7$,
Claim~\ref{cl:pathexists} gives a subpath
$C^*_q:=C_{j_q}[\alpha_q,\beta_q]$ between $S_1$ and $S_s$ with
$\alpha_q$ being an internal vertex of $S'_1[v_{q-1},v_{q+1}]$. This
means that $C^*_1$, $C^*_3$, $C^*_5$, $C^*_7$ are distinct and connect
$S_1$ and $S_s$ in this order.

\begin{claim}\label{cl:middleintersect}
Every segment $S$ intersecting $C^*_3$ or $C^*_5$ is nested
between $S_1$ and $S_s$.
\end{claim}
\begin{proof}
  By Claim~\ref{cl:pathexists}, if $S$ is a segment of $P_b$, then it
  is one of $S_1$, $\dots$, $S_s$, hence the claim is certainly true.
  Consider now a segment $S$ of $P_{b'}$ for some $b'\neq b$ and let
  $Q$ be a subpath of $S$ from one of its endpoints on $C_0$ to a
  vertex $z$ of $C^*_3$ or $C^*_5$. Clearly, $Q$ has to intersect
  either $C^*_1$ or $C^*_7$ (see Figure~\ref{fig:final}). Let us assume that $S$ intersects $C^*_1$
  (the case when $S$ intersects $C^*_7$ is similar).

\begin{figure}
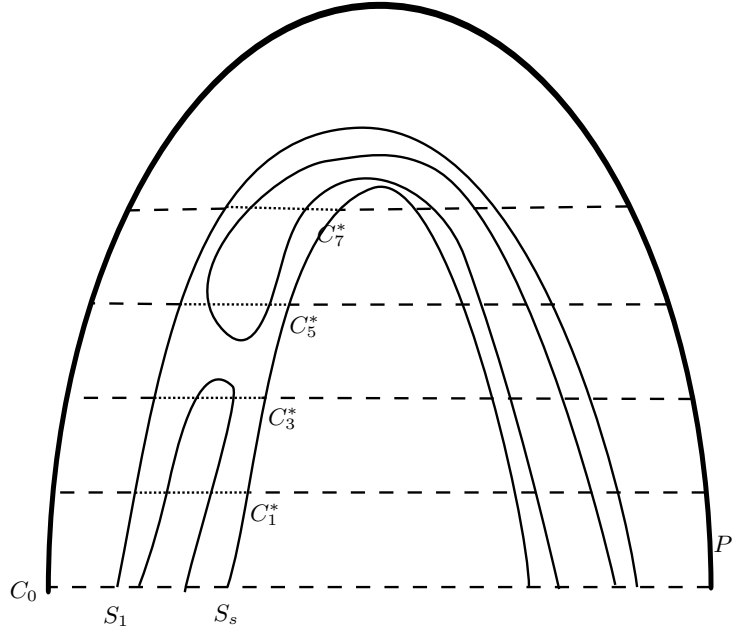

{\footnotesize 
\begin{center}
\svg{0.65\linewidth}{final}
\end{center}
}
\caption{Claim \ref{cl:middleintersect}: The segments $C^*_1$, $C^*_3$, $C^*_5$, $C^*_7$ connecting $S_1$ and $S_s$, and two additional segments 
that intersect $C^*_3$ or $C^*_5$.}\label{fig:final}
\end{figure}

Assume by contradiction that $S$ is not nested between $S_1$ and $S_s$.
Let $w_1$ and $w_2$ be the endpoints of $C^*_1$; 
  clearly, $S$ does not enclose $w_1$ and $w_2$ (as they are on
  $S_1$ and $S_s$, respectively).  Let us use Lemma~\ref{lem:pathparity0},
  on the line $C^*_1[w_1,w_2]$, on the cycle formed by $S$ and the subpath
  of $C_0$ connecting the endpoints of $S$, and on the subpath $Q$ of
  $S$ connecting $C_0$ and $z$ while intersecting $C^*_1$. We get a
  subpath $S[w_1,w_2]$ containing $z$ such that $w_1,w_2\in C^*_1$ and
  $C^*_1[w_1,w_2]$ contains no vertex of $S$. Let us observe that the
  internal vertices of $C^*_1[w_1,w_2]$ are not on any segment of
  $P_b$: any such segment would be enclosed by $S$, thus it is in
  contradiction with Claim~\ref{cl:pathexists}, which states that this
  segment has to be one of $S_1$, $\dots$, $S_s$.  Since $S[w_1,w_2]$
  contains $z$, which is on $C^*_3$ or $C^*_5$, the path $S[w_1,w_2]$
  intersects both $C^*_1$ and $C^*_3$. By
  Lemma~\ref{lem:bendinbendmain}, there is an $m$-bend
  $(S[w_1,w_2],C'_0,\dots,C'_m)$, with $C'_0=C^*_1[w_1,w_2]$.  The
  type of this bend is at most $t-1$ (as both $S[w_1,w_2]$ and
  $C^*[w_1,w_2]$ are disjoint from $P_b$), which contradicts $m>
  f(k,t-1)$.  \cqed \end{proof}

We have shown that the conditions of Lemma~\ref{lem:nestedrerouting}
hold for $S^1:=S_1$, $S^2:=S_s$, $C_x[x^1,x^2]:=C^*_3$, and
$C_y[y^1,y^2]:=C^*_5$. Thus there are less than $2^k=s$ segments
nested between $S_1$ and $S_s$, contradicting the existence of the
sequence $S_1$, $\dots$, $S_s$.
\end{proof}


\section{Min-max theorems for paths, cycles, and cuts}\label{sec:spiral}

\subsection{Framework}

In this section we consider graphs embedded on surfaces. By abusing the notation, we identify the graph with its image in the embedding.

\begin{defin}
\label{def:non-degenerate}
Let $G$ be a digraph embedded on a compact surface $\Sigma$. A directed curve $N$ on $\Sigma$ is called {\em{non-degenerate}}, if $N$ intersects embedding of $G$ in a finite number of points and for each intersection point $x$ of $N$ and $G$ there exists a neighbourhood $U_x$ of $x$ such that $N\cap U_x$ separates some non-empty subsets of $(G\cap U_x)\setminus N$ in $U_x$. For a non-degenerate curve $N$, let 
\begin{itemize}
\item $\vnoose(N)=(x_1,x_2,\ldots,x_p)$ be the sequence of vertices and edges through which $N$ passes, in their order of appearance on $N$; 
\item $\snoose(N)=(S_1,S_2,\ldots,S_p)$ be the sequence of subsets of $\{-1,+1\}$, defined as follows:
\begin{itemize} 
\item if $x_i$ is a vertex, then $-1\in S_i$ if there exists an arc entering $x_i$ from the left of $N$ and an arc leaving $x_i$ to the right of $N$, and $+1\in S_i$ if there exists an arc entering $x_i$ from the right of $N$ and an arc leaving $x_i$ to the left of $N$;
\item if $x_i$ is an edge, then $S_i=\{-1\}$ if the $x_i$ traverses $N$ from the left to right, and $S_i=\{+1\}$ if $x_i$ traverses $N$ from right to left.
\end{itemize}
\end{itemize}
\end{defin}

Intuitively a directed curve is non-degenerate if it does not touch an edge (or vertex) and return to the same face.
We point out that a non-degenerate curve is not necessarily non-self-crossing, it is just a smooth, regular image of an interval $[0,1]$. A curve is called {\emph{simple}} if it visits every vertex, arc, and face of $G$ at most once; observe that every simple curve is also non-degenerate. We often consider closed curves, that is, smooth and regular images of a circle, and call such a curve a {\emph{noose}}. The sequences $\vnoose$ and $\snoose$ are defined in the same manner in this situation; note that they are unique modulo cyclic shifts.

From now on we assume that all the considered curves are non-degenerate
(with a single exception of spiral cuts defined in Section~\ref{sec:bundles-and-bundle-words});
hence we ignore stating this attribute explicitely.

When we consider a curve or a noose in our algorithms, we may represent it as a sequence consisting of alternately vertices or edges and faces which the curve traverses. However, for some proofs it will be useful to imagine the curve as an actual topological object being an image of a circle or a closed interval.

\begin{defin}
For a sequence $(S_1,S_2,\ldots,S_p)$ where $S_i\subseteq \{-1,+1\}$ we say that a sequence $(s_1,\ldots,s_q)$ is {\emph{embeddable}} into $(S_1,S_2,\ldots,S_p)$ if there exists an increasing function $\iota:[q]\to [p]$ such that $s_i\in S_{\iota(i)}$. In this case, function $\iota$ is called an {\emph{embedding}}.
\end{defin}

We are now ready to state the following Theorem of Ding, Schrijver, and Seymour~\cite{dss:green-line}.

\begin{theorem}[\cite{dss:green-line}]\label{thm:dss}
Let $G$ be a digraph embedded on a torus $\Sigma$, and let $C_1,C_2,\ldots,C_k$ be closed, non-crossing directed curves on $\Sigma$ of homotopies $\{(0,c_i)\}$, where $c_i=\pm 1$, located in this order on the torus. Then one can find vertex-disjoint directed cycles $D_1,D_2,\ldots,D_k$ in $G$ homotopic to $C_1,C_2,\ldots,C_k$ if and only if there does not exist a noose $N$ with the following property: if $(p,q)$ is the homotopy of $N$, then $p\geq 0$ and no cyclic shift of $(c_1,c_2,\ldots,c_k)^p$ is embeddable into $\snoose(N)$.
\end{theorem}

Let us remark that by {\emph{homotopic}} we mean that there exists a {\emph{continuous shift}} of $\Sigma$ that simultaneously shifts all the cycles $C_1,C_2,\ldots,C_k$ to $D_1,D_2,\ldots,D_k$, i.e., a continuous map $h:\Sigma\times [0,1]\to \Sigma$ such that $\Sigma(\cdot,0)$ is identity and $\sigma(\cdot,1)$ maps every cycles $C_i$ to the corresponding cycle $D_i$.

In the original paper of Ding, Schrijver, and Seymour~\cite{dss:green-line}, the noose $N$ was required only to traverse vertices and faces, that is, it was a face-vertex curve. However, we may relax the requirement on the curve to just being non-degenerate, as after relaxation the assumed object still remains a counterexample, and this relaxation simplifies notation in the further parts of the article.

Let $G$ be a digraph embedded into a ring $R$ with outer face $C_1$ and inner face $C_2$. We choose an arbitrary curve $W$ in $R$ that (i) avoids $V(G)$, (ii) crosses $E(G)$ in a finite number of points, (iii) connects $C_1$ and $C_2$, and (iv) is directed from $C_1$ to $C_2$, as the {\emph{reference curve}}. For every path $P$ connecting $C_1$ and $C_2$ in $G$, by $W(P)$, the {\emph{winding number}} of $P$, we denote the number of signed crossings of $W$ defined in the following manner: we traverse $P$ in the direction from $C_1$ to $C_2$ and each time we intersect $W$, if $W$ crosses from left to right (with respect to the directed curve $W$) then this contributes $+1$ to the number of crossings, and from right to left we count $-1$ crossing. Note that even if $P$ is in fact a directed path directed from $C_2$ to $C_1$, we count the crossings by traversing from $C_1$ to $C_2$.

If a ring is equipped with a reference curve, we call it a {\emph{rooted ring}}. We would like to stress here that the reference curve is required to be fully embedded into the ring, i.e., between cycles $C_1$ and $C_2$. This property will be (implicitly) important to many claims, so we will always make sure that the constructed reference curves have this property.

Let $\Pp=(P_1,P_2,\ldots,P_k)$ be a family of directed paths connecting $C_1$ and $C_2$. For a sequence $\Cc=(c_1,c_2,\ldots,c_k)$, where each $c_i$ is $\pm 1$, we say that $\Pp$ is {\emph{compatible}} with $\Cc$ if for each index $i$, path $P_i$ is directed from $C_1$ to $C_2$ if $c_i=+1$, and from $C_2$ to $C_1$ otherwise. Let us note the following observation.

\begin{observation}\label{obs:pm1}
Let $G$ be a digraph embedded into a rooted ring with outer face $C_1$ and inner face $C_2$. Let $s_1,s_2,\ldots,s_k$ be terminals lying on the outer face in clockwise order, and let $t_1,t_2,\ldots,t_k$ be terminals lying on the inner face in clockwise order. Let us fix a sequence $\Cc=(c_1,c_2,\ldots,c_k)$, consisting of entries $\pm 1$. Let $\Pp=(P_1,P_2,\ldots,P_k)$ be a family of directed vertex-disjoint paths connecting corresponding $s_i$ with $t_i$, compatible with $\Cc$. Then the winding numbers of paths $P_i$ differ by at most $1$.
\end{observation}

For such family of paths $\Pp=(P_1,P_2,\ldots,P_k)$, we denote $W(\Pp):=W(P_1)$. Note that for every $i$ we have that $|W(P_i)-W(\Pp)|\leq 1$.

\subsection{Understanding homotopies in a ring}

We use Theorem~\ref{thm:dss} to prove the following result. Intuitively, it says that the possible homotopies of families of paths crossing a ring behave roughly in a convex way.

\begin{theorem}\label{thm:spiral}
Let $G$ be a digraph embedded into a rooted ring with outer face $C_1$ and inner face $C_2$. Let $s_1,s_2,\ldots,s_k$ be terminals lying on the outer face in clockwise order and $t_1,t_2,\ldots,t_k$ be terminals lying on the inner face in clockwise order; assume further that the reference curve $W$ connects the interval between $s_k$ and $s_1$ on the boundary of $C_1$ and the interval between $t_k$ and $t_1$ on the boundary of $C_2$. Let us fix a sequence $\Cc=(c_1,c_2,\ldots,c_k)$, consisting of entries $\pm 1$. Let $\Pp=(P_1,P_2,\ldots,P_k)$ be a family of directed vertex-disjoint paths connecting corresponding $s_i$ with $t_i$, compatible with $\Cc$. Assume further that there are some families $\Qq=(Q_1,Q_2,\ldots,Q_k)$ and $\Rr=(R_1,R_2,\ldots,R_k)$ of directed vertex-disjoint paths connecting $C_1$ and $C_2$ that are also compatible with $\Cc$. Then, for every number $\alpha$ such that $W(\Qq)+6\leq \alpha \leq W(\Rr)-6$, there exists a family $\Pp'=(P_1',P_2',\ldots,P_k')$ such that $\Pp'$ is compatible with $\Cc$, $P_i'$ connects $s_i$ and $t_i$, and $W(\Pp')=\alpha$.
\end{theorem}
\begin{proof}
By the assumed property of the reference curve $W$ we have that for every sequence $\Aa$ of vertex-disjoint paths connecting $s_1,s_2,\ldots,s_k$ with corresponding $t_1,t_2,\ldots,t_k$, all the paths of $\Aa$ have the same winding number, equal to $W(\Aa)$. Let us change the curve $W$ to any curve $W'$ that has the same endpoints as $W$, but has winding number $\alpha$ with respect to $W$. Then, winding numbers with respect to $W'$ are exactly the same as with respect to $W$, but with $-\alpha$ additive constant. Therefore, without loss of generality we may assume that $\alpha=0$.

Compactify the plane using one point in the infinity. Take any disk $D_1$ outside $C_1$, any disk $D_2$ inside $C_2$, and cut them out of the compacted plane. Then identify boundaries of $D_1$ and $D_2$, thus obtaining a torus $\Sigma$. For every index $i$, insert an arc $(t_i,s_i)$ if $c_i=+1$, and an arc $(s_i,t_i)$ if $c_i=-1$. Observe that these new arcs can be realized on $\Sigma$ without introducing crossings by drawing them through the identified boundaries of $D_1$ and $D_2$. Let us denote the new graph by $\Ge$; we will also refer to it as to {\emph{extended}} $G$. For $i=1,2,\ldots,k$, let the face $F_i$ be the face of $\Ge$ enclosed by arcs between $s_i,t_i$, between $s_{i+1},t_{i+1}$, and fragments of boundaries of $C_1$, $C_2$ between terminals $s_i,s_{i+1}$ and $t_i,t_{i+1}$, respectively, where $s_{k+1}=s_1$ and $t_{k+1}=t_1$. Moreover, we can extend the curve $W$ to a closed curve $\We$ on $\Sigma$ by closing it through the face $F_k$. We can set the homotopy group on $\Sigma$ so that $\We$ has homotopy $(0,1)$ and $\De:=\partial D_1=\partial D_2$ has homotopy $(1,0)$. 

\begin{figure}[h!]
\begin{center}
\includegraphics[scale=0.6]{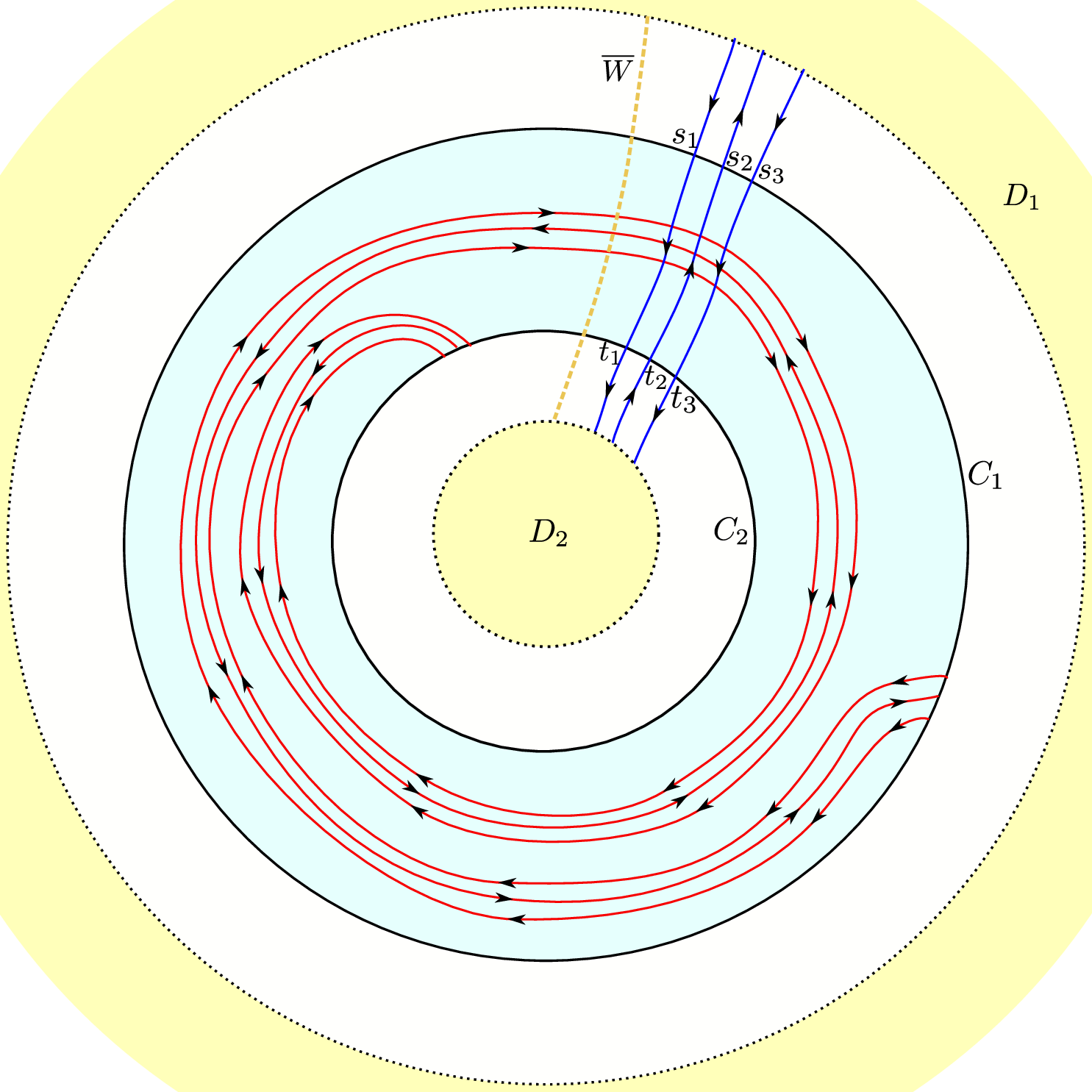}
\caption{Initial situation in the proof of Theorem~\ref{thm:spiral}. Paths of the family $\Ppe$ are coloured blue, and paths of the family $\Qq$ are coloured red. Family $\Rr$ is not depicted in order not to overcomplicate the picture. The cycles in the family $\Ppe$ already have homotopies $(0,c_i)$ (i.e., $\beta=0$), hence they may serve as $\Ppe'$.}
\label{fig:ring1}
\end{center}
\end{figure}

Similarly, using new arcs connecting $s_i$ and $t_i$ we can close each $P_i$ into a directed cycle $\Pe_i$, thus obtaining a family of vertex-disjoint cycles $\Ppe$. Note that a cycle $\Pe_i$ has homotopy $(c_i\cdot \beta,c_i)$ for $\beta=W(\Pp)$. Also, the postulated family of paths $\Pp'$ exists if and only if there exists the same family closed by the new arcs connecting $s_i$ and $t_i$; that is, in the theorem statement we postulate existence of a family $\Ppe'=(\Pe_1',\Pe_2',\ldots,\Pe_k')$ of vertex-disjoint cycles on $\Sigma$, located in this order, such that $\Pe_i'$ has homotopy $(0,c_i)$. 

For the sake of contradiction assume that the assumed family $\Ppe'$ does not exist. By Theorem~\ref{thm:dss} there exists a noose $N$ of some homotopy $(p,q)$, where $p\geq 0$, such that $(c_1,c_2,\ldots,c_k)^p$ is not embeddable in $\snoose(N)$.

We now construct an infinite graph $\Gf$, called also {\emph{unravelled}} $G$, by taking the torus $\Sigma$, cutting it along $\We$ and $\De$ thus obtaining a rectangle, and filling the plane with copies of this rectangle in a grid-like manner. The halves of arcs cut when cutting $\We$ and $\De$ are joined with the corresponding halves from a neighbouring rectangle. The rectangles are called {\emph{cells}}; we index them naturally by $\Z\times \Z$, where the first coordinate corresponds to the direction of $\De$, and the second to the direction of $\We$. Let us locate the cells on the plane in such a manner that a cell $(a,b)$ is the square with coordinates of corners $(a-1,b-1),(a-1,b),(a,b),(a,b-1)$. By $F_i^{(a,b)}$ we denote the face of $\Gf$ originating in the face $F_i$ of $\Ge$ that lies between cells $(a,b)$ and $(a,b+1)$; the copy of the face $F_k$ that lies on the meet of cells $(a,b)$, $(a+1,b)$, $(a+1,b+1)$ and $(a,b+1)$ has index $(a,b)$. Let $\Ff$ be the set of all the faces of the form $F_i^{(a,b)}$. Note that then a line $y=c$ for $c\in \Z$ traverses all the faces of $\Ff$ with the second coordinate of the upper index equal to $c$; moreover, no vertex of $
\Gf$ has integer second coordinate.

Let us define a family $\Ppf$ as an infinite family of parallel vertex-disjoint infinite paths that are results of unravelling the cycles of $\Ppe$ on the plane. Also, we can define unravelled families $\Qqf$ and $\Rrf$ as infinite families of vertex-disjoint paths originating in unravelling of $\Qq$ and $\Rr$, respectively. Note that each path from $\Qqf$ or $\Rrf$ is finite and contained in one row of the grid.

Let us elaborate a bit on the family $\Ppf$. Paths of $\Ppf$ divide the plane into infinite number of strips, each with two ends and bounded by two paths of $\Ppf$. The strips can be labeled by integer numbers so that each strip is neighbouring only with the strips with numbers differing by one. Observe that every path of $\Ppf$ traverses every line $y=c$ for $c\in \Z$ exactly once, using one of the arcs between neighbouring faces of $\Ff$. Thus, $\Ppf$ divides each line $y=c$ for $c\in \Z$ into intervals, corresponding to intersections of this line with strips. We call these intervals {\emph{strip intervals}}.

Moreover, observe that the family $\Ppf$ is invariant with respect to integer shifts, i.e., for any $(a,b)\in \Z\times \Z$ we have that $\Ppf+(a,b)=\Ppf$. Observe also that a shift $(a,0)$ acts on a line $y=c$ for $c\in \Z$ by shifting the strip intervals $k\cdot a$ to the right. As strip intervals are intersections of strips with this line, shift $(a,0)$ acts on the plane divided into strips by mapping each strip into the strip $k\cdot a$ to the right, in the order of strips.

We now move to families $\Qqf$ and $\Rrf$. Take each path $Q$ from $\Qqf$ contained in $\{y\in [c,c+1]\}$ and extend it by short curves contained in single faces of $\Ff$ at which $Q$ starts and ends, so that $Q$ connects the line $\{y=c\}$ with the line $\{y=c+1\}$ (recall that each path in $\Qqf$ connects $C_1$ with $C_2$, not $D_1$ with $D_2$). The original part of $Q$, being its intersection with $\Gf$, is called the {\emph{core}} of $Q$, and the two small curves at the ends are called {\emph{extensions}} of $Q$. Clearly, we may add the extensions in such a manner that (i) curves originating in the same path from $\Qq$ are extended in identical manner, and (ii) all the curves from $\Qqf$ are still pairwise disjoint. Perform the same construction for the family $\Rrf$.

Observe that families $\Qqf$ and $\Rrf$ have a similar periodic behaviour inside parts of the plane $\{y\in [c,c+1]\}$ for $c\in \Z$. Curves from $\Qqf$ ($\Rrf$) divide each $\{y\in [c,c+1]\}$ into small strips, naturally linearly ordered along $\{y\in [c,c+1]\}$. The partition is invariant with respect to shifts by vectors $(a,0)$ for $a\in \Z$: each such shift maps every strip to the strip that is $a\cdot k$ to the right of it, counting in the natural order of the strips.

Let us now examine what happens with the noose $N$. Clearly, $N$ also unravels into an infinite family of infinite-length parallel curves (see Fig.~\ref{fig:unraveling}). We take one of them and declare it $\Nf$, the unravelled noose $N$. By slightly perturbing $N$ we may assume that intersections of $\Nf$ with lines $\{y=c\}$ for $c\in \Z$ are not contained in $\Gf$. 

Note that $\Nf$ is periodic in the following sense: $\Nf+(p,q)=\Nf$. If we take a fragment of $\Nf$ between any $x\in \Nf\setminus \Gf$ and $x+(p,q)\in \Nf\setminus \Gf$, denote it by $\Nf[x]$, then this corresponds to the closed curve $N$ on $\Sigma$ cut in the vertex $x$. In particular, no cyclic shift of $\Cc^p$ can be embedded into $\snoose(\Nf[x])$.

Assume now that $q=0$, that is, $N$ does not wind in the direction of $\We$. Take any $x\in \Nf\setminus \Gf$ and consider the curve $\Nf[x]$. By the previous observations, $x+(p,0)$ is located in the strip $k\cdot p$ to the right with respect to the strip containing $x$. It follows that any noose travelling from $x$ to $x+(p,0)$ must cross the $k\cdot p$ consecutive paths from $\Ppf$ between these strips; this applies in particular to $\Nf[x]$. We may now define an embedding of a cyclic shift of $\Cc^p$ into $\snoose(\Nf[x])$, by embedding a $\pm 1$ into each crossing of these paths with $\Nf[x]$, where we choose the sign depending on the direction in which a corresponding path is directed. This is a contradiction with the assumed properties of $N$.

Hence, we infer that $q\neq 0$. For the rest of the proof we assume that $q>0$ and we will use only the family $\Qqf$. The proof for $q<0$ is symmetric and uses only family $\Rrf$.

Let us perform another slight modification to $\Nf$. By slightly perturbing $\Nf$ within faces it travels through, we may assume that whenever $\Nf$ travels through a face $F\in \Ff$ (say, lying on a line $\{y=c\}$), it firstly touches the line $\{y=c\}$, then crosses all the extensions of paths from $\Qqf$, then again touches the line $\{y=c\}$, and then finally leaves the face $F$. In other words, we may assume that before the first intersection and after the last intersection with the line $\{y=c\}$, no extension of a path from $\Qqf$ is crossed. This purely technical property will be used later to avoid counting crossings of extensions of paths from $\Qqf$ as true crossings contributing to $\vnoose(N)$.

\begin{figure}[h!]
\begin{center}
\includegraphics[scale=0.8]{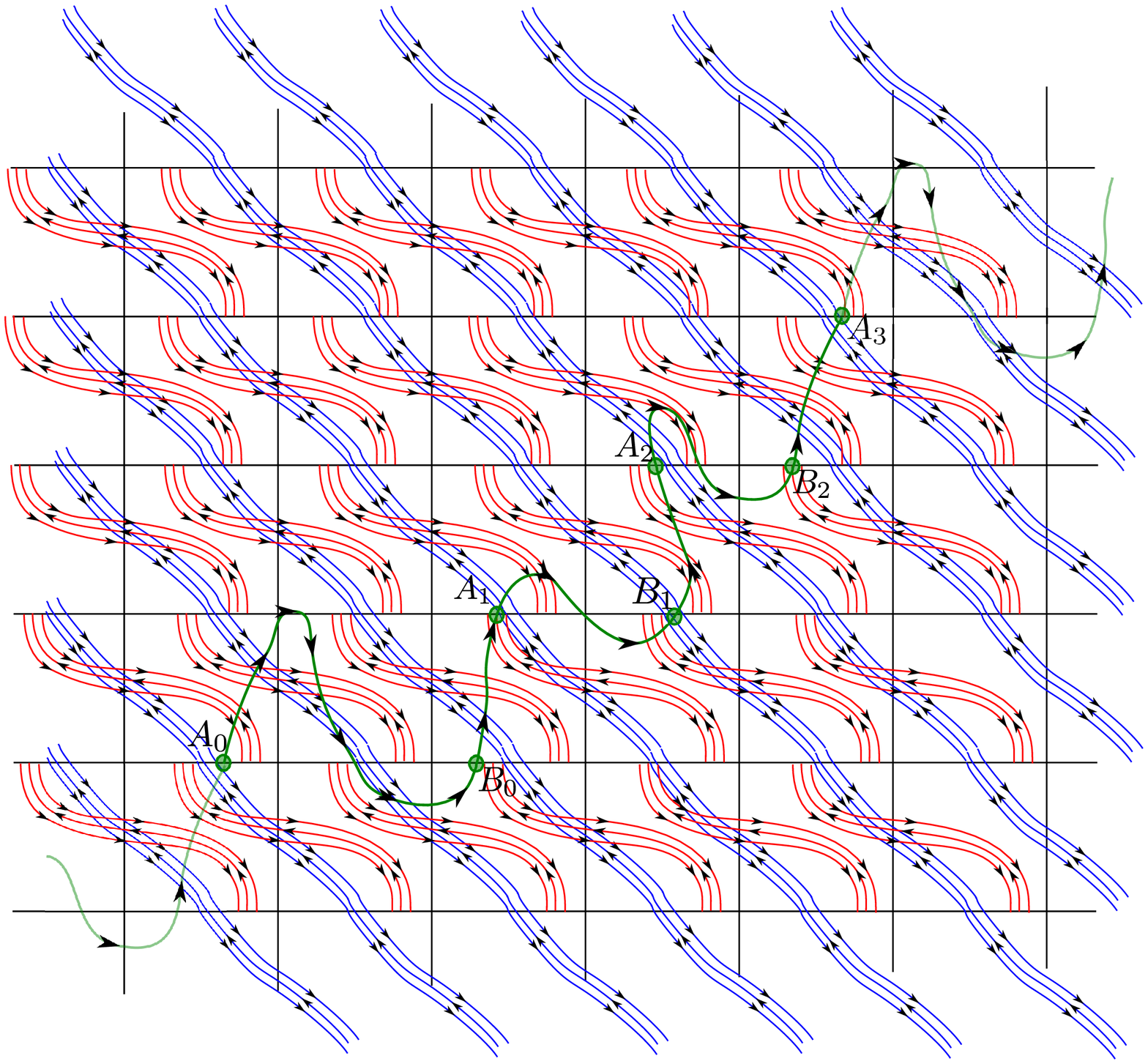}
\caption{The unravelled graph $\Gf$ with families $\Ppf$ and $\Qqf$ depicted (blue and red, respectively), where $\beta=-1$ and $(p,q)=(4,3)$. On the curve $\Nf$, the interval $\Nf[x]$ between $x$ and $x+(4,3)$ is highlighted, and partitioned into parts $K_i$, $L_i$ by points $x=A_0,B_0,A_1,B_1,A_2,B_2,A_3=x+(4,3)$. For the sake of simplicity, paths from the family $\Qq$ have winding numbers $-1$ instead of required at most $-6$.}
\label{fig:unraveling}
\end{center}
\end{figure}

We now need the following two claims.

\begin{claim}\label{clm:samelevel}
Let $M$ be a curve with an origin in a point $(x_1,c)\notin \Gf$ and end in a point $(x_2,c)\notin \Gf$, where $c\in \Z$. If $r=\max(0,\lfloor x_2\rfloor -\lceil x_1\rceil)$, then $\Cc^{r}$ is embeddable in $\snoose(M)$. 
\end{claim}
\begin{proof}
If $r=0$ then the claim is trivial, so assume that $r>0$. As $(x_1,c)$ is located in the interval $(\lceil x_1\rceil-1,\lceil x_1\rceil]$ on the line $y=c$ and $(x_2,c)$ is located in the interval $[\lfloor x_2\rfloor,\lfloor x_2\rfloor+1)$ on the same line, it follows that one can find $r$ consecutive bundles of paths from $\Ppf$ corresponding to cycles $\Pe_1,\Pe_2,\ldots,\Pe_k$, separating $(x_1,c)$ from $(x_2,c)$. Each such bundle is formed by paths of $\Ppf$ that intersect an interval $[\lambda,\lambda+1)$ on the line $\{y=c\}$ for $\lceil x_1\rceil \leq \lambda < \lfloor x_2\rfloor$, $\lambda\in \Z$. Similarly to the proof for the case $q=0$, crossings of $M$ with these paths define an embedding of $\Cc^r$ into $\snoose(M)$.
\end{proof}

\begin{claim}\label{clm:jumplevel}
Let $M$ be a curve contained in the part of the plane $\{y\in [c,c+1]\}$, where the origin is of the form $(x_1,c)$ and the end is of the form $(x_2,c+1)$. Moreover, assume that $M$ does not cross any extension of any path from $\Qqf$. If $r=\max(-4,\lfloor x_2\rfloor -\lceil x_1\rceil)$, then $\Cc^{r+4}$ is embeddable in $\snoose(M)$.
\end{claim}
\begin{proof}
If $r=-4$ then the claim is trivial, so assume that $r>-4$. We examine the mutual location of $(x_1,c)$ and $(x_2,c+1)$ with respect to the strips into which $\Qqf$ divides $\{y\in [c,c+1]\}$. We claim that the strip containing $(x_2,c)$ is at least $(r+5)\cdot k$ strips to the right with respect to the one containing $(x_1,c)$. By the properties of $\Qqf$ with respect to shifts, we have that $(x_1+r,c)$ is $r\cdot k$ strips to the right when compared with $(x_1,c)$; note that possibly $r<0$ --- then it means $-r\cdot k$ strips to the left. Now note that $\lfloor x_2\rfloor\geq x_1+r$, so $(\lfloor x_2\rfloor,c)$ is at least $r\cdot k$ strips to the right when compared with $(x_1,c)$. Observe that $(\lfloor x_2\rfloor,c+1)$ is at least $5\cdot k$ strips to the right when compared to $(\lfloor x_2\rfloor,c)$. This follows from the fact that the the segment between points $(\lfloor x_2\rfloor,c)$ and $(\lfloor x_2\rfloor,c+1)$ is exactly the line $W$, and as $W(\Qq)\leq -6$, then by Observation~\ref{obs:pm1} each curve of $\Qq$ has signed crossing number at least $5$ with respect to $W$. This means that during its travel from $(\lfloor x_2\rfloor,c)$ to $(\lfloor x_2\rfloor,c+1)$ we travel at least $5\cdot k$ strips to the right in the partition of $\{y\in [c,c+1]\}$ into the strips by $\Qqf$. As $\lfloor x_2\rfloor \leq x_2$, point $(x_2,c+1)$ can be only even further to the right when compared with $(\lfloor x_2\rfloor,c+1)$

Hence, in total $(x_2,c+1)$ lies at least $(r+5)\cdot k$ strips to the right with respect to $(x_1,c)$. As in the previous proofs, it follows that every curve $M$ travelling from $(x_1,c)$ to $(x_2,c+1)$, contained in $\{y\in [c,c+1]\}$ and not crossing any extensions of paths from $\Qqf$, must admit an embedding of a cyclic shift of $\Cc^{(r+5)\cdot k}$ into $\vnoose(M)$. As every cyclic shift of $\Cc^{(r+5)\cdot k}$ has $\Cc^{(r+4)\cdot k}$ as a subword, we have that $\Cc^{(r+4)\cdot k}$ is embeddable into $\snoose(M)$.
\end{proof}

Now observe that 
\begin{eqnarray*}
\max(0,\lfloor x_2\rfloor -\lceil x_1\rceil)\geq x_2-x_1-2, \\
\max(-4,\lfloor x_2\rfloor -\lceil x_1\rceil)+4 \geq x_2-x_1+2.
\end{eqnarray*}

Let us now examine the curve $\Nf$. As $\Nf$ is invariant under the shift $(p,q)$ and $q\neq 0$, it follows that the set of intersections of $\Nf$ with $\{y=0\}$ is nonempty and contained in a compact interval on $\Nf$. Let then $x$ be the first intersection of $\Nf$ with $\{y=0\}$. Consider the curve $\Nf[x]$; similarly as in the case $q=0$, we have that $\snoose(\Nf[x])$ does not admit embedding of any cyclic shift of $\Cc^p$. By the choice of $x$ and invariance of $\Nf$ under the shift $(p,q)$ we have that $x+(p,q)$ is the first intersection of $\Nf$ with the line $\{y=q\}$.

We would like now to divide $\Nf[x]$ into smaller curves. Let $A_0=x$. We inductively define points $B_0,A_1,B_1,\ldots,B_{q-1},A_q$ as follows: $B_i$ is the last intersection of the part of $\Nf[x]$ after $A_i$ with line $\{y=i\}$, while $A_{i+1}$ is the first intersection of the part of $\Nf[x]$ after $B_i$ with the line $\{y=i+1\}$. As $x+(p,q)$ is the first intersection of $\Nf$ with $\{y=q\}$, it follows that $A_q=x+(p,q)$. Let $A_i=(a_i,i)$ and $B_i=(b_i,i)$. For $i=0,1,\ldots,q-1$ let $K_i$ be the part of $\Nf[x]$ between $A_i$ and $B_i$, and let $L_i$ be the part of $\Nf[x]$ between $A_i$ and $B_{i+1}$.

Therefore, $\Nf[x]$ consists of curves $K_0,L_0,K_1,L_1,\ldots,K_{q-1},L_{q-1}$ concatenated in this order. Each curve $K_i$ satisfies the conditions of Claim~\ref{clm:samelevel}, while each curve $L_i$ satisfies the conditions of Claim~\ref{clm:jumplevel}; note that the claim for curves $L_i$ follows from the technical property about visits of $\Nf$ in faces of $\Ff$ that we ensured before introducing Claims~\ref{clm:samelevel} and~\ref{clm:jumplevel}. By applying these two claims we have that $\snoose(\Nf[x])$ admits an embedding of $\Cc^t$, where
\begin{eqnarray*}
t \geq \sum_{i=0}^{q-1} (b_i-a_i-2) + \sum_{i=0}^{q-1} (a_{i+1}-b_i+2) = a_q-a_0 = q
\end{eqnarray*}
This is a contradiction with the assumed properties of $N$.
\end{proof}

We now provide two corollaries of Theorem~\ref{thm:spiral}, which will be used in our algorithm.

\begin{lemma}[Rerouting in a ring]\label{lem:ringhomotopy}
Let $G$ be a digraph in a rooted ring with outer face $C_1$ and inner face $C_2$, where $W$ is the reference curve.
\begin{itemize}
\item Let $p^1_1$, $\dots$, $p^1_s$ be vertices on $C_1$ (clockwise order) and let $p^2_1$, $\dots$, $p^2_s$ be vertices on $C_2$ (clockwise order).
\item Let $q^1_1$, $\dots$, $q^1_s$ be vertices on $C_1$ (clockwise order) and let $q^2_1$, $\dots$, $q^2_s$ be vertices on $C_2$ (clockwise order).
\item Let $P_1$, $\dots$, $P_s$ be a set of pairwise vertex-disjoint paths contained between $C_1$ and $C_2$ such that the endpoints of $P_i$ are $p^1_i$ and $p^2_i$.
\item Let $Q_1$, $\dots$, $Q_s$ be a set of pairwise vertex-disjoint paths contained between $C_1$ and $C_2$ such that the endpoints of $Q_i$ are $q^1_i$ and $q^2_i$.
\item $P_i$ and $Q_i$ go in the same direction, i.e., $P_i$ goes from $p^1_i$ to $p^2_i$ if and only if $Q_i$ goes from $q^1_i$ to $q^2_i$.
\end{itemize}
Then there is set $P'_1$, $\dots$, $P'_s$ of pairwise vertex-disjoint paths between $C_1$ and $C_2$ such that $P'_i$ and $P_i$ have the same start/end vertices and $|W(P'_i)-W(Q_i)| \le  6$.
\end{lemma}
\begin{proof}
The Lemma follows from the application of Theorem~\ref{thm:spiral} to families 
$$\Pp,\Qq,\Rr=\{P_1,P_2,\ldots,P_s\}, \{P_1,P_2,\ldots,P_s\}, \{Q_1,Q_2,\ldots,Q_s\},$$ 
or 
$$\Pp,\Qq,\Rr=\{P_1,P_2,\ldots,P_s\}, \{Q_1,Q_2,\ldots,Q_s\}, \{P_1,P_2,\ldots,P_s\},$$
depending on the inequality between $W(\{P_1,P_2,\ldots,P_s\})$ and $W(\{Q_1,Q_2,\ldots,Q_s\})$.
\end{proof}

\begin{lemma}[One-way Spiral Lemma]\label{lem:onewayspiral}
Let $G$ be a digraph in a rooted ring with outer face $C_1$ and inner face $C_2$, where $W$ is the reference curve. Assume moreover that $W$ induces a directed path $W^*$ in the dual of $G$, directed from $C_1$ to $C_2$.
\begin{itemize}
\item Let $p^1_1$, $\dots$, $p^1_s$ be vertices on $C_1$ (clockwise order) and let $p^2_1$, $\dots$, $p^2_s$ be vertices on $C_2$ (clockwise order).
\item Let $q^1_1$, $\dots$, $q^1_s$ be vertices on $C_1$ (clockwise order) and let $q^2_1$, $\dots$, $q^2_s$ be vertices on $C_2$ (clockwise order).
\item Let $P_1$, $\dots$, $P_s$ be a set of pairwise vertex-disjoint paths contained between $C_1$ and $C_2$ such that $P_i$ goes from $p^1_i$ to $p^2_i$.
\item Let $Q_1$, $\dots$, $Q_s$ be a set of pairwise vertex-disjoint paths contained between $C_1$ and $C_2$ such that $Q_i$ goes from $q^1_i$ to $q^2_i$.
\end{itemize}
Then there is set $P'_1$, $\dots$, $P'_s$ of pairwise vertex-disjoint paths between $C_1$ and $C_2$ such that $P'_i$ and $P_i$ have the same start/end vertices and $-6 \le |E(P'_i)\cap E(W^*)|-|E(Q_i)\cap E(W^*)| \le 6$.
\end{lemma}
\begin{proof}
The Lemma follows from Lemma~\ref{lem:ringhomotopy} and the assumption that $W$ crosses only arcs directed in one direction, so for any path $P$ the winding number with respect to $W$ is the number of arcs of $E(W^*)$ used by $P$.
\end{proof}

\subsection{Min-max theorems for alternating cuts}

\subsubsection{Alternations}

We now introduce the language and basic facts about measuring the number of alternations along a curve.

\begin{defin}
An {\emph{alternating sequence}} is a sequence of entries $\pm 1$, where $+1$ and $-1$ appear alternately.
\end{defin}

Note that for every $p$ there are two alternating sequences of length $p$: one starting with $+1$ and one starting with $-1$.

\begin{defin}
{\emph{Alternation}} of sequence $A$ of subsets of $\{-1,+1\}$, denoted $\alt(A)$, is the largest possible length 
of an alternating sequence embeddable into $A$.
Alternation of a non-closed curve $N$, denoted $\alt(N)$, is defined as $\alt(\snoose(N))$,
while alternation of a noose $N$, also denoted $\alt(N)$, is defined as the maximum of $\alt(\snoose(N))$ 
over all cyclic shifts of $\snoose(N)$.
\end{defin}

We now define a notion of a {\emph{pretty curve}} that will be useful for cutting the digraph using small number of alternations. Intuitively, a pretty curve uses as little crossings of vertices as possible.

\begin{defin}
\label{def:pretty}
A non-degenerate curve $N$ is called {\emph{pretty}} if all $S_i$-s in $\snoose(N)$ corresponding to crossings of vertices are equal to $\{-1,+1\}$ or $\emptyset$.
\end{defin}

We state now the following observation that shows that pretty curves behave robustly with respect to alternations. We stress that in the following two observations we are considering non-closed curves, not nooses.

\begin{observation}\label{obs:circumventing}
Let $G$ be a digraph embedded on a compact orientable surface $\Sigma$ and let $N$ be a non-degenerate curve on $\Sigma$. Then there exists a pretty curve $N'$ that has the same endpoints as $N$, is homotopic to $N$ and $\alt(N')=\alt(N)$. \end{observation}
\begin{proof}
We just need to make some simple local modifications to $N$. Whenever $N$ traverses a vertex $v_i$ with $S_i=\{-1\}$, we have that all arcs adjacent to $v_i$ from the right of $N$ are directed from $v_i$, or all arcs adjacent to $v_i$ from the left of $N$ are directed from $v_i$. Then in $N'$ we may either circumvent $v_i$ by traversing it from the left or from the right side, thus traversing all the arcs adjacent to $v_i$ from the left or from the right of $N$, respectively. Note that the corresponding change in $\snoose(N')$ is substitution of one term $\{-1\}$ with an arbitrary number of terms $\{-1\}$, which does not change the alternation of the sequence. We perform a symmetric operation for every traversed $v_i$ with $S_i=\{+1\}$. Observe that the modifications performed do not spoil non-degeneracy.
\end{proof}

We now show that one can conveniently reduce any pretty curve to a simple curve.

\begin{observation}\label{obs:subcurve}
Let $N$ be a pretty curve in a graph $G$ embedded on a compact orientable surface $\Sigma$. Then there exists a simple curve $N'$ with the same endpoints as $N$, such that $\alt(N')\leq \alt(N)$, and $N'$ is constructed by taking disjoint subcurves of $N$, and appending them in the order of their appearance on $N$, and possibly some consecutive two using small shortcutting curves, each entirely contained in the interior of one face of $G$.
\end{observation}
\begin{proof}
We build $N'$ by travelling along $N$, and whenever we encounter a vertex or face visited more than once, we continue further from the last visit, thus omitting the interval on $N$ between the first and the last visit (in case of double visit of a face, we may have to make a shortcut inside the face). $N'$ built in such a manner is simple. Hence, we need to argue that $\alt(N')\leq \alt(N)$.

Observe that cutting out a segment between two visits of the same face corresponds to taking a curve with new $\snoose$ being a subsequence of the old one; hence, the alternation cannot increase. The non-trivial part is that when we cut out a segment between the first and the last visit of some vertex $v$, then the alternation also cannot increase. Let $N_0$ be the original curve and $N_0'$ be the curve after cutting out. Let $S_f$ be the set in $\snoose(N_0)$ corresponding to the first visit of $v$ and $S_l$ be the set in $\snoose(N_0)$ corresponding to the last visit of $v$; by the fact that $N$ is pretty we have that $S_f,S_l$ are equal to $\{+1,-1\}$ or $\emptyset$. Let $S$ be the set in $\snoose(N_0')$ corresponding to the only visit of $v$ after the shortening. 

If $S_f$ or $S_l$ is equal to $\{+1,-1\}$, then we have that $S\subseteq S_f\cup S_l$. This suffices for our purposes, as then every embedding of an alternating sequence 
into a cyclic shift of $\snoose(N_0')$ can be lifted to an embedding into an appropriate cyclic shift of $\snoose(N_0)$ by mapping the term mapped to $S$ in $\snoose(N_0')$ either to $S_f$ or to $S_l$.

If $S_f=S_l=\emptyset$ then $v$ is a sink or source in $G$, which means that also $S=\emptyset$. Hence $\snoose(N_0')$ is a subsequence of $\snoose(N)$ and the claim follows.
\end{proof}

We remark that in Observation~\ref{obs:subcurve} the assumption of $N$ being pretty is necessary. Moreover, the simple curve given by Observation~\ref{obs:subcurve} is not necessary pretty; however, it is non-degenerate as it is simple.

\subsubsection{Many alternating cycles, or a short cut of a ring}

The first min-max theorem follows from the torus theorem of Ding, Schrijver, and Seymour, that is, Theorem~\ref{thm:dss}.

\begin{lemma}[Alternating cycles/cut duality]\label{lem:dualcycles}
Let $G$ be a graph embedded in the plane and let $\gamma_1, \gamma_2$ be two disjoint closed Jordan curves not containing any point of $G$. Let $G_\gamma$ be the subgraph of $G$ that consists of the vertices and edges that lie between $\gamma_1$ and $\gamma_2$, and let $F_1,F_2$ be the faces of $G_\gamma$ that contain $\gamma_1, \gamma_2$, respectively (it is possible that $F_1 = F_2$). Then, for any positive even integer $r$, in time polynomial in $r$ and $G$, one can find either:
  \begin{enumerate}
    \item a sequence $C_1, C_2, \ldots, C_r$ of alternating concentric cycles in $G_\gamma$ that separate $f_1$ from $f_2$,  and $C_i$ being separated from $f_2$ by $C_j$ for $i < j$; or
    \item a simple curve $M$ that starts in $F_1$, ends in $F_2$, and such that $M$ has alternation at most $r$.
  \end{enumerate}
\end{lemma}
\begin{proof}
Observe that finding a sequence $C_1,C_2,\ldots,C_r$ can be performed in polynomial time using the algorithm of Schrijver \cite{schrijver:xp}, while finding curve $N$ can be easily done using breadth-first search in polynomial time. Hence, we are left with proving that one of these objects always exists.

Let us perform a similar operation to that from the proof of Theorem~\ref{thm:spiral}: we compactify the plane using a point in the infinity, cut the plane along $\gamma_1$ and $\gamma_2$ constraining ourselves to the part between them, and identify $\gamma_1$ and $\gamma_2$ thus creating a torus $\Sigma$. Thus we may imagine that that $G$ is a graph embedded on torus $\Sigma$. We choose the homotopy group on $\Sigma$ so that the identified $\gamma_1,\gamma_2$ have homotopy $(0,1)$ while an arbitrarily chosen simple curve connecting $\gamma_1$ and $\gamma_2$ has homotopy $(1,0)$.

By Theorem~\ref{thm:dss}, if there is no feasible sequence $C_1,C_2,\ldots,C_r$, this means that there exists a noose $N$ of homotopy $(p,q)$ for $p>0$, such that no cyclic shift of $(-1,+1)^{pr/2}$ is embeddable in $\snoose(N)$ (recall that $r$ is even). As $N$ winds $p$ times in the direction from $\gamma_1$ to $\gamma_2$, it follows that one can find $p$ disjoint segments $N_1,N_2,\ldots,N_p$ on $N$ such that every $N_i$ travels from $\gamma_1$ to $\gamma_2$ and does not touch $\gamma_1$ or $\gamma_2$ apart from the endpoints (is entirely contained in the part of the plane between $\gamma_1$ and $\gamma_2$). Assume that every $N_i$ has alternation at least $r+1$, hence some alternating  sequence of length $r+1$ can be embedded into $\snoose(N)$. It follows that one can embed $(-1,+1)^{r/2}$ into every $\snoose(N_i)$. Therefore, one can embed $(-1,+1)^{pr/2}$ into $\snoose(N)$, which is a contradiction with the assumed properties of $N$. We infer that there exists some $M_0=N_i$ such that $\snoose(M_0)$ has alternation at most $r$. By applying Observation~\ref{obs:circumventing} and Observation~\ref{obs:subcurve} to the curve $M_0$ we obtain the desired curve $M$.
\end{proof}

\subsubsection{Many alternating paths, or a short circular cut in a ring}

The second min-max theorem is considerably more difficult to prove. It is possible to prove it also by the means of Theorem~\ref{thm:dss}, but in order to avoid unnecessary technical details, we choose to use the algorithm for cohomology feasibility problem of Schrijver~\cite{schrijver:xp}.

Let $\group$ be a free group on $r$ generators $g_1,g_2,\ldots,g_r$. In the cohomology feasibility problem we are given a directed graph $G$ together with some labeling $\phi$ of arcs with elements of $\group$, and a downward-closed set $H(a)\subseteq \group$ for each arc $a$ (i.e., if a word $w$ belongs to $H(a)$, then so does every subsequence of $w$). We say that a labeling $\psi\hbox{ : }E(G)\to \group$ is {\emph{cohomologous}} to $\phi$ if there exists $F\hbox{ : }V(G)\to \group$, such that $\psi((u,v))=F(u)^{-1}\cdot\phi((u,v))\cdot F(v)$ for every $(u,v)\in E(G)$. The question asked in the problem is whether there is some $\psi$ cohomologous to $\phi$ such that $\psi(a)\in H(a)$ for every $a\in E(G)$. The following theorem is the fundamental result standing behind the algorithm of Schrijver~\cite{schrijver:xp}.

\begin{theorem}[\cite{schrijver:xp}]\label{thm:coh-alg}
The cohomology feasibility problem is polynomial time solvable, even if sets $H(a)$ are given as polynomial-time oracles that check belonging.
\end{theorem}

We describe the instance of the cohomology feasibility problem as a quadruple $(G,\group,\phi,H)$. We remark that the algorithm works also in a more general setting than just for free groups, i.e., for free partially commutative groups, but we do not need this in this work. However, we will use the fact that in the algorithm one can require for an arbitrary subset of vertices $S\subseteq V(G)$ that $F(s)=1_\group$ for every $s\in S$. In this case, we add set the $S$ as the fifth coordinate of the problem description.

In \cite{schrijver:preprint}, Schrijver explicitely describes the obstacles for existence of the solution.

\begin{theorem}[\cite{schrijver:preprint}, Theorem $3$ with adjusted terminology]\label{thm:coh-obstacles}
Let $I=(G,\group,\phi,H)$ be an instance of the cohomology feasibility problem. Then there is a solution $\psi$ to the instance $I$ if and only if for every vertex $u$ and every two undirected walks $P,Q$ from $u$ to $u$ there exists $x\in \group$ such that $x^{-1}\cdot \phi(P)\cdot x\in H(P)$ and $x^{-1}\cdot \phi(Q)\cdot x\in H(Q)$, where for an undirected walk $S$, $\phi(S)$ denotes the product of group elements along $S$, while $H(S)$ denotes the set of all possible such products where the  factors are taken from $H(a)$ for consecutive arcs $a$ of $S$. Moreover, the algorithm of Theorem~\ref{thm:coh-alg} can provide walks $P$, $Q$ contradicting this assumption, in case no solution was found.
\end{theorem}

We remark that in the sense of the above theorem, a walk is undirected, that is, it does not necessarily respect the direction of arcs; it can go via an arc in the reverse direction, and if this is the case, the contribution to $\phi(S)$ and $H(S)$ is the normal contribution reversed.

Let $G$ be a directed graph embedded into a ring $R$, where $C_1$ and $C_2$ are cycles being boundaries of the faces outside and inside the ring, respectively. Contrary to Lemma~\ref{lem:dualcycles}, in the following we assume that these two faces are different, and in fact $C_1$ and $C_2$ are disjoint. Let $r$ be an even number. We say that a family of vertex-disjoint paths $\Pp=(P_1,P_2,\ldots,P_r)$ is an {\emph{alternating join}} of $C_1$ and $C_2$ of size $r$, if every path $P_i$ connects $C_1$ and $C_2$ and is directed from $C_1$ to $C_2$ if $i$ is odd, and from $C_2$ to $C_1$ if $i$ is even, and paths $P_i$ are located in clockwise order in the ring. 

We now show how to, given a digraph $G$, construct an instance of the cohomology feasibility problem that encodes existence of an alternating join; the construction closely follows the lines of constructions of Schrijver for various other problems, but we include it for the sake of completeness. 

We define an {\emph{extended dual}} $G^+$ of $G$ as follows: we construct the classical dual $G^*$, and for every vertex $v$ and every two faces sharing $v$ and not sharing an edge adjacent to $v$,  we add an additional arc between the corresponding pair of vertices in the dual, in an arbitrary direction; we call these arcs {\emph{added}}. We delete the vertices in the dual corresponding to faces with boundaries $C_1,C_2$, i.e., we delete the inner and outer of the ring. Note that $G^+$ is not necessarily planar.

We take $\group$ to be a free group on $r$ generators $g_1,g_2,\ldots,g_r$. For every original arc $a$ of the dual we define $H(a)=\{1,g_1,g_2,\ldots,g_r\}$, and for every added arc $a^+$ we define $H(a^+)=\{1,g_1,g_2,\ldots,g_r,g_1^{-1},g_2^{-1},\ldots,g_r^{-1}\}$. Take any path $P$ in $G$ connecting $C_1$ and $C_2$ and put $\phi(a^*)=g_1\cdot g_2^{-1}\cdot g_3\cdot \ldots\cdot g_{r-1}\cdot g_r^{-1}$ for all arcs $a$ of $P$, where $a^*$ is the arc in the dual corresponding to $a$. Moreover, if $a^+$ is an added arc connecting two faces $F_1,F_2$ sharing a vertex $v$, take the boundary of the face in the dual corresponding to the vertex $v$, define $R$ to be any path on this boundary connecting $F_1$ and $F_2$, and put $\phi(a^+)=\phi(R)$. Note that after deleting the outer and inner face, at least one such path exists. Put $\phi(a)=1$ for all the other arcs.

\begin{figure}
\begin{center}
\includegraphics{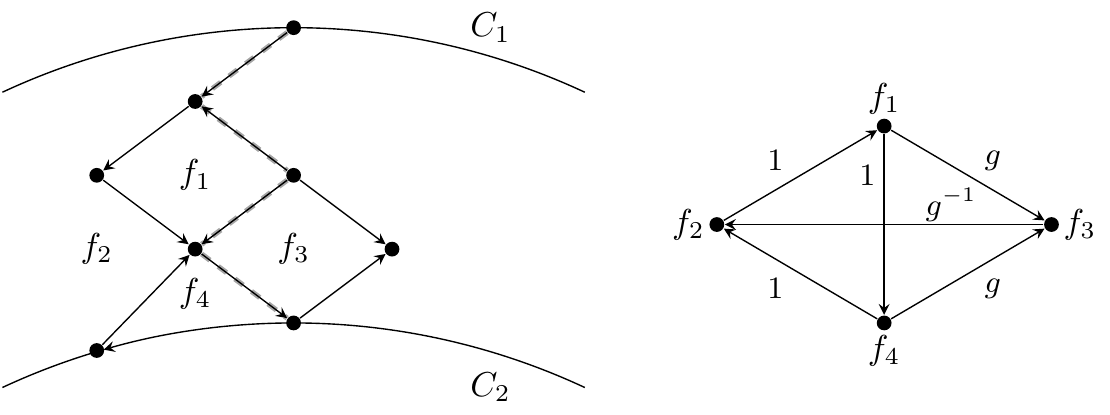}
\caption{The marked edges belong to the path $P$. On the right we have a subgraph of the extended dual depicting the values $\phi(a)$, where $g=g_1\cdot g_2^{-1}\cdot g_3\cdot \ldots\cdot g_{r-1}\cdot g_r^{-1}$.}
\label{fig:construction-s3}
\end{center}
\end{figure}

Having defined the instance, we prove the following claim.

\begin{lemma}\label{lem:coh-formulation}
Let $I=(G^+,\group,\phi,H)$ be the defined instance of the cohomology feasibility problem. Then $I$ has a solution $\psi$ if and only if there exists an alternating join of $C_1$ and $C_2$ of size $r$. Moreover, given a solution $\psi$ one can construct the alternating join in polynomial time.
\end{lemma}
\begin{proof}
Assume first that an alternating join $\Pp=(P_1,P_2,\ldots,P_r)$ exists. Put $\psi(a^*)=g_i$ if $a\in P_i$, $\psi(a^*)=1$ for every other arc of the dual, and $\psi(a^+)=\psi(R)$ for every added arc, where $R$ is defined as in the definition of the instance. It is easy to verify that $\psi$ is indeed a solution: satisfaction of constraints imposed by the function $H$ follows from vertex-disjointness, while being cohomologous can be easily seen by drawing paths aside one by one, modifying each path by a sequence of operations of jumping over a single face at a time.

Assume then that we are given a solution $\psi$ to the instance $I$. Take any vertex $v\in V(G)$. Let $a_1,a_2,\ldots,a_t$ be arcs adjacent to $v$ that do not lie on $C_1$ or $C_2$, in this order on the plane. As constraints imposed by the function $H$ are satisfied, each of the corresponding arcs in the dual can accommodate only a symbol $g_i$ or $1$. The added arcs ensure that there are no two arcs accommodating two different generators and that every two consecutive arcs accommodating a generator have different directions. Hence, if by $G_i$ we denote the subgraph of $G$ consisting of arcs $a$ such that $\psi(a^*)=g_i$ and vertices adjacent to them, then subgraphs $G_i$ are vertex-disjoint.

Consider the graph $G_i$. Partition arcs of $G_i$ into cycles and paths as follows: if we have an arc $(u,v)\in E(G_i)$, then we say that the next arc on the path or cycle is the arc $(v,w)\in E(G_i)$ that is the next arc incident to $v$ accommodating a nontrivial symbol in clockwise order. Thus, $G_i$ is partitioned into paths and cycles, where each path begins and ends on $C_1$ or $C_2$. These paths and cycles are edge-disjoint and noncrossing by their construction; moreover, the paths can begin and finish only on $C_1$ or $C_2$. We say that a path $P$ is a {\emph{connector}} if it connects from $C_1$ to $C_2$; the symbol associated with a connector from $G_i$ is $g_i$ if it goes from $C_1$ to $C_2$, and $g_i^{-1}$ if it goes from $C_2$ to $C_1$. Let $\Pp$ be the family of all connectors for all $G_i$-s. As connectors from $\Pp$ are non-crossing, they can be naturally ordered along the cycle. Let $\Cc$ be the sequence of symbols associated with connectors in this order; note that $\Cc$ is defined uniquely up to a cyclic shift. As the other paths and cycles from partitionings of $G_i$-s are edge-non-crossing, one can find a cycle $C$ in $G^+$ winding one time around the ring such that $\psi(C)=\Cc$.

As $\psi$ is cohomological to $\phi$, we have that $\psi(C)$ is conjugate to $\phi(C)$. It follows that $\Cc$ must admit a cyclic shift of the sequence $(g_1,g_2^{-1},g_3,\ldots,g_{r-1},g_r^{-1})$ as an embedding. Therefore, we can find $r$ connectors, that is, paths connecting $C_1$ and $C_2$, such that every two clockwise consecutive traverse the ring in different directions. Such a family is an alternating join of size $r$. Careful inspection of the proof shows that all the steps of the construction of this family can be performed in polynomial time, given a solution $\psi$.
\end{proof}

Now, using obstacle characterization of Theorem~\ref{thm:coh-obstacles} and cohomological formulation of Lemma~\ref{lem:coh-formulation}, we are able to prove the following min-max theorem.

\begin{lemma}[Alternating paths/circular cut duality]\label{lem:dualaltcut}
Let $G$ be a graph embedded into a ring with $C_1$ being the boundary of the outer face and $C_2$ the boundary of the inner face. Assume moreover that $C_1$ and $C_2$ are disjoint. Let $r$ be an even integer. Then there exists a polynomial-time algorithm that returns either
\begin{itemize}
\item an alternating join of $C_1$ and $C_2$ of size $r$, or
\item a simple noose inside the ring, separating $C_1$ and $C_2$ and having at most $r+4$ alternations.
\end{itemize}
\end{lemma}
\begin{proof}
Before we start the proof, without loss of generality we assume that $C_1$ and $C_2$ are in fact directed cycles going clockwise. Indeed, we can redirect the arcs of $C_1$ and $C_2$ in any manner, as we can safely assume that the alternating join will not use any of these arcs, and also redirecting these arcs do not influence alternation of any noose contained inside the ring.

We construct the instance of the cohomology feasibility problem $I=(G^+,\group,\phi,H)$, as in Lemma~\ref{lem:coh-formulation}, and apply the algorithm of Theorem~\ref{thm:coh-alg}. If the algorithm returns a solution, by Lemma~\ref{lem:coh-formulation} we can extract a sufficiently large alternating join and return it. Otherwise, by Theorem~\ref{thm:coh-obstacles} we are left with two closed undirected walks $P,Q$ in $G^+$ rooted in some vertex $u\in V(G^+)$, such that for no element $x\in \group$ we have that $x^{-1}\phi(P)x\in H(P)$ and $x^{-1}\phi(Q)x\in H(Q)$. Fix an arbitrary reference curve and let $w_P$ and $w_Q$ be the winding numbers of $P$ and $Q$ in the ring; by reversing each cycle if necessary, we may assume that $w_P,w_Q\geq 0$. Note that then $\phi(R)=\left(g_1\cdot g_2^{-1}\cdot g_3\cdot \ldots\cdot g_{r-1}\cdot g_r^{-1}\right)^{w_R}$ for $R\in \{P,Q\}$.

For clarity, in the following, whenever considering alternation, we think of $P,Q$ as non-closed undirected walks that only by coincidence begin and end in the same point. That is, we consider embeddings of alternating sequences into sequences $\snoose(P),\snoose(Q)$, where these sequences begin and end in $v$.

We now claim that if ($w_R\cdot r\leq\alt(R)$ or $w_R=0$) for both $R=P,Q$, then for $x=1$ we have that $x^{-1}\phi(P)x\in H(P)$ and $x^{-1}\phi(Q)x\in H(Q)$, which is a contradiction. Indeed, if $w_P=0$ then $\phi(P)=1\in H(P)$, and if $\alt(P)\geq w_P\cdot r$, then we can embed a sequence $(+1,-1)^{w_P\cdot r/2}$ in $\snoose(P)$; by taking consecutive $g_i^{\pm 1}$ from the sets $H(a)$ for images of the embedding, and $1$ from all sets $H(a)$ for all the other arcs of $P$, we see that $\phi(P)\in H(P)$. The same argument holds for $Q$.

Without loss of generality assume then that $w_P\cdot r > \alt(P)$ and $w_P>0$. By somehow abusing notation, from now on we identify the undirected walk $P$ in $G^+$ with a naturally corresponding non-degenerated noose in $G$; by Observation~\ref{obs:circumventing} we may assume that $P$ is pretty. We consecutively extract nooses from $P$ keeping the invariant that $w_P\cdot r> \alt(P)$ and that $P$ is pretty. At each step we either find a simple noose with winding number $1$ and alternation at most $r+4$, which can be output by the algorithm, or shorten (with respect to some measure to be defined) noose $P$ keeping the invariant. Moreover, if the extraction step cannot be applied, then $P$ is already simple, has winding number $1$ and (by the invariant) at most $r$ alternations when treated as a noose, hence it can be output by the algorithm.

We proceed with a similar cutting scheme as in the proof of Observation~\ref{obs:subcurve}. Assume then that the noose $P$ is not simple, hence some vertex or face is visited more than once. Let us take a shortest interval on $P$ between two consecutive visits of the same face or vertex; by minimality it follows that we may partition $P$ into a simple noose $N$ traversing this interval and the resulting noose $P'$ that is $P$ with the interval cut out (in case of visiting the same face twice, we may need to add small connections within this face). Note that the winding number of $N$ is of absolute value at most $1$, as it is simple.

If the cutting was performed due to visiting the same face twice, we have a simple situation: $\snoose(P')$ is $\snoose(P)$ with $\snoose(N)$ carved out, so $\alt(P') \le \alt(P)$ and $\alt(P')\leq \alt(P)-\alt(N)+3$ (we may lose at most two alternations on the cut points and potentially one
    alternation on a cyclic shift of $\snoose(N)$). Hence if the winding number of $N$ equals zero or $\alt(N)\geq r+3$, we are keeping the invariant, and otherwise we may output $N$ (or $N$ reversed if its winding number is $-1$). Note that thus $P'$ is still pretty.

Assume now that cutting was performed due to visiting the same vertex $w$ twice. Note that $\snoose(P')$ is basically $\snoose(P)$ with $\snoose(N)$ carved out, only with possible manipulations on the term corresponding to passing through the vertex $w$. Again, it follows that $\alt(P')\leq \alt(P)-\alt(N)+5$: if we can embed some alternating sequence into a cyclic shift of $\snoose(P')$ and some other alternating sequence into a cyclic shift of $\snoose(N)$, then after removing terms corresponding to passing through $w$ from both embeddings, shifting the alternating sequence embedded into a cyclic shift of $\snoose(N)$ so that it corresponds to the shift rooted at $w$, and gluing the embedding together, we obtain an embedding into $\snoose(P)$ of an alternating sequence of length at least $\alt(N)+\alt(P')-5$ (we may lose additional $2$ alternations on gluing and one on shifting $\snoose(N)$). Hence, if $N$ has nonzero winding number, then either we may reduce $P$ by shortening it to $P'$ (in case $\alt(N)\geq r+5$) or output $N$ (in case $\alt(N)\leq r+4$). In order to ensure that $P'$ is still pretty, we may need to apply Observation~\ref{obs:circumventing} to it. Note that this application can only make a small circumvent of the vertex $w$, and by the assumptions that $C_1$ and $C_2$ are directed cycles, this circumvent will not make $P'$ go outside the ring.

Now consider the case when $N$ has zero winding number; we need to argue that $\alt(P')\leq \alt(P)$, as then we can shortcut $P$ to $P'$, possibly again applying Observation~\ref{obs:circumventing} to it. This, however, follows from the same argumentation as in the proof of Observation~\ref{obs:subcurve}: from the fact that $P$ is pretty we can argue that the term corresponding to passing $w$ in $\snoose(P')$ is contained in the union of terms corresponding to passing $w$ in $\snoose(P)$, which means that every sequence embeddable into $\snoose(P')$ is also embeddable into $\snoose(P)$.

We are left with arguing that the presented procedure will terminate in polynomial number of steps. Note that in every cutting step we either decrease the number of vertices visited by $P$, or do not increase the number of vertices visited by $P$ but decrease the number of faces and arcs visited by $P$. Hence, the maximum number of steps performed is at most the number of vertices visited by $P$ times the size of the graph.
\end{proof}

The following corollary will be used in the algorithm.

\begin{lemma}[Handling a large ring]\label{lem:largering}
Let $G$ be a directed graph on the plane with some terminals. Let $r$ be an integer and let $C_1$, $\dots$, $C_{2r+3}$ be an alternating sequence of
concentric cycles in $G$ with no terminal between $C_1$ and $C_{2r+3}$. Then there exists a polynomial-time algorithm that outputs either:
\begin{itemize}
\item a simple noose separating $C_1$ and $C_{2r+3}$ having at most $2r+8$ alternations, or
\item a vertex $v$ surrounded by a sequence of $r$ alternating concentric cycles with no terminals inside them.
\end{itemize}
\end{lemma}
\begin{proof}
Let $R$ be a ring with $C_1$ being the outer cycle and $C_{2r+3}$ being the inner cycle. We use the algorithm of Lemma~\ref{lem:dualaltcut} to either find an simple noose having at most $2r+8$ alternations that separate $C_1$ from $C_{2r+3}$, which can be returned by the algorithm, or an alternating join of size $2r+4$. Assume then that the join $P_1,P_2,\ldots,P_{2r+4}$ has been found, where the indices reflect clockwise order of the paths on the ring.

Let $v$ be any vertex that is on the intersection of cycle $C_{r+2}$ and $P_{r+2}$. We modify the ring $R$ using with the following operations. We delete the part of the ring between $P_{2r+3}$ and $P_1$, thus cutting through the ring and creating a graph with outer face boundary consisting of paths $P_1$, $P_{2r+3}$ and subpaths of cycles $C_1$ and $C_{2r+3}$ (see Fig.~\ref{fig:concentric-cycles}). Moreover, we delete the vertex $v$ from the graph and declare the face in which it was embedded the inner face. Thus, the resulting graph may be viewed as embedded into a ring $R'$ with the aforementioned inner and outer face.

\begin{figure}
\begin{center}
\includegraphics{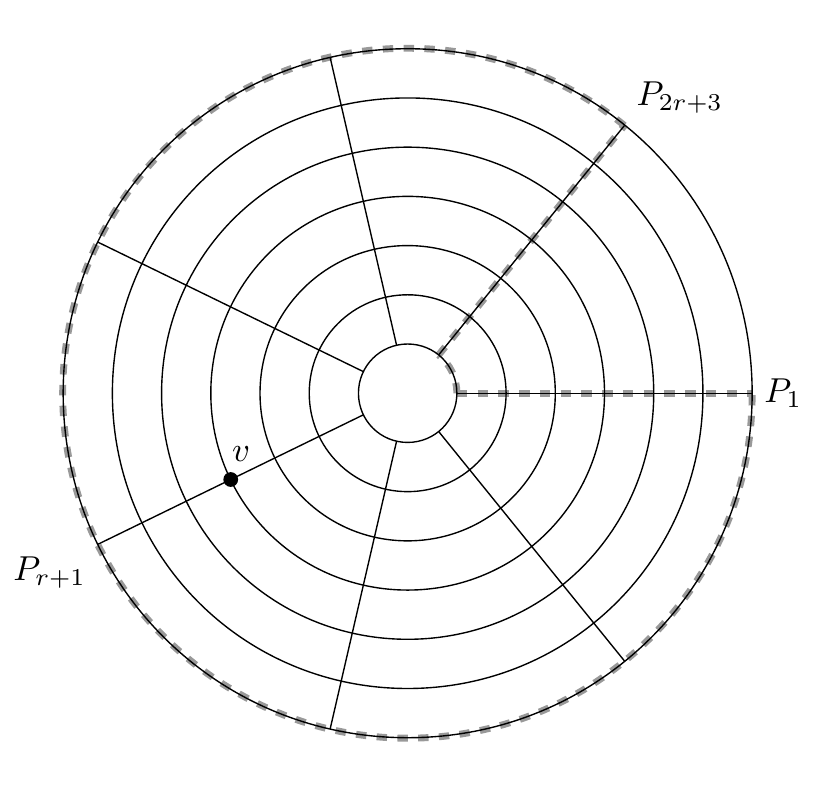}
\caption{The outer face of the ring $R'$ is marked with dashed gray.}
  \label{fig:concentric-cycles}
  \end{center}
  \end{figure}

We now apply Lemma~\ref{lem:dualcycles} to the ring $R'$. We either find $r$ alternating cycles around the inner face, which constitute $r$ alternating cycles around $v$ with no terminals embedded that can be returned by the algorithm, or a simple curve $M$ starting in the inner face and ending in the outer face, having at most $r$ alternations. Assume then that curve $M$ was found.

Assume first that $M$ reaches the part of the boundary of the outer face of $R'$ that is on $C_1$. Then $M$ must have passed through cycles $C_1,C_2,\ldots,C_{r+1}$; these passages define an embedding of an alternating sequence of size $r+1$ into $\snoose(M)$, which contradicts the fact that $\alt(M)\leq r$. Similarly, if $M$ reaches $C_{2r+3}$, $P_1$, or $P_{2r+3}$, then it must have passed through sequences of paths $(C_{r+3},C_{r+4},\ldots,C_{2r+3})$, $(P_{1},P_{2},\ldots,P_{r+1})$, or $(P_{r+3},P_{r+4},\ldots,P_{2r+3})$, respectively, and in each case we get a contradiction. Hence, the curve $M$ could not be found and we are done.
\end{proof}

\section{Decomposition}\label{sec:decomp}

In this section we show how to decompose the graph $G$ into a bounded number of weakly connected subgraphs (called henceforth {\em{components}}), such that
the interaction between the components is somehow limited, and the terminals are kept outside of the components.
We start with defining a notion of an alternation suitable for weakly connected subgraphs of $G$; the main property of our components
is that we control their alternation.
\begin{defin}[incident arcs]
Let $G$ be a graph and let $\anycomp$ be its subgraph.
By $\incarcs{\anycomp}$ we denote the set of arcs of $G$ incident to at least one vertex of $\anycomp$,
   but not belonging to $\anycomp$; that is, $\incarcs{\anycomp} = E(G) \setminus (E(\anycomp) \cup E(G[V(G) \setminus V(\anycomp)]))$.
\end{defin}
\begin{defin}[alternation of a face of a subgraph]
Let $G$ be a plane graph, let $\anycomp$ be a weakly connected subgraph of $G$ and let $f$ be a face of $\anycomp$.
Consider an undirected walk $P$ in $\anycomp$ that goes around the face $f$ in such a direction that it leaves the face $f$ to the right,
and the subgraph $\anycomp$ to the left; that is, $P$ goes counter-clockwise if $f$ is the outer face of $\anycomp$, and clockwise otherwise.
Consider a cyclic sequence $P(\incarcs{\anycomp})$ of elements of $\{+1,-1\}$ constructed as follows: we go along $\incarcs{\anycomp}$
and insert into $P(\incarcs{\anycomp})$ an element $+1$ or $-1$ for each endpoint of an arc of $\incarcs{\anycomp}$
we encounter along the walk $P$, depending on whether this is a starting or ending point of the arc.
If multiple arcs are encountered
at one vertex $v \in V(\anycomp)$, consider them in the counter-clockwise order as they appear on the face $f$.
The {\em{alternation}} of the face $f$ in $\anycomp$ is the alternation of the sequence $P(\incarcs{\anycomp})$.
\end{defin}
Note that each arc $e \in \incarcs{\anycomp}$ corresponds to exactly one entry in $P(\incarcs{\anycomp})$ if it has exactly one endpoint in $V(\anycomp)$,
and to two entries if it has both endpoints in $P(\incarcs{\anycomp})$.

We also note that
for any weakly connected subgraph $\anycomp$ and its face $f$, there exists a noose $N(f,\anycomp)$ that goes parallely and very closely to the walk $P$, is contained in the face $f$, does
not visit any vertex of $G$ and $\snoose(N(f,\anycomp)) = P(\incarcs{\anycomp})$ (up to a cyclic shift).
We call such a noose a {\em{border noose}} of $\anycomp$.
Moreover, we may assume that a border noose of a $\anycomp$
is sufficiently close to $\anycomp$, so that for any two disjoint weakly connected
subgraphs of $G$, their border nooses and the areas enclosed by them are disjoint.

We are now ready to define components in our decomposition.

\begin{defin}[disc component]
Let $G$ be a plane graph and let $\anycomp$ be a weakly connected subgraph of $G$.
We call $\anycomp$ a {\em{disc component}} of $G$ if every arc in $\incarcs{\anycomp}$ is contained in the outer face of $\anycomp$.
The {\em{alternation}} of a disc component $\anycomp$ is the alternation of its outer face.
\end{defin}

\begin{defin}[ring component]
Let $G$ be a plane graph and let $\anycomp$ be a weakly connected subgraph of $G$.
We call $\anycomp$ a {\em{ring component}} of $G$ if there exists a face $f_{IN}$ of $\anycomp$ different than its outer face (called the {\em{inner face}} of the component)
such that every arc in $\incarcs{\anycomp}$ is contained either in the outer face of $\anycomp$ or in $f_{IN}$.
The {\em{alternation}} of a ring component $\anycomp$ is the maximum of alternations of its outer face and the inner face $f_{IN}$.
\end{defin}
By $\alt(\anycomp)$ we denote an alternation of a component $\anycomp$.

Note in both component definitions, we do not require that the graph $\anycomp$ is an {\em{induced}} subgraph of $G$.
In other words, we allow arcs in $\incarcs{\anycomp}$ that have both endpoints in $\anycomp$. 

\begin{defin}[decomposition]
  Let $G$ be a plane graph having a set $T\subseteq  V(G)$ of terminals. 
  Then a set $\decomp$ of (disc and ring) components of $G$ is a {\em{decomposition}} of $G$ iff
  \begin{itemize}
   \item every vertex of $G$ is in exactly one component of $\decomp$;
   \item for each terminal $t \in T$ there exists a disc component $\disccomp_t \in \decomp$ that consists of the vertex $t$ only.
  \end{itemize}

   The {\em{disc (ring) alternation}} of the decomposition $\decomp$ is the maximum alternation of a disc (ring) component in $\decomp$.
\end{defin}
We will control two natural measures of a quality of a decomposition: its alternation and the number of its components.
Moreover, we will require that ring components are embedded into a decomposition in a special way.

\begin{defin}[isolating component]
Consider a decomposition $\decomp$ of a plane graph $G$ with terminals $T$. We say that a disc component $\disccomp \in \decomp$
is a {\em{$d$-isolating}} component if the subgraph of $G$ induced by the vertices of $\disccomp$ (i.e., the subgraph $\disccomp$ together
    with all the arcs of $\incarcs{\disccomp}$ that have both endpoints in $\disccomp$)
contains a sequence of $d$ alternating concentric cycles and each edge of $\incarcs{\disccomp}$
that has exactly one endpoint in $\disccomp$ either lies inside the innermost of these cycles (an {\em{inner edge}})
or lies outside of the outermost of these cycles (an {\em{outer edge}}).
\end{defin}
We note that, if the graph $G$ is weakly connected, in the space enclosed between the innermost and the outermost cycle from the definition
of the isolating component $\disccomp$ there are only arcs and vertices of $\disccomp$ and arcs of $\incarcs{\disccomp}$ that
have both endpoints in $\disccomp$.

\begin{defin}[ring isolation]\label{def:isolation}
Consider a decomposition $\decomp$ of a plane graph $G$ with terminals $T$. We say that the decomposition has {\em ring isolation $(\Lambda,d)$}
if for every ring component $\ringcomp \in \decomp$ with outer face $f_{OUT}$ inner face $f_{IN}$ there exist $2\Lambda$
  disc components $\disccomp_{IN,\lambda}$, $\disccomp_{OUT,\lambda} \in \decomp$, $1 \leq \lambda \leq \Lambda$, such that each of these components is $d$-isolating and:
 \begin{itemize}
   \item the components $\disccomp_{IN,\lambda}$, $1 \leq \lambda \leq \Lambda$, are contained inside $f_{IN}$ and the components
   $\disccomp_{OUT,\lambda}$, $1 \leq \lambda \leq \Lambda$ are contained in $f_{OUT}$;
   \item for each $1 \leq \lambda < \Lambda$, each inner arc of $\disccomp_{IN,\lambda}$ has the second endpoint in $\disccomp_{IN,\lambda+1}$,
   and each outer arc of $\disccomp_{OUT,\lambda}$ has the second endpoint in $\disccomp_{OUT,\lambda+1}$;
   \item for each $1 <\lambda \leq \Lambda$, each outer arc of $\disccomp_{IN,\lambda}$ has the second endpoint in $\disccomp_{IN,\lambda-1}$,
   and each inner arc of $\disccomp_{OUT,\lambda}$ has the second endpoint in $\disccomp_{OUT,\lambda-1}$;
   \item each outer arc of $\disccomp_{IN,1}$ and each inner arc of $\disccomp_{OUT,1}$ has the second endpoint in $\ringcomp$.
 \end{itemize}
 Moreover, we require that no disc component serves as an isolation to two ring components,
 that is, the components $\disccomp_{IN,\lambda}$ and $\disccomp_{OUT,\lambda}$, $1 \leq \lambda \leq \Lambda$ are pairwise distinct
 for different ring components of $\decomp$.
\end{defin}
Note that each of the $d$ alternating concentric cycles of $\disccomp_{IN,\lambda}$ (forming a $d$-isolation) is contained in the inner face $f_{IN}$,
while each of the $d$ alternating concentric cycles of $\disccomp_{OUT,\lambda}$ encloses $\ringcomp$.

\begin{figure}
\begin{center}
\includegraphics[width=0.5\textwidth]{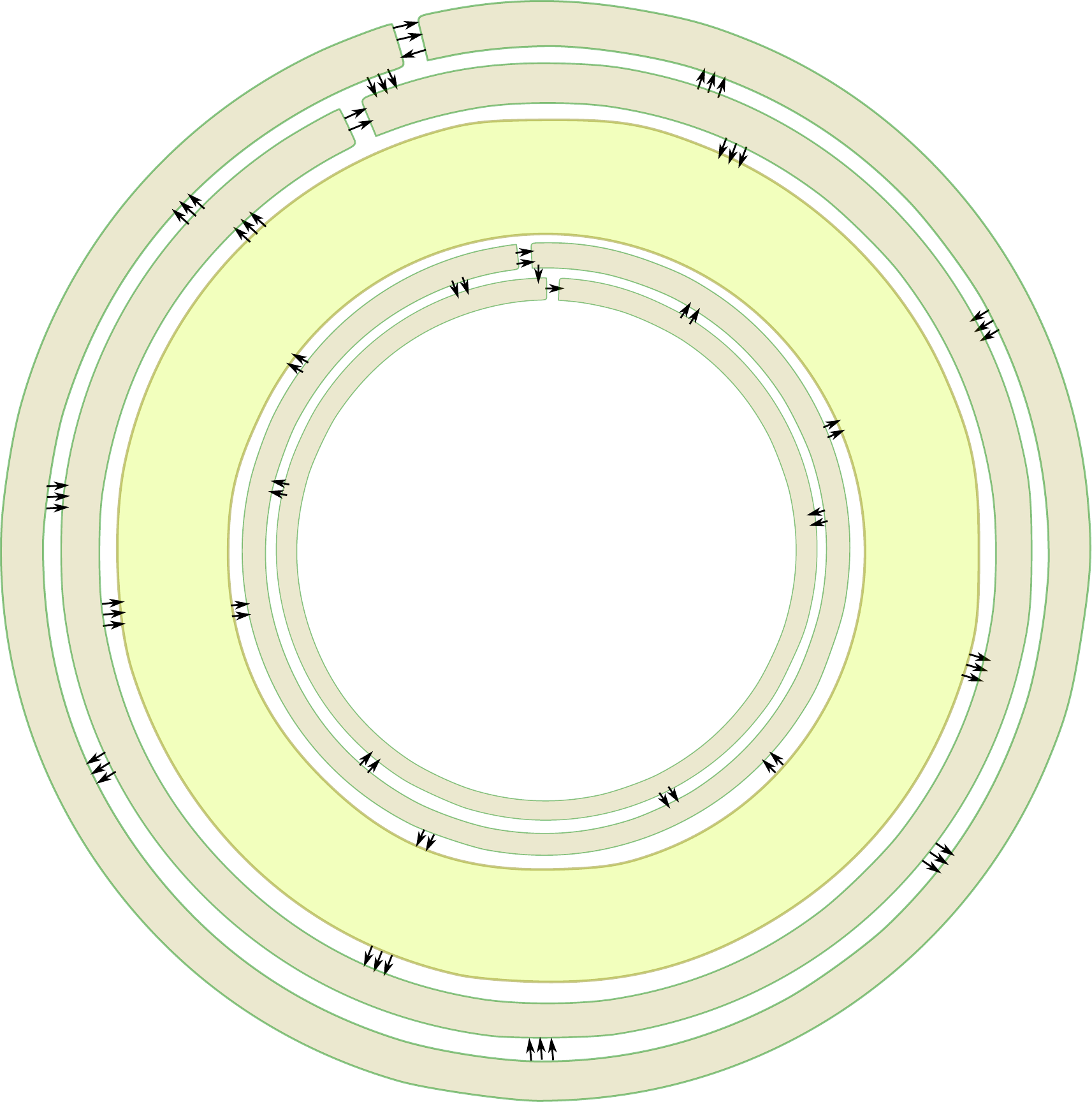}
\caption{A ring component with two levels of isolation.}
  \label{fig:spiraling-ring}
  \end{center}
  \end{figure}

We say that a decomposition has {\em{positive isolation}}
if it has isolation $(\Lambda,d)$ for some $\Lambda,d > 0$.

We are now ready to state our decomposition theorem.

\begin{theorem}[Main Decomposition Theorem]\label{th:finddecomp}
Let $G$ be a plane graph with $k$ terminal pairs, each of them having degree 1.
Let $\Lambda$, $d$ and $r$ be positive integers.
Then in $O^*(2^{O(\Lambda(d+r)k^2)})$ time we can either find a set
of $r$ alternating concentric cycles with no terminal inside the outermost
cycle, or compute a set of at most $2^{O(\Lambda(d+r)k^2)}$ pairs $(G_i, \decomp_i)$ where:
\begin{enumerate}
\item each $G_i$ is a plane graph of size polynomially bounded
in the size of $G$, with $k$ terminal pairs, where each terminal is of degree one
in $G_i$;
\item $\decomp_i$ is a decomposition of the graph $G_i$ with $O(\Lambda k^2)$ components, ring isolation $d$ and alternation $O(\Lambda(d+r)k^2)$;
\item $G$ is a YES-instance to \probshort{}, if and only if there exists $i$ such that $G_i$ is a YES-instance to \probshort{}.
\end{enumerate}
\end{theorem}

\begin{proof}
We decompose the graph in an iterative manner.
At $i$-th iteration, we are given a graph $G_0$ with $k$ terminal pairs
and family $\mathcal{C}$ of pairwise disjoint disc components of $G_0$
that are to be partitioned further.
Moreover, we assume that each terminal in $G_0$ has degree one.
Each iteration decreases the total number of terminals in the components
of $\mathcal{C}$, thus the iteration ends after at most $2k$ steps.

In each iteration, we first filter out any component $H \in \mathcal{C}$
that does not contain any terminals. Such a component may be output
as a disc component in the final decomposition.
If $\mathcal{C}$ becomes empty, we finish the algorithm.

Otherwise, we pick any component $\anycomp_0 \in \mathcal{C}$ to decompose
it further (and remove it from $\mathcal{C}$).
By $\kappa$ denote the number of terminals contained in $\anycomp_0$.
If there is a terminal on the outer face of $\anycomp$, we create a new disc component that contains it, delete the terminal and its incident arc from $\anycomp_0$ and add back
$\anycomp_0$ to $\mathcal{C}$.
Each such step produces a new disc component, decreases $\kappa$ and increases $\alt(\anycomp_0)$ by a constant.
Moreover, it maintains the connectivity of $\anycomp_0$.

Thus, from this point we may assume that $\anycomp_0$ contains $\kappa > 0$ terminals, but none of these terminals lies in the outer face of $\anycomp_0$.
Let us pick one terminal $t \in T \cap V(\anycomp_0)$ and let $F_t$ be the face that contains $t$. Let $\gamma_t$ be an arbitrary circle with centre $t$ and with sufficiently
small radius such that $\gamma_t \cap G_0$ consists of a single point which is an intersection of $\gamma_t$ with the arc incident to $t$.
Invoke Lemma \ref{lem:dualcycles} on the graph $G_0 \setminus \{t\}$, curves $\gamma_t$ and a border noose
$N=N(f_0, \anycomp_0)$, where $f_0$ is the outer face of $\anycomp_0$.
Find the largest even $a$ for which we obtain a sequence $C_1,C_2, \ldots, C_a$
of alternating concentric cycles in $\anycomp_0$. Moreover, invoke
Lemma \ref{lem:dualcycles} for $a+2$ and obtain a simple curve $M$ connecting $\gamma_t$ and $N$ with at most $a+2$
alternations.

Assume $C_1$ is the innermost and $C_a$ is the outermost
of the constructed cycles. For $s=1,2,\ldots,\kappa$
let $i(s)$ be the largest integer $0 \leq i(s) \leq a$ such that
there are at most $s$ terminals contained in the area
enclosed by the cycle $C_{i(s)}$ (note that no terminal lies on any of the cycles $C_1,C_2, \ldots, C_a$, as terminals are of degree one in $G_0$). Note that $i(\kappa) = a$; we
denote $i(0) = 0$.

Let $q = 2r+1$ and let $S^{\textrm{big}} \subseteq \{1,2,\ldots,\kappa\}$ be the set of such
integers that $s \in S^{\textrm{big}}$ whenever $i(s) - i(s-1) \geq 2\Lambda d + 2(\Lambda + 1)q + 2$.
For each $s \in S^{\textrm{big}}$, apply Lemma \ref{lem:largering} to the following
$2(\Lambda + 1)$ sets of $q+2=2r+3$ concentric cycles:
\begin{itemize}
\item $C_{i(s-1)+(\Lambda-\lambda)(d+q) + 1}, C_{i(s-1)+(\Lambda-\lambda)(d+q)+2}, \ldots, C_{i(s-1)+(\Lambda-\lambda)(d+q) + q+2}$ for $0 \leq \lambda \leq \Lambda$;
\item $C_{i(s)-(\Lambda-\lambda)(d+q)-q-1}, C_{i(s)-(\Lambda-\lambda)(d+q)-q}, \ldots, C_{i(s)-(\Lambda-\lambda)(d+q)}$ for $0 \leq \lambda \leq \Lambda$.
\end{itemize}
If any of the applications of Lemma \ref{lem:largering} returns
  a sequence of $r$ concentric cycles, we return it as an outcome of the algorithm.
Otherwise, we obtain $2\Lambda+2$ circular cuts $N_s^{IN,\lambda}$ and $N_s^{OUT,\lambda}$, $0 \leq \lambda \leq \Lambda$:
the cut $N_s^{IN,\lambda}$ separates
$C_{i(s-1)+(\Lambda-\lambda)(d+q) + 1}$ from $C_{i(s-1)+(\Lambda-\lambda)(d+q) + q+2}$
and the cut $N_s^{OUT,\lambda}$ separates
$C_{i(s)-(\Lambda-\lambda)(d+q)-q-1}$ from $C_{i(s)-(\Lambda-\lambda)(d+q)}$.
As $i(s) - i(s-1) \geq 2\Lambda d + 2(\Lambda + 1)q + 2$, we have that
$$i(s-1) + \Lambda(d+q) + q+1 \leq i(s) - \Lambda (d+q) -q -1$$
and the cycle $C_{i(s-1)+\Lambda(d+q)+q+2}$ lies between $N_s^{IN,0}$ and $N_s^{OUT,0}$.

Note that, for any $1 \leq \lambda \leq \Lambda$:
\begin{itemize}
\item the $d$ cycles $C_{i(s-1)+(\Lambda-\lambda)(d+q)+q+2}, C_{i(s-1)+(\Lambda-\lambda)(d+q)+q+3}, \ldots, C_{i(s-1)+(\Lambda-\lambda+1)(d+q)+1}$ are contained between $N_s^{IN,\lambda}$ and $N_s^{IN,\lambda-1}$,
\item the $d$ cycles $C_{i(s)-(\Lambda-\lambda-1)(d+q)}, C_{i(s)-(\Lambda-\lambda-1)(d+q)+1}, \ldots, C_{i(s)-(\Lambda-\lambda)(d+q)-q-1}$ are contained between $N_s^{OUT,\lambda-1}$ and $N_s^{OUT,\lambda}$.
\end{itemize}

As the simple curve $M$ connects $\gamma_t$ with $N$, it crosses all cycles $C_i$
as well as all circular cuts $N_s^{IN,\lambda}$, $N_s^{OUT,\lambda}$ for $s \in S^{\textrm{big}}$,
$1 \leq \lambda \leq \Lambda$. By slightly perturbing $M$, we may assume that
$M$ crosses these circular cuts either in a vertex of $G$ or inside a face of $G$,
  never on an edge of $G$.
For each $s \in S^{\textrm{big}}$, remove from
$M$ the subcurve between the intersection with $N_s^{IN,0}$ that is closest to $\gamma_t$
on $M$ and the intersection with $N_s^{OUT,0}$ that is closest to $N$ on $M$;
if the intersection happens in the vertex of $G$, we remove a slightly shorter
part of $M$, so that the remaining curve still traverses the intersection vertex
and ends in a face of $G$.
Denote by $\mathcal{M}$ the set of remaining subcurves of $M$ (note that
the are at most $|S^{\textrm{big}}|+1 \leq \kappa+1$ of them).

We now note that all curves in $\mathcal{M}$ intersect $O(\Lambda(d+r)\kappa)$ alternating
cycles from the sequence $C_1,C_2,\ldots, C_a$, as they intersect
$O(\Lambda(d+r))$ cycles between $C_{i(s-1)+1}$ and $C_{i(s)}$ for each $1 \leq s \leq \kappa$.
Recall that the alternation of $M$ is at most $a+2$, while $M$ intersects
all cycles $C_1,C_2,\ldots, C_a$. Therefore the sum of alternations of all curves
in $\mathcal{M}$ is $O(\Lambda(d+r)\kappa)$.


We are now going to cut the graph along the circular cuts $N_s^{IN,\lambda}$ and $N_s^{OUT,\lambda}$ for $s \in S^\textrm{big}$, $1 \leq \lambda \leq \Lambda$, as well as along the curves of $\mathcal{M}$.
However, both circular cuts and the curves in $\mathcal{M}$ may traverse a vertex, whereas
we want to cut arcs only.
To cope with this, we introduce the following bounded search tree strategy.

Recall that any curve of $\mathcal{M}$ as well as any cut $N_s^{IN,\lambda}$, $N_s^{OUT,\lambda}$ is
a simple curve. 
Therefore we can apply Observation~\ref{obs:circumventing} to each of
those curves, making them pretty.
Let $N_0$ be any of these curves, and let $v$ be a vertex on $N_0$.
By the definition of a pretty curve (Definition~\ref{def:pretty}), the set $S_v \in \snoose(N_0)$ that corresponds
to $v$ equals $\{+1,-1\}$ or $\emptyset$.

For a given curve $N_0$, we first resolve all vertices with $S_v = \{+1,-1\}$.
For each such vertex, we branch into four cases. We choose one side of
the curve $N_0$ (left or right) around the vertex $v$ and one type of arcs incident to $v$
(incoming or outgoing arcs). In each branch, we delete from the graph $G_0$ the arcs
of the chosen type from the chosen side of the curve
and perturb $N_0$ slightly to omit the vertex $v$ from the
chosen side. In this way, we do not increase $\alt(N_0)$, as
in $\snoose(N_0)$ we exchange the $S_v = \{+1,-1\}$ 
for a sequence (of arbitrary length) of equal one-element sets, corresponding
to the remaining arcs from the chosen side of $N_0$. 
Moreover, for any solution to the \probshort{} problem,
there exists a choice where we do not delete any arcs of the solution,
as the vertex $v$ may lie only on one path of the solution, and this path
cannot enter and leave the vertex $v$ from both sides of $N_0$.

Note that the number of these vertices is at most $\alt(N_0)$, thus we create 
at most $4^{\alt(N_0)}$ branches.

Once we have resolved all vertices with $S_v = \{+1,-1\}$, we move to the second
case. For a vertex $v$ we may have $S_v = \emptyset$ in two situations.
First, $N_0$ may visit the vertex $v$, but 
leaves all incident edges on one of its sides. 
However this would contradict the assumption that $N_0$ is non-degenerate.

Second, $v$ may be a source or a sink. We consider here two options: either we modify
the curve $N_0$ to omit the vertex $v$ from the left or from the right.
Let us investigate how such a change will influence $\snoose(N_0)$ and $\alt(N_0)$.
We replace $S_v = \emptyset$ with a sequence (of arbitrary length) of equal one-element
sets corresponding to the traversed arcs of $G$; such a situation may increase $\alt(N_0)$
by one. However, we have a choice of whether we insert a sequence of sets $\{+1\}$
or $\{-1\}$: if we omit the vertex $v$ from the left, we cross the arcs incident
to it in a different orientation than if we omit the vertex $v$ from the right.

If in $\snoose(N_0)$ there exists a vertex $v$ with $S_v = \emptyset$, but one neighbouring
set $S$ with $|S|=1$, we may modify $N_0$ around $v$ so that we replace $S_v$
with a sequence of sets equal to $S$. It may be easily seen that if we apply this operation to all the sets $S_v$ equal to $\emptyset$, then the alternation of $N_0$ does not increase in case $N_0$ is a non-closed curve (from $\mathcal{M}$) or increases by at most $1$ in case $N_0$ is a noose (of form $N_s^{IN,\lambda}$ or $N_s^{OUT,\lambda}$). As we have removed all two-element sets from $\snoose(N_0)$, we may not perform
the above operation only if $\snoose(N_0)$ consists only of empty sets. In this case,
we modify $N_0$ around each vertex, so that $\snoose(N_0)$ is a sequence of sets $\{+1\}$;
note that the alternation of $N_0$ changes from $0$ to $1$ in this case.

Let us now summarize. Recall that:
\begin{itemize}
\item $|\mathcal{M}| \leq \kappa+1$;
\item there are at most $2(\Lambda+1)\kappa$ curves
  $N_s^{IN,\lambda}$, $N_s^{OUT,\lambda}$, $s \in S^{\textrm{big}}$, $0 \leq \lambda \leq \Lambda$;
\item the sum of alternations of all curves of $\mathcal{M}$ is $O(\Lambda(d+r)\kappa)$;
\item each curve $N_s^{IN,\lambda}$ and $N_s^{OUT,\lambda}$ has alternation at most $2r+8$ (by Lemma \ref{lem:largering}).
\end{itemize}
Therefore the aforementioned procedure generates
$4^{O(\Lambda(d+r)\kappa)}$ subcases, and increases alternation
of each curve by at most one.

Before we describe the outcome of the partitioning of $\anycomp_0$,
let us define the notion of {\em{connectifying}} a component.
Assume we are given an open connected subset $A$ of the plane, isomorphic to a disc or to a ring,
such that no vertex of $G_0$ lies on the border of $A$ (i.e., not in $A$, but in the closure of $A$).
For each of the (one or two) borders of $A$ that are homeomorphic to a circle, travel along the border
and, for each arc that it intersects, subdivide it, inserting a new vertex inside $A$. For each two consecutive
newly added vertices, connect them with a length-two path inside $A$, where the middle vertex of the path is a sink (i.e.,
both arcs from the path point from the subdivided arcs of $G_0$ towards the vertex in the middle).
As $G_0$ is weakly connected, after this operation, the subgraph of $G_0$ consisting of all arcs and vertices completely contained in $A$
is weakly connected, whereas the answer to \probshort{} on $G_0$ does not change, as the added arcs
are useless from the point of view of constructing directed paths.

We are now ready to partition the graph into the following components.
\begin{itemize}
\item We create a disc component containing the terminal $t$ only.
\item For each $s \in S^\textrm{big}$, we create the following $2\Lambda+1$ components.
  \begin{itemize}
  \item A ring component $\ringcomp_s$ that is a connectification
  of a subgraph consisting all vertices and edges
  of the graph $G_0$ contained between the noose $N_s^{IN,0}$ and the noose $N_s^{OUT,0}$.
  The face of $\ringcomp_s$ that contains $N_s^{IN,0}$ is the inner face of $\ringcomp_s$,
  and the face that contains $N_s^{OUT,0}$ is the outer face.
  Note that, as $\ringcomp_s$ contains the cycle $C_{i(s-1)+\Lambda(d+q)+q+2}$,
  these faces are distinct.
  \item For any $1 \leq \lambda \leq \Lambda$, 
  a disc component $\disccomp_{s,IN,\lambda}$ that is a connectification
  of a subgraph of $G_0$ enclosed in the area
  with its border being a concatenation of the following four
  curves: a minimal segment of a curve of $\mathcal{M}$ connecting $N_s^{IN,\lambda-1}$ with $N_s^{IN,\lambda}$,
  the curve $N_s^{IN,\lambda}$, the same minimal segment of a curve of $\mathcal{M}$,
  but now traversed backwards, and the curve $N_s^{IN,\lambda-1}$.
  Note that, as the chosen subset of the plane is homeomorphic to a disc, 
  the $\disccomp_{s,IN,\lambda}$ is in fact a disc component.
  \item For any $1 \leq \lambda \leq \Lambda$,
  a disc component $\disccomp_{s,OUT,\lambda}$
  created in the same manner as $\disccomp_{s,IN,\lambda}$,
  but between curves $N_s^{OUT,\lambda-1}$ and $N_s^{OUT,\lambda}$.
  \end{itemize}
  We note that the $d$ alternating cycles $C_{i(s-1)+(\Lambda-\lambda)(d+q)+q+2}, \ldots, C_{i(s-1)+(\Lambda-\lambda+1)(d+q)+1}$ are 
  contained in the subgraph of $G$ induced by the vertices of $\disccomp_{s,IN,\lambda}$,
  fulfilling all the requirements to make $\disccomp_{s,IN,\lambda}$ a $d$-isolating component.
  Similarly, the components $\disccomp_{s,OUT,\lambda}$ are $d$-isolating as well.
  Therefore $\ringcomp_s$ has $(\Lambda,d)$-isolation, as required.
\item For each $s \in S^\textrm{big} \cup \{\kappa+1\}$, we insert into
  $\mathcal{C}$ the connectification of a subgraph enclosed by the following noose.
  We denote $N_0^{OUT,\Lambda} = \gamma_t$ and $N_{\kappa+1}^{IN,\Lambda} = N$. We concatenate
  a minimal segment of the curve of $\mathcal{M}$ connecting $N_{{\rm pred}(s)}^{OUT,\Lambda}$ and
  $N_s^{IN,\Lambda}$, the curve $N_s^{IN,\Lambda}$, again the same minimal segment but traversed backwards,
  and the curve $N_{{\rm pred}(s)}^{OUT,\Lambda}$, where ${\rm pred}(s)$ is
  the maximum element of $S^\textrm{big}$ smaller than $s$, or $0$ if
  it does not exist.
  Note that, as in the case of previous two components, the aforementioned area
  in the plane is homeomorphic to a disc, thus we insert into $\mathcal{C}$
  disc components only.
\end{itemize}
We note that at each step, we invoke the connectification algorithm
a few times, but in pairwise disjoint subsets of the plane. Therefore,
the enhanced graph $G_0$ remains plane.

Note that, after this step, the total number of terminals in $\mathcal{C}$
decreased by one, as the terminal $t$ is put in its own disc component
and each other terminal is put into exactly one new recursive call.
Therefore, the number of iterations is at most $2k$.
Moreover, as the step creates at most $\kappa$ ring components
and $2\Lambda\kappa$ disc components that are not components with a terminal,
we obtain at most $O(\Lambda k^2)$ components in total
and at most $2^{O(\Lambda(d+r)k^2)}$ subcases.

Let us now bound the alternations of the constructed components.
Each constructed ring component $\ringcomp_s$
has two borders $N_s^{IN,0}$ and $N_s^{OUT,0}$ of alternation at most $2r+8$, with possible slight modifications due to the branching procedure that can add $1$ to the alternation,
 thus its alternation is at most $2r+9$.
Similarly, each isolating component $\disccomp_{s,IN,\lambda}$ and $\disccomp_{s,OUT,\lambda}$ is surrounded
by two circular cuts of alternation $2r+9$ each and two subcurves of a
curve of $\mathcal{M}$, again with possible slight modifications due to the branching procedure that can add $1$ to the alternation. Recall that the sum of alternations of all curves
in $\mathcal{M}$ is $O(\Lambda(d+r)\kappa)$, thus each isolating disc component
has alternation $O(\Lambda(d+r)k)$.
Moreover, $\gamma_t$ has alternation $1$.

At each iteration, we put into $\mathcal{C}$
components surrounded by two copies of a subcurve of a curve from $\mathcal{M}$
and two nooses being either circular cuts $N_s^{IN,\Lambda}$, $N_s^{OUT,\Lambda}$, the curve $\gamma_t$ or $N$.
An operation of cutting away a disc component with a terminal may increase
the alternation of a component by at most one per terminal, thus $2k$ in total.
We infer that at iteration $i$, any component in $\mathcal{C}$ has alternation
bounded by $O(i(\Lambda(d+r)k))$.
  Therefore any computed
  disc component has alternation $O(\Lambda(d+r)k^2)$.

This concludes the proof of the decomposition theorem.
\end{proof}

\section{Bundles and bundle words}
\label{sec:bundles-and-bundle-words}

Our decomposition theorem, Theorem \ref{th:finddecomp},
provides us either with a situation where an irrelevant vertex
can be found, or with a bounded number of subcases, each of the subcase
being a modified graph $G$ with a decomposition $\decomp$
of bounded number of components and with bounded alternation.
In this section we focus on solving one fixed subcase, 
that is, we focus on a single graph $G$ with a decomposition $\decomp$.

\subsection{Bundle arcs, bundles and component multigraph}

Consider a graph $G$ with its decomposition $\decomp$.
For any $v \in V(G)$, let $\anycomp(v)$ be the component of $\decomp$
that contains $v$.
Let $\barcs(G,\decomp)$ be the set of arcs of $G$ that are not contained
in any component of $\decomp$, that is, $\barcs(G,\decomp) = E(G) \setminus \bigcup_{\anycomp \in \decomp} E(\anycomp)$. Any element of $\barcs(G,\decomp)$ is called
a {\em{bundle arc}}.

The set of bundle arcs form a structure of a multigraph on the set of
components $\decomp$; we call this multigraph a {\em{component multigraph}}.
If the number of components and the alternation of $\decomp$ is bounded,
the set of arcs of the component multigraph can be decomposed into a bounded number of {\em{bundles}} of arcs that go parallely and in the same direction.
We now formalize these notions.

\begin{defin}[component multigraph]
Let $G$ be a plane graph and let $\decomp$ be its decomposition.
The {\em{component multigraph}} of $G$ and $\decomp$, denoted
$\compgraph(G,\decomp)$, is a multigraph with vertex set $\decomp$
and arcs set $\{(\anycomp(u),\anycomp(v)): (u,v) \in \barcs(G,\decomp)\}$.
\end{defin}
Note that the component multigraph is planar, and the embedding
of $G$ naturally imposes a non-standard embedding of $\compgraph(G,\decomp)$, where each
disc component is contracted into a single point and each ring component is contracted into a closed curve, separating
the part of the graph inside the ring component from the part outside it.

We sometimes abuse the notation and
identify the bundle arc being an arc in the multigraph with the vertex set $\decomp$
with the corresponding arc in $G$.

We now define a notion of a {\em{bundle}}, that gathers together bundle arcs that ``serve the same role''.
\begin{defin}[bundle]
Let $G$ be a graph and let $\decomp$ be its decomposition.
A sequence $B=(b_1,b_2,\ldots,b_s)$ of bundle arcs is called a {\em{bundle}} if
\begin{itemize}
  \item there exist two components $\anycomp_1,\anycomp_2$, such that each bundle arc $b_i$ leads from $\anycomp_1$ to $\anycomp_2$ (possibly $\anycomp_1=\anycomp_2$);
  \item for any $1 < i \leq s$,
  in the graph $\anycomp_1 \cup \anycomp_2 \cup \{b_{i-1}, b_i\}$
  the unique face to the left of $b_{i-1}$ and to the right of $b_i$
  is empty, that is, does not contain any point of the embedding of the graph $G$.
\end{itemize}
\end{defin}

Note that there may exist other arcs from $H_1$ to $H_2$ in $G$, that do not belong to $B$.

Consider faces of $G$. We say that a face is a {\emph{component face}} if it belongs to some component $\anycomp$; otherwise it is called a {\emph{mortar face}}. A mortar face $F$ can be either a {\emph{bundle face}} if all the bundle arcs on the boundary of $F$ belong to the same bundle (in which case there is at most $2$ of them), or an {\emph{end-face}} otherwise.

\begin{lemma}[bundle in the dual]\label{lem:bundle-dual}
Let $G$ be a graph, $\decomp$ be its decomposition and $B=(b_1,b_2,\ldots,b_s)$ be a bundle.
Then $B$ is a directed path or a directed cycle in the dual of $G$.
\end{lemma}
\begin{proof}
Let $B$ connect $\anycomp_1$ with $\anycomp_2$ in $\decomp$.
By the definition of a bundle, in the graph $\anycomp_1 \cup \anycomp_2 \cup B$
for any $1 < i \leq s$ there exists a face $f_i$ whose border
consists of $b_{i-1}$, $b_i$ and parts of $\anycomp_1$ and $\anycomp_2$;
moreover, $f_i$ lies to the left of $b_{i-1}$ and to the right of $b_i$.
Therefore, in the dual of $G$, for $1 < i \leq s$, $b_i$ is an arc between $f_{i-1}$
and $f_i$. Moreover, no face $f_i$ is an endpoint of $b_s$ nor a starting point of
$b_1$. Therefore $B$ is indeed a directed path or a cycle in the dual of $G$.
\end{proof}

We now show that, given a graph $G$ with a decomposition $\decomp$ with bounded
alternation, we can efficiently
partition the bundle arcs into a bounded number of bundles.
\begin{lemma}[bundle recognition]\label{lem:bundle-recognition}
Given a graph $G$ embedded in a plane, together with its decomposition $\decomp$ of alternation $\alt(\decomp)$, one can in polynomial time
partition the bundle arcs into a set $\mathcal{B}$ of {\em{bundles}}, such that each component $\anycomp \in \decomp$ with
is incident to $O(\alt(\decomp)|\decomp| + |\decomp|^2)$ bundles.
\end{lemma}

\begin{proof}
Consider first any two components $\anycomp_1$ and $\anycomp_2$, $\anycomp_1 \neq \anycomp_2$.
We are to partition the bundle arcs with one endpoint in $\anycomp_1$ and second endpoint in $\anycomp_2$ into $O(\alt(\decomp)+|\decomp|)$ bundles.
If there is a constant number of them, the task is trivial, so let us assume there are at least three bundle arcs between these two components.

Note that, if any of the components $\anycomp_1,\anycomp_2$ is a ring component, the bundle arcs between $\anycomp_1$ and $\anycomp_2$ lie either in the inner face of 
the ring component, or in the outer face. Moreover, due the assumption of isolation
it may not happen that both $\anycomp_1$ and $\anycomp_2$ are ring components.
Let $b_1,b_2,\ldots,b_m$ be the bundle arcs with one endpoint in $\anycomp_1$ and the second endpoint in $\anycomp_2$,
 in the order of their appearance on a walk around the face of $\anycomp_1$ that contains $\anycomp_2$ (i.e., the outer face
    if $\anycomp_1$ is a disc component, and the outer or the inner face if $\anycomp_1$ is a ring component).
By planarity, this is also a reversed order of they appearance in a walk around the face of $\anycomp_2$
that contains $\anycomp_1$.

Denote $b_0=b_m$ and let $A_i$, $1 \leq i \leq m$, be the unique face
of the graph $\anycomp_1 \cup \anycomp_2 \cup \{b_{i-1}, b_i\}$ that does not contain
any arc $b_j$ for $j \not\in \{i-1,i\}$.
Note that $b_{i-1}$ and $b_i$ may not be grouped together in the same bundle for the following reasons.
\begin{enumerate}
\item $b_{i-1}$ and $b_i$ go in different directions, that is, one of this bundle arcs leads from $\anycomp_1$ to $\anycomp_2$ and the second one from $\anycomp_2$ to $\anycomp_1$.
However, this may happen at most $O(\alt(\decomp))$ times.
\item There is a component inside $A_i$; as the areas of the components are disjoint and disjoint with bundle arcs, this may happen at most $|\decomp|$ times.
\item There is a bundle arc that is a self-loop of $\anycomp_1$ or $\anycomp_2$ in the component multigraph contained in $A_i$. However, as both its endpoints are contained in $A_i$,
  this happens at most $O(\alt(\decomp))$ times.
\end{enumerate}
If none of the aforementioned events happens between the arcs $b_i, b_{i+1}, b_{i+2}, \ldots, b_j$, then all these arcs form a bundle (possibly in the reversed order).
As there are $O(\alt(\decomp)+|\decomp|)$ aforementioned
events, the claim is proven.

We are left with the task of partitioning self-loops (in the component multigraph) of an arbitrarily chosen component $\anycomp$ into $O(\alt(\decomp)+|\decomp|)$ bundles.
The arguments are similar to the previous case, but the topology of the situation is a bit different.

We focus on bundle arcs that lie in a single face $f$ of $\anycomp$; there is
one such face in the case of a disc component and two such faces in the case
of a ring component. Denote the set of bundle arcs in question by $\mathcal{B}$.
The bundle arcs of $\mathcal{B}$ partition $f$ into a number of areas; let us denote the set of this areas as $\mathcal{A}$.
Let us define a (natural) structure of a tree on the set $\mathcal{A}$: two areas are adjacent in the tree $T_\mathcal{A}$ if they share a bundle arc of $\mathcal{B}$
as a part of their borders.
Note that the edges of $T_\mathcal{A}$ are in a bijection with the set $\mathcal{B}$.

First, note that the edge of $T_\mathcal{A}$ incident to a leaf in the tree $T_\mathcal{A}$ corresponds to a bundle arc $b \in \mathcal{B}$ whose both endpoints are subsequent
intersections of $\mathcal{B}$ with a border noose of $\anycomp$.
Therefore, the number of leaves of $T_\mathcal{A}$ is bounded by
$\alt(\anycomp) \leq \alt(\decomp)$.
Consequently, the same bound holds for the number of vertices of $T_\mathcal{A}$ of degree at least three.

Second, as in the case of bundle arcs between two different components, there are:
\begin{enumerate}
\item at most $|\decomp|$ areas $A \in \mathcal{A}$ that contain a component inside;
\item $O(\alt(\decomp))$ areas $A \in \mathcal{A}$ that are of degree two in $T_\mathcal{A}$, but the two bundle arcs on the border of $A$
do not form a bundle, as they are both oriented clockwise or both oriented counter-clockwise around the area $A$;
\end{enumerate}
These $O(\alt(\decomp)+|\decomp|)$ vertices of $T_\mathcal{A}$, together with at most $\alt(\decomp)$ vertices of degree at least $3$, partition the set edges of $T_\mathcal{A}$
into a set of $O(\alt(\decomp)+|\decomp|)$ paths, and the edges on each path form a bundle. This concludes the proof of the lemma.
\end{proof}

The ability to partition the bundle arcs of a given decomposition
into bundles motivates the following definition.
\begin{defin}[bundle multigraph]
For a plane graph $G$, its decomposition $\decomp$
and partition $\bundleset$ of bundle arcs into bundles,
we define a {\em{bundle multigraph}} $\bundlegraph(G,\decomp,\bundleset)$
as a multigraph obtained from $\compgraph(G,\decomp)$ by identifying
bundle arcs of one bundle in $\bundleset$.
In other words, $\bundlegraph(G,\decomp,\bundleset)$ is a multigraph
with vertex set $\decomp$ and edge set $\bundleset$.
\end{defin}
Note that $\bundlegraph(G,\decomp,\bundleset)$ is planar
and the non-standard embedding of $\compgraph(G,\decomp)$ naturally imposes
a non-standard embedding of $\bundlegraph(G,\decomp,\bundleset)$
(i.e., where each disc component is represented by a point, and each ring component
 by a closed curve).

\begin{defin}[bundled instance]
A {\em{bundled instance}} is a tuple consisting of a \probshort{} input $G$,
  its decomposition $\decomp$ of positive isolation
  and partition into bundles $\bundleset$,
  equipped with an embedding of $G$ into a plane together with corresponding
  embeddings of $\compgraph(G,\decomp)$ and $\bundlegraph(G,\decomp,\bundleset)$.
\end{defin}
We somewhat abuse the notation and denote a bundled instance
as $(G,\decomp,\bundleset)$, making the embeddings implicit.

\subsection{Bundle words and spirals}

The following notion is crucial for the rest of this section.
\begin{defin}[bundle word]
Let $(G,\decomp,\bundleset)$ be a bundled instance.
Then, for a directed walk $P$ in $G$ we define a word $\bundleword(P)$
over the alphabet $\bundleset$, called the {\em{bundle word of $P$}},
as follows: we follow $P$ and for each bundle arc $e$
visited by $P$ we append the bundle $B \in \bundleset$ that contains $e$
to $\bundleword(P)$.
\end{defin}
Intuitively, $\bundleword(P)$ describes precisely how does the path $P$ circulate
around the component multigraph $\compgraph(G,\decomp)$. Note that it does not distinguish
different bundle arcs inside one bundle.

Assume for a moment that $\decomp$ does not contain any ring components.
Consider a solution $(P_i)_{i=1}^k$ to \probshort{} in $G$.
The core observation is that if we know $\bundleword(P_i)$ for each $1 \leq i \leq k$,
then the \probshort{} problem can be solved using the cohomology algorithm
of Schrijver \cite{schrijver:xp}; a generalization of this statement, including some technical difficulties around ring components,
   is proven in Section \ref{s:words-to-paths}.
The goal of Section \ref{s:word-guessing} is to show that there is only a bounded number
of {\em{reasonable}} choices for $\bundleword(P_i)$; in other words (but still ignoring the ring components) we show
how to construct an $f(k)$-sized family of sequences $(p_i)_{i=1}^k$
such that if $G$ is a yes instance to \probshort{}, then there exists
a solution $(P_i)_{i=1}^k$ and a sequence $(p_i)_{i=1}^k$ such that
$\bundleword(P_i) = p_i$ for each $1 \leq i \leq k$.
In this section we prepare some preliminary results for
both claims.

We start with an insight into properties of bundle words
$(\bundleword(P_i))_{i=1}^k$ for a solution to \probshort{}.
\begin{defin}[spiral in a bundle word]
Let $P$ be a directed path in a bundled instance $(G,\decomp,\bundleset)$.
A subword $w$ of $\bundleword(P)$ is called a {\em{spiral}}
if it starts and ends with the same bundle $B$, but all other letters in $w$
are pairwise distinct and different than $B$.
\end{defin}
In other words, a {\em{spiral}} corresponds to one `turn' of a path, between two arcs in the same bundle.

Note that a spiral $w$ in $\bundleword(P)$ corresponds to a closed walk
without self-intersections in $\bundlegraph(G,\decomp,\bundleset)$.
Moreover, no arc incident to a terminal may be a part of a spiral,
as terminals are of degree one in $G$, and each terminal is in a single-vertex disc component of $\decomp$.

We now make a key observation about the spirals.

\begin{defin}[spiral cut]
Let $P$ be a directed path in a bundled instance
$(G,\decomp,\bundleset)$ and let $w$ be a spiral on $P$.
Assume that the first and the last symbol of $w$ is
$B$ and let $b^1, b^2 \in B$ be two arcs on $P$
that correspond to these two occurrences of $B$ in $w$.
Then a {\em{spiral cut}} associated with the spiral $w$
of the path $P$ is a closed curve $\gamma$ that consists
of a subset of the drawing of $P$ between a midpoint of $b^1$ and
a midpoint of $b^2$, and a drawing of a directed
path in the dual of $G$ that connects $b^1$ and $b^2$
and uses only the arcs of $B$ (an existence of such path is guaranteed
by Lemma \ref{lem:bundle-dual}).
\end{defin}

We would like to note that a spiral cut is degenerate (see Definition~\ref{def:non-degenerate}),
as it goes along an arc.

\begin{lemma}\label{lem:spiral-split}
Let $(G,\decomp,\bundleset)$ be a bundled instance,
let $P$ be a path in $G$, and let $w$ be a spiral in $\bundleword(P)$
that starts and ends with a bundle $B$, and let
$b^1$ and $b^2$ be two arcs of $B$ that correspond to the two
occurrences of $B$ on $w$.
Let $\gamma$ be the spiral cut associated with the spiral $w$
on $P$. 
Then any path $P'$ that does not contain
any vertex that lies on $\gamma$ (i.e.,
does not contain any inner vertex of the subpath of $P$
that corresponds to the spiral $w$)
intersects $\gamma$ at most once.
Such an intersection, if exists, is contained in an arc of $B$ between $b^1$ and $b^2$
and $P'$ traverses $\gamma$ in the same direction (i.e., either from the inside to the outside or from the outside to the inside)
as the path $P$.

In particular, the path $P$ does not intersect $\gamma$ except
for the subpath that corresponds to the spiral $w$.
\end{lemma}
\begin{proof}
The statement follows from the fact that $(\gamma \cap G) \setminus P$ is contained in the arcs of $B$ between $b^1$ and $b^2$;
moreover, all these arcs are directed in the same direction as $b^1$ and $b^2$.
\end{proof}

Let $P$ be a path in $G$ and let $w$ be a spiral in $\bundleword(P)$.
Assume moreover that $B$ contains only bundle arcs with both endpoints on disc components.
Then, knowing only the spiral $w$ (as a bundle word), for each bundle $B$ that is not in $w$,
we know on which side of $w$ it lies, and the notions of bundles {\em{inside $w$}} and {\em{outside $w$}}
are well-defined.
Note that this statement is not true if $w$ contains a subword $B_1B_2$, where $B_1$ leads to some
ring component $\ringcomp$ and $B_2$ leaves $\ringcomp$ on the same side as $B_1$: in this situation
from $w$ we cannot deduce on which side of the spiral cut associated with the spiral $w$ on $P$ lies the other
side of the ring component $\ringcomp$.


Note the following.

\begin{corollary}\label{cor:bwrep}
Let $P$ be a directed path in a bundled instance $(G,\decomp,\bundleset)$,
let $w$ be a spiral in $P$ and let $\gamma$ be the spiral cut associated with the spiral $w$ on $P$. Assume $\bundleword(P) = u_1 w u_2$.
Then one of the following holds.
\begin{itemize}
\item The word $u_1$ contains only bundles from $w$ or outside $\gamma$ and the word $u_2$ contains
only bundles from $w$ or inside $\gamma$. Moreover,
for any path $P'$ that is vertex-disjoint with $P$ there exists a partition
$\bundleword(P') = u_1' u_2'$ (where $u_1'$ or $u_2'$ may be empty)
such that $u_1'$ contains only bundles from $w$ or outside $\gamma$
and $u_2'$ contains only bundles from $w$ or inside $\gamma$.
\item The same situation happens as in the previous point,
  but with the role of inside and outside exchanged:
  $u_1$ and $u_1'$ are allowed to contain only bundles from $w$
or inside $\gamma$, whereas $u_2$ and $u_2'$ are allowed to contain only bundles from 
  $w$ and outside $\gamma$.
\end{itemize}
If $w$ contains only bundles with endpoints in disc components, then the statements `inside/outside $\gamma$' can be replaced
with `inside/outside $w$'.
\end{corollary}

A similar statement as Corollary \ref{cor:bwrep} is also true if we consider paths traversing a single bundle.
Consider the following definitions.
\begin{defin}[bundle profile]
Let $(G,\decomp,\bundleset)$ be a bundled instance, let $B=(b_1,b_2,\ldots,b_s)$ be a bundle in $G$
and let $(P_i)_{i=1}^k$ be a solution to \probshort{} in $G$. 
Then a {\em{bundle profile}} of the bundle $B$, denoted $\bundleprof(B,(P_i)_{i=1}^k)$, is a word
over the alphabet $\{1,2,\ldots,k\}$ constructed in the following manner:
we inspect arcs $b_1,b_2,\ldots,b_s$ of the bundle $B$ in this order and whenever
$b_j \in P_i$, we append the symbol $i$ to the word $\bundleprof(B,(P_i)_{i=1}^k)$.
\end{defin}
We sometimes shorten $\bundleprof(B)$ if the solution $(P_i)_{i=1}^k$ is clear from the context.
\begin{defin}[spiral pack]
Let $(G,\decomp,\bundleset)$ be a bundled instance, let $B=(b_1,b_2,\ldots,b_s)$ be a bundle in $G$
and let $(P_i)_{i=1}^k$ be a solution to \probshort{} in $G$. 
A {\em{spiral pack}} is a subword $w$ of the word $\bundleprof(B)$ such
that $w$ starts and ends with the same symbol, but all other symbols of $w$ are pairwise
distinct and different than the starting symbol of $w$.
\end{defin}

\begin{lemma}\label{lem:bporder}
Let $(G,\decomp,\bundleset)$ be a bundled instance, let $B=(b_1,b_2,\ldots,b_s)$ be a bundle in $G$
and let $(P_i)_{i=1}^k$ be a solution to \probshort{} in $G$.
Let $1 \leq j(i,1), j(i,2), \ldots, j(i,J(i)) \leq s$ be such indices
that $b_{j(i,1)}, b_{j(i,2)}, \ldots, b_{j(i,J(i))}$ are precisely
the arcs of $B$ that appear on $P_i$, in the order of their appearance
on $P_i$. Then:
\begin{enumerate}
\item for any $1 \leq i \leq k$ and $1 \leq \iota < J(i)$,
  the subword of $\bundleword(P_i)$ between the 
  occurrences of $B$ that correspond to $j(i,\iota)$ and $j(i,\iota+1)$
  is a spiral;
\item $j(i,1) < j(i,2) < j(i,3) < \ldots < j(i,J(i))$
or $j(i,1) > j(i,2) > j(i,3) > \ldots > j(i,J(i))$;
\item if for some $i \neq i'$, the order from the previous
point is different
(i.e., $J(i),J(i') > 1$, but $j(i,1) < j(i,2)$ and $j(i',1) > j(i',2)$)
then either $j(i,\iota) > j(i',\iota')$ for any $\iota,\iota'$
or $j(i,\iota) < j(i',\iota')$ for any $\iota,\iota'$.
\end{enumerate}
\end{lemma}

\begin{proof}
Let $1 \leq i \leq k$ be such an index that $J(i) > 1$
and let $w$ be a subword of $\bundleword(P_i)$ that starts
and ends with $B$, but contains the symbol $B$ only twice
(i.e., $w$ corresponds to a part of the path $P_i$ between
 two consecutive arcs on $B$). We first prove that $B$
is a spiral. Assume otherwise: as $w$ does not contain $B$
except for the first and last symbol, $w$ needs to contain
a different bundle twice, thus $w$ needs to contain a spiral $w'$.
Moreover, this spiral does not contain $B$.
However, $\bundleword(P_i)$ contains $B$ both before the spiral
$w'$ and after it, a contradiction to Corollary \ref{cor:bwrep}.

We infer that, by Lemma \ref{lem:spiral-split},
the path $P_i$ does not intersect the spiral cut
associated with $w$ except for the spiral $w$ itself, and this
spiral cut separates the arcs of $B$ that lie on $P_i$ after $w$
from the arcs of $B$ that lie before $w$. This proves that
$j(i,1), j(i,2), \ldots, j(i,J(i))$ are ordered monotonically.

Take any other path $P_{i'}$, $i' \neq i$, and assume that
$P_{i'}$ contains an arc of $B$ that lies on $B$ between the arcs
of $B$ on the spiral $w$. In this situation, $P_{i'}$ intersects
the spiral cut associated with $w$. By Lemma \ref{lem:spiral-split},
$P_{i'}$ contains the arcs of $B$ in the same order as $P_i$,
and the lemma is proven.
\end{proof}

\begin{corollary}\label{cor:bprep}
Let $(G,\decomp,\bundleset)$ be a bundled instance, let $B=(b_1,b_2,\ldots,b_s)$ be a bundle in $G$
and let $(P_i)_{i=1}^k$ be a solution to \probshort{} in $G$. 
Let $w$ be a spiral pack in $\bundleprof(B)$ and assume $\bundleprof(B) = u_1 w u_2$.
Then if a symbol $\iota \in \{1,2,\ldots,k\}$ appears both in $u_1$ and in $u_2$, then
$\iota$ appears in $w$.
\end{corollary}
\begin{proof}
By Lemma \ref{lem:bporder}, if $i$ is the first and last symbol
of $w$, then the subpath of $P_i$ that connects the arcs
that correspond to the first and the last symbol of $w$
is a spiral. If a symbol $i' \neq i$ does not appear in $w$,
the path $P_i'$ does not intersect the spiral cut associated
with the aforementioned spiral on $P_i$ and cannot contain
arcs of $B$ both inside and outside this spiral.
\end{proof}

We now investigate a situation where the same spiral
appears on $\bundleword(P)$ several times in a row.
\begin{defin}[spiraling ring]
Let $P$ be a directed path in a bundled instance $(G,\decomp,\bundleset)$
and let $u$ be a word over $\bundleset$, such 
that $u$ contains each symbol of $\bundleset$ at most once,
$u$ contains only bundles with both endpoints in disc components,
and $\bundleword(P)$ contains $u^rB$ as a subword for some $r \geq 2$,
where $B$ is the first symbol of $u$.
Let $w_1$ and $w_2$ be the first and the last of the $r$
spirals $uB$ of $\bundleword(P)$ that appear in the subword $u^rB$
and let $\gamma_1$ and $\gamma_2$ be the spiral cuts associated
with these spirals; in case $r=2$ we slightly modify $\gamma_1$ and $\gamma_2$
around the midpoint of the second arc of $P \cap B$ so that they do not intersect.
The closed area of the plane contained between $\gamma_1$
and $\gamma_2$ is called a {\em{spiraling ring}}.
The curve $\gamma_1$ is the {\em{input border}} of the spiraling
ring, and the curve $\gamma_2$ is the {\em{output border}}.
\end{defin}
Note that, by Lemma \ref{lem:spiral-split},
the input and output borders do not intersect each other,
and any spiraling ring is homeomorphic to a closed ring.

\begin{lemma}\label{lem:bwrep:spiraling-ring}
Let $P$ be a directed path in a bundled instance $(G,\decomp,\bundleset)$
and let $u$ be a word over $\bundleset$, such 
that $u$ contains each symbol of $\bundleset$ at most once,
$u$ contains only bundles with both endpoints in disc components,
and $\bundleword(P)$ contains $u^rB$ as a subword for some $r \geq 2$,
where $B$ is the first symbol of $u$.
Let $A_R$ be the spiraling ring associated with $u^rB$ and
let $\gamma_1$ and $\gamma_2$ be its input and output borders.
Let $P'$ be a directed path in $G$ that is contained in $A_R$,
except for possibly the first and last arc that may belong to the bundle $B$
 intersecting $\gamma_1$ and $\gamma_2$, respectively.
Then
\begin{enumerate}
\item $\bundleword(P')$ is a subword of $u^{r'}$ for some integer $r'$;
\item if
$P'$ starts and ends with an arc of $B$ intersecting
$\gamma_1$ and $\gamma_2$, respectively, then $\bundleword(P')=u^{r'}B$
for some integer $r' \geq 1$;
\item if, additionally to the previous point, we assume that
$P'$ is vertex-disjoint with $P$, then $r'=r-1$, i.e,
  $\bundleword(P') = u^{r-1}B$.
\end{enumerate}
\end{lemma}
\begin{proof}
Let $u = B_1B_2\ldots B_s$, where $B_1=B$.
For $1 \leq j \leq s$ let $H_j$ be the component of $\decomp$
that contains the ending points of arcs in the bundle $B_j$.
Note that, although the bundles $B_j$ are pairwise distinct,
the components $H_j$ may be equal for different indices $j$.
For $1 \leq j \leq s$, denote by $G_j$
the subgraph of $H_j$ contained
between the subpath of $P$
between arcs on $B_j$ and $B_{j+1}$ (with $B_{s+1}=B_1=B$) that
is contained in $\gamma_1$, and the subpath of $P$
between arcs on $B_j$ and $B_{j+1}$
that is contained in $\gamma_2$ (including paths and vertices
on these paths as well). As the bundles $B_j$ are pairwise distinct,
the graphs $G_j$ are pairwise disjoint. Moreover, the subgraph
of $G$ contained in the spiraling ring $A_R$ consists of the
graphs $G_j$, $1 \leq j \leq s$, and parts of bundles $B_j$,
$1 \leq j \leq s$ that connect $G_j$ with $G_{j+1}$; note that here we used
the fact that all components $H_j$ are disc components.

Therefore, if some path, that is contained in the subgraph of $G$
contained in $A_R$, starts in $G_j$ and ends in $G_{j'}$, then
its bundleword equals $B_jB_{j+1}\ldots B_s u^{r'} B_1\ldots B_{j'-1}$
for some integer $r'$ (or $B_jB_{j+1}\ldots B_{j'-1}$
provided that $j \leq j'$). This proves the first two claims
of the lemma. For the third claim, we apply the second claim
for $r-1$ subwords $u^2B$ of the subword $u^rB$ of $\bundleword(P)$.
\end{proof}

\begin{corollary}\label{cor:bwrep:equalrep}
Let $P$ be a directed path in a bundled instance $(G,\decomp,\bundleset)$
and let $u$ be a word over $\bundleset$, such that $u$ contains each symbol
of $\bundleset$ at most once,
$u$ contains only bundles with both endpoints in disc components,
and $\bundleword(P)$ contains $u^r$ as a subword, for
some $r \geq 2$. Let $B$ be the first symbol of $u$ and let $w=uB$ be any
of the $r-1$ spirals $w$ contained in the subword $u^r$ of $\bundleword(P)$.
Then for any directed path $P'$ in $G$ that contains arcs from bundles
both outside and inside $w$ and
is vertex-disjoint with $P$, the bundle word of $P'$ contains $u^{r-2}B$
as a subword.
\end{corollary}
\begin{proof}
We apply Lemma \ref{lem:bwrep:spiraling-ring}
for the spiraling ring that corresponds to the subword $u^{r-1}B$
of $\bundleword(P)$.
\end{proof}

\subsection{Structure of bundle words and bundle profiles}

Corollaries \ref{cor:bwrep} and \ref{cor:bprep} imply some structure on bundle words and bundle profiles.
To recognize it, we need the following observation.
\begin{defin}[spiral decomposition of a word]
Let $u$ be a word over an alphabet $\Sigma$.
Let $s$ and $(r_j)_{j=1}^s$ be positive integers and let $(u_j)_{j=1}^s$ be words over
$\Sigma$. We call $u_1^{r_1} u_2^{r_2} \ldots u_s^{r_s}$ a {\em{spiral decomposition}}
of $u$ if:
\begin{enumerate}
\item $u = u_1^{r_1} u_2^{r_2} \ldots u_s^{r_s}$;
\item each word $u_j$ contains each symbol of $\Sigma$ at most once;
\item for each $1 \leq j < s$, there exists a symbol of $\Sigma$ that appears
in exactly one of the two words $u_j$ and $u_{j+1}$;
\item for each letter $\sigma \in \Sigma$ that appears in $u$, 
  there exists $j(\sigma,1)$, $j(\sigma,2)$ such that $\sigma$ appears
  in $u_j$ if and only if $j(\sigma,1) \leq j \leq j(\sigma,2)$.
\end{enumerate}
\end{defin}

\begin{lemma}\label{lem:rep-len}
Let $u$ be a word over an alphabet $\Sigma$
and let $u_1^{r_1} u_2^{r_2} \ldots u_s^{r_2}$ be a spiral decomposition
of $u$. Then $s \leq 2|\Sigma|$.
\end{lemma}
\begin{proof}
The claim follows from the fact that, for each $1 \leq i \leq s$,
there exists a symbol $\sigma \in \Sigma$, such that $u_i$ is either the first
word that contains $\sigma$, or the last one.
\end{proof}

\begin{defin}[nonnested word]
A word $u$ over an alphabet $|\Sigma|$ is called {\em{nonnested}}
if it satisfies the following property: for any letters $a,b \in \Sigma$, $a \neq b$,
and any (possibly empty) words $w_1,w_2,w_3$ such that $u = w_1 a w_2 a w_3$, 
if $b$ appears in $w_1$ and $w_3$, then $b$ appears in $w_2$ as well. 
\end{defin}

\begin{lemma}\label{lem:rep}
Let $u$ be a nonnested word over an alphabet $|\Sigma|$.
Then one can in polynomial time compute a spiral decomposition $u_1^{r_1} u_2^{r_2} \ldots u_s^{r_s}$ of $u$.
\end{lemma}
\begin{proof}
Consider the following greedy algorithm that decomposes
$u$ as $v_1v_2\ldots v_q$: construct $v_j$ consecutively, taking $v_j$ to be the longest
possible word that does not contain the same letter of $\Sigma$ twice.
By construction, the words $v_j$ contain each letter of $\Sigma$ at most once
and $v_{j+1}$ starts with a letter that appears in $v_j$ for each $1 \leq j < q$.

Now assume that $\sigma \in \Sigma$ appears both in $v_{i_1}$ and $v_{i_3}$, but does not appear in $v_{i_2}$ for some $i_1 < i_2 < i_3$. Let $\sigma'$ be the first letter of $v_{i_2+1}$;
by the properties of the words $v_i$, $\sigma'$ appears in $v_{i_2}$ and $\sigma' \neq \sigma$.
However, this contradicts the assumptions on the word $u$ for $a=\sigma'$ and $b=\sigma$.
Therefore, for each letter $\sigma \in \Sigma$ that appears in $u$, there exist
$i(\sigma,1)$, $i(\sigma,2)$ such that $\sigma$ appears in $v_i$ iff $i(\sigma,1) \leq i \leq i(\sigma,2)$.

Now assume that, for some $1 \leq i < q$, $v_i$ and $v_{i+1}$ consist of the same
letters of $\Sigma$, but in different order. Let $v_i = v \sigma_i v_i'$, $v_{i+1} = v \sigma_{i+1} v_{i+1}'$ where $\sigma_i \neq \sigma_{i+1}$, i.e., the first position
on which $v_i$ and $v_{i+1}$ differ is $|v|+1$. However, in this case $u$ contains
$v_i' v \sigma_{i+1}$ as a subword that does contain the letter $\sigma_{i+1}$ twice
and does not contain $\sigma_i$, whereas both $v \sigma_i$ and $v_{i+1}'$ contain $\sigma_i$,
a contradiction with the properties of $u$.

Therefore, if we define $u_1,u_2,\ldots, u_s$ as the sequence of pairwise
distinct words of the sequence $v_1,v_2,\ldots,v_q$, and $r_j$ as the number
of consecutive occurrences of $u_j$ in the sequence $v_1,v_2,\ldots,v_q$, we obtain
a correct spiral decomposition of $u$.
\end{proof}

By Corollaries \ref{cor:bwrep} and \ref{cor:bprep}, both bundle words and bundle profiles are nonnested.
We infer the following decomposition statements.

\begin{corollary}[Bundle word decomposition]\label{cor:bwdecomp}
Let $(G,\decomp,\bundleset)$ be a bundled instance and let $P$ be a path in $G$.
Then there exists a spiral decomposition 
$\bundleword(P) = u_1^{r_1} u_2^{r_2} \ldots u_s^{r_s}$ where $s \leq 2|\bundleset|$.
Moreover, if $u_j$ is a cyclic shift of $u_{j'}$ for some $1 \leq j < j' \leq s$ and $r_j,r_{j'} > 1$,
  then $u_j$ contains a bundle incident to a ring component.
\end{corollary}
\begin{proof}
The first part of the corollary is straightforward.
For the second part, first note that,
as $u_j$ and $u_{j+1}$ need to differ on at least one symbol, we have $j' > j+1$.
As $u_j$ and $u_{j'}$ are cyclic shifts, $u_{j+1}$ contains a bundle $B$
that does not appear in $u_j$ nor in $u_{j'}$.
However, if $r_j,r_{j'} > 1$ then the terms
$u_j^{r_j}$ and $u_{j'}^{r_{j'}}$ contain the same spiral, a contradiction
to Lemma \ref{lem:spiral-split} unless $u_j$ contains a bundle incident to a ring component.
\end{proof}

\begin{corollary}[Bundle profile decomposition]\label{cor:bpdecomp}
Let $(G,\decomp,\bundleset)$ be a bundled instance, let $B$ be a bundle in $G$ and let $(P_i)_{i=1}^k$ be a solution to \probshort{} in $G$.
Then there exists a spiral decomposition
$\bundleprof(P) = u_1^{r_1} u_2^{r_2} \ldots u_s^{r_s}$ where $s \leq 2|\bundleset|$.
\end{corollary}

We also note the following.
\begin{lemma}\label{lem:rep-must}
Let $w$ be a word over an alphabet $\Sigma$ and assume $w$ contains
a subword $v^q\sigma$, where $v$ contains each letter of $\Sigma$ at most once,
  $\sigma$ is the first letter of $v$ and $q \geq 3$.
  Then any spiral decomposition of $w$ needs to contain
a term $u^r$, where $u$ is a cyclic shift of $v$ and $r \geq q-1$.
\end{lemma}
\begin{proof}
Consider a spiral decomposition $w = u_1^{r_1} u_2^{r_2} \ldots u_s^{r_s}$.
Let $\tau$ be any letter of the middle $v^{q-2}\sigma$ subword of $v^q\sigma$ in $w$.
Assume that this occurrence of $\tau$ appears
in $u_j$ in the spiral decomposition. Then, if $|u_j| > |v|$, $u_j$ contains
the same symbol twice, and if $|u_j| < |v|$, $u_j$ is a subword of $v^q\sigma$
and the word $v^q\sigma$ contains the same symbol on both sides of the subword $u_j$,
and this symbol does not appear in $u_j$. In both cases we have reached a contradiction,
so $|u_j| = |v|$ and $u_j$ is a cyclic shift of $v$. As the choice of $\tau$
was arbitrary, we infer that $r_j \geq q-1$.
\end{proof}

\begin{lemma}\label{lem:rep-perm}
Let $(G,\decomp,\bundleset)$ be a bundled instance. 
Let $P$ be a path in $G$ with $\bundleword(P) = u^3B$ where $u$ contains each bundle at most once,
   $u$ contains only bundles with both endpoints in disc components,
 and $B$ is the first symbol of $u$.
Moreover, let $Q$ be a path in $G$ with $\bundleword(Q) = vB'$ where $v$ is a permutation of $u$ that is not a cyclic shift of $u$ and $B'$ is the first symbol of $v$.
Then $P$ and $Q$ share an internal vertex.
\end{lemma}
\begin{proof}
Assume the contrary: $P$ and $Q$ do not intersect, except possibly for the endpoints.

Let $b_1,b_2,b_3,b_4$ be the four arcs of $P \cap B$
in the order of their appearance on $P$
and let $R$ be the subpath of $P$ that starts with $b_2$ and ends with $b_3$;
note that $\bundleword(R) = uB$.
Let $\gamma_R$ and $\gamma_Q$ be spiral cuts corresponding to $R$ and $Q$, respectively.

We claim that $\gamma_Q$ and $\gamma_R$ do not intersect.
Assume the contrary; let $p \in \gamma_Q \cap \gamma_R$.
As $Q$ and $R$ are vertex-disjoint except for possibly the endpoints,
the point $p$ must belong to $\gamma_Q \setminus Q$ or $\gamma_R \setminus R$
(i.e., to the part of $\gamma_Q$ inside $B'$ or the part of $\gamma_R$ inside $B$).
If $p \in \gamma_Q \setminus Q$ then, as $\gamma_R$ visits each bundle of $u$ exactly once,
the midpoints of the two arcs of $Q \cap B'$ lie on different sides of $\gamma_R$,
    and $Q$ needs to contain an arc of $B$ that lies between $b_2$ and $b_3$.
Otherwise, if $p \in (\gamma_R \setminus R) \cap Q$, then, as $\gamma_R \setminus R$ connects $b_2$ with $b_3$ inside $B$,
$Q$ again contains an arc of $B$ that lies between $b_2$ and $b_3$.

If $Q$ is contained in the spiraling ring $A_P$ associated with the path $P$ (recall $\bundleword(P) = u^3B$)
then, by the first claim of Lemma \ref{lem:bwrep:spiraling-ring}, $\bundleword(Q)$ is a subword of $u^{r'}$ for some $r'$, a contradiction
to the assumption that $v$ is not a cyclic shift of $u$. Otherwise, $Q$ contains an arc $b' \in B$, contained either
between $b_1$ and $b_2$ or between $b_3$ and $b_4$. Consider the first case; the second one is symmetrical.
Let $P'$ be the subpath of $P$ between $b_1$ and $b_3$; note that $\bundleword(P') = u^2B$.
Let $A_R'$ be the spiraling ring associated with $P'$ and let $Q'$ be the subpath of $Q$ contained in $A_R'$ that contains the aforementioned
arcs of $B$  between $b_1$ and $b_2$ as well as between $b_2$ and $b_3$.
By the third claim of Lemma \ref{lem:bwrep:spiraling-ring}, applied to the spiraling ring $A_R'$ and the path $Q'$, $\bundleword(Q') = uB$, a contradiction.
Hence, $Q$ cannot contain an arc $b \in B$ that lies between $b_2$ and $b_3$, and $\gamma_Q$ and $\gamma_R$ do not intersect.

As $v$ is a permutation of $u$, but not a cyclic shift of $u$, there exist three pairwise distinct
bundles $B_1,B_R,B_Q$ and a component $D$ such that $\gamma_R$ first goes through
the part of the plane occupied by $B_1$, then through $D$ and then through the part of the plane occupied by $B_R$, whereas $\gamma_Q$ goes thought the
part of the plane occupied by $B_1$, then through $D$ and then through $B_Q$ (see Figure \ref{fig:lem:rep-perm}).
Let $\gamma_Q'$ be the subcurve of $\gamma_Q$ between the midpoint of the traverse of $B_1$
and the midpoint of the traverse of $B_R$ that includes the traverse of $B_Q$. Obtain a curve $\gamma$ by closing $\gamma_Q'$ as follows:
from the midpoint of the traverse of $B_R$ go along the bundle $B_R$ in the dual of $G$ to the curve $\gamma_R$, follow $\gamma_R$
backwards through $G$ to the midpoint of the traverse of $\gamma_R$ though $B_1$ and then go along the bundle $B_1$ in the dual of $G$
to the start of $\gamma_Q'$.

Note that the only intersection of $\gamma$ with $\gamma_R$ is the part of $\gamma_R$
between the midpoint of its traverse of $B_1$ and the midpoint of its traverse of $B_R$, containing the traverse of $D$.
As $\gamma_Q$ lies on the same side of $\gamma_R$ both in $B_1$ and in $B_R$, $\gamma_R$ leaves $\gamma$ to the same side
in $B_1$ and in $B_R$. That is,
the start and the end of the part of curve $\gamma_R$ that traverses $B_1$, $D$ and $B_R$ lies on one of the two sides of $\gamma$.
However, since in $B_1$, $B_Q$ and $B_R$ the curve $\gamma_Q$ needs to lie on the same side of $\gamma_R$,
the part of the curve $\gamma_R$ that traverses $B_Q$ lies on the other side of $\gamma$, a contradiction.
This finishes the proof of the lemma.
\end{proof}

\begin{figure}
\begin{center}
\includegraphics{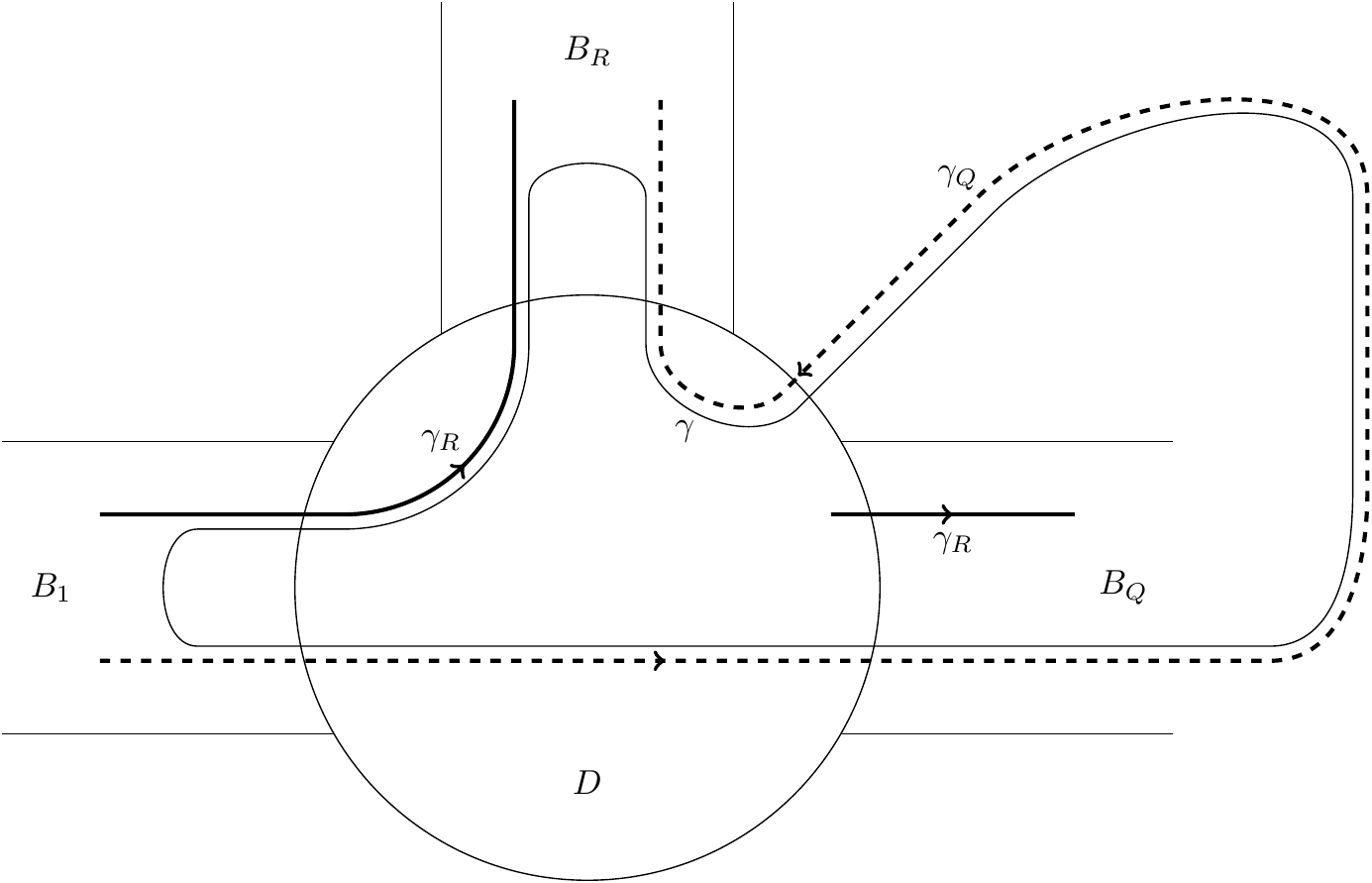}
\caption{An illustration of the proof of Lemma \ref{lem:rep-perm}.}
  \label{fig:lem:rep-perm}
  \end{center}
  \end{figure}

\begin{corollary}\label{cor:copy-word}
Let $\mathcal{P}$ be a family of pairwise vertex-disjoint paths in a 
bundled instance $(G,\decomp,\bundleset)$ with fixed
bundle word decompositions for each $P \in \mathcal{P}$.
Assume there exists a path $P_0 \in \mathcal{P}$ and a subword $u^rB$ of $\bundleword(P_0)$
such that $u$ contains each bundle at most once,
$u$ contains only bundles with both endpoints in disc components,
$B$ is the first symbol of $u$ and $r \geq 4$.
Moreover, assume that any path $P$ starts and ends with an arc that belongs to a bundle
that is not contained in $u$.
Then, for any $P \in \mathcal{P}$
either
\begin{enumerate}
\item the first and the last arc $P$ lie on the same side of the spiral $uB$ and
there does not exist any term $v^q$ in the bundle word decomposition of $\bundleword(P)$
such that $q > 1$ and $v$ is a permutation of $u$;
\item the first and the last arc of $P$ lie on different sides of the spiral $uB$
and there exists exactly one term $v^q$ in the bundle word decomposition of $\bundleword(P)$
where $q>1$ and $v$ is a permutation of $u$; moreover,
in this case $v$ is a cyclic shift of $u$ and $q \geq r-2$.
\end{enumerate}
\end{corollary}

\begin{proof}
First, observe that by Corollary~\ref{cor:bwdecomp} 
there is at most one $v^q$ in the bundle word decomposition of $\bundleword(P)$,
where $v$ is a cyclic shift of $u$ and $q > 1$.
Let us prove that for any $P \in \mathcal{P}$ the bundle word $\bundleword(P)$ does not contain
a subword $v^2$ where $v$ is a permutation of $u$ that is not a cyclic shift of $u$.
This is clearly true for $P \neq P_0$, due to Lemma \ref{lem:rep-perm} and the assumption of the subword $u^rB$ in $\bundleword(P_0)$.
Moreover, any such subword $v^2$ of $\bundleword(P_0)$ would be disjoint with one of the two subwords $u^{r-1}B$ of $u^rB$, which is
in turn a subword of $\bundleword(P_0)$. As $r \geq 4$, the claim follows from Lemma \ref{lem:rep-perm}, as we have
two disjoint (except for possibly the endpoints) subpaths of $P_0$ with bundlewords $v^2$ and $u^3B$.

Now let us consider two cases.
First assume that $P = P_0$. Then the endpoints of $P_0$ lie on different sides of the spiral $uB$, and the bundle word decomposition of $\bundleword(P_0)$ obviously contains $u^r$.

Finally, assume that $P \neq P_0$.
Note that if the endpoints of $P \in \mathcal{P}$ lie on different sides of the spiral $uB$ then, by the third claim of Lemma \ref{lem:bwrep:spiraling-ring}
combined with Lemma \ref{lem:rep-must}, the bundle word decomposition of $P$ needs to contain the term $v^q$ for $v$ being a cyclic shift of $u$
and $q \geq r-2 > 1$. In the reverse direction, if $\bundleword(P)$ contains a subword $v^2$ where $v$ is a cyclic shift of $u$, then
$\bundleword(P)$ contains a subword $uB$ and, by Corollary \ref{cor:bwrep}, $P$ has endpoints on different
sides of the spiral $uB$. 
\end{proof}

\subsection{Ring components: basics}

The ring components have more complicated topology structure than the disc ones,
but, thanks to the isolation property, we are able to prove that, if we are given a YES-instance to \probshort{},
there exists a solution with the interaction with ring components that is somehow bounded.
In this section we give some basic notions towards proving this statement.

Assume we are given a bundled instance $(G,\decomp,\bundleset)$, such that $\decomp$ has isolation $(\Lambda,d)$, $\Lambda,d > 0$.
Consider a ring component $\ringcomp$ with isolation components
$\disccomp_{IN,\lambda}$ and $\disccomp_{OUT,\lambda}$, $1 \leq \lambda \leq \Lambda$.
Fix a choice of the $d$ concentric cycles in the subgraph induced by the vertices
of $\disccomp_{IN,\lambda}$ promised by the definition of an isolating component
and let $C_{IN,\lambda}(\ringcomp)$ be the innermost of these cycles. Similarly,
let $C_{OUT,\lambda}(\ringcomp)$ be the outermost of a fixed choice of a alternating
sequence of cycles in $\disccomp_{OUT,\lambda}$.
\begin{defin}[closure of a ring component]
Let $(G,\decomp,\bundleset)$ be a bundled instance, such that $\decomp$ has isolation $(\Lambda,d)$, $\Lambda,d > 0$.
Let $\ringcomp$ be a ring component in $\decomp$ and let $1 \leq \lambda \leq \Lambda$.
Then $\ringcl^\lambda(\ringcomp)$, the {\em{level-$\lambda$ closure of $\ringcomp$}}, is the subgraph of $G$ that contains
$\ringcomp \cup \bigcup_{1 \leq \lambda' \leq \lambda} \disccomp_{IN,\lambda'} \cup \disccomp_{OUT,\lambda'}$ as well as all
bundle arcs of $G$ that are contained between $C_{IN,\lambda}(\ringcomp)$
and $C_{OUT,\lambda}(\ringcomp)$ or on one of these cycles.
Moreover, we define the closure of $\ringcomp$, $\ringcl(\ringcomp)$ as the $\Lambda$-level closure of $\ringcomp$,
  and the $0$-level closure of $\ringcomp$, $\ringcl^0(\ringcomp)$, as $\ringcomp$ itself.
\end{defin}
Note that, since each disc component may serve as an isolation only to one ring component, closures are pairwise disjoint.
By splitting some bundles incident to the isolation components of $\ringcomp$ into smaller bundles, we may assume
that either none or all arcs of a single bundle are contained in the closure of a ring component.
By Lemma \ref{lem:spiral-split}, a cycle in $G$ may contain at most one arc from a single bundle.
Therefore, the aforementioned splitting operation partitions each bundle with arcs with both endpoints in an
isolation component into at most three parts, and, consequently, at most triples the number of such bundles.

\begin{defin}[normal, ring and isolation bundles]
Let $(G,\decomp,\bundleset)$ be a bundled instance, such that $\decomp$ has isolation $(\Lambda,d)$, $\Lambda,d>0$.
A bundle $B \in \bundleset$ is a {\em{ring bundle}} if it is contained in a closure of some ring component,
and a {\em{normal bundle}} otherwise. 
A {\em{level}} of a ring bundle $B$ contained in the closure of a ring component $\ringcomp$,
  is the minimum nonnegative integer $\lambda$, such that
$B$ contains arcs incident to $\ringcl^\lambda(\ringcomp)$.
A ring bundle is an {\em{isolation bundle}} if it has both endpoints
in the same isolating disc component.
\end{defin}

\begin{defin}[ring part]
Let $P$ be a directed path in a bundled instance $(G,\decomp,\bundleset)$ with isolation $(\Lambda,d)$, $\Lambda,d > 0$.
Then a maximal subword of $\bundleword(P)$ that consists of only ring bundles
and contains at least one level-$\lambda$ ring bundle for some $\lambda < \Lambda$
is called a {\em{ring part}} of $\bundleword(P)$.
The normal bundles preceding and succeeding a ring part of $\bundleword(P)$ are called the
{\em{predecessor}} and {\em{successor}} of the ring part, respectively.

For a ring part $w$ of $\bundleword(P)$, a {\em{ring part}} of $P$ is maximal possible subpath
of $P$ that corresponds to $w$. The {\em{predecessor}} and the {\em{successor}} of a ring part of $P$
is an arc that precedes (resp. succeeds) the ring part on the path $P$.
\end{defin}
Note that the predecessor or successor of a ring part may not be defined if $P$ ends or starts in the closure of the ring component;
however, they are always well-defined for paths connecting a terminal pair.

Note also that ring parts of a bundle word are pairwise disjoint and
there is at least one normal bundle between any ring parts
in the bundle word $\bundleword(P)$.
Moreover, as each disc component may serve as an isolation component to at most one ring component
and closures of ring components are pairwise disjoint, each ring part of a bundle word
contains bundles from a closure of a single ring component and the corresponding ring part of a path is contained
in the closure of this ring component.

\begin{defin}[isolation passage]
Let $P$ be a directed path in a bundled instance $(G,\decomp,\bundleset)$ of isolation $(\Lambda,d)$, $\Lambda,d>0$,
let $\ringcomp$ be a ring component in $\decomp$ and let $\disccomp$ be any isolation component of
$\ringcomp$ such that $P$ does not start nor end in $\disccomp$.
Then an {\em{isolation passage}} of $P$ through $\disccomp$ is
any maximal subpath of $P$ that is contained in $\ringcl(\ringcomp)[V(\disccomp)]$ (i.e., $\disccomp$
    with arcs of its isolation bundles), but the arc preceding and succeeding the subpath on $P$
lie on different faces of $\ringcl(\ringcomp)[V(\disccomp)]$.
The level of the isolation passage is the level of $\disccomp$.
\end{defin}
Note that $P$ starts and ends outside $\disccomp$ for example, if $P$ connects a terminal pair in $G$.
Moreover, an isolation passage of $P$ though $\disccomp$ is preceded and succeeded by an arc connecting $\disccomp$ with 
a neighbouring isolation component, a normal bundle (possible if $\disccomp$ is of level $\Lambda$)
or $\ringcomp$ (possible if $\disccomp$ is of level $1$).

\begin{defin}[ring passage, ring visitor]
Let $P$ be a directed path in a bundled instance $(G,\decomp,\bundleset)$.
Let $w$ be a ring part in $\bundleword(P)$ in the closure of the ring component $\ringcomp$, for which the predecessor and the successor are well-defined.
Then $w$ is a {\em{ring passage}} if the predecessor and the successor of $w$ lie on different faces of the closure of $\ringcomp$ (equivalently, lie on different
faces of $\ringcomp$ itself), and {\em{ring visitor}} otherwise.
\end{defin}

\begin{defin}[level-$\lambda$ ring passage]
Let $P$ be a directed path in a bundled instance $(G,\decomp,\bundleset)$ with isolation $(\Lambda,d)$, $\Lambda,d > 0$.
Let $0 \leq \lambda \leq \Lambda$, let $w$ be a ring part in $\bundleword(P)$ in the closure of the ring component $\ringcomp$
and assume that $P$ starts and ends outside $\ringcl^\lambda(\ringcomp)$.
Then any maximal subpath of a ring part $w$ of $P$ that is contained in $\ringcl^\lambda(\ringcomp)$
that has a preceding and succeeding arc on $P$ contained in different faces of $\ringcl^\lambda(\ringcomp)$ is called a {\em{level-$\lambda$ ring passage}}.
\end{defin}
Note that a level-$\Lambda$ ring passage is simply a ring passage.

We also note the following.
\begin{lemma}\label{lem:spirals-in-rings}
Let $(G,\decomp,\bundleset)$ be a bundled instance with positive isolation and let
$P$ be an arbitrary path connecting a terminal pair in $G$.
Then any spiral in $\bundleword(P)$ contains at least one normal bundle or is contained in a ring passage of $\bundleword(P)$.
In particular, any ring visitor of $\bundleword(P)$ consists of pairwise distinct bundles.
\end{lemma}
\begin{proof}
Let $w$ be a spiral in $\bundleword(P)$ that contains only ring bundles and let $\gamma$ be the corresponding spiral cut.
As $w$ contains only ring bundles, the curve $\gamma$ is completely contained in $\ringcl(\ringcomp)$ for one ring component $\ringcomp$.
If the curve $\gamma$ separates the inner face of $\ringcl(\ringcomp)$ from the outer face,
Lemma \ref{lem:spiral-split} implies that the predecessor and the successor of $w$ lie in different faces of $\ringcl(\ringcomp)$.
In other case, the graph enclosed inside $\gamma$ is contained in $\ringcl(\ringcomp)$.
This contradicts Lemma \ref{lem:spiral-split}, as $P$ connects a terminal pair in $G$ and there is no terminal contained inside $\gamma$.
\end{proof}

By the definition of ring isolation, all bundle arcs incident to at least one vertex
of $\ringcl^\lambda(\ringcomp)$ (for some $0 \leq \lambda \leq \Lambda$ and some ring component $\ringcomp$),
   but not belonging to $\ringcl^\lambda(\ringcomp)$,
lie either in the outer face of $\ringcl^\lambda(\ringcomp)$, or in one other face
of $\ringcl^\lambda(\ringcomp)$ inside $C_{IN,\lambda}(\ringcomp)$ (or $\ringcomp$ if $\lambda=0$). 
\begin{lemma}\label{lem:refcurve-choice}
Let $(G,\decomp,\bundleset)$ be a bundled instance
of isolation $(\Lambda,d)$, $\Lambda,d>0$.
Then in polynomial time we may compute a set of reference curves
$\refcurve^\lambda(\ringcomp)$ for all $0 \leq \lambda \leq \Lambda$
and $\ringcomp$ being a ring component of $\decomp$ in such a manner that:
\begin{enumerate}
\item $\refcurve^\lambda(\ringcomp)$ connects the outer face
of $\ringcl^\lambda(\ringcomp)$ with the inner face;
\item for each bundle $B$, either $\refcurve^\lambda(\ringcomp)$ does not intersect $B$, or contains a subcurve that intersects each arc of $B$ and nothing else,
\item $\refcurve^\lambda(\ringcomp)$ is a subcurve of $\refcurve^{\lambda+1}(\ringcomp)$ for $0 \leq \lambda < \Lambda$
and $\refcurve^\lambda(\ringcomp) \cap \ringcl^\lambda(\ringcomp) = \refcurve^{\lambda+1}(\ringcomp) \cap \ringcl^\lambda(\ringcomp)$;
\item $\refcurve^\lambda(\ringcomp) \cap \disccomp = \emptyset$
for any disc component $\disccomp$ that isolates $\ringcomp$.
\end{enumerate}
\end{lemma}
\begin{proof}
For each ring component $\ringcomp$, $\Gamma \in \{IN,OUT\}$
and $1 \leq \lambda \leq \Lambda$, choose an undirected path
in the dual of $G$ that connects the outer and inner face of
$\ringcl(\ringcomp)[V(\disccomp_{\Gamma,\lambda})]$
and contains bundle arcs only.
Connect these paths with parts of (again undirected) cycles in the dual
of $G$ that contain bundle arcs connecting $\disccomp_{\Gamma,\lambda}$
with $\disccomp_{\Gamma,\lambda-1}$ (or $\ringcomp$ if $\lambda=1$).
The drawings of these undirected paths in the dual of $G$
yield the desired paths $\refcurve^\Lambda(\ringcomp)$;
the paths $\refcurve^\lambda(\ringcomp)$ are their appropriate
subcurves.
\end{proof}

In this manner, the subgraph $\ringcl^\lambda(\ringcomp)$ becomes a rooted ring, and we may
use the notion of a winding number with respect to $\refcurve^\lambda(\ringcomp)$.
In particular, given a path $P$ that connects a terminal pair in $G$,
and a level-$\lambda$ ring passage $w$ of $\bundleword(P)$, we define the winding number
$\refcurve^\lambda(\ringcomp)(w)$ of the level-$\lambda$ ring passage $w$ as the winding number of the corresponding
subpath of the path $P$.

\subsection{Bundle words with holes}

In this section we introduce a modification of the definition
of bundle words that takes care also of ring components.
Informally speaking, if we want to apply Schijver's cohomology algorithm,
we need not only to know bundle words of the paths, but also the number
of turns a path makes when it makes a ring passage.

\begin{defin}[bundle word with ring holes]
Let $(G,\decomp,\bundleset)$ be a bundled instance of isolation $(\Lambda,d)$ for $\Lambda,d>0$. Let $0 \leq \lambda < \Lambda$.
A pair $((p_j)_{j=0}^q, (w_j)_{j=1}^q)$ is called a {\em{bundle word with level-$\lambda$ ring holes}} if:
\begin{enumerate}
\item each $p_j$, $0 \leq j \leq q$ is a bundle word in $(G,\decomp,\bundleset)$ 
that does not contain two bundles that lie on different sides of
$\ringcl^\lambda(\ringcomp)$ for any ring component \ringcomp;
\item each $w_j$, $1 \leq j \leq q$ is an integer;
\item for each $1 \leq j \leq q$, there exists a ring component $\ringcomp_j \in \decomp$, such that
 \begin{enumerate}
 \item the last bundle of $p_{j-1}$ is a
 bundle that contains arcs with ending points in $\ringcl^\lambda(\ringcomp_j)$
 but starting points not in $\ringcl^\lambda(\ringcomp_j)$;
 \item the first bundle of $p_j$ is a a bundle that contains arcs with starting points in $\ringcl^\lambda(\ringcomp_j)$
 but ending points not in $\ringcl^\lambda(\ringcomp_j)$;
 \item the two aforementioned bundles lie in different faces of $\ringcl^\lambda(\ringcomp_j)$
   (i.e., one lies in the outer face, and one in the inner face).
 \end{enumerate}
\end{enumerate}

We say that a path $P$ is {\em{consistent}} with a bundle word with ring holes $((p_j)_{j=0}^q, (w_j)_{j=1}^q)$ if 
$\bundleword(P) = p_0 r_1 p_1 r_2 p_2 \ldots r_q p_q$, where $r_1, r_2, \ldots, r_q$ are exactly the level-$\lambda$
ring passages of $\bundleword(P)$,
  and for each $1 \leq j \leq q$, the subpath of $P$ corresponding to the ring passage $r_j$ has winding number $w_j$
  in the level-$\lambda$ closure of the appropriate ring component.
\end{defin}

We first note that, knowing a bundle word with ring holes for some level,
we in fact know the bundle words with ring holes for all higher levels.

\begin{defin}[projection of bundle words with ring holes]
Let $(G,\decomp,\bundleset)$ be a bundled instance of isolation $(\Lambda,d)$ for $\Lambda,d>0$
and let $((p_j)_{j=0}^q, (w_j)_{j=1}^q)$ be a bundle word with level-$\lambda$ ring holes for some $0 \leq \lambda < \Lambda-1$ such that $p_0$ starts with a normal bundle and $p_q$
ends with a normal bundle.
Let $\lambda < \lambda' < \Lambda$.
A {\em{level-$\lambda'$ projection}} of $((p_j)_{j=0}^q, (w_j)_{j=1}^q)$ is
a bundle word with level-$\lambda'$ ring holes $((p_j')_{j=0}^{q'}, (w_j')_{j=1}^{q'})$ defined as follows.
Let $p = p_0 w_1 p_1 w_2 p_2 \ldots w_q p_q$ be a word over alphabet $\bundleset \cup \Z$ and let
$p = p_0' x_1 p_1' x_2 \ldots x_{q'} p_q'$ be such that each $x_j$, $1 \leq j \leq q'$ is a maximal subword
of $p$ that contains at least one integer $w_\iota$,
and all the bundles it contains are in $\ringcl^{\lambda'}(\ringcomp)$ for some ring component $\ringcomp$.
Let $w_{j'}'$ for $1 \leq j' \leq q'$ be equal to the sum of integers $w_j$ contained in $x_{j'}$ plus $+1$ or $-1$
for each bundle in $x_{j'}$ that crosses $\refcurve^{\lambda'}(\ringcomp)$; the sign depends on the direction of the crossing.
\end{defin}

\begin{lemma}\label{lem:bwholes-projection}
Let $(G,\decomp,\bundleset)$ be a bundled instance of isolation $(\Lambda,d)$ for $\Lambda\geq 2, d > 0$.
Let $\bwholes$ a bundle word with level-$\lambda$ ring holes for some $1 \leq \lambda < \Lambda$.
and let $P$ be a path consistent with $\bwholes$. Assume that $\bwholes$ starts and ends
with a normal bundle.
Let $\lambda \leq \lambda' \leq \Lambda$ and let $\bwholes'$ be the level-$\lambda'$ projection of $\bwholes$.
Then $P$ is consistent with $\bwholes'$ as well.
\end{lemma}
\begin{proof}
Let $\bwholes = ((p_j)_{j=0}^q, (w_j)_{j=1}^q)$ and let $q'$, $p_j'$, $x_j$ and $w_j'$ be as in the definition
of the level-$\lambda'$ projection; in particular $\bwholes' = ((p_j')_{j=0}^{q'}, (w_j')_{j=1}^{q'})$.
Let $\bundleword(P) = p_0 r_1 p_1 r_2 \ldots r_q p_q$ and let $r_1,r_2,\ldots,r_q$ be the level-$\lambda$ ring passages of $P$.
Note that, as $p_0$ starts with a normal bundle and $p_q$ ends with a normal bundle,
by the definition of level-$\lambda'$ ring passages, $\bundleword(P) = p_0' r_1' p_1' \ldots r_{q'} p_{q'}$ where
$r_1', r_2', \ldots, r_{q'}'$ are exactly the level-$\lambda'$ ring passages of $P$.
Moreover, as the reference curve $\refcurve^{\lambda'}(\ringcomp)$ contains $\refcurve^\lambda(\ringcomp)$,
does not contain any point of $\ringcl^\lambda(\ringcomp)$ outside $\refcurve^\lambda(\ringcomp)$
and is disjoint with isolating components of $\ringcomp$, $w_j'$ equals the winding number of $r_j'$.
\end{proof}

Finally, we note an easy fact that any path in a bundled instance $(G,\decomp,\bundleset)$
yields a unique bundle word with holes for each fixed level $\lambda$.

\begin{lemma}\label{lem:path-to-bwholes}
Let $(G,\decomp,\bundleset)$ be a bundled instance of isolation $(\Lambda,d)$ for $\Lambda\geq 2, d > 0$.
Let $P$ be a path in $G$, whose both starting and ending points are not contained
in any closure of a ring component. Then, for any $0 \leq \lambda < \Lambda$, there
exists a unique level-$\lambda$ bundle word with ring holes that is consistent with
$P$.
\end{lemma}
\begin{proof}
First, let us construct a level-$0$ bundle word with ring holes for $P$.
We start with the bundle word $\bundleword(P)$.
For any level-$0$ ring passage $P^\ast$ of $P$, let $b_1$ and $b_2$
be the arcs on $P$ preceding and succeeding $P^\ast$ and let $B_1,B_2 \in \bundleset$
be such that $b_1 \in B_1$, $b_2 \in B_2$;
take the subword $B_1 \bundleword(P^\ast) B_2$ of $\bundleword(P)$ that corresponds
to the subpath $P^\ast$ with arcs $b_1$ and $b_2$ and replace $\bundleword(P^\ast)$
with the winding number of $P^\ast$ in its ring component $\ringcomp$
(note that $P^\ast$ may contain arcs of bundles with both endpoints in $\ringcomp$).
Performing this operation for each level-$0$ ring passage of $P$, we obtain
a bundle word with level-$0$ holes that is consistent with $P$ and, moreover,
  starts and ends with a normal bundle.
By Lemma \ref{lem:bwholes-projection}, a level-$\lambda$ projection of this bundle word
with level-$0$ ring holes is a bundle word with level-$\lambda$ ring holes consistent with
$P$. The uniqueness follows from the fact that in the consistency definition
we require that the integers $w_j$ correspond exactly to the level-$\lambda$ ring
passages of $P$.
\end{proof}

We also define the following property of bundle words with ring holes that will be
  useful in the future.

\begin{defin}[flat bundle word with ring holes]
Let $(G,\decomp,\bundleset)$ be a bundled instance of isolation $(\Lambda,d)$ for $\Lambda\geq 2$, $d>0$.
Let $\bwholes$ a bundle word with level-$\lambda$ ring holes for some $0 \leq \lambda < \Lambda-1$ that starts and ends with a normal bundle. We say that $\bwholes$ is {\em{flat}} if
the level-$(\lambda+1)$ projection of $\bwholes$ does not contain
any bundles of level $\lambda$ or lower.
\end{defin}

\subsection{Minimal solution}

We now define the notion of a {\em{minimal}} solution to a bundled
instance $(G,\decomp,\bundleset)$. 
Assume that $\decomp$ has ring isolation $(\Lambda,d)$.

\begin{defin}[minimal solution]
A solution $(P_i)_{i=1}^k$ to \probshort{} on a bundled instance
$(G,\decomp,\bundleset)$ is called {\em{minimal}} if:
\begin{enumerate}
\item the words $(\bundleword(P_i))_{i=1}^k$ have minimal possible
  total number of bundles that are not isolation bundles;
\item satisfying the above,
 the words $(\bundleword(P_i))_{i=1}^k$ have minimal possible
  total length.
\end{enumerate}
\end{defin}

One of the ways we use minimality of a solution is the following.
\begin{lemma}\label{lem:minsol-spiral}
Let $(G,\decomp,\bundleset)$ be a bundled instance
and let $(P_i)_{i=1}^k$ be a minimal solution to \probshort{} in $G$.
Assume that for some $1 \leq \iota \leq k$, $\bundleword(P_\iota)$ contains a subword $u^r$,
where $u$ contains each bundle at most once,
$u$ contains only bundles with both endpoints in disc components,
and $r \geq 10$.
Let $A_R$ be the spiraling ring (as a closed subspace of a plane)
associated with the subword $u^r$ in $\bundleword(P_\iota)$ and borders
$\gamma_1$ and $\gamma_2$. Let $G_R$ be the subgraph of $G$ 
consisting of those vertices and arcs of $G$ that lie in $A_R \setminus \gamma_2$.
Let $f_1$ be the face of $G_R$ that contains the part of the plane separated
from $A_R$ by $\gamma_1$ and let $f_2$ be defined analogously for $\gamma_2$.
Let $I \subseteq \{1,2,\ldots,k\}$ be the set of indices of those
paths that intersect $G_R$.
Then, for any set of vertex-disjoint paths $(Q_i)_{i \in I}$
such that for each $i \in I$, $Q_i$ starts in a vertex on $f_1$ and ends in a vertex on $f_2$,
the bundle word of each $Q_i$ contains $u^{r-10}$ as a subword. 
\end{lemma}
\begin{proof}
First note that, by Lemma \ref{lem:spiral-split},
the intersection of a path $P_i$ with $G_R$ is either empty
  or is a single path connecting a vertex on $f_1$ with a vertex on $f_2$.

Let $B$ be the first bundle on $u$. We treat $G_R$ as a rooted ring
with faces $f_1$ and $f_2$, containing the parts of the plane separated
by the curves $\gamma_1$ and $\gamma_2$ from $A_R$,
and with a reference curve being the bundle $B$ in the dual of $G$.

Assume the contrary: let $(Q_i)_{i \in I}$ be such that $\bundleword(Q_{i'})$
does not contain $u^{r-10}$ as a subword.
By Lemma \ref{lem:bwrep:spiraling-ring}, the winding number of $Q_{i'}$
is at most $r-10$. By Observation \ref{obs:pm1}, the winding number
of any $Q_i$ for $i \in I$ is at most $r-9$.

By Lemmata \ref{lem:spiral-split} and \ref{lem:bwrep:spiraling-ring},
for each $i \in I$, the intersection of $P_i$ with $G_R$
is a path $P_{i,\ast}$ that connects a vertex on $f_1$ with a vertex on $f_2$
such that $B\,\bundleword(P_{i,\ast}) = u^{r-1}$ and, consequently,
the winding number of $P_{i,\ast}$ is $r-2$.
Note that for any $i \notin I$, $P_i$ does not intersect $G_R$.

By Lemma \ref{lem:onewayspiral}, there exists a sequence
of vertex-disjoint paths $(P_{i,\ast}')_{i \in I}$ in $G_R$ such that
$P_{i,\ast}'$ connects the same pair of vertices as $P_{i,\ast}$,
but has winding number at most $r-3$. Therefore
$B\,\bundleword(P_{i,\ast}') = u_{r(i)}$ for some $r(i) < r-1$.

For each $1 \leq i \leq k$, construct a path $P_i'$ from $P_i$ by replacing
the subpath $P_{i,\ast}$ with $P_{i,\ast}'$. As $P_{i,\ast}$ and $P_{i,\ast}'$
have the same endpoints, $P_i'$ connects the $i$-th terminal pair.
As $P_{i,\ast}'$ were pairwise disjoint, and the intersection of $(P_i)_{i=1}^k$
with $G_R$ is exactly $(P_{i,\ast})_{i \in I}$, the sequence $(P_i')_{i=1}^k$
is a solution to \probshort{} in $G$. However, $\bundleword(P_i) = \bundleword(P_i')$
for $i \notin I$ and $\bundleword(P_i')$ is a proper subsequence of $\bundleword(P_i)$
  for $i \in I$. As $\iota \in I$, this contradicts the minimality of $(P_i)_{i=1}^k$.
\end{proof}

\subsection{Ring components: bound on the interaction}\label{ss:rings:bounds}

The goal of this section is to show that a minimal solution
has a limited interaction with ring components.
As a tool, we use the following routing argument.
\begin{lemma}[one-directional routing]\label{lem:flow-routing}
Let $H$ be a connected graph, embedded in a plane. Let $\ell > 0$ be an integer and
let $(p_j)_{j=1}^{2\ell}$, $(q_j)_{j=1}^{2\ell}$ be pairwise distinct vertices
that lie on the outer face of $H$ in the order $p_1,p_2,\ldots,p_{2\ell},q_{2\ell}, q_{2\ell-1}, \ldots, q_1$.
Moreover, let $(a_j)_{j=1}^\ell$, $(b_j)_{j=1}^\ell$ be pairwise
distinct vertices that lie on the outer face of $H$ in the order $b_1,b_2,\ldots,b_\ell,a_\ell, a_{\ell-1}, \ldots, a_1$
and such that $a_1,a_2,\ldots,a_\ell$ lie on the outer face of $H$ between $p_{2\ell}$ and $q_{2\ell}$ (possibly $p_{2\ell} = a_1$ or $q_{2\ell}=a_\ell$)
and $b_1,b_2,\ldots,b_\ell$ lie on the outer face of $H$ between $q_1$ and $p_1$ (possibly $p_1 = b_1$ or $q_1=b_\ell$).
Assume that there exists a set of $2\ell$ vertex-disjoint paths $(P_j)_{j=1}^{2\ell}$ in $H$ such that
for $1 \leq j \leq 2\ell$, $P_j$ starts in $p_j$ and ends in $q_j$,
and a set of $\ell$ vertex-disjoint paths $(C_j)_{j=1}^\ell$ in $H$ such that
for $1 \leq j \leq \ell$, $C_j$ starts in $a_j$ and ends in $b_j$.

Then there exist a sequence $(R_j)_{j=1}^\ell$ of vertex-disjoint paths in $H$ such that for $1 \leq j \leq \ell$, the path $R_j$
starts in $p_j$ and ends in $q_{j+\ell}$.
\end{lemma}
\begin{proof}
We first note that, by planarity, any set of $\ell$ vertex-disjoint paths that connects 
$\{p_1,p_2,\ldots,p_\ell\}$ with $q_{\ell+1},q_{\ell+2},\ldots,q_{2\ell}$ needs to connect
$p_\iota$ with $q_{\ell+\iota}$ for $1 \leq \iota \leq \ell$. Therefore, by Menger's theorem
if the desired paths $(R_j)_{j=1}^\ell$ does not exist, there exists a set $X \subseteq V(H)$, $|X| < \ell$,
such that in $H \setminus X$ no vertex $q_{\ell+\iota'}$ is reachable from $p_\iota$ for any $1 \leq \iota,\iota' \leq \ell$,
However, as $|X| < \ell$, there exists $1 \leq \iota,\iota' \leq \ell$ such that $X$ does not contain any vertex of
$P_\iota$, $P_{\ell+\iota}$ and $C_{\iota'}$. Due to the order of the vertices $p_j$, $q_j$, $a_j$ and $b_j$ on the outer face of $H$,
the union of these three paths contain a path from $p_\iota$ to $q_{\ell+\iota}$, a contradiction.
\end{proof}

We now use the isolation components $\disccomp_{OUT,\lambda}$ and $\disccomp_{IN,\lambda}$
to show that there is a bounded number of ring parts in a minimal solution.

\begin{theorem}[bound on the number of isolation passages]\label{thm:isolation-passages}
Let $(G,\decomp,\bundleset)$ be a bundled instance, such that $\decomp$ has isolation $(\Lambda,d)$ where $\Lambda > 0$ and $d \geq 2k$.
Assume that $G$ is a yes instance to \probshort{} and let $(P_i)_{i=1}^k$ be a minimal solution.
Then the paths $P_i$, in total, contain at most $4|\bundleset|^2 k^2$ isolation passages.
\end{theorem}
\begin{proof}
Consider one ring component $\ringcomp \in \decomp$ and one its isolation component $\disccomp$. Fix one bundle $B^1$ that contains
arcs with ending points contained in $\disccomp$ fix one bundle $B^2$ that contains arcs with starting points in $\disccomp$.
Consider isolation passages of paths $(P_i)_{i=1}^k$ such that
that $B^1$ contains the arc preceding the isolation part and $B^2$ contains the arc succeeding it.
Note that between $B^1$ and $B^2$ the path $P_i$ may contain arcs only in bundles
that are isolation bundles of component $\disccomp$.

To prove the lemma, we need to show that in a minimal solution there are at most $4k^2$ isolation passages.
for a fixed choice of $\ringcomp$, $\disccomp$, $B^1$ and $B^2$. Note that the choice of $B^1$ determines $\ringcomp$ and $\disccomp$,
thus there are less than $|\bundleset|^2$ choices.

Consider a choice of $\ringcomp$, $\disccomp$, $B^1$ and $B^2$ where there are more than $4k^2$ such isolation passages.
Denote all such passages, in the order of their appearance on $B^1$, as $Q_1, Q_2, \ldots, Q_r$.
Assume $Q_j$ belongs to a path $P_{i_j}$.
Let $\mathrm{prec(Q_j)} \in B^1$ be an arc preceding $Q_j$ on $P_{i_j}$ and $\mathrm{succ}(Q_j) \in B^2$ be an arc succeeding $Q_j$ on $P_{i_j}$.

Let $H = \ringcl(\ringcomp)[\disccomp]$.
By the definition of isolation, all arcs of $G$ that are not contained in $H$, but are incident to at least one vertex of $H$,
are contained in one of the two faces of $H$: the one containing $\ringcomp$ and $B^2$, and the one containing $B^1$.
Thus, by planarity, the paths $Q_1, Q_2, \ldots, Q_r$ arrive at bundle $B^2$ in order $Q_{q+1}, Q_{q+2}, \ldots, Q_r, Q_1, Q_2, \ldots, Q_q$
for some $1 \leq q \leq r$. We assume $q \geq r/2$; the second case is symmetrical. Note that $q \geq r/2$ implies $q > 2k^2$.

Consider an area $A_0$ of the plane enclosed by:
\begin{itemize}
\item the path $Q_1$, together with parts of arcs of $\mathrm{prec}(Q_1)$ and $\mathrm{succ}(Q_1)$;
\item a drawing of a path in the dual of $G$, connecting the arc $\mathrm{succ}(Q_1)$ and the arc $\mathrm{succ}(Q_q)$, that intersects only the arcs of $B^2$ between
$\mathrm{succ}(Q_1)$ and $\mathrm{succ}(Q_q)$;
\item the path $Q_q$, together with parts of arcs of $\mathrm{prec}(Q_q)$ and $\mathrm{succ}(Q_q)$;
\item a drawing of a path in the dual of $G$, connecting the arc $\mathrm{prec}(Q_1)$ and the arc $\mathrm{prec}(Q_q)$, that intersects only the arcs of $B^1$ between
$\mathrm{prec}(Q_1)$ and $\mathrm{prec}(Q_q)$.
\end{itemize}
Note that any path may enter $A_0$ only via an arc of $B^1$ between $\mathrm{prec}(Q_1)$ and $\mathrm{prec}(Q_q)$
and leave $A_0$ only via an arc of $B^2$ between $\mathrm{succ}(Q_1)$ and $\mathrm{succ}(Q_q)$.
Moreover, $A_0$ does not contain any terminal, as it contains only elements of $H$. Therefore the paths $Q_1,Q_2,\ldots,Q_q$ (with preceding and succeeding arcs)
are the only intersections of the solution $(P_i)_{i=1}^k$ with the closure of the area $A_0$.

By Corollary \ref{cor:bpdecomp}, the word $i_1i_2\ldots i_q$ over alphabet $\{1,2,\ldots,k\}$ can be decomposed into powers of at most $2k$ words, such that
each of these words contains each symbol at most once. As $q > 2k^2$, there exist a subword $w$ or $i_1i_2\ldots i_q$ of the form $u^2$, where $u$ contains
each symbol at most once. Assume $|u| = \ell$ and $w = i_j i_{j+1} \ldots, i_{j+2\ell-1}$, $i_\iota=i_{\iota+\ell}$ for $j \leq \iota < j+\ell$.

Consider now a subarea $A$ of $A_0$, defined similarly as $A_0$, but with paths $Q_j$ and $Q_{j+2\ell-1}$ as borders. 
The intersection of the solution $(P_i)_{i=1}^k$ with the closure of $A$ is exactly the set of paths $Q_j, Q_{j+1}, \ldots, Q_{j+2\ell-1}$, together
with parts of their preceding and succeeding arcs of $B^N$ and $B^R$.
For $j \leq \iota < j+2\ell$, let $x_\iota$ and $y_\iota$ be the first and the last point of $Q_\iota$, respectively.
For $j \leq \iota < j+\ell$, the paths $Q_\iota$ and $Q_{\iota+\ell}$ belong to $P_{i_\iota}$ and $i_\iota$ are pairwise distinct for $j \leq \iota < j+\ell$.

Let $H_A$ be a subgraph of $H$ consisting of vertices and edges that are contained entirely in the closure of $A$.
Assume that $Q_j$ appears earlier on the path $P_{i_j}$ than $Q_{j+\ell}$; in the opposite case, we may consider
a mirror image of $H_A$.
By Lemma \ref{lem:bporder}, for any $j \leq \iota < j+\ell$, the path $Q_\iota$ appears on $P_{i_\iota}$ earlier than $Q_{\iota+\ell}$.

We now note that we may apply Lemma \ref{lem:flow-routing} to the graph $H_A$ with paths $Q_\iota$ connecting $(x_\iota)_{\iota=j}^{j+2\ell-1}$
with $(y_\iota)_{\iota=j}^{j+2\ell-1}$. Indeed, recall that $d\geq 2k$, thus $H$ contains $2k$ alternating cycles and each of these cycles
intersects all paths $(Q_\iota)_{\iota=j}^{j+2\ell-1}$. Therefore $k$ of these cycles (in one of the directions) yield
the promised paths $(C_j)_{j=1}^\ell$. We infer that in $H_A$ there exist a sequence of vertex-disjoint paths $(R_\iota)_{\iota=j}^{j+\ell-1}$
such that $R_\iota$ starts in $x_\iota$ and ends in $y_{\iota+\ell}$.

Consider the following set of paths $(P_i')_{i=1}^k$. For each $j \leq \iota < j+\ell$, we remove from $P_{i_\iota}$
a subpath starting from $x_\iota$ and ending at $y_{\iota+\ell}$ and replace it with $R_\iota$. Since the intersection of $H_A$
with the solution $(P_i)_{i=1}^k$ consists of the paths $\{Q_\iota: j \leq \iota < j+2\ell\}$ only, $(P_i')_{i=1}^k$ is a solution to \probshort{} on $G$.
Moreover, each path $P_{i_\iota}'$ contains strictly less arcs of non-isolating bundles than $P_{i_\iota}$, as we removed from $P_{i_\iota}$ an arc $\mathrm{prec}(Q_{\iota+\ell})$, and the paths $R_\iota$ contain
only isolation arcs of $\disccomp$.
This contradicts the minimality of $(P_i)_{i=1}^k$ and concludes the proof of theorem.

\end{proof}

We now show that a minimal solution cannot oscillate between isolation layers.
\begin{theorem}[no oscillators in the ring]\label{thm:oscillators}
Let $(G,\decomp,\bundleset)$ be a bundled instance, such that $\decomp$ has isolation
$(\Lambda,d)$ where $\Lambda \geq 1$ and $d \geq f(k,k) + 4$,
where $f(k,t)=2^{O(kt)}$ is the bound on the type-$t$ bend promised by Lemma \ref{lem:bendbound}.
Let $\disccomp$ be an isolation component of a ring component $\ringcomp$,
and let $f_1$ and $f_2$ be the two faces of $\ringcl(\ringcomp)[V(\disccomp)]$ that contain
vertices and edges of $G \setminus \disccomp$.
Let $(P_i)_{i=1}^k$ be a minimal solution to \probshort{} on $G$ and assume that, for some $1 \leq \iota \leq k$,
the path $P_\iota$ contains a subpath $P$ that starts with an arc $e_1$ with ending point in $\disccomp$, 
ends with an arc $e_2$ with starting point in $\disccomp$, is contained in $\ringcl(\ringcomp)$ and both $e_1$ and $e_2$ lie in $f_1$.
Then $P$ does not contain any arc that lies in $f_2$.
\end{theorem}
\begin{proof}
Assume the contrary, and the path $P$ contains some arc in $f_2$;
in particular, $P$ contains a bundle arc $e$ of bundle $B$ leading from $\disccomp$ to another
disc or ring component of $\ringcl(\ringcomp)$ that lies in $f_2$.
Note that $B$ is not an isolating component.
Our goal is to apply the bound from Lemma \ref{lem:bendbound} to reroute
$(P_i)_{i=1}^k$ so that it does not use the arc $e$. In this way we would strictly
decrease the number of non-isolation bundles of the solution, a contradiction
to its minimality. However, to reuse Lemma \ref{lem:bendbound},
we need to make arguments similar to the proof of Theorem \ref{th:irrelevant}
using Lemma \ref{lem:unique}.

Let $C_0^\ast, C_1^\ast, \ldots, C_{d-1}^\ast$ be a fixed choice of alternating cycles in $\ringcl(\ringcomp)[V(\disccomp)]$, separating
$f_1$ from $f_2$.
Without loss of generality assume that $f_1$ is the infinite face
of $\ringcl(\ringcomp)[V(\disccomp)]$,
and $C_i^\ast$ encloses $C_j^\ast$ whenever $i < j$ (i.e., the cycle $C_0^\ast$ is close to $f_1$
and $C_{d-1}^\ast$ is close to $f_2$).
Let $H^\ast$ be a subgraph of $G$ consisting of the solution $(P_i)_{i=1}^k$
and the cycles $(C_j^\ast)_{j=0}^{d-1}$. Let $H$ be a graph obtained from
$H^\ast$ by contracting any arc that belongs both to some path $P_i$ and to some cycle
$C_j^\ast$. Each path $P_i$ projects to a path $Q_i$ in $H$, $P$ projects to $Q$,
a subpath
of $Q_\iota$, and
each cycle $C_j^\ast$ projects to a cycle $C_j$; note that $C_j$ may consist of a single
arc (self-loop at some vertex), but does not disappear completely.
Note that $(Q_i)_{i=1}^k$ is a solution to \probshort{} on $H$
and the cycles $(C_j)_{j=0}^{d-1}$ are free.

Let $x_0$ and $y_0$ be the two points of $Q \cap C_0$ closest (on $Q$)
to the chosen edge $e$.
As $P$ is contained in $\ringcl(\ringcomp)$, for one of the two subpaths of $C_0$
between $x_0$ and $y_0$, denote it $C_0'$, we have that the cycle $C_0' \cup Q[x_0,y_0]$
does not enclose any point of $f_2$ and, therefore, does not enclose any terminal.
By Lemma \ref{lem:concentricbend2}, there exists subpaths $C_i'$ of $C_i$ with endpoints
$x_i$ and $y_i$ for $1 \leq i < d-1$,
such that $(Q[x_0,y_0],C_0',C_1',\ldots,C_{d-2}')$ is a $(d-2)$-bend.

Let $H^1$ be a subgraph of $H$ that consists of the paths $(Q_i)_{i=1}^k$
and the chords $(C_j')_{j=0}^{d-2}$. Clearly, $(Q_i)_{i=1}^k$ is a solution
to \probshort{} on $H^1$ as well, and
$(R,C_0',C_1',\ldots,C_{d-2}')$ is a $(d-2)$-bend in $H^1$.
Let $(Q_i^\circ)_{i=1}^k$ be a solution
to \probshort{} on $H^1$ that uses minimum possible number of edges that do
not lie on the chords $(C_j')_{j=0}^{d-2}$. We claim the following.

\begin{claim}\label{cl:no-isolation-bend}
The paths $(Q_i^\circ)_{i=1}^k$ do not use the arc $e$.
\end{claim}
\begin{proof}
Construct a graph $H^2$ from $H^1$ as follows: first, for any $0 \leq j \leq d-2$,
connect $x_j$ with $y_j$ outside the bend $(R,C_0',C_1',\ldots,C_{d-2}')$
with an arc $a_j$ in a direction such that $C_j^2 := C_j' \cup \{a_j\}$ is a directed cycle.
Note that this can be done such that $H^2$ is planar
and without changing the embedding of $H^1$, as the vertices $x_j$ and $y_j$
lie in a good order along the path $Q[x_0,y_0]$:
simply draw the arc $a_j$ parallely to the path $Q[x_0,y_0]$.
In this manner we obtain that $(C_j^2)_{j=0}^{d-2}$ is a set of
alternating concentric cycles not enclosing any terminal.

Now construct a graph $H^3$ from $H^2$ by first removing
any arc that does not belong to any cycle $C_j^2$ nor any
path $Q_i^\circ$ and then by contracting any arc that belongs
both to some path $Q_i^\circ$ and some cycle $C_j^2$. Note that
$a_j$ does not become contracted in this manner, and the cycle $C_j^2$
projects to a cycle $C_j^3$ in $H^3$. Moreover, the solution
$(Q_i^\circ)_{i=1}^k$ projects to a solution $(Q_i^3)_{i=1}^k$ to \probshort{} in $H^3$
and the cycles $C_j^3$ are free with respect to this solution.

Assume that $e \in Q_\eta^\circ$ for some $1 \leq \eta \leq k$; as $e$ does not lie on any cycle $C_j'$,
$e \in Q_\eta^3$ as well.
Note that $C_0'$ separates $e$ from the terminals in $H^1$, and thus $C_0^3$ separates $e$ from the terminals in $H^3$.
Thus, there exist vertices $x_0^3, y_0^3 \in Q_\eta^3 \cap C_0^3$
that are closest to $e$ on $Q_\eta^3$;
let $C_0^{3+}$ be the subpath of $C_0^3$ between $x_0^3$ and $y_0^3$ that
does not contain $a_0$. Note that the cycle $C_0^{3+} \cup Q_\eta^3[x_0^3,y_0^3]$
does not enclose
any terminal nor any arc $a_j$, contains a vertex of $C_{d-2}^3$.
By Lemma \ref{lem:concentricbend2}, there is a $(d-4)$-bend on $Q_\eta^3$
that does not enclose any arc $a_j$, $0 \leq j \leq d-2$.

As $d \geq f(k,k) + 4$, by Lemma \ref{lem:bendbound}, $(Q_i^3)_{i=1}^k$
is not a minimal solution to \probshort{} in $H^3 \setminus \{a_j:0 \leq j \leq d-2\}$.
Let $(R_i^3)_{i=1}^k$ be a different solution.
Let $E_R = \bigcup_{i=1}^k E(R_i^3)$, $E_Q = \bigcup_{i=1}^k E(Q_i^3)$
and $E_C = \bigcup_{j=0}^{d-2} E(C_j^3)$. Note that $E(H^3) = E_Q \cup E_C$,
and $E_Q$ is a set of disjoint paths. Hence $E_R \setminus E_C$ is a proper subset of $E_Q$.

Similarly as in the proof of Theorem \ref{th:irrelevant}, the solution $(R_i^3)_{i=1}^k$
yields a solution $(R_i^\circ)_{i=1}^k$ in $H^1$. Indeed, $H^3$ is created
from $H^1$ by removing some arcs, adding arcs $a_j$ (not used by any path $R_i^3$)
and contracting some arcs: however, as $H^2$ consisted
of cycles $(C_j^2)_{j=0}^{d-2}$ and paths $(Q_i^\circ)_{i=1}^k$,
if we uncontract an arc $xy$ the arc $xy$ is the only arc leaving $x$ and entering $y$,
and any path going through the image of $xy$ in $H^3$ can be redirected via the arc $xy$ in $H^2$.
However, the solution $(R_i^\circ)_{i=1}^k$ uses strictly less arcs outside the
chords $(C_j')_{j=0}^{d-2}$ than $(Q_i^\circ)_{i=1}^k$, a contradiction to the choice
of $(Q_i^\circ)_{i=1}^k$. This finishes the proof of the claim.
\end{proof}

Clearly, $(Q_i^\circ)_{i=1}^k$ is a solution to \probshort{} on both $H$ and $H^1$
that does not use the arc $e$.
Again, by the construction of the graph $H$ from $H^\ast$,
the solution $(Q_i^\circ)_{i=1}^k$ in $H$ yields a solution $(P_i^\circ)_{i=1}^k$
in $H^\ast$ that does not use the arc $e$.
Hence, by the construction of $H^\ast$,
the bundle arcs of non-isolation bundles of the solution $(P_i^\circ)$
is a proper subset of the bundle arcs of non-isolation bundles of the solution
$(P_i)_{i=1}^k$, a contradiction to the minimality of $(P_i)_{i=1}^k$.
\end{proof}

\begin{corollary}[no ring visitors]\label{cor:no-visitors}
Let $(G,\decomp,\bundleset)$ be a bundled instance, such that $\decomp$ has isolation
$(\Lambda,d)$ where $\Lambda \geq 1$ and $d \geq f(k,k)+4$,
where $f(k,t)=2^{O(kt)}$ is the bound on the type-$t$ bend promised by Lemma \ref{lem:bendbound}.
Assume that $G$ is a YES-instance to \probshort{} and let $(P_i)_{i=1}^k$ be a minimal solution.
Then there are no ring visitors on any path $P_i$.
\end{corollary}
\begin{proof}
We apply Theorem \ref{thm:oscillators} for $\disccomp = \disccomp_{\Gamma,\Lambda}$
for $\Gamma \in \{IN,OUT\}$ and $\disccomp_{\Gamma,\Lambda}$ be an isolating component
 of a ring component $\ringcomp$.
\end{proof}

\begin{corollary}[bundle words with ring holes for paths in a minimal solution]\label{cor:bwholes-minsol}
Let $(G,\decomp,\bundleset)$ be a bundled instance, such that $\decomp$ has isolation
$(\Lambda,d)$ where $\Lambda \geq 2$ and $d \geq f(k,k)+4$,
where $f(k,t)=2^{O(kt)}$ is the bound on the type-$t$ bend promised by Lemma \ref{lem:bendbound}.
Assume that $G$ is a YES-instance to \probshort{} and let $(P_i)_{i=1}^k$ be a minimal solution.
Then, for any $0 \leq \lambda < \Lambda-1$ and $1 \leq i \leq k$,
the unique bundle word with level-$\lambda$ ring holes that is consistent with
$P_i$ is flat, that is, the its level-$(\lambda+1)$ projection
does not contain any arcs of bundles of level $\lambda$ or lower.
\end{corollary}
\begin{proof}
As a bundle word with level-$(\lambda+1)$ ring holes does not contain
any symbol corresponding to an arc that belongs to a level-$(\lambda+1)$ ring passage,
any level-$\lambda$ ring bundle in such a bundle word with ring holes
needs to correspond to an arc on a structure forbidden by Theorem \ref{thm:oscillators}.
\end{proof}

\section{From bundle words to disjoint paths}\label{s:words-to-paths}

In this section we prove the following theorem that will be the main tool both for measuring the good guesses for winding numbers, as well as for seeking the ultimate solution.

\begin{theorem}\label{thm:words-to-paths}
For any bundled instance $(G,\decomp,\bundleset)$ of ring isolation $(\Lambda,d)$, $\Lambda,d>0$, any $1 \leq \lambda <\Lambda$
and any sequence $(\bwholes_i)_{i=1}^k$ of bundle words with level-$\lambda$ ring holes that do not contain any level-$0$ bundles,
one may in polynomial time either:
\begin{enumerate}
\item correctly conclude that there is no solution $(P_i)_{i=1}^k$ to \probshort{} on $G$ such that $P_i$ is consistent with $\bwholes_i$ for each $1 \leq i \leq k$; or
\item compute a solution $(P_i)_{i=1}^k$ to \probshort{} on $G$ with the following property: for any $1 \leq i \leq k$ and any bundle $B$
that is a normal bundle or has at least one endpoint in an isolation bundle of level higher than $\lambda$, the number of appearances of $B$ in $\bundleword(P_i)$
is not greater than the number of appearances of $B$ in $\bwholes_i$.
\end{enumerate}
\end{theorem}

The rest of this Section is devoted to the proof of Theorem~\ref{thm:words-to-paths}. Intuitively, the argument is based on designing an instance of cohomology problem in a natural manner, and then running the algorithm of Theorem~\ref{thm:coh-alg}. However, formal construction of the instance and proof that the output of the algorithm of Theorem~\ref{thm:coh-alg} is meaningful requires technical effort. In particular, we need to check that bundle words with level-$\lambda$ ring holes $(\bwholes_i)_{i=1}^k$ contain already enough information to deduce the topology of path network in the graph.

Apart from the decomposition $\decomp$, we will be working as well on a modified decomposition $\decomp'$ where for every ring component $\ringcomp$, all the isolation components of $\ringcomp$ up to level $\lambda$ are included into $\ringcomp$. In other words, for every ring component $\ringcomp$, we identify component $\ringcomp$ together with first $\lambda$ layers of IN- and OUT- isolation into one ring component $\ringcl^\lambda(\ringcomp)$. We also define $\bundleset'$ to be the set of these bundles $B \in \bundleset$ that are normal bundles
or have at least one endpoint in a isolation component of level higher than $\lambda$.

In decomposition $\decomp'$, for every component $\anycomp$ choose an arbitrary vertex $v(\anycomp)\in V(\anycomp)$, and for every bundle $B \in \bundleset'$ choose an arbitrary arc $a(B)\in B$. Moreover, if $B$ connects components $\anycomp_1$ and $\anycomp_2$, let us choose arbitrary undirected paths $P(B,\anycomp_1)$ and $P(B,\anycomp_2)$ in $\anycomp_1,\anycomp_2$ that connect $v(\anycomp_1),v(\anycomp_2)$ with $a(B)$, respectively. Note that this is possible as every component is weakly connected. Moreover, we can choose the paths so that they do not cross each other in the planar embedding (though may use the same arcs). In addition, for every ring component $\ringcomp \in \decomp'$ choose an arbitrary undirected cycle $C(\ringcomp)$ of homotopy $1$ that passes through $v(\ringcomp)$, again in such a manner that this cycle does not cross any path $P(B,\ringcomp)$ for any bundle $B$ adjacent to $\ringcomp$.

Now, given a sequence $(\bwholes_i)_{i=1}^k$ of bundle words with level-$\lambda$ ring holes in decomposition $\decomp$, we construct undirected paths $(Q_i)_{i=1}^k$ connecting respective terminals. Intuitively, paths $Q_i$ are going to model seeken paths $P_i$ in the solution. Consider one bundle word with level-$\lambda$ holes $\bwholes_i=((p_{i,j})_{j=0}^{q(i)}, (w_{i,j})_{j=1}^{q(i)})$. We say that a bundle is {\emph{deep}} if both its endpoints are either in ring components or in isolation components of level lower or equal to $\lambda$, i.e., it does not appear in $\bundleset'$; otherwise the bundle is {\emph{shallow}}. Note that by our assumption on $(\bwholes_i)_{i=1}^k$, every deep bundle of each $\bwholes_i$ has both endpoints in an isolation component of level lower or equal to $\lambda$. In each bundle word $p_{i,j}$, distinguish maximal subwords consisting of deep bundles only. Then (note that in the following we refer to components in decomposition $\decomp$):
\begin{itemize}
\item For every two consecutive shallow bundles $B_1,B_2$ on $p_{i,j}$, such that $B_1$ is directed toward a disc component $\disccomp$ while $B_2$ is directed from $\disccomp$, take concatenation of paths $P(B_1,\disccomp)$ and $P(B_2,\disccomp)$ (note here that $\disccomp\in \decomp'$).
\item Consider now any two shallow bundles $B_1,B_2$ on $p_{i,j}$ such that the subword $v$ between $B_1,B_2$ consists of deep bundles only. Then bundles $B_1,B_2$ are on the same side of $\ringcomp_*$, and the subword $v$ between them can visit only isolation components of $\ringcomp$ of levels smaller or equal than $\lambda$ on the same side as $B_1,B_2$. Examine word $v$ and compute the signed number $w$ of how many times bundles from $v$ cross the reference curve of $\ringcl^\lambda(\ringcomp)$. Observe that, by the construction of the reference curve of $\ringcl^\lambda(\ringcomp)$, any path consistent with bundle word $B_1vB_2$ will have winding number $w$ in $\ringcl^\lambda(\ringcomp)$. Note that any such path leads from one face of $\ringcl^\lambda(\ringcomp)$ to the same face of $\ringcl^\lambda(\ringcomp)$, so its winding number with respect to the reference curve of $\ringcl^\lambda(\ringcomp)$ must belong to the set $\{-1,0,1\}$. Hence, if $w\notin \{-1,0,1\}$, then we can conclude that there cannot be any solution consistent with bundle words with level-$\lambda$ ring holes $(\bwholes_i)_{i=1}^k$, and we terminate the whole algorithm returning the negative answer. Otherwise, we take the concatenation of paths $P(B_1,\ringcl^\lambda(\ringcomp))$, $C(\ringcl^\lambda(\ringcomp))^{w'}$ and $P(B_2,\ringcl^\lambda(\ringcomp))$, where $w'$ is chosen such that the resulting path has winding number $w$ in $\ringcl^\lambda(\ringcomp)$. Note here that of course $B_1$ and $B_2$ are adjacent to $\ringcl^\lambda(\ringcomp)$ in $\decomp'$.
\item For every two bundles $B_1,B_2$ such that $B_1$ is the last bundle of $p_{i,j}$ and $B_2$ is the first bundle of $p_{i,j+1}$, perform the following construction. Let $\ringcomp$ be the ring component such that $B_1,B_2$ are on different sides of $\ringcl^\lambda(\ringcomp)$ and let $w:=w_{i,j+1}$ be the winding number between them. Take concatenation of $P(B_1,\ringcl^\lambda(\ringcomp))$, $C(\ringcl^\lambda(\ringcomp))^{w'}$ and $P(B_2,\ringcl^\lambda(\ringcomp))$, where $w'$ is chosen such that the resulting path has winding number $w$.
\end{itemize}
Finally, obtain $Q_i$ by concatenating constructed undirected paths in a natural order imposed by $(\bwholes_i)_{i=1}^k$, and gluing them in between by arcs $a(B)$ for $B\in \bwholes_i$. If some arc $a(B)$ is traversed in the constructed path in a wrong direction, we may immediately terminate the computation providing a negative answer: the corresponding path of the solution would also need to traverse the same bundle in the wrong direction, which is impossible.

We have already constructed models $Q_i$ for paths $P_i$, but to finish the construction of the instance of the cohomology problem, we need to choose consistently the order of paths on shared arcs.

Consider two paths $Q_i,Q_j$ (possibly $i=j$) and let $M_i,M_j$ be some (possibly consisting of a single vertex) maximal fragments of $Q_i,Q_j$ that are common, i.e., $M_i$ and $M_j$ are the same path, but the arcs $p_i,p_j$ preceding $M_i$ and $M_j$ on $Q_i,Q_j$, respectively, are different, as well as arcs $s_i,s_j$ succeeding $M_i$ and $M_j$ on $Q_i,Q_j$ are different. Contract fragments $M_i,M_j$ to one vertex $v$ and consider the positioning of arcs $p_i,p_j,s_i,s_j$ around $v$. Observe that if there is a solution $(P_i)_{i=1}^k$ consistent with $(\bwholes_i)_{i=1}^k$, then $p_i$ and $s_i$ cannot separate $p_j$ from $s_j$: otherwise, the parts of paths $P_i,P_j$ between the last bundles preceding $M_i,M_j$ and the first bundles succeeding $P_i,P_j$ would need to cross.
Therefore, if we think of paths $Q_i,Q_j$ as curves on the plane, it is possible to spread them slightly along the common path $M_i,M_j$ so that the curves do not intersect along this path. The same reasoning can be performed when paths $Q_i,Q_j$ traverse the common path in different directions, i.e., $M_i,M_j$ are two maximal subpaths such that $M_i$ is $M_j$ reversed, the arc preceding $M_i$ on $Q_i$ is different than the arc succeeding $M_j$ on $Q_j$, and the arc preceding $M_j$ on $Q_j$ is different than the arc succeeding $M_i$ on $Q_i$.

We perform this reasoning for every pair of indices $i,j$, every two maximal common fragments, and both directions of traversal; if for any pair of fragments we obtain inconsistency, we terminate the algorithm providing an answer that no solution can be found. As a result, for every arc that is traversed by some path at least two times, we obtain a constraint between every two traversals, which traversal should be spread towards the face on one side of the arc, and which should be spread to the other side. It can be easily seen that this constraints are transitive, i.e., if traversal $a$ must be left to traversal $b$ and $b$ must be to the left to traversal $c$, then traversal $a$ must be to the left to traversal $c$; if this is not the case, we can again terminate the algorithm providing a negative answer. Hence, we may order the traversals of every arc in a linear order, so that slightly spreading all the paths $(Q_i)_{i=1}^k$ according to these orders on arcs yields a family of pairwise non-intersecting curves, denoted $(\tilde{Q})_{i=1}^k$.

We now consider a free group $\group$ on $k$ generators $g_1,g_2,\ldots,g_k$, corresponding to paths $P_1,P_2,\ldots,P_k$. Define a function $\phi:E(G)\to \group$ by putting for every arc $a$ the product of generators or their inverses corresponding to traversals of this arc: if $P_i$ traverses $a$ respecting direction of $a$ then we take $g_i$, if disrespecting then we put $g_i^{-1}$, and we multiply these group elements according to the linear order of traversals on this arc, found as in the previous paragraph. As the family $(\tilde{Q})_{i=1}^k$ is pairwise non-intersecting, it follows that for every non-terminal vertex $v$, if we compute the product of group elements on the arcs incident to $v$ in the order imposed by the planar embedding, where the value of every arc is taken as its inverse if the arc is incoming to $v$ and as the value itself otherwise, then we obtain $1_\Lambda$.

Let $G^*$ be the dual of $G$ and let $G^+$ be the extended dual of $G$. Clearly, as $E(G)=E(G^*)$, we may consider $\phi$ also as a function from arcs of $G^*$ to $\group$. Moreover, we may naturally extend $\phi$ from $E(G^*)$ to $\phi^+$ defined on $E(G^+)$, as in the proof of Theorem~\ref{lem:dualaltcut}.

We are now going to define an instance of cohomology feasibility problem similarly as in the proof of Theorem~\ref{lem:dualaltcut}. For every arc $a\in E(G^+)$, let $H(a)=\{1,g_1,g_2,\ldots,g_k\}$ if $a\in E(G)$, and let $H(a)=\{1,g_1,g_2,\ldots,g_k,g_1^{-1},g_2^{-1},\ldots,g_k^{-1}\}$ if $a\in E(G^+)\setminus E(G)$, i.e., $a$ was added in the extension. Moreover, let $S$ be the set of end-faces in decomposition $\decomp'$; we would like to stress that we consider decomposition $\decomp'$, hence we do not put any restrictions on faces that are contained in level-$\lambda$ closures of ring components. Thus we have defined an instance of cohomology feasibility problem $(G^+,\group,\phi,H,S)$, to which we can apply the algorithm of Theorem~\ref{thm:coh-alg}. The rest of the proof of Theorem~\ref{thm:words-to-paths} follows from the two lemmata that show equivalence of the cohomology instance and the statement of Theorem~\ref{thm:words-to-paths}.

\begin{lemma}
If there is a sequence of disjoint paths $(P_i)_{i=1}^k$ consistent with $(\bwholes_i)_{i=1}^k$, then instance $(G^+,\group,\phi,H,S)$ has some solution $\psi$.
\end{lemma}
\begin{proof}
Assume that we have a sequence of disjoint paths $(P_i)_{i=1}^k$ consistent with $(\bwholes_i)_{i=1}^k$. Similarly to paths $(Q_i)_{i=1}^k$ and function $\phi$, sequence $(P_i)_{i=1}^k$ also naturally defines a function $\psi:E(G^+)\to \group$, as in the proof of Theorem~\ref{lem:dualaltcut}. Moreover, $\psi$ satisfies all the constraints imposed by $H$. It remains to prove that $\psi$ is cohomologous to $\phi$ via a function $F$ that assigns $1_\group$ to all the end-faces in decomposition $\decomp'$.

For every path $P_i$, decompose it into sequence $(b_{i,0},P_{i,1},b_{i,1},P_{i,2},\ldots,b_{i,t},P_{i,t},b_{i,t+1})$, where $P_{i,j}$ are maximal subpaths contained in decomposition $\decomp'$ and $b_{i,j}$ are bundle arcs from bundles of decomposition $\decomp'$ that appear on $\bwholes_i$ (note that not every bundle appearing in $\bwholes_i$ is represented by some arc $b_{i,j}$, as we omit the bundles that are not in $\bundleset'$). As both $P_i$ and $Q_i$ are consistent with $\bwholes_i$ and $\bwholes_i$ is a level-$\lambda$ bundle word with holes, path $Q_i$ can be similarly decomposed into sequence $(b_{i,0}',Q_{i,1},b_{i,1}',Q_{i,2},\ldots,b_{i,t}',Q_{i,t},b_{i,t+1}')$, where arcs $b_{i,j}$ and $b_{i,j}'$ belong to the same bundle for every $j$, while $P_{i,j}$ and $Q_{i,j}$ traverse the same component from the same bundle to the same bundle. Moreover, if $P_{i,j}$ and $Q_{i,j}$ are connecting sides of the same ring component $\ringcomp$ (regardless whether the same or different), then $P_{i,j}$ and $Q_{i,j}$ have the same winding number in $\ringcomp$. 

It follows that for every component $\anycomp$, all the traversals of paths $(Q_i)_{i=1}^k$ via $\anycomp$ can be simultaneously shifted to $(P_i)_{i=1}^k$ using a by a continuous shift of the interior of $\anycomp$ (formally, by a continuous function $f: U(\anycomp)\times [0,1]\to U(\anycomp)$ such that $f(\cdot,0)$ is identity and in $f(\cdot,1)$ paths $Q_i$ are mapped to respective paths $P_i$). Note that we may choose this shift so that end-faces of bundles adjacent to $\anycomp$ stay invariant, as ends of paths in a particular bundle must be mapped to ends of paths in the same bundle. Now observe that this shift can be modelled by a function $F$ defined at vertices of $G^+$, which measures which paths are shifted over a given face and in which order. If we define this function $F$ for each component separately, then $F$ will be defined consistently in the same manner on the mortar faces of $\decomp'$. Moreover, invariance of the end-faces in the homeomorphisms means that end-faces will be assigned $1_\group$. Hence $\phi$ and $\psi$ are cohomologous via a cohomology that fixes end-faces.
\end{proof}

\begin{lemma}
If the instance $(G^+,\group,\phi,H,S)$ has some solution $\psi$, then there is a solution $(P_i)_{i=1}^k$ to \probshort{} on $G$ with property (2) defined in the statement of Theorem~\ref{thm:words-to-paths}.
\end{lemma}
\begin{proof}
Assume that $\psi:E(G^+)\to \group$ is a function that is consistent with $H$ and is cohomologous to $\psi$ via a function $F$ that assigns $1_\group$ to all the end-faces of $\decomp'$. For $i=1,2,\ldots,k$, let $G_i$ be the subgraph of $G$ consisting of all the arcs $a$ of $G$ such that $\psi(a^*)=g_i$. Note that conditions imposed in the cohomology instance imply that:
\begin{itemize}
\item subgraphs $G_i$ are vertex-disjoint;
\item for every $i=1,2,\ldots,k$, all the vertices of $G_i$ have the same in- and outdegrees, apart from $s_i$ and $t_i$ that have exactly one outgoing arc in $G_i$ and exactly one incoming arc in $G_i$, respectively.
\end{itemize}
By the standard degree counting argument, we infer that for each $i=1,2,\ldots,k$ vertices $s_i$ and $t_i$ must be in the same weakly connected component of $G_i$. Moreover, this weakly connected component must contain an eulerian tour from $s_i$ to $t_i$. By shortcutting this eulerian tour appropriately, we obtain a simple path $P_i$ in $G_i$ leading from $s_i$ to $t_i$. As subgraphs $G_i$ are vertex-disjoint, so do paths $P_i$. It remains to prove that paths $P_i$ satisfy the property that for each shallow bundle $B$, the number of occurrences of $B$ in $\bundleword(P_i)$ is not larger than the number of occurrences of $B$ on $\bwholes_i$. To this end, we will prove that the subgraph $G_i$ in total contains exactly the same number of arcs of $B$ as the number of occurrences of $B$ on $\bwholes_i$; as $P_i$ is a subgraph of $G_i$, the lemma statement follows.

Consider any component $\anycomp$ of $\decomp'$ and any bundle $B \in \bundleset'$ adjacent to $\anycomp$. Let $P^*_B$ be the directed path or directed cycle in $G^+$ traversing the bundle $B$, whose existence is guaranteed by Lemma~\ref{lem:bundle-dual}. Let $f_0,f_1,\ldots,f_r$ are the consecutive faces of $G$ visited by $P^*_B$ and $a_1,a_2,\ldots,a_r$ are consecutive arcs. Suppose for a moment that $f_0$ and $f_r$ are end-faces. We claim that then the products of elements assigned by $\psi$ and by $\phi$ to the arcs $a_1^*,a_2^*,\ldots,a_r^*$ are equal. Indeed, we have that
\begin{eqnarray*}
\prod_{i=1}^r \psi(a_i^*) = \prod_{i=1}^r F^{-1}(f_{i-1})\phi(a_i^*)F(f_i)=F^{-1}(f_0)\left(\prod_{i=1}^r \phi(a_i^*)\right)F(f_r)=\prod_{i=1}^r \phi(a_i^*),
\end{eqnarray*}
where the last equality follows from the fact that $F$ assigns $1_\group$ to all the end-faces. If now $f_0$ and $f_r$ are not end-faces, then $f_0=f_r$, $B$ is the only bundle adjacent to $\anycomp$, and $\prod_{i=1}^r \psi(a_i^*)$ and $\prod_{i=1}^r \phi(a_i^*)$ are conjugate using $F(f_0)=F(f_r)$. 

For $\iota=1,2,\ldots,k$, consider a homomorphism $h_\iota: \Lambda\to \mathbb{Z}$ defined on generators by $h_\iota(g_\iota)=1$ and $h_\iota(g_{\iota'})=0$ for $\iota'\neq \iota$; this homomorphism just counts the number of $g_\iota$-s. It follows that in both cases $h_\iota\left(\prod_{i=1}^r \psi(a_i^*)\right)=h_\iota\left(\prod_{i=1}^r \phi(a_i^*)\right)$: either simply $\prod_{i=1}^r \psi(a_i^*)=\prod_{i=1}^r \phi(a_i^*)$, or
\begin{eqnarray*}
h_\iota\left(\prod_{i=1}^r \psi(a_i^*)\right) & = & h_\iota\left(F^{-1}(f_0)\left(\prod_{i=1}^r \phi(a_i^*)\right)F(f_r)\right) \\
& = & h_\iota\left(\prod_{i=1}^r \phi(a_i^*)\right)+h_\iota(F(f_r))-h_\iota(F(f_0))=h_\iota\left(\prod_{i=1}^r \phi(a_i^*)\right).
\end{eqnarray*}

Note that all the elements $\psi(a_i^*)$ are simply generators, as $\psi$ respects constraints imposed by $H$, and $\phi(a_i)$ are only multiplications of positive powers of generators, as we excluded the cases where paths $Q_i$ traverse bundles in the wrong direction. Hence, $h_\iota\left(\prod_{i=1}^r \psi(a_i^*)\right)$ is the number of arcs of $G_\iota$ contained in the bundle $B$, and $h_\iota\left(\prod_{i=1}^r \phi(a_i^*)\right)$ is the number of passages of path $Q_\iota$ via arc $a(B)$, which, by the construction of $Q_\iota$, is equal to the number of occurrences of $B$ on $\bwholes_\iota$.
\end{proof}

\newcommand{\unknown}{?}
\newcommand{\passset}{\mathcal{P}}
\newcommand{\smallexp}{4}
\newcommand{\zbundleset}{\hat{\bundleset}}
\newcommand{\zdecomp}{\hat{\decomp}}

\section{Guessing bundle words and their winding numbers}\label{s:word-guessing}

The aim of this section is to prove that, for a given bundled instance $(G,\decomp,\bundleset)$ with sufficient isolation,
the number of reasonable bundle words with holes for a minimal solution on $G$ is bounded by a function of $|\bundleset|$, $|\decomp|$ and $k$.
Note that, as each terminal lies in its own component of $\decomp$ and
$G$ is weakly connected, we may assume $2k \leq |\decomp| \leq |\bundleset|+1$.

Formally, we prove the following theorem.

\begin{theorem}\label{thm:word-guessing}
Let $(G,\decomp,\bundleset)$ be a bundled instance of isolation $(\Lambda, d)$, where $\Lambda \geq 3$, $d \geq \max(2k,f(k,k)+4)$, and $f(k,t)=2^{O(kt)}$ is the bound on the type-$t$ bend promised by Lemma \ref{lem:bendbound}.
Then in $2^{O(k^2|\bundleset|^2 \log |\bundleset|)} |G|^{O(1)}$
time one can compute a family of at most
$2^{O(k^2|\bundleset|^2 \log |\bundleset|)}$
sequences $(\bwholes_i)_{i=1}^k$ of bundle words with level-$2$ ring
holes, not containing any level-$0$ bundle,
such that if $(G,\decomp,\bundleset)$ is a YES-instance to \probshort{},
  then there exists a solution $(P_i)_{i=1}^k$ to \probshort{} on $(G,\decomp,\bundleset)$
and a sequence $(\bwholes_i)_{i=1}^k$ in the generated set such that
$P_i$ is consistent with $\bwholes_i$ for each $1 \leq i \leq k$.
\end{theorem}

The proof of Theorem \ref{thm:word-guessing} consists of two steps.
In Section \ref{ss:guessing:words} we show that there is only a bounded number
of reasonable word parts of the promised bundle words with ring holes.
Then, in Section \ref{ss:guessing:windings}, we show that we may pick winding numbers
from a bounded set of candidates.

\subsection{Zooms}\label{sec:zooms}

In this section we introduce a toolbox for measuring spiraling properties of 
parts of the graph $G$. Informally speaking, Lemma \ref{lem:minsol-spiral}
implies that in a minimal solution to \probshort{}, the length of any spiral
is (close to) minimum possible in $G$ --- otherwise, we should be able to reroute
the solution to a shorter spiral, contradicting the minimality of the solution.
We now introduce some gadgets and auxiliary graphs that allow us to measure
this minimum possible size of spirals in $G$.

\begin{defin}[zoom]
Let $(G,\decomp,\bundleset)$ be a bundled instance with isolation $(\Lambda,d)$.
A {\em{zoom}} is a pair $(\zdecomp,\zbundleset)$ such that
$\zbundleset$ is a subset of $\bundleset$ that does not contain any bundle
with an arc incident to a terminal, $\zdecomp \subseteq \decomp$ and
$\anycomp \in \zdecomp$ if and only if there exists $B \in \zbundleset$
whose arcs start or end in $\anycomp$.

We say that a zoom $(\zdecomp,\zbundleset)$
is {\em{level-$\lambda$}} safe for some $0 \leq \lambda \leq \Lambda$,
if for any ring component $\ringcomp \in \decomp$, either the set
$\zbundleset$ contains all bundles with arcs with both endpoints
in $\ringcl^\lambda(\ringcomp)$ or $\zdecomp$ is disjoint with the set
of components of $\ringcl^\lambda(\ringcomp)$.
\end{defin}

\begin{defin}[zoom pass, pack of zoom passes]
Let $(\zdecomp,\zbundleset)$ be a zoom in a bundled instance $(G,\decomp,\bundleset)$
with isolation $(\Lambda,d)$
that is $\lambda$-safe for some $0 \leq \lambda < \Lambda$.
A {\em{level-$\lambda$ zoom pass}} in $(\zdecomp,\zbundleset)$ is a bundle word with level-$\lambda$ ring holes $\bwholes = ((p_j)_{j=0}^q, (w_j)_{j=1}^q)$ in 
$(G,\decomp,\bundleset)$ such that:
\begin{enumerate}
\item $q \geq 1$ or $|p_0| \geq 2$, i.e., the bundle words of $\bwholes$
 contain at least two symbols in total;
\item the first bundle of $p_0$ (called {\em{the first bundle of $\bwholes$}})
  does not belong to $\zbundleset$, but its arcs end
in a component belonging to $\zdecomp$;
\item the last bundle of $p_q$ (called {\em{the last bundle of $\bwholes$}})
  does not belong to $\zbundleset$, but its arcs start
in a component belonging to $\zdecomp$;
\item all other bundles of $p_j$, $0 \leq j \leq q$, belong to $\zbundleset$.
\end{enumerate}.
A {\em{pack of level-$\lambda$ zoom passes}} in $(\zdecomp,\zbundleset)$ is a pair
$((\bwholes_\tau)_{\tau \in I}, (\psi_{B,\alpha})_{B \in \bundleset, 1 \leq \alpha \leq 2})$ where
  \begin{enumerate}
  \item $(\bwholes_\tau)_{\tau \in I}$ is a sequence of level-$\lambda$ zoom passes for some index set $I$;
  \item $\psi_{B,1}$ is a permutation of those indices $\tau \in I$ for which $\bwholes_\tau$ starts
  with $B$, and $\psi_{B,2}$ is a permutation of those indices $\tau \in I$
  for which $\bwholes_\tau$ ends with $B$.
  \end{enumerate}
\end{defin}

Informally speaking, a zoom is a part of a graph $G$ which we investigate, and
zoom passes are bundle words with ring holes
of chosen parts of a solution that pass through a zoom.
The permutations $\psi_{B,1}$ and $\psi_{B,2}$ are intented to correspond to the order of the 
chosen parts on bundle $B$ when entering and leaving the zoom, respectively.
Note that we do not require $(\zdecomp,\zbundleset)$ to be an induced subgraph of
$(\decomp,\bundleset)$ and, consequently, we allow a zoom pass $p$ to start or end with a bundle with arcs
with both endpoints in components of $\zdecomp$ (but this bundle cannot belong to $\zbundleset$).

We now note that the number of choices for $\psi_{B,\alpha}$ is bounded.
\begin{lemma}\label{lem:psi-bound}
Let $(\zdecomp,\zbundleset)$ be a zoom in a bundled instance $(G,\decomp,\bundleset)$ with isolation $(\Lambda,d)$
and let $(\bwholes_\tau)_{\tau \in I}$ be a sequence of level-$\lambda$ zoom passes 
in $(\zdecomp,\zbundleset)$ for some $0 \leq \lambda < \Lambda$.
Then there exist at most $(|I|!)^2$ sequences
$(\psi_{B,\alpha})_{B \in \bundleset, 1 \leq \alpha \leq 2}$
such that
$((\bwholes_\tau)_{\tau \in I}, (\psi_{B,\alpha})_{B \in \bundleset, 1 \leq \alpha \leq 2})$
is a pack of zoom passes in $(\zdecomp,\zbundleset)$.
\end{lemma}
\begin{proof}
For $B \in \bundleset$, let $I_{B,1}$ be a set of indices $\tau \in I$ such that
$\bwholes_\tau$ starts with $B$, and $I_{B,2}$ be a set of indices $\tau \in I$ such that
$\bwholes_\tau$ ends with $B$. By definition, $\psi_{B,\alpha}$
is a permutation of $I_{B,\alpha}$ for any $B \in \bundleset$, $1 \leq \alpha \leq 2$.
The lemma follows from the convexity of the factorial function and the fact that
$$\sum_{B \in \bundleset} |I_{B,1}| = \sum_{B \in \bundleset} |I_{B,2}| = |I|.$$
\end{proof}

We now define how a zoom with a pack of zoom passes defines a sub-instance of
the original \probshort{} bundled instance $(G, \decomp,\bundleset)$.

\begin{defin}[Zoom pass starting and ending gadget]
Consider a bundled instance $(G,\decomp,\bundleset)$ and a bundle $B \in \bundleset$.
Let $I$ be an index set and let $\psi$ be a permutation of $I$.

The {\em{zoom pass starting gadget}} is constructed as follows.
First, take a bi-directional grid of size $|I| \times |B|$, that is,
take a set of vertices $\{v_{\tau,b} : \tau \in I, b \in B\}$ and 
connect $v_{\tau,b}$ and $v_{\tau,b'}$ with arcs in both directions
for $\tau \in I$ and $b$, $b'$ being two consecutive arcs of $B$,
 as well as $v_{\tau,b}$ and $v_{\tau',b}$ for $b \in B$ and $\tau$, $\tau'$ being
 two consecutive indices in the permutation $\psi$ of $I$.
Let $\tau_1 \in I$ be the first element of $I$ in the permutation $\psi$
and $b_1$ be the first arc of $B$. For each $b \in B$, make an arc
from $v_{\tau_1,b}$ to the endpoint of $b$ in $G$.
For each $\tau \in I$, create a new source terminal $s_\tau$ and connect it
with $v_{\tau,b_1}$. See Figure \ref{fig:zoom-gadget} for an illustration.

The {\em{zoom pass ending gadget}} is constructed similarly as
the zoom pass starting gadget, except for three differences:
the arcs are reversed, the constructed terminals are sink terminals,
and the gadget is attached to the starting points of the arcs of $B$.
\end{defin}

\begin{figure}
\begin{center}
\includegraphics{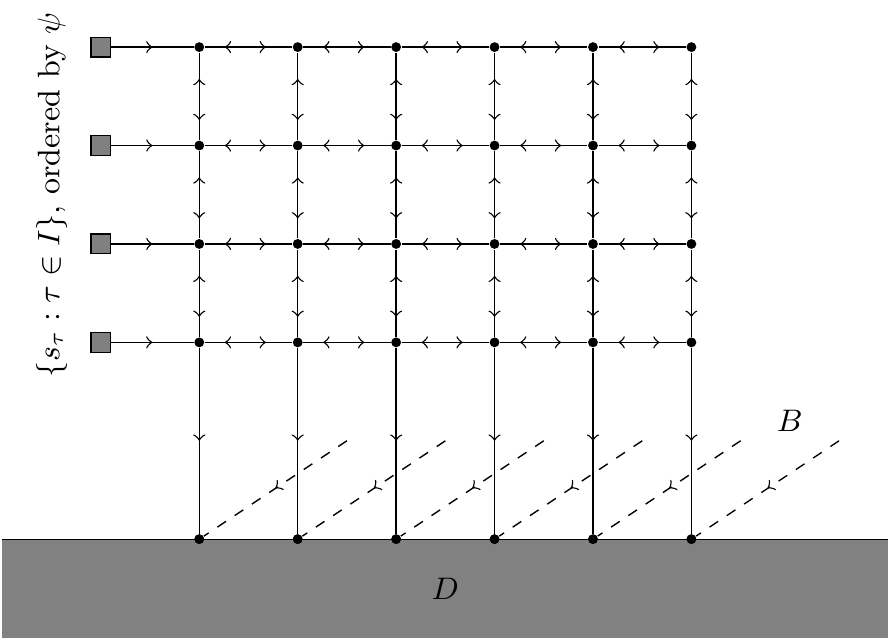}
\caption{An illustration of the zoom pass starting gadget, attached to bundle $B$ with endpoints in component $D$.
The gray squares represent the new terminals, and the dashed lines represent the arcs of the bundle $B$.}
  \label{fig:zoom-gadget}
  \end{center}
  \end{figure}

\begin{defin}[zoom auxiliary graph]
Let $(\zdecomp,\zbundleset)$ be a zoom in a bundled instance $(G,\decomp,\bundleset)$
and let $((\bwholes_\tau)_{\tau \in I}, (\psi_{B,\alpha})_{B \in \bundleset, 1 \leq \alpha \leq 2})$
be a pack of zoom passes in $(\zdecomp,\zbundleset)$.

The {\em{zoom auxiliary graph}} is constructed as follows.
First, we take a subgraph of $G$ consisting of all components
of $\zdecomp$ and all arcs of bundles of $\zdecomp$.
Then, for each bundle $B \in \bundleset$ that is the first
symbol of some bundle word with ring holes $\bwholes_\tau$, we define $I_{B,1}$
to be the set of those indices $\tau \in I$ for which $\bwholes_\tau$ starts with $B$,
and we construct a zoom pass starting gadget
for bundle $B$, index set $I_{B,1}$ and permutation $\psi_{B,1}$.
Similarly, for each bundle $B \in \bundleset$ that is the last
symbol of some $\bwholes_\tau$, we define $I_{B,2}$
to be the set of those indices $\tau \in I$ for which $\bwholes_\tau$ ends with $B$,
and we construct a zoom pass ending gadget
for bundle $B$, index set $I_{B,2}$ and permutation $\psi_{B,2}$.
\end{defin}

We note that a zoom auxiliary graph is always planar: any bundle $B$ 
for which we construct a starting or ending gadget does not belong to $\zbundleset$,
and we can embed these gadgets in the plane in the space occupied by $B$
in the embedding of the graph $G$. Note that this claim remains true also if a bundle $B$
appears both as the first bundle of some zoom pass and the last bundle of some
(possibly other) zoom pass.

Moreover, note that the aforementioned embedding imposes a natural decomposition of 
the zoom auxiliary graph. Formally, as a set of components we take $\zdecomp$,
    the set of bundles $\zbundleset$ and we add the following components and bundles:
\begin{enumerate}
\item for each zoom pass starting gadget for a bundle $B$, we take a disc component
$\disccomp_{B,1}$ that contains the bidirectional grid, and a bundle $B_{B,1}$
that contains arcs connecting $\disccomp_{B,1}$ with the endpoints
of the bundle $B$;
\item similarly, for each zoom pass ending gadget we define disc component
$\disccomp_{B,2}$ and a bundle $B_{B,2}$;
\item each introduced terminal is embedded in its own disc component,
  and its incident arc is embedded in its own bundle.
\end{enumerate}
In this manner we define at most $4|I|$ new components and bundles,
and we create a bundled instance with $|I|$ terminal pairs indexed by $I$.
Let us call this instance {\em{a zoom auxiliary instance}}.

Let $(\zdecomp,\zbundleset)$ be a level-$\lambda$ safe
zoom in a bundled instance $(G,\decomp,\bundleset)$ with isolation $(\Lambda,d)$ for
some $0 \leq \lambda < \Lambda$.
Let $(P_i)_{i=1}^k$ be a solution to \probshort{} on $G$ and
assume $\zdecomp$ does not contain any component with a terminal.
For one path $P_i$, we say that its subpath $P_\tau$ is a {\em{zoom incident}} (with respect to the zoom $(\zdecomp,\zbundleset)$)
  if $P_\tau$:
\begin{enumerate}
\item starts with an arc $b_{\tau,1}$ that belongs to a bundle $B_{\tau,1} \notin \zbundleset$,
  but the ending point of $b_{\tau,1}$ belongs to a component of $\zdecomp$;
\item ends with an arc $b_{\tau,2}$ that belongs to a bundle $B_{\tau,2} \notin \zbundleset$,
  but the starting point of $b_{\tau,2}$ belongs to a component of $\zdecomp$;
\item all other arcs of $P_{\tau}$ either belong to a bundle of $\zbundleset$
  or a component of $\zdecomp$.
\end{enumerate}
Note that all zoom incidents of a path $P_i$ are pairwise disjoint, except for possibly overlapping the first and last arcs; let $P_{i,1},P_{i,2},\ldots,P_{i,q(i)}$ be the sequence
of all zoom incidents on $P_i$ (with respect to $(\zdecomp,\zbundleset)$) in the order of they appearance on $P_i$.
Let $I = \{(i,j): 1 \leq i \leq k, 1 \leq j \leq q(i)\}$.

For $\tau=(i,j) \in I$, let $\bwholes_\tau$ be a bundle word with level-$\lambda$ ring holes constructed as follows:
we take $\bundleword(P_\tau)$ and for each level-$\lambda$ ring passage of $P_\tau$
through $\ringcl^\lambda(\ringcomp)$ for some $\ringcomp \in \zdecomp$, we replace
the subword corresponding to this ring passage with its winding number.
The fact that $\bwholes_\tau$ satisfies all properties of a bundle word with level-$\lambda$
ring holes follows from the assumption that the zoom $(\zdecomp,\zbundleset)$ is level-$\lambda$ safe. 

For any $B \in \bundleset$ and $1 \leq \alpha \leq 2$, let $I_{B,\alpha} = \{\tau \in I: b_{\tau,\alpha} \in B\}$ and $\psi_{B,\alpha}$ be the order of the indices $\tau \in I_{B,\alpha}$ in which arcs $b_{\tau,\alpha}$ appear on $B$.
Observe the following.
\begin{observation}\label{obs:zoom-solution}
The pair
$((\bwholes_{\tau})_{\tau \in I}, (\psi_{B,\alpha})_{B \in \bundleset, 1 \leq \alpha \leq 2})$ defined above is a pack of zoom passes in $(\zdecomp,\zbundleset)$. Moreover,
  the zoom auxiliary instance constructed for this pack of zoom passes
has a solution $(P_\tau^\ast)_{\tau \in I}$, where $P_\tau^\ast$
is constructed from $P_\tau$ by modifying its first and last arcs:
we replace $b_{\tau,1}$ with a path from $s_\tau$ through vertices $v_{\tau,b}$ up to index
$b_{\tau,1}$ and then through vertices $v_{\tau,b_{\tau,1}}$ up to $v_{\tau_1,b_{\tau,1}}$
and to the ending point of $b_{\tau,1}$; a similar replacement is made for $b_{\tau,2}$.
Moreover, the bundle word with level-$\lambda$ ring holes for $P_\tau^\ast$ in the created
instance is equal to $\bwholes_\tau^\ast$, except for the prefix and suffix that correspond to
the passage in the zoom pass starting and ending gadgets.
\end{observation}

\subsection{Deducing bundle words}\label{ss:guessing:words}

In this section we present the first half of the algorithm of Theorem \ref{thm:word-guessing}:
we aim to guess the bundle word part of the required bundle words with holes; the winding number part is guessed in the next subsection.

Consider a minimal solution $(P_i)_{i=1}^k$ to the bundled instance $(G,\decomp,\bundleset)$.
By Corollary \ref{cor:bwdecomp}, the bundle word $\bundleword(P_i)$ can be decomposed into $u_{i,1}^{r_{i,1}} u_{i,2}^{r_{i,2}} \ldots u_{i,s(i)}^{r_{i,s(i)}}$
for some $s(i) \leq 2|\bundleset|$. Note that there is only a bounded number of choices for the words $u_{i,j}$; the difficult
part is to guess the exponents $r_{i,j}$ for those $j$ where $r_{i,j}$ is large.
In this section we guess only exponents $r_{i,j}$ for which $u_{i,j}$ contains only bundles with both endpoints
in disc components (i.e., does not contain any level-$0$ ring bundles); in the next section we `hide' the unknown
exponents by going to the level-$1$ projection of the bundle words and using results of Section \ref{ss:rings:bounds}.

However, if some path $P_\iota$ has subpath with bundle word $u^r$ for some large integer $r$ and $u$ not containing any level-$0$ ring bundles,
we have a spiraling ring $A$ and we may use Lemma \ref{lem:minsol-spiral}: we are not able
to route the same paths through the spiraling ring using less than $r-10$ turns.
Lemma \ref{lem:spiral-split} gives us an answer how many times the paths
$(P_i)_{i=1}^k$ cross the spiraling ring $A$ and in which direction.
Our approach is to create a zoom containing the spiral $u$, attach to it all paths
that want to cross the spiraling ring $A$ and measure the minimal number of turns
they need to make along the spiral $u$.

However, there is a significant difficulty with this approach --- it is unclear where to
attach the paths crossing $A$ to the zoomed spiral $u$.
To create a zoom instance, we need that each zoom pass starts and ends in a bundle
that does not appear in the zoom (in our case, does not appear on the spiral $u$).
However, any path $P_i$, before spiraling on the spiral $u$, may go along
a bundle word $v^q$, where $v$ contains a proper subset of the symbols of $u$.
Thus, we need to include such spirals as well in our zoom instance. Luckily,
we can assume that the exponent $q$ is already known to us, if we assume that
we guess exponents $r_{i,j}$ in the order of increasing lengths $|u_{i,j}|$.

The approach from the previous paragraph, although resolves the issue of
where to attach the zoom passes in the created zoom instance, creates a new problem.
If we apply Theorem \ref{thm:words-to-paths} to the created zoom instance,
the returned paths follow our guessed bundle words in a very relaxed manner.
As we have included in our zoom instance all parts of the solution that
cross the spiraling ring $A$, by Lemma \ref{lem:minsol-spiral}
in the returned solution the path $P_\iota$ needs to spiral at least $r-10$
times in this ring. However, it is no longer true for other terms $v^q$
where $v$ contains a proper subset of the symbols of $u$: the solution
returned by Theorem \ref{thm:words-to-paths} may ``borrow'' some bundles from such
terms, in order to spiral along $u$ significantly {\em{more}} times than $r$.

To circumvent this problem, we add even more passes to our zoom instance:
for each path $P_i$ we add a zoom pass that corresponds to a maximal subpath of
$P_i$ that does not contain arcs from bundles that do not appear in $u$,
and contains a subword $v_0^2$ for some $v_0$. Corollary \ref{cor:copy-word}
ensures that the only unknown exponent in our zoom instance is the number 
of turns along the spiral $u$, while the presence of all these zoom passes
allow us to use Lemma \ref{lem:minsol-spiral} for any term $v^q$ in $\bundleword(P_\iota)$,
not only to $u^r$. Hence, in the solution returned by Theorem \ref{thm:words-to-paths}
the simulated part of the path $P_\iota$ needs to behave very similarly as
in the original solution  $(P_i)_{i=1}^k$, and we
are able to guess the exponent $r$.

Let us now proceed to the formal argumentation. We perform multiple branching steps,
one for each unknown exponent $r_{i,j}$.
We start with the following set of definitions that formalizes the state of this branching procedure.

\begin{defin}[potential spiral]
Let $(G,\decomp,\bundleset)$ be a bundled instance. 
Let $u$ be a nonempty word over alphabet $\bundleset$ that contains each letter of $\bundleset$ at most once.
Then $u$ is a {\em{potential spiral}} if, for each two consecutive bundles $B_1B_2$ on $u$, 
the components where the arcs of $B_1$ end and the arcs of $B_2$ start are equal.
Moreover, $u$ is a {\em{potential closed spiral}} if the components where the arcs
of the last bundle of $u$ end and the arcs of the first bundle of $u$ start are equal.
\end{defin}
\begin{defin}[potential long spiral]
Let $(G,\decomp,\bundleset)$ be a bundled instance of isolation $(\Lambda,d)$ where $\Lambda \geq 2$, $d \geq \max(2k,f(k,k)+4)$, and $f(k,t)=2^{O(kt)}$ is the bound on the type-$t$ bend promised by Lemma \ref{lem:bendbound}.
A potential closed spiral $u$ is a {\em{potential long spiral}} if either:
\begin{enumerate}
\item $u$ contains at least one normal bundle and does not contain any ring bundle of level $\Lambda-1$ or lower; or
\item $u$ contains only ring bundles of $\ringcl(\ringcomp)$ for some ring component $\ringcomp$, and does not contain
non-isolation bundles of different levels.
\end{enumerate}
\end{defin}
\begin{defin}[partial bundle word]
Let $(G,\decomp,\bundleset)$ be a bundled instance of isolation $(\Lambda,d)$ where $\Lambda \geq 2$, $d \geq \max(2k,f(k,k)+4)$, and $f(k,t)=2^{O(kt)}$ is the bound on the type-$t$ bend promised by Lemma \ref{lem:bendbound}.
A {\em{partial bundle word}} is a (formal) string $\pi = u_1^{\rho_1} u_2^{\rho_2} \ldots u_s^{\rho_s}$, where
\begin{itemize}
\item $s \leq 2|\bundleset|$;
\item each $u_j$ is a potential spiral;
\item for each $1 \leq j < s$, there exists a symbol of $\bundleset$ that appears
in exactly one of the two words $u_j$ and $u_{j+1}$;
\item if a bundle $B \in \bundleset$, if $B$ appears in $u_{j_1}$ and in $u_{j_2}$, then $B$ appears in all words $u_j$ for $j_1 \leq j \leq j_2$;
\item each $\rho_j$ is either a positive integer or a symbol $\unknown$; moreover, if $\rho_j = \unknown$ then $u_j$ is a potential long spiral and if $\rho_j \neq 1$ then
$u_j$ is a potential closed spiral.
\end{itemize}
\end{defin}
\begin{defin}[consistent with partial bundle word]
A bundle word $u$ is {\em{consistent}} with a partial bundle word $u_1^{\rho_1} u_2^{\rho_2} \ldots u_s^{\rho_s}$
if there exists positive integers $r_1,r_2,\ldots,r_s$ such that $u = u_1^{r_1} u_2^{r_2} \ldots u_s^{r_s}$ and,
   for each $1 \leq j \leq s$, either $r_j = \rho_j$, or $\rho_j = \unknown$;
in the second case we require that $r_j > \smallexp$.
\end{defin}
\begin{defin}[semi-complete partial bundle word]
A partial bundle word $u_1^{\rho_1} u_2^{\rho_2} \ldots u_s^{\rho_s}$
is {\em{semi-complete}} if $\rho_j = \unknown$ implies that $u_j$ contains at least one level-0 ring bundle.
\end{defin}

Note that, if $\rho_j = \unknown$ in some semi-complete partial bundle word, then (as $(G,\decomp,\bundleset)$ has isolation $(\Lambda,d)$ with $\Lambda \geq 2$) the word $u_j$ contains only level-0, level-1 and level-2 ring bundles and does not contain any normal bundle.

We now note the following about a minimal solution to \probshort{}.
\begin{lemma}\label{lem:sol-to-pbw}
Let $(G,\decomp,\bundleset)$ be a bundled instance of isolation $(\Lambda,d)$ where $\Lambda \geq 2$, $d \geq \max(2k,f(k,k)+4)$, and $f(k,t)=2^{O(kt)}$ is the bound on the type-$t$ bend promised by Lemma \ref{lem:bendbound}.
Assume $(G,\decomp,\bundleset)$ is a YES-instance to \probshort{} and let $(P_i)_{i=1}^k$ be a minimal solution.
For each $1 \leq i \leq k$, let $u_{i,1}^{r_{i,1}} u_{i,2}^{r_{i,2}}, \ldots, u_{i,s(i)}^{r_{i,s(i)}}$
be a decomposition of $\bundleword(P_i)$ given by Corollary \ref{cor:bwdecomp}.
Let $\rho_{i,j} = r_{i,j}$ if $r_{i,j} \leq \smallexp$ or $u_{i,j}$ is not a potential long spiral,
and $\rho_{i,j} = \unknown$ otherwise.
Moreover, let $\rho'_{i,j} = \unknown$ if $\rho_{i,j} = \unknown$ and $u_{i,j}$
contains at least one level-0 ring bundle, and $\rho'_{i,j} = r_{i,j}$ otherwise.
Then, for each $1 \leq i \leq k$,:
\begin{enumerate}
\item $u_{i,1}^{\rho_{i,1}} u_{i,2}^{\rho_{i,2}} \ldots u_{i,s(i)}^{\rho_{i,s(i)}}$
is a partial bundle word consistent with $\bundleword(P_i)$;\label{cl:sol-to-pbw:rho}
\item $u_{i,1}^{\rho'_{i,1}} u_{i,2}^{\rho'_{i,2}} \ldots u_{i,s(i)}^{\rho'_{i,s(i)}}$
is a semi-complete partial bundle word consistent with $\bundleword(P_i)$;\label{cl:sol-to-pbw:rho2}
\item for each $1 \leq j \leq s(i)$, if $\rho_{i,j} \neq \unknown$
then $\rho'_{i,j} = \rho_{i,j} = r_{i,j}$;\label{cl:sol-to-pbw:rho-rho2}
\item for each $1 \leq j \leq s(i)$, if $\rho_{i,j} \neq \unknown$
then $\rho_{i,j} \leq 4|\bundleset|^2 k^2 + 1$.\label{cl:sol-to-pbw:rhobound}
\end{enumerate}
\end{lemma}
\begin{proof}
Each word $u_{i,j}$ is a potential spiral due to the properties promised
by Corollary \ref{cor:bwdecomp} and the fact that they are subwords
of bundle words of paths in $G$. Moreover, Corollary \ref{cor:bwdecomp}
as well as the definitions of $\rho_{i,j}$ and $\rho'_{i,j}$ directly
imply Claims \ref{cl:sol-to-pbw:rho}, \ref{cl:sol-to-pbw:rho2}.
Claim \ref{cl:sol-to-pbw:rho-rho2} follows directly from the definitions.
What remains is to show Claim \ref{cl:sol-to-pbw:rhobound}.

To this end, fix a choice of $1 \leq i \leq k$ and $1 \leq j \leq s(i)$
and assume that $\rho_{i,j} = r_{i,j}$. If $r_{i,j} \leq \smallexp$, the claim
is obvious, so let $r_{i,j} > \smallexp$ which, by the definition of $\rho_{i,j}$,
implies that $u_{i,j}$ is not a potential long spiral. 
By the definition of a potential long spiral, $u_{i,j}$ contains at least
one arc of bundle of level $\Lambda-1$ or lower. 
If $u_{i,j}$ contains at least one normal bundle, then the word
$u_{i,j}^{r_{i,j}}$ (which is a subword of $\bundleword(P_i)$)
  contains at least $r_{i,j}-1$ pairwise equal disjoint subwords, each corresponding
 to an isolation passage or a ring visitor. By Theorem \ref{thm:isolation-passages}
 and Corollary \ref{cor:no-visitors}, $r_{i,j} \leq 4|\bundleset|^2 k^2 + 1$, as desired.

In the other case, assume that all bundles of $u_{i,j}$ are ring bundles of 
some ring component $\ringcomp$, and the subpath of $P_i$ that corresponds
to the bundle word $u_{i,j}^{r_{i,j}}$ is completely contained in
$\ringcl(\ringcomp)$. If $u_{i,j}$ contains bundle arcs from non-isolation bundles of different levels,
then a situation forbidden by Theorem \ref{thm:oscillators}
appears on the path $P_i$. This completes the proof of Claim \ref{cl:sol-to-pbw:rhobound}.
\end{proof}

Lemma \ref{lem:sol-to-pbw} allows us to perform the following branching step.
\begin{lemma}\label{lem:bpwguess}
Let $(G,\decomp,\bundleset)$ be a bundled instance of isolation $(\Lambda,d)$ where $\Lambda \geq 2$ and $d \geq \max(2k,f(k,k)+4)$,
where $f(k,t)=2^{O(kt)}$ is the bound on the type-$t$ bend promised by Lemma \ref{lem:bendbound}.
Then in $O(2^{O(k |\bundleset|^2 \log |\bundleset|)} |G|^{O(1)})$ time one can enumerate a set of at most
$2^{O(k|\bundleset|^2 \log |\bundleset|)}$ sequences $(\pi_i)_{i=1}^k$ of partial bundle words, such that for any minimal solution $(P_i)_{i=1}^k$ to \probshort{}
there exists a generated sequence $(\pi_i)_{i=1}^k$ such that $\bundleword(P_i)$
is consistent with 
$\pi_i$ for all $1 \leq i \leq k$.
\end{lemma}
\begin{proof}
By Lemma \ref{lem:sol-to-pbw}, it is sufficient to prove
that for any $\beta > 1$
there are at most $2^{O(|\bundleset|^2 \log |\bundleset| + |\bundleset|\log \beta)}$
partial bundle words $u_1^{\rho_1} u_2^{\rho_2} \ldots u_s^{\rho_s}$
with $s \leq 2|\bundleset|$ and $\rho_j \leq \beta$ whenever $\rho_j \neq \unknown$
(recall that $2k \leq |\bundleset|-1$) and they can be generated
in $O(2^{O(|\bundleset|^2 \log |\bundleset| + |\bundleset|\log \beta))} |G|^{O(1)})$ time.

To this end, note that
there are less than $|\bundleset| \cdot |\bundleset|!$ words $u$ over alphabet $\bundleset$ that do not contain a symbol twice, and can be enumerated with polynomial delay.
As we require $s \leq 2|\bundleset|$,
there are only $2^{O(|\bundleset|^2 \log |\bundleset|)}$ choices for the values of $s$
and the strings $u_j$, $1 \leq j \leq s$.
Moreover, there are at most $\beta+2$ choices of each value $\rho_j$.
Note that one can verify in polynomial time if a given choice of $s$, $u_j$s and $\rho_j$s yields indeed a partial bundle word.
\end{proof}

Lemma \ref{lem:bpwguess} enables us to guess partial bundle words that are consistent with bundle words
of a minimal solution to the given bundled instance $(G,\decomp,\bundleset)$.
Thus, henceforth we assume that, apart from a bundled instance $(G,\decomp,\bundleset)$, we are given a set $(\pi_i)_{i=1}^k$ of partial bundle words
and we look for a minimal solution $(P_i)_{i=1}^k$ such that $\bundleword(P_i)$ is consistent with $\pi_i$ for each $1 \leq i \leq k$.
By Lemma \ref{lem:bpwguess}, there are at most $2^{O(k|\bundleset|^2 \log |\bundleset|)}$ subcases to consider.

Our goal now is to show a branching procedure that yields a bounded
in $|\decomp|$, $k$ and $|\bundleset|$ number of subcases of evaluating all values $\rho_{i,j}$
for which $\rho_{i,j} = \unknown$ but $u_{i,j}$ does not contain any
level-0 ring bundle.
In other words, we aim to produce a bounded number of semi-complete partial bundle words
and seek for minimal solutions consistent with one of them.

To achieve this goal, we analyze exponents $\rho_{i,j}=\unknown$ one by one, and prove that there is only a bounded number of choices for $r_{i,j}$, given the guesses made so far.

Formally, let $\pi_i = u_{i,1}^{\rho_{i,1}} u_{i,2}^{\rho_{i,2}} \ldots u_{i,s(i)}^{\rho_{i,s(i)}}$. Assume that $1 \leq \iota \leq k$ and $1 \leq \eta \leq s(\iota)$ are chosen
that $\rho_{\iota,\eta} = \unknown$, $u_{\iota,\eta}$ does not contain any level-0 ring bundle, but for any $i$ and $j$ such that $|u_{i,j}| < |u_{\iota,\eta}|$, if the exponent $\rho_{i,j} = \unknown$ then $u_{i,j}$ contains least one level-0 ring bundle.
In other words, we guess the exponents $\rho_{\iota,\eta}$, starting from the shorter strings $u_{\iota,\eta}$.

In one step, we identify a set of possible values of for the exponent $r_{\iota,\eta}$, whose size is bounded as a function of $|\decomp|$, $|\bundleset|$ and $k$,
and branch into a number of subcases, replacing the value of $\rho_{\iota,\eta}$ with one of the elements of the identified set.

Note that a bundle word of a path $P_i$ needs to start with the bundle that consists of the arc incident to the first terminal of the $i$-th terminal pair,
and ends with a bundle that consists of the arc incident to the second terminal of the $i$-th terminal pair. Moreover, these arcs appear only once in the bundle
word of $P_i$. Therefore, we may assume that for each $i$, $\rho_{i,1} = \rho_{i,s(i)} = 1$ and the corresponding words $u_{i,1}$ and $u_{i,s(i)}$ start and end respectively
with the appropriate bundle; if that is not the case, we may terminate the current branch.

Let $(P_i)_{i=1}^k$ be a (hypothetical) solution to \probshort{} on $G$, such that $\bundleword(P_i)$ is consistent with $\pi_i$ for all $1 \leq i \leq k$,
i.e., let $\bundleword(P_i) = u_{i,1}^{r_{i,1}} u_{i,2}^{r_{i,2}} \ldots u_{i,s(i)}^{r_{i,s(i)}}$.
As $\rho_{\iota,\eta} = \unknown$, $r_{\iota,\eta} > \smallexp$ and $u_{\iota,\eta}^{r_{\iota,\eta}}$ contains a spiral $u_{\iota,\eta} B$, where $B$ is the first symbol of $u_{\iota,\eta}$.
Note that $u_{\iota,\eta}$ needs to be a potential long spiral in the bundle graph
(as $\rho_{\iota,\eta} = \unknown$), and it does not contain any level-$0$ ring bundles. In particular, we may use Corollary \ref{cor:bwrep}
and speak of bundles inside and outside $u_{\iota,\eta}$.

Directly from Corollary \ref{cor:bwrep} we have the following.
\begin{lemma}\label{lem:bw-guess:spiral}
Consider a term $u_{\iota,\zeta}^{r_{\iota,\zeta}}$ in $\bundleword(P_\iota)$ for some $1 \leq \zeta \leq s(\iota)$ where $u_{\iota,\zeta}$ does not contain any level-$0$ ring bundles
and $r_{\iota,\zeta} > \smallexp$.
Then $u_{\iota,\zeta}$ is a potential long spiral and the terminals of the $\iota$-th pair lie on different sides
of the closed walk $u_{\iota,\zeta}$ in the bundle graph $(\decomp,\bundleset)$.
Moreover, for any $1 \leq i \leq k$ either:
\begin{enumerate}
\item both terminals of the $i$-th pair
and all bundles of $\bundleword(P_i)$ lie inside the closed walk $u_{\iota,\zeta}$;
\item both terminals of the $i$-th pair
and all bundles of $\bundleword(P_i)$ lie outside the closed walk $u_{\iota,\zeta}$;
\item the terminals of the $i$-th pair lie on different sides of the closed
walk $u_{\iota,\zeta}$ and
the starting terminal of the $i$-th pair lies inside $u_{\iota,\zeta}$
if and only if the starting terminal of the $\iota$-th pair does.
\end{enumerate}
\end{lemma}
Let $I^\to_\zeta$ be the set of those indices $i$ for which $P_i$ satisfies the last option in Lemma \ref{lem:bw-guess:spiral}
for the term $u_{\iota,\zeta}^{r_{\iota,\zeta}}$.

Let $u_{i,j}^{r_{i,j}}$ be a term of the bundle word decomposition of $P_i$
where $r_{i,j} \geq 2$ and $u_{i,j}$ does not contain any level-$0$ ring bundles. Let $B_{i,j,1}$ be the last bundle
on $\bundleword(P_i)$ that lies on the same side of $u_{i,j}$ as the starting terminal
of the $i$-th pair, and $B_{i,j,2}$ be the first bundle of $\bundleword(P_i)$
that lies on the same side of $u_{i,j}$ as the ending terminal of the $i$-th pair.
Let $u_{i,\eta(i,j,1)}$ be the word that contains the last occurrence of $B_{i,j,1}$
on $\bundleword(P_i)$ and $u_{i,j,\eta(i,2)}$ be the word that contains
the first occurrence of $B_{i,j,2}$ on $\bundleword(P_{i})$.
Let $P_{i,j}$ be the subpath of $P_i$ between the ending point of the arc corresponding
to the last occurrence of $B_{i,j,1}$ and the starting point of the arc corresponding
to the first occurrence of $B_{i,j,2}$ in $\bundleword(P_i)$.

\begin{figure}[h!]
\begin{center}
\includegraphics[width=0.8\textwidth]{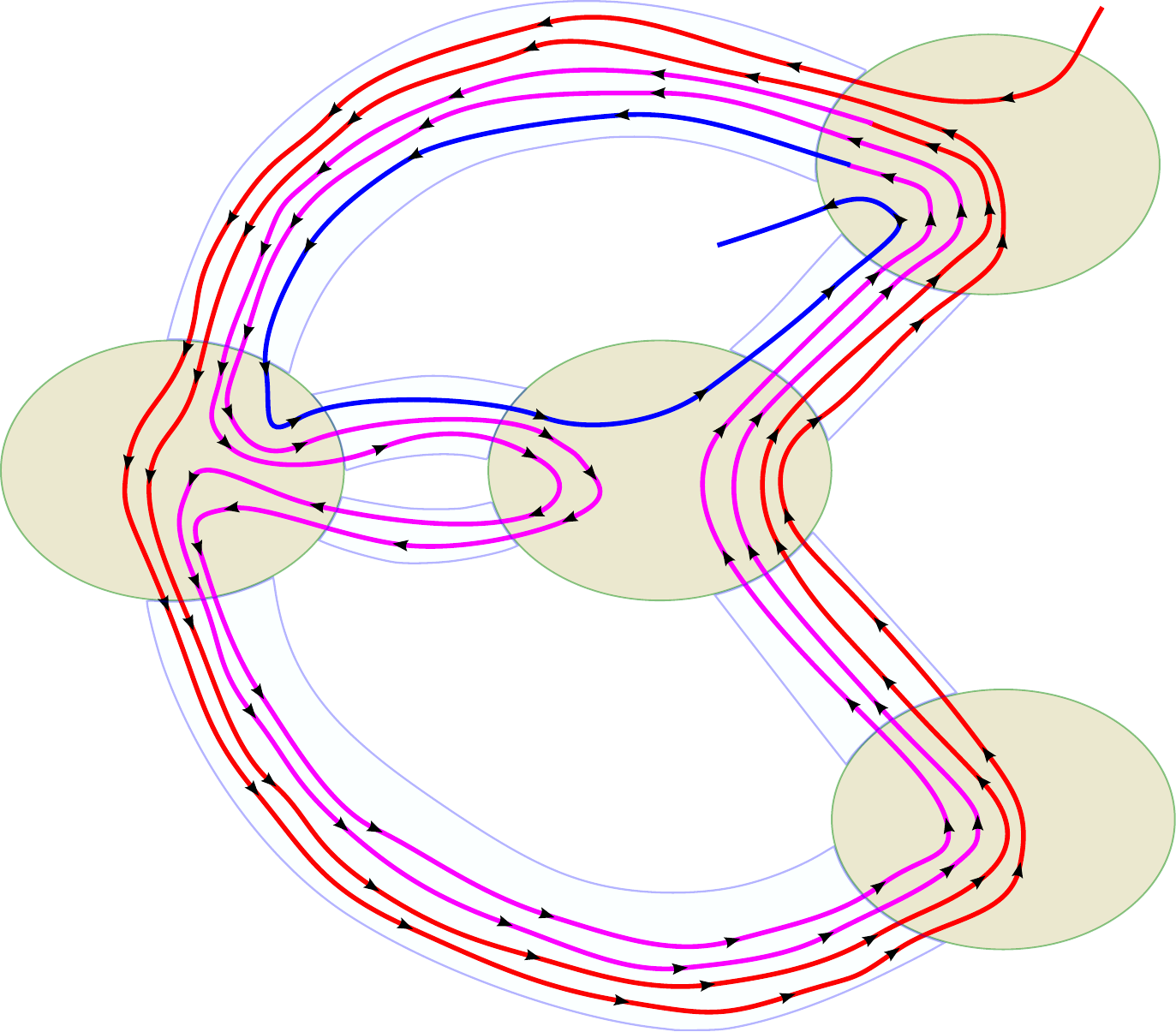}
\caption{A nontrivial situation around spiral $u_{i,j}^{r_{i,j}}$ on $P_i$. Path $P_i$ first spirals two times using four bundles (red spiral), then makes the $u_{i,j}^{r_{i,j}}$ spiral by spiraling two times using six bundles (pink spiral), and finally makes one turn of a spiral using three bundles (blue spiral). Note that the set of bundles used in the red spiral and in the blue spiral is a subset of the set of bundles used in the pink spiral, which we aim to measure. The crucial observation of this section is that these sets of bundles must be in fact {\bf{proper}} subsets of the set used by the pink spiral.}
\label{fig:nontrivial-measuring}
\end{center}
\end{figure}

\begin{lemma}\label{lem:bw-guess:etas}
Let $u_{i,j}^{r_{i,j}}$ be a term of the bundle word decomposition of $P_i$
where $r_{i,j} \geq 2$ and $u_{i,j}$ does not contain any level-$0$ ring bundles.
Then $B_{i,j,1}$ occurs on $\bundleword(P_i)$ before $B_{i,j,2}$
and $\eta(i,j,1) < j < \eta(i,j,2)$.
Moreover, for any $j' \neq j$, $\eta(i,j,1) < j' < \eta(i,j,2)$
we have $|u_{i,j'}| < |u_{i,j}|$.
\end{lemma}
\begin{proof}
The first claim follows directly from Corollary \ref{cor:bwrep}.
As for the second claim, by the properties of the bundle word decomposition,
if $\eta(i,j,1) + 1 < j$, there exists a symbol
$\bar{B}_{i,j,1}$ that appears in $u_{i,j}$ but not in $u_{i,j'}$
for any $\eta(i,j,1) < j' < j$. Thus, as for each such $j'$ the set of symbols of $u_{i,j'}$ is not only a subset of the set of symbols of $u_{i,j'}$, but also a proper subset, and consequently $|u_{i,j'}| < |u_{i,j}|$.
A symmetrical argument holds for the case $j' > j$
and a symbol $\bar{B}_{i,j,2}$ is missing from all words $u_{i,j'}$ for $j < j' < \eta(i,j,2)$.
\end{proof}

\begin{lemma}\label{lem:bw-guess:equalspiral}
Consider a term $u_{\iota,\zeta}^{r_{\iota,\zeta}}$ in $\bundleword(P_\iota)$ for some $1 \leq \zeta \leq s(\iota)$ where $u_{\iota,\zeta}$ does not contain any level-$0$ ring bundles
and $r_{\iota,\zeta} > \smallexp$. Let $B_{\iota,\zeta}$ be the first symbol of 
$u_{\iota,\zeta}$.
Then, for any $i \in I^\to_\zeta$
  \begin{enumerate}
  \item there exists a unique index $1 \leq j(i,\zeta) \leq s(i)$
  such that $r_{i,j(i,\zeta)} \geq 2$ and $u_{i,j(i,\zeta)}$ is a cyclic shift of
  $u_{\iota,\zeta}$;
  \item $|r_{i,j(i,\zeta)} - r_{\iota,\zeta}| \leq 3$.
  \end{enumerate}
\end{lemma}
\begin{proof}
The first claim, as well as inequality $r_{i,j(i,\zeta)} \geq r_{\iota,\zeta} - 3$
follows from Corollary \ref{cor:copy-word} applied for $P_0 = P_\iota$ and $\mathcal{P}= \{P_0, P_i\}$.
On the other hand, assume for the sake of contradiction that $r_{\iota,\zeta}<r_{i,j(i,\zeta)}-3$. If we apply Corollary \ref{cor:copy-word}
to $P_0 = P_i$ and $\mathcal{P} = \{P_i,P_\iota\}$ (note that we can do it as $r_{i,j(i,\zeta)}\geq 5$), we infer that $\bundleword(P_\iota)$ must contain a term $v^q$ where $v$ is a cyclic shift of $u_{i,j(i,\zeta)}$ and $q\geq r_{i,j(i,\zeta)}-3$. As $\bundleword(P_\iota)$ already contains $u_{\iota,\zeta}^{r_{\iota,\zeta}}$, this implies that $v=u_{\iota,\zeta}$, $q=r_{\iota,\zeta}$, and, consequently $r_{\iota,\zeta} \geq r_{i,j(i,\zeta)}-3$. This is a contradiction.
\end{proof}

Note that the values of $I^\to_\zeta$, $j(i,\zeta)$ as well as
$B_{i,j,1}$, $B_{i,j,2}$, $\eta(i,j,1)$ and $\eta(i,j,2)$ for
valid values of $i$, $j$ and $\zeta$, are known,
given the partial bundle words $(\pi_i)_{i=1}^k$: the values of
$r_{i,j}$ for $\rho_{i,j} = \unknown$ are not necessary to compute these values.
If there is some inconsistency in the current branch (say, there is no unique candidate
for $j(i,\zeta)$), we terminate the current branch.

Moreover, $\bundleword(P_{i,j})$ is consistent with the partial bundle word
$$\hat{u}_{i,\eta(i,j,1)} u_{i,\eta(i,j,1)+1}^{\rho_{i,\eta(i,j,1)+1}} \ldots u_{i,\eta(i,j,2)-1}^{\rho_{i,\eta(i,j,2)-1}} \hat{u}_{i,\eta(i,j,2)},$$
where $\hat{u}_{i,\eta(i,j,1)}$ is the maximal suffix of $u_{i,\eta(i,j,1)}$ that does not contain
$B_{i,j,1}$ and $\hat{u}_{i,\eta(i,j,2)}$ is the maximal prefix of $u_{i,\eta(i,j,2)}$ that
does not contain $B_{i,j,2}$ (note that any of these two words may be empty).

Recall that we aim to guess $r_{\iota,\eta}$.
We claim that, if there exists a minimal solution $(P_i)_{i=1}^k$
such that $P_i$ is consistent with $\pi_i$ for each $1 \leq i \leq k$, then
for any $i \in I_\eta^\to$ and any $\eta(i,j(i,\eta),1) < j < \eta(i,j(i,\eta),2)$
the exponent $\rho_{i,j}$ may be equal to $\unknown$ only for $j = j(i,\eta)$.
Indeed, take any such $j$. As $u_{\iota,\eta}$ does not contain any level-0
ring bundles, $u_{i,j}$ does not contain as well. 
By Lemma \ref{lem:bw-guess:etas}, $|u_{i,j}| < |u_{\iota,\eta}|$ unless
$j = j(i,\eta)$. The claim follows from our chosen order of guessing of the
exponents $\rho_{i,j}$. Therefore we may safely terminate branches 
where the claim is not satisfied.

Moreover, by Lemma \ref{lem:bw-guess:equalspiral}, there exist integers
$(\alpha_i)_{i \in I_\eta^\to}$, $-3 \leq \alpha_i \leq 3$, $\alpha_\iota=0$,
and a single integer $\smallexp < \aleph \leq n$ such that
$\rho_{i,j(i,\eta)} = \aleph + \alpha_i$ for any $i \in I_\eta^\to$.
We branch into at most $7^{k-1}$ options, guessing the values of $\alpha_i$
for $i \in I_\eta^\to$. If for any $i \in I_\eta^\to$, the value $\rho_{i,j(i,\eta)}$ does not
equal $\unknown$, the value of $\rho_{\iota,\eta}$ is determined.
Thus, henceforth we assume that this is not the case.

To choose a good value for $\rho_{\iota,\eta}$, we construct a zoom and a pack of zoom passes.
Let $\zbundleset$ be the set of bundles that appear in $u_{\iota,\eta}$,
and let $\zdecomp$ be the set of components that contain endpoints of arcs of bundles
of $\zbundleset$. Clearly, $(\zdecomp,\zbundleset)$ is a zoom. As $u_{\iota,\eta}$
does not contain any level-$0$ ring bundles, $\zdecomp$ does not contain
any ring components, and hence is level-$0$ safe.

Fix $\smallexp < \aleph \leq n$. We are to construct a pack of zoom passes, parameterized
by $\aleph$.
Intuitively, we want to reproduce what happens in all the spirals
$u_{i,\zeta}^{r_{i,\zeta}}$ for $\eta(i,\eta,1) < \zeta < \eta(i,\eta,2)$,
so that the equivalent of the path $P_{\iota,\eta}$ will behave in our zoom instance
in a very similar way to the original (unknown to us) path $P_{\iota,\eta}$.

To this end, for each partial bundle word $\pi_i$ we say that a
pair $(a,b)$, $1 \leq a \leq b \leq s(i)$ is {\em{relevant}} if 
\begin{enumerate}
\item each $u_{i,j}$, $a \leq j \leq b$ contains only symbols that appear in $u_{\iota,\eta}$;\item at least one exponent $\rho_{i,j}$, $a \leq j \leq b$, does not equal $1$.
\end{enumerate}
A pair $(a,b)$ is a {\em{maximal relevant}} pair if neither $(a-1,b)$ nor $(a,b+1)$
is a relevant pair. Recall that $u_{i,1}$ and $u_{i,s(i)}$ contains bundles incident
to terminals, and thus $a > 1$ and $b < s(i)$ for each relevant pair $(a,b)$ and,
consequently, for any relevant pair $(a,b)$ there exists a unique maximal relevant pair
$(a',b')$ with $a' \leq a \leq b \leq b'$.

By definition, for any $i \in I_\eta^\to$, the pair
$(\eta(i,j(i,\eta),1)+1, \eta(i,j(i,\eta),2)-1)$ is a maximal relevant pair
in $\pi_i$.
We now note that, in the parts of the partial bundle words $\pi_i$
that correspond to relevant pairs, almost every exponent $\rho_{i,j}$ is known.

\begin{figure}[h!]
\begin{center}
\includegraphics[width=0.8\textwidth]{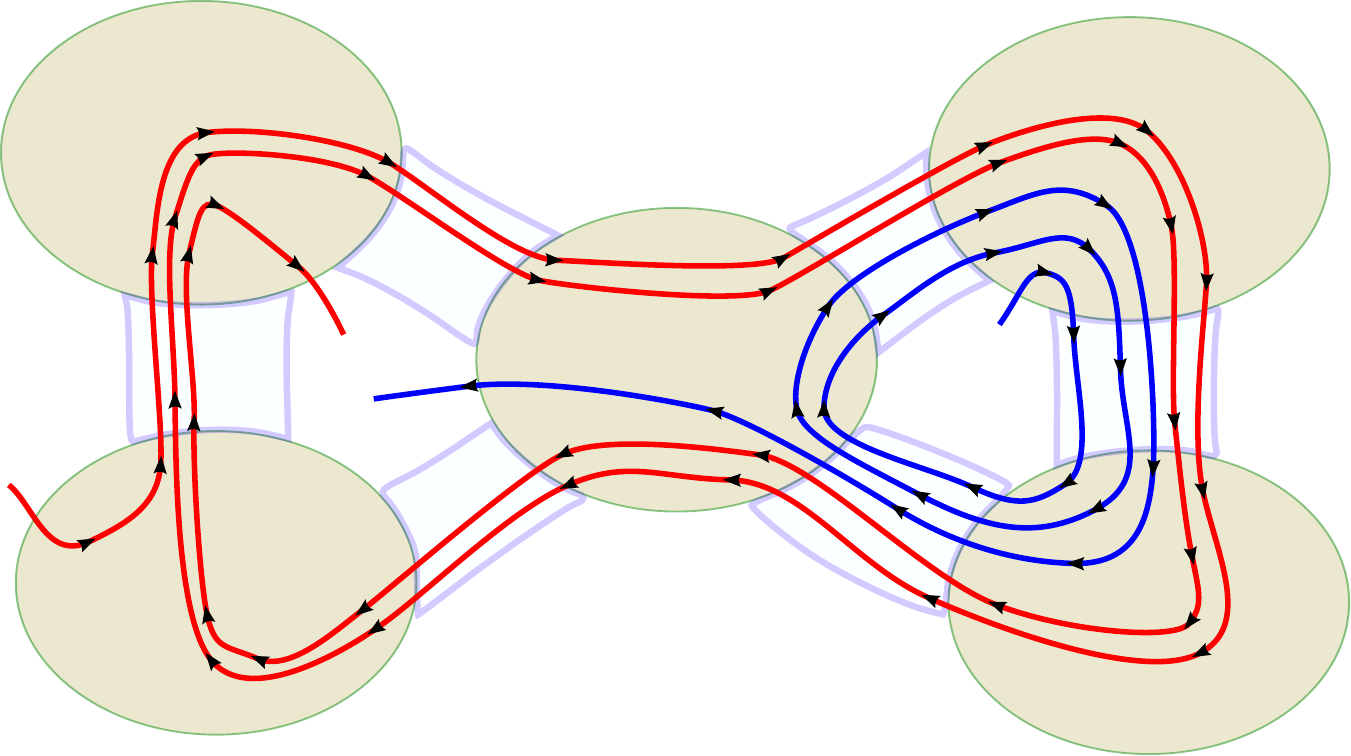}
\caption{An example showing that there may be relevant pairs that satisfy the second option in Lemma~\ref{lem:bw-guess:known-exp}. The path $P_1$ spirals multiple times using five bundles (the red spiral), while path $P_2$ spirals using a subset of these bundles (blue spiral) creating a relevant pair. Note that this happens even though both of the terminals of $P_2$ are enclosed by the red spiral.}
\label{fig:relevant}
\end{center}
\end{figure}

\begin{lemma}\label{lem:bw-guess:known-exp}
Let $(a,b)$ be a maximal relevant pair in $\pi_i$.
Then exactly one of the following holds:
\begin{enumerate}
\item $i \in I_\eta^\to$, $a = \eta(i,j(i,\eta),1)+1$, $b = \eta(i,j(i,\eta),2)-1$
(in particular, $\rho_{i,j} = \unknown$ for $a \leq j \leq b$
if and only if $j = j(i,\eta)$); or
\item for any $a \leq j \leq b$, $|u_{i,j}| < |u_{\iota,\eta}|$ or $\rho_{i,j}=1$
(in particular, $\rho_{i,j} \neq \unknown$ for $a \leq j \leq b$).
\end{enumerate}
\end{lemma}
\begin{proof}
Note that if the first option is satisfied, then there exists $j=j(i,\eta)$, for which $\rho_{i,j}=\unknown$ and $|u_{i,j}|=|u_{\iota,\eta}|$, so the second option is not satisfied. We are left with proving that if the second option is not satisfied, then the first is.

Let $(a,b)$ be a maximal relevant pair in $\pi_i$ that does not satisfy the second
option from the statement of the lemma. That is, there exists
$j$, $a \leq j \leq b$, such that $\rho_{i,j} \neq 1$ (and, consequently,
$r_{i,j} > 1$) but $|u_{i,j}| = |u_{\iota,\eta}|$.
Thus, $u_{i,j}$ is a permutation of $u_{\iota,\eta}$.
By Corollary \ref{cor:copy-word}, $u_{i,j}$ is a cyclic shift of $u_{\iota,\eta}$,
$i \in I_\eta^\to$ and, by the uniqueness of $j(i,\eta)$, $j=j(i,\eta)$.
\end{proof}

For a fixed value of $\smallexp < \aleph \leq n$,
for any $\pi_i$ and any maximal relevant pair $(a,b)$ in $\pi_i$,
we define the following bundle word:
$$q_{i,(a,b)}(\aleph) = B_{i,(a,b),1}' \hat{u}_{i,(a,b),1} u_{i,a}^{\rho'_{i,a}}
u_{i,a+1}^{\rho'_{i,a+1}} \ldots u_{i,b}^{\rho'_{i,b}} \hat{u}_{i,(a,b),2} B_{i,(a,b),2}',$$
where:
\begin{enumerate}
\item $\rho'_{i,j(i,\eta)} = \aleph + \alpha_i$ and 
  $\rho'_{i,j} = \rho_{i,j}$ for $j \neq j(i,\eta)$;
\item $\hat{u}_{i,(a,b),1}$ is the maximal suffix of $u_{i,a-1}$
that contains only symbols that appear in $u_{\iota,\eta}$ and
$B_{i,(a,b),1}'$ is the symbol of $u_{i,a-1}$ immediately preceding $\hat{u}_{i,(a,b),1}$;
\item symmetrically, 
$\hat{u}_{i,(a,b),2}$ is the maximal prefix of $u_{i,b+1}$
that contains only symbols that appear in $u_{\iota,\eta}$ and
$B_{i,(a,b),2}'$ is the symbol of $u_{i,b+1}$ immediately succeeding $\hat{u}_{i,(a,b),2}$.
\end{enumerate}
Note that both $\hat{u}_{i,(a,b),1}$ and $\hat{u}_{i,(a,b),2}$ may be empty.

We observe that, by Lemma \ref{lem:bw-guess:known-exp},
   $q_{i,(a,b)}(\aleph)$ is a bundle word: all exponents
are integers. As $\zdecomp$ does not contain any ring components,
$q_{i,(a,b)}(\aleph)$ is also a level-$0$ zoom pass in $(\zdecomp,\zbundleset)$.
Let $I$ be the set of pairs $(i,(a,b))$ where $1 \leq i \leq k$ and $(a,b)$
is a maximal relevant pair in $\pi_i$. As $s(i) \leq 2|\bundleset|$ for
each $i$, we infer that $|I| \leq 2|\bundleset|k$.

Using Lemma \ref{lem:psi-bound}, branch into at most $(|I|!)^2 \leq ((2|\bundleset|k)!)^2$
subcases,
guessing the permutations $(\psi_{B,\alpha})_{B \in \bundleset, 1 \leq \alpha \leq 2}$
for which $((q_\tau(\aleph))_{\tau \in I}, (\psi_{B,\alpha})_{B \in \bundleset, 1 \leq \alpha \leq 2})$
is a pack of level-$0$ zoom passes in $(\zdecomp,\zbundleset)$ for any $\smallexp < \aleph \leq n$ (note that the set of possible options for permutations $\psi_{B,\alpha}$ does
not depend on $\aleph$).
For $B \in \bundleset$ and $1 \leq \alpha \leq 2$, define 
$I_{B,\alpha} = \{\tau \in I: B_{\tau,\alpha} = B\}$; 
the permutation $\psi_{B,\alpha}$ permutes $I_{B,\alpha}$.

Construct the zoom auxiliary graph and instance for this pack of zoom passes
in $(\zdecomp,\zbundleset)$ and denote it $(H,\zdecomp_H,\zbundleset_H)$.

The discussion in Section \ref{sec:zooms} concluded with Observation \ref{obs:zoom-solution}
immediately yields the following.
\begin{lemma}\label{lem:bw-guess:equiv1}
If there exists a solution $(P_i)_{i=1}^k$ to the bundled instance
$(G,\decomp,\bundleset)$ such that
\begin{enumerate}
\item for each $1 \leq i \leq k$, $\bundleword(P_i)$ is consistent with $\pi_i$,
    $\bundleword(P_i) = u_{i,1}^{r_{i,1}} u_{i,2}^{r_{i,2}} \ldots u_{i,s(i)}^{r_{i,s(i)}}$;
\item for each $i \in I$, $r_{i,j(\eta)} = \alpha_i + r_{\iota,\eta}$;
\item if we denote for $\tau = (i,(a,b)) \in I$ by $Q_\tau$ the subpath of $P_i$
that corresponds to the subword $q_\tau(r_{\iota,\eta})$ of $\bundleword(P_i)$,
     then for each $B \in \bundleset$ the permutation $\psi_{B,1}$
     is equal to the order of first arcs of paths $Q_\tau$ on $B$ for $\tau \in I_{B,1}$
     and the permutation $\psi_{B,2}$ is equal to the order of the last
     arcs of paths $Q_\tau$ on $B$ for $\tau \in I_{B,2}$;
\end{enumerate}
then there exists a solution $(Q_\tau')_{\tau \in I}$
to the constructed zoom auxiliary instance
such that $\bundleword(Q_\tau')$ equals $q_\tau(r_{\iota,\eta})$
up to the prefix and suffix that corresponds to the part of the path
contained in the zoom starting and ending gadgets.
\end{lemma}

Let us solve the constructed zoom auxiliary instance
$(H,\zdecomp_H,\zbundleset_H)$ using Theorem \ref{thm:words-to-paths},
(note that $\zdecomp$ and $\zdecomp_H$ do not contain any ring components).
For a fixed choice of the permutations $\psi_{B,\alpha}$, $B \in \bundleset$,
    $1 \leq \alpha \leq 2$, we find a minimum $\aleph_0$
 such that Theorem \ref{thm:words-to-paths} returns
 a solution for bundled instance $(H,\decomp_H,\bundleset_H)$ and
 bundle words $(q_\tau(\aleph_0))_{\tau \in I}$.
If there is no such $\aleph_0$, by Lemma \ref{lem:bw-guess:equiv1},
we may terminate the current branch. Otherwise we note the following.

\begin{lemma}\label{lem:bw-guess:equiv2}
Let $(P_i)_{i=1}^k$ be a solution to the bundled instance $(G,\decomp,\bundleset)$
  as in Lemma \ref{lem:bw-guess:equiv1}, and suppose that $(P_i)_{i=1}^k$ is minimal.
Then $\aleph_0 \leq r_{\iota,\eta} \leq \aleph_0+32|\bundleset|$.
\end{lemma}
\begin{proof}
The inequality $\aleph_0 \leq r_{\iota,\eta}$ is straightforward by
the choice of $\aleph_0$ and Lemma \ref{lem:bw-guess:equiv1}.

Let $(Q_\tau')_{\tau \in I}$ be the family of paths returned
by Theorem \ref{thm:words-to-paths} for bundle words
$(q_\tau(\aleph_0))_{\tau \in I}$.
Let $t = (\iota,(\eta(\iota,\eta,1)+1,\eta(\iota,\eta,2)-1)) \in I$.
We claim that for each $\eta(\iota,\eta,1) < \zeta < \eta(\iota,\eta,2)$
such that $r_{\iota,\zeta} > 10$,
the bundle word $\bundleword(Q_t')$ contains $u_{\iota,\zeta}^{r_{\iota,\zeta}-10}$
as a subword and those subwords are pairwise disjoint for different choices
of $\zeta$.

Indeed, consider the subword $u_{\iota,\zeta}^{r_{\iota,\zeta}}$ of
$\bundleword(P_\iota)$ and the corresponding spiraling ring $A_\zeta$
associated with the subpath of $P_\iota$ corresponding
to $u_{\iota,\zeta}^{r_{\iota,\zeta}}$ with borders $\gamma_{\zeta,1}$
and $\gamma_{\zeta,2}$ and faces $f_{\zeta,1}$ and $f_{\zeta,2}$.
As $r_{\iota,\zeta} > 10$, for any $i \in I_\zeta^\to$ 
the index $j(i,\zeta)$ is defined,
$u_{i,j(i,\zeta)}$ is a cyclic shift of $u_{\iota,\zeta}$
and $r_{i,j(i,\zeta)} \geq r_{\iota,\zeta}-3 > 7$.
Therefore there exists an element $\tau_i = (i,(a_i,b_i)) \in I$
such that $a_i \leq j(i,\zeta) \leq b_i$, the terminals of the 
pair $\tau_i$ in the bundled instance $(H,\zdecomp_H,\zbundleset_H)$
lie on different sides of the spiral $u_{\iota,\zeta}$
and the path $Q_{\tau_i}'$ contains a subpath $R_{\tau_i}'$ that starts
in a vertex on $\gamma_{\zeta,1}$ and ends in a vertex on $\gamma_{\zeta,2}$.
As $(P_i)_{i=1}^k$ is a minimal solution, by Lemma \ref{lem:minsol-spiral}
the bundle word of each path $R_{\tau_i'}$
contains $u_{\iota,\zeta}^{r_{\iota,\zeta}-10}$ as a subword. 
Since the spiraling rings $A_\zeta$ are disjoint for different 
choices of $\zeta$, the paths $R_{\tau_i'}$ are edge-disjoint for different choices
of $\zeta$. As $t \in I_\zeta^\to$ for any choice of $\zeta$, the claim is proven.

We infer that
$$|\bundleword(Q_t')| \geq \sum_{\zeta=\eta(\iota,\eta,1)+1}^{\eta(\iota,\eta,2)-1} (r_{\iota,\zeta}-10)|u_{\iota,\zeta}|.$$
On the other hand, as $\bundleword(Q_t')$ contains a subset (as a multiset)
of the symbols of $q_t(\aleph_0)$, we have that
$$|\bundleword(Q_t')| \leq |q_t(\aleph_0)| = 2 + |\hat{u}_{t,1}| + |\hat{u}_{t,2}| + \sum_{\zeta=\eta(\iota,\eta,1)+1}^{\eta(\iota,\eta,2)-1} \rho'_{\iota,\zeta} |u_{\iota,\zeta}|.$$
Recall $\rho'_{\iota,\zeta} = r_{\iota,\zeta}$ for $\zeta \neq \eta$ 
and $\rho'_{\iota,\eta} = \aleph_0$.
Moreover, $|u_{\iota,\zeta}| \leq |u_{\iota,\eta}|$ and $|\hat{u}_{t,\alpha}| \leq |u_{\iota,\eta}|$ for $\alpha = 1,2$. We infer that
$$(r_{\iota,\eta}-10-\aleph_0)|u_{\iota,\eta}| \leq 2+10 \cdot (2+\eta(\iota,\eta,2)-\eta(\iota,\eta,1)-2)|u_{\iota,\eta}|.$$
As $\eta(\iota,\eta,2)-\eta(\iota,\eta,1) < s(\iota) \leq 2|\bundleset|$,
we have $r_{\iota,\eta} \leq \aleph_0 + 32|\bundleset|$, as desired.
\end{proof}

Lemma \ref{lem:bw-guess:equiv2} allows us to conclude with the following lemma that
summarizes the branching steps made in this section.

\begin{lemma}\label{lem:word-guessing}
Let $(G,\decomp,\bundleset)$ be a bundled instance of isolation $(\Lambda, d)$
where $\Lambda \geq 2$ and $d \geq \max(2k,f(k,k)+4)$,
where $f(k,t)=2^{O(kt)}$ is the bound on the type-$t$ bend promised by Lemma \ref{lem:bendbound}.
Then in $2^{O(k^2|\bundleset|^2 \log |\bundleset|)} |G|^{O(1)}$
time one can compute a family of at most
$2^{O(k^2|\bundleset|^2 \log |\bundleset|)}$ semi-complete sequences of partial bundle words
$(\pi_i)_{i=1}^k$
such that
for any minimal solution $(P_i)_{i=1}^k$ to \probshort{} on $(G,\decomp,\bundleset)$,
there exists a generated sequence $(\pi_i)_{i=1}^k$ in the set such that
$P_i$ is consistent with $\pi_i$ for each $1 \leq i \leq k$.
\end{lemma}

\begin{proof}
We first branch into $2^{O(k |\bundleset|^2 \log |\bundleset|)}$ subcases,
guessing the initial partial bundle word $\pi_i$ for each $1 \leq i \leq k$, using Lemma \ref{lem:bpwguess}.
Then, for each unknown exponent $\rho_{\iota,\eta}$, in the order
of increasing lengths of $|u_{\iota,\eta}|$, we guess the value of $\rho_{\iota,\eta}$.
Recall that this includes guessing the values $\alpha_i$ (at most $7^{k-1}$ options)
permutations $(\psi_{B,\alpha})_{B \in \bundleset, 1 \leq \alpha \leq 2}$
(at most $(|I|!)^2 \leq ((2|\bundleset|k)!)^2$ options) and a value $r_{\iota,\eta}$ between $\aleph_0$
and $\aleph_0+32|\bundleset|$. Therefore we have at most $2^{O(k |\bundleset|\log |\bundleset|)}$ subcases
for each exponent $\rho_{\iota,\eta}$ to guess.

Recall that in each $\pi_i$ we have $s(i) \leq 2|\bundleset|$. Therefore,
we perform the aforementioned guessing step at most $2|\bundleset| k $ times.
The promised bound follows.

Finally, note that if $\rho_{i,j} = \unknown$ implies that $u_{i,j}$
contains a level-$0$ ring bundle for any $1 \leq i \leq k$, $1 \leq j \leq s(i)$,
then $(\pi_i)_{i=1}^k$ are semi-complete by the definition.
\end{proof}

\subsection{Ring components: deducing winding numbers}\label{ss:guessing:windings}

In the previous section we have shown that there is a bounded number
of semi-complete partial bundle words to consider.
Here our goal is to change this semi-complete partial bundle words
into bundle words with ring holes. The main difficulty is to find
a set of good candidates for paths' winding numbers in the ring components.
To cope with this, we use Lemma \ref{lem:ringhomotopy}: if we know which parts of paths
traverse a ring component, and we find one way to route them through a ring component,
there exists a solution that winds in the ring component similarly as the way we have
found.

However, there are two main technical problems with this approach. First,
the paths may visit an isolation of a ring component, but do not traverse
the ring component itself (i.e., there are ring visitors). These visitors
block space for rerouting: we cannot use Lemma \ref{lem:ringhomotopy}
directly to a ring component or some fixed closure of it. Here the rescue comes
from results developed in Section \ref{ss:rings:bounds} that help us
control the behaviour of a minimal solution in the closure of a ring component.

A second problem is that, if we ask Theorem \ref{thm:words-to-paths}
to provide us with some canonical way to route ring passages through (a closure of)
a ring component, the returned solution follows our guidelines (i.e., bundle words
with ring holes) in a quite relaxed way. To cope with that, we employ
a similar line of reasoning as in the previous subsection:
if in a minimal solution a ring passage spirals along a bundle word $u^r$, for some
large $r$, then Lemma \ref{lem:minsol-spiral} forces any canonical way found
by Theorem \ref{thm:words-to-paths} to spiral at least $r-10$ times
(i.e., to contain $u^{r-10}$ in its bundle word). Hence, the solution
returned by Theorem \ref{thm:words-to-paths} can differ from the minimal solution only by a limited number bundles, which implies that their winding numbers also do not differ much.

In this section we assume that the isolation of our decomposition
is $(\Lambda,d)$ for $\Lambda \geq 3$ and $d \geq \max(2k,f(k,k)+4)$.
The assumption $d \geq \max(2k,f(k,k)+4)$ allows us to use the results
of Section \ref{ss:rings:bounds}. The assumed $3$ layers of isolation gives
us space to carefully extract the ring on which Lemma \ref{lem:ringhomotopy}
is applied. It is worth noticing that all essential argumentation happens
in layers 1 and 2; the last layer are added only for the sake of clarity
of the presentation (for example, we do not need to care about normal bundles
with both endpoints in the same level-$\Lambda$ isolation component etc.).

Let us now proceed with a formal argumentation.
We first note that a semi-complete partial bundle word contains more
information than the bundle word part of a bundle word with level-$1$ ring holes.

\begin{lemma}\label{lem:pbw-to-holes}
Let $(G,\decomp,\bundleset)$ be a bundled instance with isolation $(\Lambda,d)$
where $\Lambda \geq 2$ and $d \geq \max(2k,f(k,k)+4)$ and
$f(k,t)=2^{O(kt)}$ is the bound on the type-$t$ bend promised by Lemma \ref{lem:bendbound}.
Let $\pi$ be a semi-complete partial bundle word in $(G,\decomp,\bundleset)$.
Then there exists a unique sequence of bundle words $(p_j)_{j=0}^h$ 
such that the following holds:
$(p_j)_{j=0}^h$ does not contain any level-$0$ bundles and
for any path $P$ that connects a terminal pair, is consistent with $\pi$ and its unique
bundle word with level-$1$ ring holes does not contain any level-$0$ ring bundle,
there exists a choice of integers $(w_j)_{j=1}^h$ such that $((p_j)_{j=0}^h, (w_j)_{j=1}^h)$
is a bundle word with level-$1$ ring holes consistent with $P$.
Moreover, the sequence $(p_j)_{j=0}^h$ can be computed in polynomial time, given $\pi$.
\end{lemma}
\begin{proof}
Let $\pi = u_1^{\rho_1} u_2^{\rho_2} \ldots u_s^{\rho_s}$.
Assume $P$ is consistent with $\pi$
and let $\bundleword(P) = u_1^{r_1} u_2^{r_2} \ldots u_s^{r_s}$.
Note that, as $P$ connects a terminal pair, $r_1=r_s=1$, $u_1$ starts with a bundle
that contains the arc incident to the starting terminal of $P$ and $u_s$ ends
with a bundle that contains the arc incident to the ending terminal of $P$.

Consider now an index $j$ for which $\rho_j = \unknown$. As $\pi$ is semi-complete,
$u_j$ contains at least one level-$0$ ring bundle, is a potential long spiral and,
consequently, does not contain non-isolation bundles of level different than $0$.
From the assumption that the bundle word with level-$1$ ring holes of $P$
does not contain any level-$0$ ring bundle, we infer that the subpath of $P$
that corresponds to the subword $u_j^{r_j}$ is a part of a level-$1$ ring passage
of $P$. As the choice of $j$ is arbitrary, we infer that the bundle word
part of the bundle word with level-$1$ ring holes of $P$ does not depend on the choice
of $P$, but only on $\pi$.

Moreover, the aforementioned argument yields a polynomial-time
algorithm to compute the bundle words $(p_j)_{j=0}^h$ from $\pi$.
We compute $p$ defined as a bundle word created from $\pi$ by evaluating each $\rho_j = \unknown$ to a fixed positive integer. Then we compute the decomposition
$p = p_0 r_1 p_1 r_2 p_2 \ldots r_h p_h$, where $(r_j)_{j=1}^h$ are all
level-$1$ ring passages in $p$, and output the sequence $(p_j)_{j=0}^h$.
\end{proof}

\begin{lemma}\label{lem:holes-vs-sol}
Let $(G,\decomp,\bundleset)$ be a bundled instance with isolation $(\Lambda,d)$
where $\Lambda \geq 2$ and $d \geq \max(2k,f(k,k)+4)$,
$f(k,t)=2^{O(kt)}$ is the bound on the type-$t$ bend promised by Lemma \ref{lem:bendbound}.
Let $(P_i)_{i=1}^k$ be a minimal solution to \probshort{} on $G$ such that
$P_i$ is consistent with $\pi_i$ for each $1 \leq i \leq k$. 
Let $(p_{i,j})_{j=0}^{h(i)}$ be a sequence computed by Lemma \ref{lem:pbw-to-holes}
for $\pi_i$. Then $\sum_{i=1}^k h(i) \leq 4|\bundleset|^2 k^2$
and there exist integers $(w_{i,j})_{1 \leq i \leq k, 1 \leq j \leq h(i)}$
such that for each $1 \leq i \leq k$ the pair $((p_{i,j})_{j=0}^{h(i)}, (w_{i,j})_{j=1}^{h(i)})$
is a bundle word with level-$1$ ring holes consistent with $P_i$.
\end{lemma}
\begin{proof}
The bound on $\sum_{i=1}^k h(i)$ follows from Theorem \ref{thm:isolation-passages}.
As for the second claim, note that by Corollary \ref{cor:bwholes-minsol}, for any $1 \leq i \leq k$,
the bundle word with level-$1$ ring holes consistent with $P_i$ does not contain
any bundle of level $0$.
\end{proof}

From this point, we assume that the isolation of 
the bundled instance $(G,\decomp,\bundleset)$ satisfies
$\Lambda \geq 3$ and $d \geq (2k,f(k,k)+4)$, as in
the assumptions of Theorem \ref{thm:word-guessing}.

Fix $1 \leq i \leq k$ and $1 \leq j \leq h(i)$. Let
$B_{i,j,1}$ be the last bundle of $p_{i,j-1}$
and $B_{i,j,2}$ be the first bundle of $p_{i,j}$. 
Directly from Lemma \ref{lem:holes-vs-sol} we obtain the following observation.
\begin{observation}\label{obs:bij-props}
If there exists a minimal solution $(P_i)_{i=1}^k$
such that $P_i$ is consistent with $\pi_i$ for each $1 \leq i \leq k$,
then for each $1 \leq i \leq k$ and $1 \leq j \leq h(i)$
there exists a ring component $\ringcomp_{i,j}$ such that 
\begin{enumerate}
\item $B_{i,j,1}$ contains arcs leading from a level-$2$ isolation component of $\ringcomp_{i,j}$ to a level-$1$ isolation component;
\item $B_{i,j,2}$ contains arcs leading from a level-$1$ isolation component of $\ringcomp_{i,j}$ to a level-$2$ isolation component;
\item $B_{i,j,1}$ and $B_{i,j,2}$ lie on different sides of $\ringcomp_{i,j}$.
\end{enumerate}
\end{observation}
Thus, if this is not the case, we may terminate the current branch.

For each ring component $\ringcomp \in \decomp$, we define
$I(\ringcomp) = \{(i,j): \ringcomp = \ringcomp_{i,j}\}$. 
Note the following, due to Lemma \ref{lem:holes-vs-sol}.
\begin{observation}\label{obs:I-bound}
If there exists a minimal solution
$(P_i)_{i=1}^k$ to \probshort{} on $G$ such that $P_i$ is consistent with $\pi_i$
for each $1 \leq i \leq k$, we have
$$\sum_{\ringcomp \in \decomp} |I(\ringcomp)| = \sum_{i=1}^k h(i) \leq 4|\bundleset|^2 k^2.$$
\end{observation}
Again, if this is not the case, we terminate the current branch.

Observe that, by Theorem \ref{thm:oscillators}, we obtain the following.
\begin{observation}\label{obs:extend-passages}
If there exists a minimal solution
$(P_i)_{i=1}^k$ to \probshort{} on $G$ such that $P_i$ is consistent with $\pi_i$,
  then for each $1 \leq i \leq k$ and $0 \leq j \leq h(i)$, the bundle word $p_{i,j}$
  contains at least one normal bundle.
\end{observation}
\begin{proof}
The claim is obvious for $j=0$ or $j=h(i)$, as then $p_{i,j}$ contains
a bundle with an arc incident to a terminal. Assume that the claim is not true for some
$1 \leq i \leq k$ and $0 < j < h(i)$. Then the subpath of $P_i$ between
the arcs corresponding to the last symbol of $p_{i,j-1}$ and the first symbol
$p_{i,j+1}$ contains the structure forbidden by Theorem \ref{thm:oscillators}.
\end{proof}
Again, if this is not the case, we terminate the current branch.

Recall that $\Lambda \geq 3$. For each $1 \leq i \leq k$ and $1 \leq j \leq h(i)$,
we define $B^\circ_{i,j,1}$ to be the last bundle of $p_{i,j-1}$ that contains
arcs leading from the level-$3$ isolation component of $\ringcomp_{i,j}$
to level-$2$ one, and $B^\circ_{i,j,2}$ to be the first bundle of $p_{i,j}$
that contains arcs leading from the level-$2$ isolation component of $\ringcomp_{i,j}$
to level-$3$ one.
Let $p^\circ_{i,j,1}$ be the suffix of $p_{i,j-1}$ starting with $B^\circ_{i,j,1}$
and $p^\circ_{i,j,2}$ be the prefix of $p_{i,j}$ ending with $B^\circ_{i,j,2}$.
Moreover, for $1 \leq i \leq k$, $0 \leq j \leq h(i)$ let
$q_{i,j}$ be the subword of $p_{i,j}$ between $B^\circ_{i,j,2}$ and $B^\circ_{i,j+1,1}$,
  where $B_{i,0,2}^\circ$ is the first symbol of $p_{i,0}$ and $B^\circ_{i,h(i)+1,1}$ is
  the last symbol of $p_{i,h(i)}$.

\begin{figure}[h!]
\begin{center}
\includegraphics[width=0.8\textwidth]{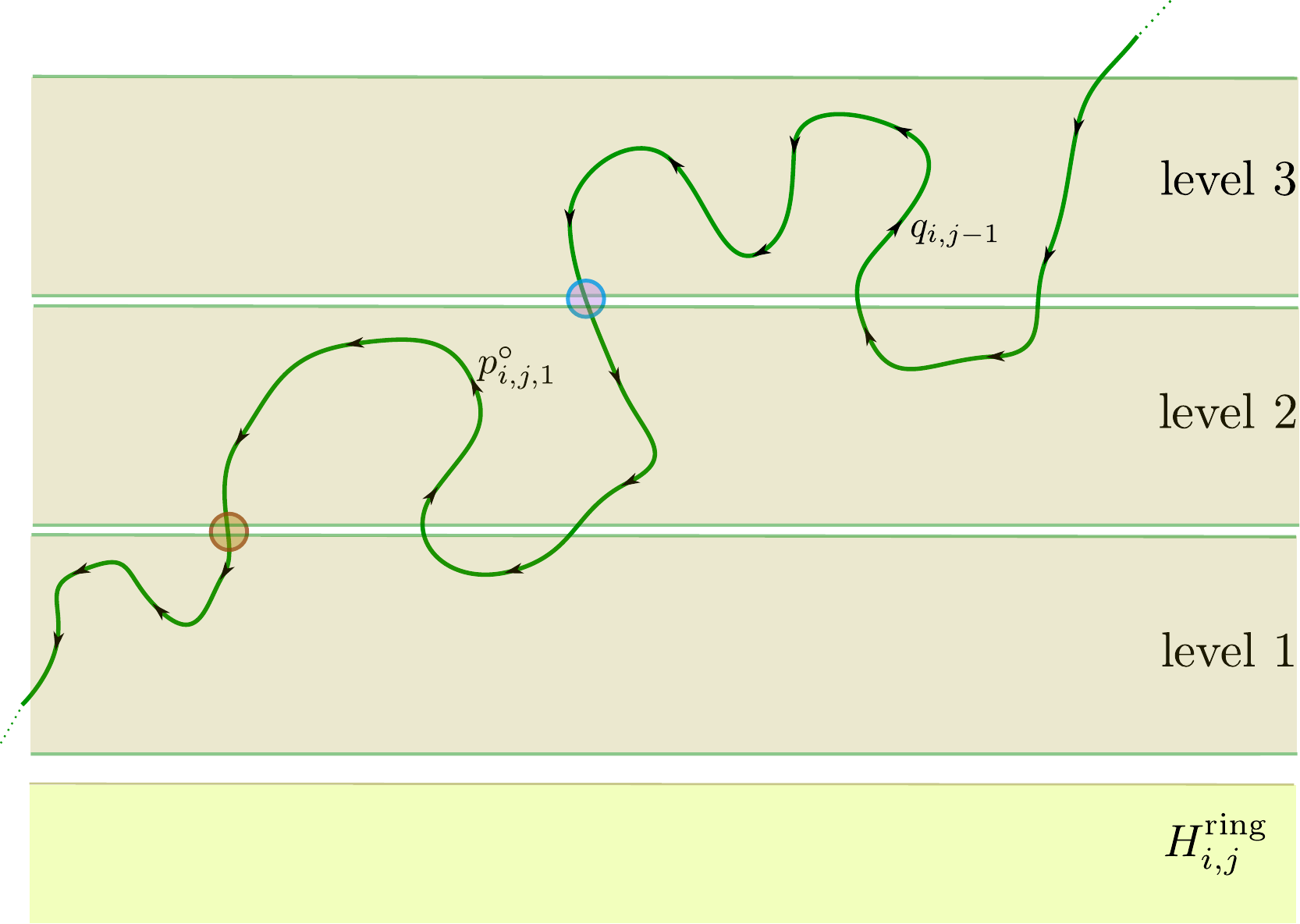}
\caption{An exemplary beginning of the $j$-th ring passage on path $P_i$. The blue circle depicts bundle $B_{i,j,1}^\circ$, where $q_{i,j-1}$ ends and $p^\circ_{i,j,1}$ starts. The orange circle depicts the last bundle of $p_{i,j}$ (thus also of $p_{i,j,1}^\circ$), where the level-$1$ hole starts.}
\label{fig:ring-passage}
\end{center}
\end{figure}

By Observation \ref{obs:extend-passages} and again
Theorem \ref{thm:oscillators} we have the following.
\begin{observation}\label{obs:levelcirc}
If there exists a minimal solution
$(P_i)_{i=1}^k$ to \probshort{} on $G$ such that $P_i$ is consistent with $\pi_i$,
then the bundles $B^\circ_{i,j,1}$ and $B^\circ_{i,j,2}$ are well defined and lie
on the opposite sides of $\ringcomp_{i,j}$.
Moreover, for $1 \leq \alpha \leq 2$, the bundle word $p_{i,j,\alpha}^\circ$, except for the symbol $B_{i,j,\alpha}$,
contains only bundles with arcs with both endpoints in level-$\lambda$, $1 \leq \lambda \leq 2$ isolation components of $\ringcomp_{i,j}$ that lie on the same side of $\ringcomp_{i,j}$
as $B^\circ_{i,j,\alpha}$. 
\end{observation}
Again, if this is not the case, we terminate the current branch.

Now note that for ring passages in the same ring component $\ringcomp$, by Observation \ref{obs:pm1}, the winding numbers cannot differ too much.
\begin{observation}\label{obs:wij-diff}
Let $(P_i)_{i=1}^k$ be a minimal solution to \probshort{} on $G$ such that
$P_i$ is consistent with $\pi_i$ for each $1 \leq i \leq k$.
For $1 \leq i \leq k$, let $((p_{i,j})_{j=0}^{h(i)}, (w_{i,j})_{j=1}^{h(i)})$ be the
bundle word with level-$1$ ring holes of $P_i$.
Then, if for some $(i,j)$ and $(i',j')$ we have $\ringcomp_{i,j} = \ringcomp_{i',j'}$,
then $|w_{i,j} - w_{i',j'}| \leq 1$.
\end{observation}

Observation \ref{obs:wij-diff} motivates us to the following branch.
For each $1 \leq i \leq k$ and for each $1 \leq j \leq h(i)$
we branch into three subcases, picking an integer $-1 \leq \alpha_{i,j} \leq 1$.
By Observation \ref{obs:I-bound}, this step leads to at most $3^{4|\bundleset|^2 k^2}$
subcases.
We say that a solution $(P_i)_{i=1}^k$ is {\em{consistent with the current branch}}
for each ring component $\ringcomp$ there exists an integer $w(\ringcomp)$ such that
for each $1 \leq i \leq k$, the bundle word with level-$1$ ring holes of $P_i$
equals 
$$((p_{i,j})_{j=0}^{h(i)}, (w(\ringcomp_{i,j}) + \alpha_{i,j})_{j=1}^{h(i)}).$$
By Observation \ref{obs:wij-diff} we obtain the following.
\begin{observation}\label{obs:good-alpha}
If there exists a minimal solution $(P_i)_{i=1}^k$ to \probshort{} on $G$ such that
$P_i$ is consistent with $\pi_i$ for each $1 \leq i \leq k$,
then there exists a subcase with a choice of integers $\alpha_{i,j}$
such that $(P_i)_{i=1}^k$ is consistent with this branch.
\end{observation}

Our goal now is, for a fixed branch with integers $\alpha_{i,j}$ and
for a fixed ring component $\ringcomp \in \decomp$ such that $I(\ringcomp) \neq \emptyset$,
to compute a set of bounded size of possible candidates for $w(\ringcomp)$.

For each ring component $\ringcomp \in \decomp$ such that $I(\ringcomp) \neq \emptyset$
and for each integer $-n \leq w \leq n$
we construct an zoom auxiliary instance $(H(\ringcomp,w), \zdecomp_H(\ringcomp,w), \bundleset_H(\ringcomp,w))$ as follows. First we take a zoom $(\zdecomp(\ringcomp),\zbundleset(\ringcomp))$ that includes all components in the level-$2$ closure of $\ringcomp$ and bundles with arcs with both endpoints in these components. Then, for each $(i,j) \in I(\ringcomp)$ we construct bundle word with level-$2$ ring holes $\bwholes^\circ_{i,j} = ((p^\circ_{i,j,1}, p^\circ_{i,j,2}), w+\alpha_{i,j})$;
note that, by the definition of $p^\circ_{i,j,1}$ and $p^\circ_{i,j,2}$, this pair is indeed
a bundle word with level-$2$ ring holes and a level-$2$ zoom pass in $(\zdecomp(\ringcomp), \zbundleset(\ringcomp))$ as well. 
Finally, using Lemma \ref{lem:psi-bound}, we branch into at most $(|I(\ringcomp)|!)^2$
subcases, guessing, for each $B \in \bundleset$ and $1 \leq \alpha \leq 2$, a permutation $\psi_{B,\alpha}(\ringcomp)$ of those indices $(i,j) \in I(\ringcomp)$ for which $B_{i,j,\alpha} = B$.
Thus, $((\bwholes^\circ_{i,j})_{(i,j) \in I(\ringcomp)}, (\psi_{B,\alpha}(\ringcomp))_{B \in \bundleset, 1 \leq \alpha \leq 2})$
is a pack of zoom passes in $(\zdecomp(\ringcomp), \zbundleset(\ringcomp))$.
The zoom auxiliary instance $(H(\ringcomp,w), \zdecomp_H(\ringcomp,w), \bundleset_H(\ringcomp,w))$ is defined as the zoom auxiliary instance for this pack of zoom passes.

The discussion in Section \ref{sec:zooms} concluded with Observation \ref{obs:zoom-solution}
immediately yields the following.
\begin{lemma}\label{lem:ring-guessing-equiv1}
If there exists a minimal solution $(P_i)_{i=1}^k$ to \probshort{} on $G$
such that
\begin{enumerate}
\item for each $1 \leq i \leq k$, $P_i$ is consistent with $\pi_i$;
\item for each ring component $\ringcomp$ where $I(\ringcomp) \neq \emptyset$ there exists an integer $w(\ringcomp)$
  such that for each $1 \leq i \leq k$ and $1 \leq j \leq h(i)$ the level-$1$ ring passage
  that corresponds to the part of the bundle word of $P_i$ between $p_{i,j-1}$ and $p_{i,j}$
  has winding number $w(\ringcomp) + \alpha_{i,j}$;
\item for each ring component $\ringcomp$, $B \in \bundleset$ and $1 \leq \alpha \leq 2$,
  the order of the arcs $b_{i,j,\alpha}$ for $(i,j) \in I(\ringcomp)$ and
  $b_{i,j,\alpha} \in B$ is equal to $\psi_{B,\alpha}(\ringcomp)$,
\end{enumerate}
then, for each $\ringcomp$ where $I(\ringcomp) \neq \emptyset$ the zoom auxiliary
instance constructed for the pair $(\ringcomp, w(\ringcomp))$ has a solution
$(P_\tau^\ast)_{\tau \in I(\ringcomp)}$, where the bundle word with level-$1$
ring holes of $P_\tau^\ast$ equals $\bwholes^\circ_\tau$ up to a prefix and suffix
that corresponds to the subpath in the zoom starting and ending gadgets.
\end{lemma}

For each ring component $\ringcomp$ where $I(\ringcomp) \neq \emptyset$ and for each
$-n \leq w \leq n$ we apply Theorem \ref{thm:words-to-paths} to 
the zoom auxiliary instance $(H(\ringcomp,w), \zdecomp_H(\ringcomp,w), \bundleset_H(\ringcomp,w))$. For each $\ringcomp$, let $w(\ringcomp)$ be an integer for which Theorem \ref{thm:words-to-paths} returned a solution; if such an integer does not exist, by Lemma \ref{lem:ring-guessing-equiv1} we may safely terminate the current branch.
Let $(P_\tau^\ast)_{\tau \in I(\ringcomp)}$ be the solution returned
by Theorem \ref{thm:words-to-paths} and let $w^\ast_\tau$ be the winding number
of $P_\tau^\ast$ in the level-$2$ closure of $\ringcomp$ (which is a subgraph
of both $G$ and $H(\ringcomp,w(\ringcomp))$).

We now prove the following crucial claim.

\begin{lemma}\label{lem:ring-guessing-equiv2}
If there exists a minimal solution $(P_i)_{i=1}^k$ to \probshort{} on $G$
such that
\begin{enumerate}
\item for each $1 \leq i \leq k$, $P_i$ is consistent with $\pi_i$, and
\item for each ring component $\ringcomp$, $B \in \bundleset$ and $1 \leq \alpha \leq 2$,
  the order of the arcs $b_{i,j,\alpha}$ for $(i,j) \in I(\ringcomp)$ and
  $b_{i,j,\alpha} \in B$ is equal to $\psi_{B,\alpha}(\ringcomp)$,
\end{enumerate}
then, there exists integers $(x_{i,j})_{1 \leq i \leq k, 1 \leq j \leq h(i)}$ such that
$|x_{i,j} - w^\ast_{i,j}| \leq 40|\bundleset|^2 +2|\bundleset|+8$
for each $1 \leq i \leq k$, $1 \leq j \leq h(i)$ and a solution $(P_i')_{i=1}^k$ to \probshort{} on $G$ (not necessarily minimal)
such that $P_i'$ is consistent with bundle word with level-$2$ ring holes
$((q_{i,j})_{j=0}^{h(i)}, (x_{i,j})_{j=1}^{h(i)})$ for each $1 \leq i \leq k$.
\end{lemma}
\begin{proof}
Fix $\ringcomp \in \decomp$ for which $I(\ringcomp) \neq \emptyset$.
Our goal is to modify $(P_i)_{i=1}^k$ in the level-$2$ closure of $\ringcomp$
so that the winding numbers of passages of $\ringcomp$, indexed with $\tau \in I(\ringcomp)$,
do not differ from $w^\ast_\tau$ much.

First, we slightly modify the graph $\ringcl^2(\ringcomp)$, so that further topological
arguments become cleaner. For each bundle $B = (b_1,b_2,\ldots,b_s)$ contained in $\ringcl^2(\ringcomp)$
we first subdivide each arc $b_j$ twice, introducing vertices $v_{j,1}$ and $v_{j,2}$, and then,
for each $1 \leq j < s$ and $1 \leq \alpha \leq 2$, we add a vertex $z_{j,\alpha}$ and arcs
$(v_{j,\alpha},z_{j,\alpha})$ and $(v_{j+1,\alpha},z_{j,\alpha})$ inside the face of $G$ between
$b_j$ and $b_{j+1}$. We do it in such a manner that, if the reference curve $\refcurve^2(\ringcomp)$
crosses $B$, it crosses arcs $(v_{j,1},v_{j,2})$ for each $1 \leq j \leq s$, i.e., is contained
between the undirected paths $v_{1,\alpha}z_{1,\alpha}v_{2,\alpha}z_{2,\alpha}\ldots v_{s,\alpha}$, $1 \leq \alpha \leq 2$.
Note that this operation does not change the answer to \probshort{} on any supergraph of $\ringcl^2(\ringcomp)$,
as the added arcs are useless from the point of view of the directed paths (vertices $z_{j,\alpha}$ are sinks).
In the new graph the bundle $B$ consists of arcs $(v_{j,1},v_{j,2})$ for $1 \leq j \leq s$, the
vertices $v_{j,1}$, $z_{j,1}$ belong to the component where the arcs of $B$ originally start,
and the vertices $v_{j,2}$, $z_{j,2}$ belong to the component where the arcs of $B$ originally end.

Now we define the graph $G^\sharp$. We start with the subgraph of $G$ induced by the vertices
of the level-$2$ closure of $\ringcomp$. Then we repeatedly take maximal subpaths
of paths $(P_i)_{i=1}^k$ that go through vertices of $\ringcl^2(\ringcomp)$ and, if such a path
starts and ends on the same side of $\ringcl^2(\ringcomp)$, we remove from $G^\sharp$
all arcs and vertices that lie on the subpath or on 
the different side of the chosen subpath than $\ringcomp$.
As $(P_i)_{i=1}^k$ is a minimal solution, by Theorem \ref{thm:oscillators},
any such path contains only vertices of level $2$, and $G^\sharp$ contains
the subgraph of $G$ induced by the vertices of $\ringcl^1(\ringcomp)$.

We identify two faces $f_1$ and $f_2$ of $G^\sharp$ that contain the outer and inner face of
$\ringcl(\ringcomp)$, respectively. 
We want to choose a subcurve of $\refcurve^2(\ringcomp)$ to be a reference curve in $G^\sharp$.
Let $f_1'$ be the last face crossed by $\refcurve^2(\ringcomp)$, contained in $f_1$,
and let $f_2'$ be the first face crossed by $\refcurve^2(\ringcomp)$.
By the construction of $\refcurve^2(\ringcomp)$ (Lemma \ref{lem:refcurve-choice}),
$f_1'$ appears on $\refcurve^2(\ringcomp)$ earlier than $f_2'$. We choose $\refcurve$ to be subcurve of
$\refcurve^2(\ringcomp)$ between leaving $f_1'$ and entering $f_2'$. Note that $\refcurve$ is a reference curve in $G^\sharp$, as
it travels from the boundary of $f_1$ to the boundary of $f_2$.

In the rest of the proof, we often measure winding numbers of different paths in $G^\sharp$ with respect to either $\refcurve$
or $\refcurve^2(\ringcomp)$; note that, although the latter may not necessarily be a proper
reference curve in $G^\sharp$ (as it may visit $f_1$ and $f_2$ several times), the notion
of winding number is properly defined. However, Lemma \ref{lem:ringhomotopy}
requires us to use a proper reference curve $\refcurve$, for this reason we need to
translate winding numbers between these two curves.

For each $\tau \in I(\ringcomp)$, the path $P_\tau^\ast$ contains a subpath in $G^\sharp$
connecting a vertex on $f_1$ with a vertex on $f_2$ (in one of the directions).
Denote the first such path as $Q_\tau^\ast$. Let $P_\tau^{\ast,1}$ be the subpath of $P_\tau^\ast$
from the start of $P_\tau^\ast$ up to the beginning of $Q_\tau^\ast$ and $P_\tau^{\ast,2}$
be the subpath of $P_\tau^\ast$ from the end of $Q_\tau^\ast$ to the end of $P_\tau^\ast$.

\begin{figure}[h!]
\begin{center}
\includegraphics[width=0.8\textwidth]{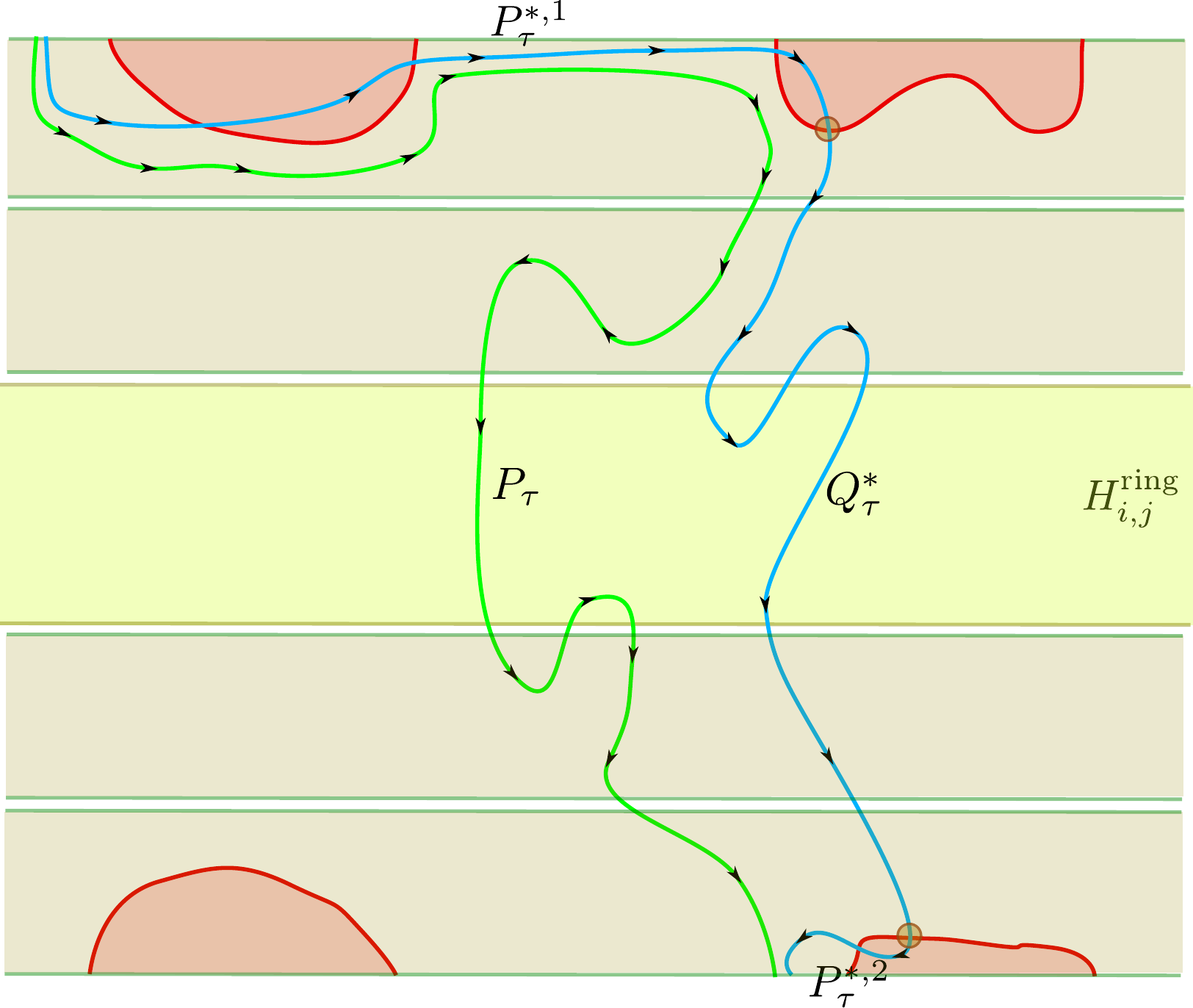}
\caption{Situation in the proof of Lemma~\ref{lem:ring-guessing-equiv2} for one passage index $\tau$. The red parts depict parts of the graph removed while constructing $G^\sharp$, that is, parts separated from the ring component by paths from the minimal solution. The blue path is the path $P^*_\tau$ found by Theorem~\ref{thm:words-to-paths}, while the green path is the passage induced by the minimal solution. The orange circles depict points where $P^*_\tau$ is split into its three parts: $P^{*,1}_\tau$, $Q^{*}_\tau$, and $P^{*,2}_\tau$. The reader may verify that this picture can be realized so that in level $2$, the blue path uses a subset of bundles used by the green path.}
\label{fig:ring-passage2}
\end{center}
\end{figure}

For $\alpha=1,2$, let $\bundleset^\circ_\alpha \subseteq \zbundleset(\ringcomp)$
be the set of bundles consisting of all bundles with at least one endpoint in level-$2$ isolation
component of $\ringcomp$ that lie on the same side of $\ringcomp$ as $f_\alpha$. Moreover, let $\bundleset^\circ = \bundleset^\circ_1 \cup \bundleset^\circ_2$.

For a path $R$ and a set $S \subseteq \bundleset$, by $|\bundleword(R) \cap S|$ we denote the number of appearances of a bundle from $S$
in $\bundleword(R)$. We claim the following.
\begin{claim}\label{clm:tight-budget}
$$|\bundleword(P_\tau^{\ast,1}) \cap \bundleset^\circ| + |\bundleword(P_\tau^{\ast,2}) \cap \bundleset^\circ| \leq 20|\bundleset|^2.$$
\end{claim}
\begin{proof}
Consider two cases.
First, assume that there are two paths $P_{\tau_1}^\ast$ and $P_{\tau_2}^\ast$, $\tau_1,\tau_2 \in I$, that go in
different direction, that is, $P_{\tau_1}^\ast$ starts on $f_1$ and ends on $f_2$
and $P_{\tau_2}^\ast$ starts on $f_2$ and ends on $f_1$. 
Then any path $P_\tau^\ast$ for $\tau \in I(\ringcomp)$ cannot contain two 
arcs of the same bundle, as otherwise its bundle word would contain a spiral
and the existence of both paths $P_{\tau_1}^\ast$ and $P_{\tau_2}^\ast$ would
contradict Lemma \ref{lem:spiral-split}.
Therefore, $P_\tau^\ast$ contains at most $|\bundleset|$ arcs that are bundle arcs
and the claim is proven.

In the other case, all paths $P_\tau^\ast$, $\tau \in \ringcomp$ go in the same
direction; without loss of generality, let us assume that they start in a vertex
on $f_1$ and end in a vertex on $f_2$. Note that, by the choice of $Q_\tau^\ast$,
any path $Q_\tau^\ast$ also starts in a vertex on $f_1$ and ends in a vertex on $f_2$.
Fix $\tau = (i,j) \in I(\ringcomp)$ and let $P_\tau$ be a level-$2$
ring passage of $\ringcomp$
of the path $P_i$ (from the solution $(P_i)_{i=1}^k$)
that corresponds to the bundle word with level-$1$ ring holes $\bwholes_\tau^\circ$.
Let $B \in \bundleset^\circ$ and assume $B$ appears in $\bundleword(P_\tau)$.
Then, as $\bwholes_\tau^\circ$ is a bundle word with level-$1$ ring holes,
$B$ appears in some term $u^r$ of a bundle word decomposition of $\bwholes_\tau^\circ$. 
As $\bwholes_\tau^\circ$ does not contain any level-$0$ bundle, neither does $u$.


Assume for a moment that $r > 10$.
Note that, by the definition of the ring $G^\sharp$, the spiraling ring $A$ associated
with the term $u^r$ in the path $P_i$ is contained in $G^\sharp$.
By Lemma \ref{lem:spiral-split}, each path $Q_{\tau'}^\ast$ for $\tau' \in I(\ringcomp)$ traverses $A$. By Lemma \ref{lem:minsol-spiral},
$\bundleword(Q_\tau^\ast)$ needs to contain $u^{r-10}$ as a subword. Moreover,
as the spiraling rings $A$ are disjoint for different terms $u^r$, the subwords $u^{r-10}$ of $\bundleword(Q_\tau^\ast)$
are pairwise disjoint for different terms $u^r$.

By the properties of the solution
$(P_\tau^\ast)_{\tau \in I(\ringcomp)}$ returned by Theorem \ref{thm:words-to-paths},
the number of appearances of $B$ in the bundle word of the path $P_\tau^\ast$
is not greater than the number of appearances of $B$ in $\bwholes_\tau^\circ$.
There are at most $2|\bundleset|$ terms in a bundle word decomposition of $\bwholes_\tau^\circ$,
and for each term $u^r$, a subword $u^{\min(0,r-10)}$ appears in $\bundleword(Q_\tau^\ast)$,
and these subwords are pairwise disjoint for different terms $u^r$.
We infer that $B$ appears at most $2|\bundleset| \cdot 10$ times in
$\bundleword(P_\tau^{\ast,1})$ and $\bundleword(P_\tau^{\ast,2})$ in total.
This finishes the proof of the claim.
\end{proof}

We also need the following observations.
\begin{claim}\label{cl:Rwinding}
For any $\alpha=1,2$ and for any maximal subpath $R$ of $P_\tau^{\ast,\alpha}$ that does not contain any bundles of $\bundleset^\circ_\alpha$
the winding number of $R$ with respect to $\refcurve^2(\ringcomp)$ equals $-1$, $0$ or $+1$.
\end{claim}
\begin{proof}
Recall that the reference curve $\refcurve^2(\ringcomp)$ has properties promised by Lemma \ref{lem:refcurve-choice}.
If $R$ is contained in the level-$2$ isolation
component of $\ringcomp$ that lies on the same side of $\ringcomp$
as $B_{\tau,\alpha}$, then $R$ does not cross the reference curve and its winding number is $0$.

Otherwise, recall that both endpoints of $Q_\tau^\ast$, and thus
$P_\tau^{\ast,\alpha}$ as well, lie outside $\ringcl^1(\ringcomp)$,
since $\ringcl^1(\ringcomp)$ is contained in $G^\sharp$. Therefore
the endpoints of $R$ lie on the inside or outside face of $\ringcl^1(\ringcomp)$, and, as we exclude only bundles from $\bundleset^\circ_\alpha$
for either $\alpha=1$ or $\alpha=2$, they lie on the same side
of $R$.
Obtain a closed curve $\gamma$ from $R$ by
connecting the endpoints of $R$
using parts of arcs that precede and succeed $R$ on $P_\tau^{\ast,\alpha}$
and the arcs connecting the level-$1$ isolation and level-$2$ isolation
component of $\ringcomp$ on the same side as $B_{\tau,\alpha}$
in the dual of $G$, in such a manner
that $\gamma$ does not separate the sides of $\ringcl(\ringcomp)$.
Note that this implies that the winding number of $\gamma$ is $0$,
whereas, by the properties of $\refcurve(\ringcomp)$,
$\gamma \setminus R$ has winding number $+1$, $0$ or $-1$.
The claim follows.
\end{proof}
\begin{claim}\label{cl:curve-change}
Let $P$ be a path in $G^\sharp$ connecting $f_1$ with $f_2$. Then the winding number of $P$ with respect to curve $\refcurve$
differs by at most $|\bundleset|$ from the winding number of $P$ with respect to $\refcurve^2(\ringcomp)$.
\end{claim}
\begin{proof}
Let $f^1, f^2, \ldots, f^s = f_1'$ be faces of $\ringcl^2(\ringcomp)$ crossed by $\refcurve^2(\ringcomp)$, contained in $f_1$,
in the order of their appearance on $\refcurve^2(\ringcomp)$. For $1 \leq j < s$, let $\refcurve^{R,j}$ be the subcurve of $\refcurve^2(\ringcomp)$ between $f^j$ and $f^{j+1}$.

Let $B$ be a bundle crossed by $\refcurve^2(\ringcomp)$. We claim that at most one curve $\refcurve^{R,j}$ may intersect arcs of $B$.
Recall the construction of $G^\sharp$; let $P$ be one of the maximal subpaths of a solution $(P_i)_{i=1}^k$
that goes though the vertices of $\ringcl^2(\ringcomp)$ and has both endpoints on the same side of $\ringcomp$ as the face $f_1$.
Let $B=(b_1,b_2,\ldots,b_s)$. Due to the subdivision of bundles we performed at the beginning of the proof,
any maximal subpath of $P$ that consists of vertices and edges incident to the union of faces between arcs $b_j$, $b_{j+1}$, $1 \leq j < s$,
in fact consists of a single arc $b_\eta$ for some $1 \leq \eta \leq s$.
However, due to Lemma \ref{lem:spiral-split}, $P$ does not traverse $B$ twice, and arcs $b_j$ for $j \geq \eta$ or $j \leq \eta$ are removed
from $G^\sharp$.
Hence, the number of curves $\refcurve^{R,j}$, $s-1$, is not larger than the number of bundles that lie on the same side of $\ringcomp$ as $f_1$.

For each $1 \leq j < s$, close the curve $\refcurve^{R,j}$ inside $f_1$ to obtain a closed curve. The winding number of $P$ with regards
to the closed curve $\refcurve^{R,j}$ is $0$, $+1$ or $-1$. 
By performing the same analysis on the side of $\ringcomp$ that contains $f_2$, the claim follows.
\end{proof}

We may now conclude with the following statement.
\begin{claim}\label{cl:Qwinding}
The winding number of $Q_\tau^\ast$ with respect to $\refcurve$ differs from $w^\ast_\tau$ by at most $40|\bundleset|^2 + |\bundleset|+2$.
\end{claim}
\begin{proof}
By Claim \ref{cl:Rwinding}, we infer that the winding number of $P_\tau^{\ast,1}$, with respect to $\refcurve^2(\ringcomp)$,
is at most $2|\bundleword(P_\tau^{\ast,1}) \cap \bundleset^\circ_1|+1$,
and a similar claim holds for $P_\tau^{\ast,2}$.
Hence, by Claim \ref{clm:tight-budget}, the winding number of $Q_\tau^\ast$ with respect to $\refcurve^2(\ringcomp)$
does not differ from $w^\ast_\tau$ by more than $40|\bundleset|^2+2$.
By Claim \ref{cl:curve-change}, the winding numbers of $Q_\tau^\ast$ with respect to $\refcurve^2(\ringcomp)$
and $\refcurve$ differ by at most $|\bundleset|$.
The claim follows by pipelining the above three bounds.
\end{proof}

Recall that, by the definition of the graph $G^\sharp$, the intersection of the solution
$(P_i)_{i=1}^k$ with $G^\sharp$ is a set of paths $(P_\tau)_{\tau \in I(\ringcomp)}$;
each such path is a level-$2$ ring passage consistent with $\bwholes^\circ_\tau$,
with the first and last bundle removed.
Moreover, note the following.
\begin{claim}\label{cl:PQorder}
The (circular) orders of the starting and ending vertices of $(P_\tau)_{\tau \in \ringcomp}$ on the faces $f_1$ and $f_2$ of $G_R$ is exactly the same as the order of the starting and ending vertices of the paths $(Q_\tau^\ast)_{\tau \in \ringcomp}$.
\end{claim}
\begin{proof}
Recall that for each ring component $\ringcomp$, $B \in \bundleset$ and $1 \leq \alpha \leq 2$,
the order of the arcs $b_{i,j,\alpha}$ for $(i,j) \in I(\ringcomp)$ and
$b_{i,j,\alpha} \in B$ is equal to $\psi_{B,\alpha}(\ringcomp)$.
By the construction of the zoom auxiliary graph $H(\ringcomp, w(\ringcomp))$,
in $\ringcl^2(\ringcomp)$, the orders of the starting and ending vertices of $(P_\tau)_{\tau \in \ringcomp}$
and the paths $(P_\tau^\ast)_{\tau \in \ringcomp}$ are equal.

Consider now the following graph $G^{\sharp/2}$, constructed similarly as $G^\sharp$, but the removing procedure is performed only on the side of $f_2$.
We start with the subgraph of $G$ induced by the vertices
of the level-$2$ closure of $\ringcomp$. Then we repeatedly take maximal subpaths
of paths $(P_i)_{i=1}^k$ that go through vertices of $\ringcl^2(\ringcomp)$ and, if such a path
starts and ends on the same side of $\ringcl^2(\ringcomp)$ {\em{as the face $f_2$}}, we remove from $G^{\sharp/2}$
all arcs and vertices that lie on the subpath or on 
the different side of the chosen subpath than $\ringcomp$.
Note that $G^{\sharp/2}$ is a supergraph of $G^\sharp$ and a subgraph of $\ringcl^2(\ringcomp)$.
One of its faces is $f_2$, and the other face is one of the faces of $\ringcl^2(\ringcomp)$ that is contained in $f_1$; let us denote it $f_1'$.

Note that
for any $\tau \in I(\ringcomp)$, since $Q_\tau^\ast$ is the {\em{first}} subpath of $P_\tau^\ast$
that connects $f_1$ with $f_2$, $P_\tau^{\ast,1} \cup Q_\tau^\ast$ connects $f_1'$ with $f_2$ inside $G^{\sharp/2}$.
Hence, the order of the starting vertices of the paths $P_\tau^\ast$ on $f_1'$ is equal to the order
of the ending vertices of the paths $Q_\tau^\ast$ on $f_2$. On the other hand, if we look at $G^\sharp$,
the order of the starting vertices of the paths $Q_\tau^\ast$ on $f_1$ is equal to the order
of the ending vertices of $Q_\tau^\ast$ on $f_2$. This concludes the proof of the claim. 
\end{proof}

Claim \ref{cl:PQorder} allows us to
apply Lemma \ref{lem:ringhomotopy} for paths $(P_\tau)_{\tau \in I(\ringcomp)}$
and $(Q_\tau^\ast)_{\tau \in I(\ringcomp)}$ in the rooted ring $G_R$ with reference curve $\refcurve$,
obtaining a sequence of vertex-disjoint paths $(P_\tau')_{\tau \in I(\ringcomp)}$,
such that for each $\tau \in I(\ringcomp)$ the path $P_\tau'$ has the same
starting and ending vertex as $P_\tau$, but the winding numbers of $P_\tau'$
and $Q_\tau^\ast$ with respect to $\refcurve$ differ by at most $6$. Let $x_\tau$ be the winding number
of $P_\tau'$ with respect to $\refcurve^2(\ringcomp)$.
\begin{claim}\label{cl:final-wind-diff}
For each $\tau \in I(\ringcomp)$ we have $|x_\tau - w^\ast_\tau| \leq 40|\bundleset|^2 + 2|\bundleset| + 8$.
\end{claim}
\begin{proof}
By Claim \ref{cl:Qwinding}, the winding number of $Q_\tau^\ast$ with respect to $\refcurve$ and $w^\ast_\tau$,  differ by
at most $40|\bundleset|^2+|\bundleset|+2$.
By Lemma \ref{lem:ringhomotopy}, the winding numbers of $Q_\tau^\ast$ and $P_\tau'$ with respect to $\refcurve$ differ by at most $6$.
By Claim \ref{cl:curve-change}, the winding number of $P_\tau'$ with respect to $\refcurve$ and $x_\tau$ differ by at most $|\bundleset|$.
The claim follows by pipelining the above three bounds.
\end{proof}

Recall that the paths $(P_\tau')_{\tau \in I(\ringcomp)}$
are vertex-disjoint, are contained in $G_R$, and $(P_\tau)_{\tau \in I(\ringcomp)}$ are
the only parts of $(P_i)_{i=1}^k$ in $G_R$.
  Thus, if we conduct the same argument for each ring component $\ringcomp$
  with $I(\ringcomp) \neq \emptyset$ and 
  replace in the solution $(P_i)_{i=1}^k$ each subpath $P_{i,j}$ with $P_{i,j}'$
  for $1 \leq i \leq k$, $1 \leq j \leq h(i)$, we obtain another solution $(P_i')_{i=1}^k$
  to \probshort{} on $G$.
Moreover, as we modified only subpaths $P_{i,j}$ for $1 \leq i \leq k$, $1 \leq j \leq h(i)$,
each path $P_i'$ for $1\leq i \leq k$ is consistent with the
  bundle word with level-$2$ ring holes
$((q_{i,j})_{j=0}^{h(i)}, (x_{i,j})_{j=1}^{h(i)})$, as desired.
This completes the proof of the lemma.
\end{proof}

We can now summarize with the following lemma.
\begin{lemma}\label{lem:holes-guessing}
Let $(G,\decomp,\bundleset)$ be a bundled instance of isolation $(\Lambda, d)$
where $\Lambda \geq 3$, $d \geq \max(2k,f(k,k)+4)$, and $f(k,t)=2^{O(kt)}$ is the bound on the type-$t$ bend promised by Lemma \ref{lem:bendbound}.
Assume we are given a sequence $(\pi_i)_{i=1}^k$ of semi-complete partial bundle words.
Then in $O(2^{O(k^2|\bundleset|^2 \log |\bundleset|)} |G|^{O(1)})$ time
one can compute a family of at most 
$2^{O(k^2|\bundleset|^2 \log |\bundleset|)}$ sequences $(\bwholes_i)_{i=1}^k$
of bundle words with level-$2$ ring holes
such that if there exists a minimal solution $(P_i)_{i=1}^k$ to \probshort{}
on $G$ such that $P_i$ is consistent with $\pi_i$ for each $1 \leq i \leq k$,
then there exists a solution $(P_i')_{i=1}^k$ to \probshort{}
on $G$ and generated sequence $(\bwholes_i)_{i=1}^k$ such that $P_i$
  is consistent with $\bwholes_i$ for each $1 \leq i \leq k$.
\end{lemma}
\begin{proof}
First recall that we branch into at most $3^{4|\bundleset|^2 k^2}$ subcases
guessing the values $\alpha_{i,j}$ and further 
at most
$$\left(\left(\sum_{i=1}^k h(i)\right)!\right)^2 \leq ((4|\bundleset|^2k^2)!)^2 = 2^{O(|\bundleset|^2 k^2 \log |\bundleset|)}$$
subcases when we choose permutations $(\psi_{B,\alpha}(\ringcomp))_{B \in \bundleset, 1 \leq \alpha \leq 2}$ for each $\ringcomp \in \decomp$ with $I(\ringcomp) \neq \emptyset$. 
Finally, once we compute integers $(w_{i,j})_{1 \leq i \leq k, 1 \leq j \leq h(i)}$,
  we may guess the values $(x_{i,j})_{1 \leq i \leq k, 1 \leq j \leq h(i)}$;
by Lemma \ref{lem:ring-guessing-equiv2} for each $1 \leq i \leq k$, $1 \leq j \leq h(i)$
there are $80|\bundleset|^2+4|\bundleset|+17$ possible values for $x_{i,j}$.
As $\sum_{i=1}^k h(i) \leq 4|\bundleset|^2 k^2$, we end up with 
the promised running time and number of subcases.
Correctness follows from Lemmata \ref{lem:ring-guessing-equiv1} and \ref{lem:ring-guessing-equiv2}.
\end{proof}

By pipelining Lemma \ref{lem:word-guessing} with Lemma \ref{lem:holes-guessing}
we finish the proof of Theorem \ref{thm:word-guessing}.
Note that the output bundle words with level-$2$ ring holes does not contain
any bundle of level $0$ due to the application of Lemma \ref{lem:pbw-to-holes}.

\section{Summary: proof of the main result}

We are now ready to summarize
the results of the previous sections by formally proving Theorem \ref{th:main}.

\begin{proof}[Proof of Theorem \ref{th:main}]
Given an instance $G$ of \probshort{} with $k$ terminal pairs,
we first apply the Decomposition Theorem (Theorem \ref{th:finddecomp})
on $G$ with constants $\Lambda = 3$, $d = \max(2k,f(k,k)+4)$
and $r=d(k)+1$, where $f(k,t) = 2^{O(kt)}$ is the bound on the type-$t$ bend
of Lemma \ref{lem:bendbound} and $d(k) = 2^{O(k^2)}$ is the bound on the number
of concentric cycles of Theorem \ref{th:irrelevant}.
If Theorem \ref{th:finddecomp} returns a set of $r$ concentric cycles,
we delete any vertex of the innermost cycle and restart the algorithm.
The correctness follows from Theorem \ref{th:irrelevant},
and the algorithm is restarted at most $|V(G)|$ times.

Otherwise, the algorithm of Theorem \ref{th:finddecomp},
 in time $O(2^{O(\Lambda(d+r)k^2)} |G|^{O(1)}) = O(2^{2^{O(k^2)}} |G|^{O(1)})$
 returns a set of $2^{O(\Lambda(d+r)k^2)} = 2^{2^{O(k^2)}}$
pairs $(G_i, \decomp_i)$; by Theorem \ref{th:finddecomp}
it suffices to check if any graph $G_i$ is a YES-instance to \probshort{}.
Thus, from this point we investigate one graph $G_i$ with decomposition
$\decomp_i$. Note that $|\decomp_i| = O(k^2)$ and the alternation
of $\decomp_i$ is $O(\Lambda(d+r)k^2) = 2^{O(k^2)}$.

We first apply the bundle recognition algorithm of Lemma \ref{lem:bundle-recognition}
to obtain a bundled instance $(G_i,\decomp_i,\bundleset_i)$ with
$|\bundleset_i| = 2^{O(k^2)}$.
Then we apply the algorithm of Theorem \ref{thm:word-guessing};
note that the values of $d$ and $\Lambda$ are large enough to allow this step.
We obtain a family $\mathcal{F}_i$ of sequences $(\bwholes_i)_{i=1}^k$
of bundle words with level-$2$ ring holes with a promise that,
if $G_i$ is a YES-instance to \probshort{}, then there exists a solution
consistent with one of the sequences.
As $|\bundleset_i| = 2^{O(k^2)}$, the size of the family
$\mathcal{F}_i$ is bounded by $2^{2^{O(k^2)}}$ and the running
time of the algorithm of Theorem \ref{thm:word-guessing}
is bounded by $O(2^{2^{O(k^2)}} |G|^{O(1)})$.

Moreover, Theorem \ref{thm:word-guessing} promises us that the bundle
words with level-$2$ ring holes of $\mathcal{F}_i$ do not contain any level-$0$ ring bundles.
Thus, we may apply the algorithm of Theorem \ref{thm:words-to-paths}
to each element of $\mathcal{F}_i$. 
If it outputs a solution for some sequence $(\bwholes_i)_{i=1}^k$,
we know that $G_i$ is a YES-instance to \probshort{}.
Otherwise, by Theorem \ref{thm:words-to-paths}, there does not exist
a solution to \probshort{} on $G_i$ consistent with an element of $\mathcal{F}_i$,
and, by Theorem \ref{thm:word-guessing}, $G_i$ is a NO-instance.
This finishes the proof of Theorem \ref{th:main}.
\end{proof}

\LetLtxMacro{\section}{\oldsection}

\bibliographystyle{abbrv}
\bibliography{dirplanarkpath}

\immediate\closeout\tempfile

\clearpage
\end{document}